\documentclass[11pt]{amsart}
\usepackage{amsfonts,amssymb,mathtools,stmaryrd, graphicx,a4,bbm,mdframed,mathrsfs,hyperref}
\graphicspath{{./figures/}}
\usepackage{float}
\usepackage{tikz, tkz-euclide}
\usepackage{tikz-cd}
\usepackage[title]{appendix}
\usetikzlibrary{calc}
\usetikzlibrary{arrows}
\usetikzlibrary{positioning}
\usetikzlibrary{decorations.markings}
\usepackage[normalem]{ulem}

\usepackage{color}\definecolor{Red}{cmyk}{0,1,1,0}

\newcommand{\spann}{span}

\newcommand*\circled[1]{\tikz[baseline=(char.base)]{
            \node[shape=circle,draw,inner sep=2pt] (char) {#1};}}

\newcommand{\Var}[0]{\text{Var}}
\newcommand\redsout{\bgroup\markoverwith{\textcolor{red}{\rule[0.5ex]{2pt}{0.4pt}}}\ULon}

\newtheorem{theorem}           {Theorem}

\theoremstyle{definition}
\newtheorem{remark}[theorem]{Remark}
\newtheorem*{theorem*}{Theorem}
\newtheorem*{conjecture*}   {Conjecture}
\newtheorem*{corollary*}   {Corollary}
  \newtheorem{lemma}[theorem]              {Lemma}
  \newtheorem*{lemma*}          {Lemma}
    \newtheorem*{claim*}          {Claim}
  \newtheorem{definition}[theorem]         {Definition}
  \newtheorem{corollary}[theorem]          {Corollary}
  \newtheorem{proposition}[theorem]      {Proposition}
  \newtheorem{notation}[theorem]          {Notation}
  
  \theoremstyle{definition}
  \newtheorem{example}[theorem]          {Example}
  
\DeclareMathOperator{\induced}{ind}

\begin{document}


\voffset=-1.5truecm\hsize=16.5truecm    \vsize=24.truecm
\baselineskip=14pt plus0.1pt minus0.1pt \parindent=12pt
\lineskip=4pt\lineskiplimit=0.1pt      \parskip=0.1pt plus1pt

\def\ds{\displaystyle}\def\st{\scriptstyle}\def\sst{\scriptscriptstyle}



\global\newcount\numsec\global\newcount\numfor
\gdef\profonditastruttura{\dp\strutbox}
\def\senondefinito#1{\expandafter\ifx\csname#1\endcsname\relax}
\def\SIA #1,#2,#3 {\senondefinito{#1#2}
\expandafter\xdef\csname #1#2\endcsname{#3} \else
\write16{???? il simbolo #2 e' gia' stato definito !!!!} \fi}
\def\etichetta(#1){(\veroparagrafo.\veraformula)
\SIA e,#1,(\veroparagrafo.\veraformula)
 \global\advance\numfor by 1
 \write16{ EQ \equ(#1) ha simbolo #1 }}
\def\etichettaa(#1){(A\veroparagrafo.\veraformula)
 \SIA e,#1,(A\veroparagrafo.\veraformula)
 \global\advance\numfor by 1\write16{ EQ \equ(#1) ha simbolo #1 }}
\def\BOZZA{\def\alato(##1){
 {\vtop to \profonditastruttura{\baselineskip
 \profonditastruttura\vss
 \rlap{\kern-\hsize\kern-1.2truecm{$\scriptstyle##1$}}}}}}
\def\alato(#1){}
\def\veroparagrafo{\number\numsec}\def\veraformula{\number\numfor}
\def\Eq(#1){\eqno{\etichetta(#1)\alato(#1)}}
\def\eq(#1){\etichetta(#1)\alato(#1)}
\def\Eqa(#1){\eqno{\etichettaa(#1)\alato(#1)}}
\def\eqa(#1){\etichettaa(#1)\alato(#1)}
\def\equ(#1){\senondefinito{e#1}$\clubsuit$#1\else\csname e#1\endcsname\fi}
\let\EQ=\Eq


\def\\{\noindent}
\let\io=\infty

\def\VU{{\mathbb{V}}}
\def\EE{{\mathbb{E}}}
\def\GI{{\mathbb{G}}}
\def\TT{{\mathbb{T}}}
\def\C{\mathbb{C}}
\def\LL{{\cal L}}
\def\RR{{\cal R}}
\def\SS{{\cal S}}
\def\NN{{\cal N}}
\def\HH{{\cal H}}
\def\GG{{\cal G}}
\def\PP{{\cal P}}
\def\AA{{\cal A}}
\def\BB{{\cal B}}
\def\FF{{\cal F}}
\def\vv{\vskip.2cm}
\def\gt{{\tilde\g}}
\def\E{{\mathcal E} }
\def\I{{\rm I}}

\def\cal{\mathcal}

\def\tende#1{\vtop{\ialign{##\crcr\rightarrowfill\crcr
              \noalign{\kern-1pt\nointerlineskip}
              \hskip3.pt${\scriptstyle #1}$\hskip3.pt\crcr}}}
\def\otto{{\kern-1.truept\leftarrow\kern-5.truept\to\kern-1.truept}}
\def\arm{{}}
\font\bigfnt=cmbx10 scaled\magstep1

\newcommand{\card}[1]{\left|#1\right|}
\newcommand{\und}[1]{\underline{#1}}
\newcommand{\Dom}[0]{\text{Dom}}
\def\1{\rlap{\mbox{\small\rm 1}}\kern.15em 1}
\def\ind#1{\1_{\{#1\}}}
\def\bydef{:=}
\def\defby{=:}
\def\buildd#1#2{\mathrel{\mathop{\kern 0pt#1}\limits_{#2}}}
\def\card#1{\left|#1\right|}
\def\proofof#1{\noindent{\bf Proof of #1. }}
\def\trp{\mathbb{T}}
\def\trt{\mathcal{T}}

\def\bfz{\boldsymbol z}
\def\bfa{\boldsymbol a}
\def\bfalpha{\boldsymbol\alpha}
\def\bfmu{\boldsymbol \mu}
\def\bfmust{\bfT^\infty(\bfmu)}
\def\bfmupr{\boldsymbol {\widetilde\mu}}
\def\bfrho{\boldsymbol \rho}
\def\bfrhost{\boldsymbol \rho^*}
\def\bfrhopr{\boldsymbol {\widetilde\rho}}
\def\bfT{{\boldsymbol T}_{\!\!\bfrho}}
\def\bfR{\boldsymbol R}
\def\bfvarphi{\boldsymbol \varphi}
\def\bfvarphist{\boldsymbol \varphi^*}
\def\bfPi{\boldsymbol \Pi}
\def\bfzero{\boldsymbol 0}
\def\bfW{\boldsymbol W}
\def\formal{\stackrel{\rm F}{=}{}}
\def\eee{{\rm e}}
\def\nnn{\mathcal N}
\def\nst{\nnn^*}
\def\Var{\text{Var}}

\thispagestyle{empty}

\begin{center}
{\LARGE Thermodynamic Formalism for Generalized Markov Shifts on Infinitely Many States}
\vskip.5cm
Rodrigo Bissacot$^{1,2}$, Ruy Exel$^{3}$, Rodrigo Frausino$^{1,4}$ and Thiago Raszeja$^{1,5}$
\vskip.3cm
\begin{footnotesize}
$^{1}$Institute of Mathematics and Statistics (IME-USP), University of S\~{a}o Paulo, Brazil\\
$^{2}$Faculty of Mathematics and Computer Science, Nicolaus Copernicus University, Poland\\
$^{3}$ Departamento de Matem\'{a}tica, Universidade Federal de Santa Catarina, Florian\'{o}polis SC, Brazil\\   $^{4}$ School of Mathematics and Applied Statistics, University of Wollongong, Australia

$^{5}$ Faculty of Applied Mathematics, AGH University of Science and Technology, Poland
\end{footnotesize}
\vskip.1cm
\begin{scriptsize}
emails: rodrigo.bissacot@gmail.com; ruyexel@gmail.com; rodrigofrausino6@gmail.com; tcraszeja@gmail.com
\end{scriptsize}

\end{center}

\def\be{\begin{equation}}
\def\ee{\end{equation}}

\vskip1.0cm
\begin{quote}
{\small

\textbf{Abstract.} \begin{footnotesize} 
Given a 0-1 infinite matrix $A$ and its countable Markov shift $\Sigma_A$, one of the authors and M. Laca have introduced a kind of {\it generalized countable Markov shift} $X_A=\Sigma_A \cup Y_A$, where $Y_A$ is a special set of finite admissible words. For some of the most studied countable Markov shifts $\Sigma_A$, $X_A$ is a compactification of $\Sigma_A$, and always it is at least locally compact. We developed the thermodynamic formalism on the space $X_A$, exploring the connections with standard results on $\Sigma_A$. New phenomena appear, such as new conformal measures and a {\it length-type phase transition}: the eigenmeasure lives on $\Sigma_A$ at high temperature and lives on $Y_A$ at low temperature. Using a pressure-point definition proposed by M. Denker and M. Yuri for iterated function systems, we proved that the Gurevich pressure is a natural definition for the pressure function in the generalized setting. For the gauge action, the Gurevich entropy is a critical temperature for the existence of new conformal measures (KMS states) living on $Y_A$. We exhibit examples with infinitely (even uncountable) many new extremal conformal measures, undetectable in the usual formalism. We prove that conformal measures always exist at low temperatures when the potential is coercive enough. We characterized a basis of the topology of $X_A$ to study the weak$^*$ convergence of measures on $X_A$, and we show some cases where the conformal measure living on $Y_A$ converges to a conformal one living on $\Sigma_A$. We prove the equivalence among several notions of conformality for locally compact Hausdorff second countable spaces, including quasi-invariant measures for generalized Renault-Deaconu groupoids.

\end{footnotesize}

}
\end{quote}
\tableofcontents
\numsec=2\numfor=1
\section*{Introduction}
\vv
\noindent

Markov shifts are important objects in dynamical systems, ergodic theory, and operator algebras. Their dynamical and ergodic importance lies, among other topics, on thermodynamic formalism, a rigorous and abstract from the point of view of mathematical physics into the foundations of equilibrium statistical mechanics, introduced by R. Bowen \cite{Bowen1975}, D. Ruelle \cite{Ruelle1967, Ruelle1978} and Y. Sinai \cite{Sinai1972}.  One of the main topics of the area is the study of phase transition phenomena \cite{BelBisEndo2020, Dobrushin1965, Dobrushin1968, EnterFernSokal1993, FriedliVelenik2018,Georgii2011, Iommi2007,KuchQuasWolf2020, Peierls1936,Pe, PirogovSinai1974, Sarig2000, Sarig2001}. Roughly speaking, a phase transition consists of an abrupt change in the system, and it can have many different meanings; some of them are equivalent, and some are not; a good reference for a discussion about phase transitions is \cite{EnterFernSokal1993}. Among the different notions of phase transitions, we have:
\begin{itemize}
    \item the change from uniqueness to non-uniqueness of the probability measures of interest in statistical mechanics, such as the Dobrushin-Lanford-Ruelle measures \cite{BisExelFrauRas2018, Dobrushin1968, Dobrushin1968b, Dobrushin1968c, LanRu1969}, also called DLR measures \cite{Dobrushin1968b, Ellis2006, Georgii2011, LanRu1969}. When the potential is regular enough and depending on the structure of the shift space, the invariant DLR measures coincide with the equilibrium states \cite{Keller1998, Kimura2015, Muir2011, Muir2011b, Ruelle1967};
    \item the existence of a point of non-analyticity for the pressure \cite{FriedliVelenik2018, KuchQuasWolf2020, Ott2019,Ott2020, Sarig2000, Sarig2001};
    \item recently, it was discovered a new type of phase transition on countable Markov shifts, when a set of DLR measures changes from finite DLR measures to the infinite ones \cite{BelBisEndo2020}, a {\it volume-type phase transition}. This phenomenon comes from the fact that, in general, in unidimensional systems, the set of DLR measures contains (sometimes coincides with) the set of eigenmeasures associated with the Ruelle operator \cite{BelBisEndo2020, CioLopesStad2020, Sarig2009} and, in countable Markov shifts, the eigenmeasures can be infinite \cite{Sarig2000, Sarig2001, Sarig2009}. 
\end{itemize}

From the point of view of operator algebras, Markov shifts are related to the Cuntz-Krieger algebras and their generalizations. The first connection between C$^*$-algebras and symbolic spaces was made in 1977 for the full shift case by J. Cuntz \cite{Cuntz1977}. It was generalized for Markov shifts in the case of a finite number of symbols three years later, in the seminal paper by J. Cuntz and W. Krieger \cite{CK1980}, where it was introduced the so-called Cuntz-Krieger algebras $\mathcal{O}_A$, with $A$ being the transition matrix. These algebras are universal C$^*$-algebras generated by a family of partial isometries indexed by the alphabet under some relations. They codify the Markov shift spaces in the sense that a product of these generators is non-zero if, and only if, the correspondent word is admissible. Moreover, $C(\Sigma_A)$ is a commutative C$^*$-subalgebra of $\mathcal{O}_A$. When the alphabet is infinite, results for the case when the matrix $A$ is row-finite were achieved by A. Kumjian, D. Pask, I. Raeburn and J. Renault \cite{KumjPaskRaeb1998, KumjPaskRaebRen1997} by using the groupoid C$^*$-algebra theory, see also \cite{Renault1980}. Later, in 1999, one of the authors and M. Laca \cite{EL1999} generalized the Cuntz-Krieger theory for the case in which the matrix $A$ is not row-finite. For this general case, $\Sigma_A$ may not be locally compact. The generalized Markov shift $X_A$ arises as a natural object to replace $\Sigma_A$, $X_A$ contains $\Sigma_A$, and the dynamics (the shift) can be extended almost to the whole space, except by what we call empty-stem configurations.

In parallel, started by the Russian school \cite{Gu1, Gu2, Gu3, Gu4}, we had an intense development of the thermodynamic formalism on countable Markov shifts (CMS) $\Sigma_A$ in the last 20 years \cite{MaulUr2001, MaulUr2003, Sarig1999, Sarig2001_Null, Sarig2015, Sarig2009}. Y. Pesin wrote an excellent survey about O. Sarig's contributions in \cite{Pe}. However, no study on thermodynamic formalism was done for $X_A$. This paper aims to start the study of the thermodynamics on $X_A$. It holds that $\Sigma_A$ is dense in $X_A$, and these sets coincide for standard locally compact countable Markov shifts. We point out that it is a very natural strategy that many authors already used to consider compactifications of $\Sigma_A$ and then try to obtain results and consequences for the thermodynamic formalism on $\Sigma_A$, and we have available a good amount of results on the literature \cite{Fiebig2013, Fiebig_Fiebig1995, Gu4, Iommi2020, Shwartz2019, Zargaryan1986}. However, the compactification constructed in \cite{EL1999} looks pretty natural since its construction depends only on the transition matrix $A$. As we will see, many of the important definitions and concepts are extended naturally to this generalized setting. We can recover the classical objects by restrictions of the potential and dynamics.

The compatibility of the Borel $\sigma$-algebras between both settings in the sense that Borel sets of $\Sigma_A$ are included in the Borel $\sigma$-algebra on $X_A$ plays a crucial role in guaranteeing that conformal measures from the standard theory can be seen as conformal ones in the generalized setting. About the topology, unlike in the $\Sigma_A$ case, the complement of cylinders is not simply a union of cylinders, but they have an extra part that lives in $Y_A = X_A\setminus \Sigma_A$. Also, $X_A$ is always locally compact and, in many cases, compact. On the other hand, the usual definition of Gurevich pressure, whose exponential is the Ruelle's operator eigenvalue for suitable potentials (recurrent), only considers the periodic points, which cannot live in $Y_A$. However, by M. Denker and M. Yuri's \cite{DenYu2015} results for Iterated Function Systems (IFS), it is possible to define the pressure at a point $P(\beta F,x)$  for a wide class of matrices and potentials that includes the main examples of the theory as the full shift, renewal shift and others. This pressure coincides with the Gurevich pressure on every point $x \in X_A$. We study the phase transition phenomena for conformal measures and eigenmeasures. We highlight a length-type phase transition phenomenon: we present an example for which we have the existence of a critical value $\beta_c$ s.t. in the region $\beta \leq \beta_c$ we have the existence of a probability eigenmeasure that gives mass zero for $Y_A$, and for $\beta > \beta_c$, the eigenmeasure still does exist, but it is concentrated on $Y_A$. The term \textit{length-type} refers to that we can see the elements of $Y_A$ as finite words; sometimes, they have some multiplicity. For the case of conformal measures, for regular enough potentials, and $X_A$ satisfying suitable notions of compactness, we have results for the existence and absence of extremal conformal measures, where the Gurevich entropy $h_{G}$ plays a fundamental role in determining the regions for $\beta$ where the new measures do exist. Also, each of these new extremal measures lives on different subsets of $Y_A$, called $Y_A$-families, and these families form a partition for $Y_A$. Besides, for infinite entropy cases, we can study a large class of coercive enough potentials by an algebraic condition that encodes them as potentials on the renewal shift.

For continuous potentials, the conformal measures are also connected with the thermodynamic formalism for the Cuntz-Krieger algebras and their generalizations since they correspond to the quasi-invariant measures on these algebras via the Renault-Deaconu groupoid approach \cite{Renault2009, Thomsen2017, Tho3}. Therefore, these measures are KMS states for the C$^*$-dynamics generated by a 1-cocycle. In particular, for potentials that are always positive (negative), every KMS state associated with the 1-cocycle dynamics comes from a conformal measure \cite{Nesh2013, Tho3}. We present a couple of examples where we have finite, countably infinite, and uncountably many extremal conformal measures (KMS states) for low enough temperatures. 

Our paper unifies what has been done in thermodynamic formalism on countable Markov shifts with the results for Cuntz-Krieger C$^*$-algebras with infinitely many states. Since $\Sigma_A \subset X_A$, we recover the standard objects by taking the restriction of the objects as potentials, measures, and shift dynamics. We prove that the Gurevich pressure is suitable for this new setting and our results show the existence of more conformal measures at low temperatures, a more natural feature compared with what happens in physical models in statistical mechanics.

The paper is organized as follows. 

In section \ref{sec:Preliminaries}, we introduce the standard theory of Markov shifts, and we state the results that we will use in the generalized setting, such as the RPF theorem and the discriminant theorem. We present the Example \ref{exa:renewal_potential_log}, which is the model used to describe the length-type phase transition for eigenmeasures of the Ruelle's operator. Also, in this section, we present the generalized Renault-Deaconu groupoid \cite{Renault2} and the 1-cocycle associated with a continuous potential. We remember the Cuntz-Krieger algebras $\mathcal{O}_A$ for infinite matrices \cite{EL1999} in a friendly fashion for the dynamical system community, and we avoid using crossed products. We present a faithful representation of these algebras in $\mathfrak{B}(\ell^2(\Sigma_A))$, which is helpful to our purposes.

Section \ref{sec:results_O_A} is focused on showing some properties of $\mathcal{O}_A$ and constructing the commutative C$^*$-subalgebra $\mathcal{D}_A$ of $\mathcal{O}_A$, corresponding to a C$^*$-algebra of diagonal operators when viewed as a subalgebra of $\mathfrak{B}(\ell^2(\Sigma_A))$. The spectrum of $\mathcal{D}_A$ is the generalized countable Markov shift space $X_A$, which we call GCMS. 

In section \ref{section:X_A}, we define $X_A$ and discuss some of its properties, such as the inclusion of $\Sigma_A$ into $X_A$, the surjectivity of such inclusion when $\Sigma_A$ is locally compact, and how the elements of $X_A$ can be seen as configurations on the Cayley tree generated by the alphabet. Moreover, we identify the set of finite words in $X_A$, denoted by $Y_A$, by the stem-root identification.  We also present an example of non-locally compact $\Sigma_A$ such that $X_A$ is not compact and where we have uncountably many elements in $Y_A$. Furthermore, we prove the compatibility between the Borel $\sigma$-algebras of $\Sigma_A$ and $X_A$. We introduce the notion of the $Y_A$-family. These objects characterize the extremal conformal measures living on $Y_A$. The results show that, differently from graph algebras \cite{Thomsen2017}, the set of infinite emitters is not in bijection with the extremal KMS states in our setting.

In section \ref{section_cylinder}, we present the topological description of $X_A$, where we define the generalized cylinder sets, and we identify the typical elements of a basis of the topology in terms of these generalized cylinders. This basis is described in total generality for irreducible matrices, and it is closed under intersections, which allows the study of weak$^*$ convergence of probability measures on $X_A$. Details of the basis description can be found in the appendix, in section \ref{ape:proof_giant_theorem_generalized_cylinders}.

Sections \ref{sec:TF_Generalized}, \ref{new_measures} and \ref{sec:eigenmeasures} comprehend the thermodynamic formalism on $X_A$ as follows.

In section \ref{sec:TF_Generalized} we present different notions of conformal measures: in both senses of Denker-Urba\'nski and Sarig, quasi-invariant measures for the Renault-Deaconu groupoids, and eigenmeasures of the Ruelle's transformation. We prove the equivalence among all these notions beyond the GCMS setting. After this, we present a natural definition for the pressure at a point on $X_A$. The definition was proposed by M. Denker and M. Yuri \cite{DenYu2015} for Iterated Function Systems (IFS) and makes sense for any point in $X_A$. For potentials with uniform bounded distortion and for a wide class of GCMS', we prove that the notions of pressure at a point and Gurevich pressure coincide.

In section \ref{new_measures} we focus on the new conformal measures that are not detected in the standard formalism by identifying a conformality condition on $Y_A$. Then, we show that the extremal conformal measures that live on $Y_A$ are precisely those that live on a unique $Y_A$-family. Further, we prove that the sign of the pressure determines the existence and the absence of these new measures. We have conformal measures living on $Y_A$ for low enough temperatures, and for sufficiently high temperatures, we have the absence of these measures. In particular, for the gauge potential, there exists a critical value of the inverse of temperature $\beta_c$ separating these regions; the value of the $\beta_c$ coincides with the Gurevich entropy $h_G$, a natural result after the theorem of K. Thomsen \cite{Thomsen2017} for graph C$^*$-algebras, where he proved the existence of KMS weights for $\beta \geq h_G$.  Also, for the renewal and pair renewal GCMS', we prove that for $\beta$ decreasing to the critical value, the set of conformal measures living on $Y_A$ collapses to the unique conformal measure living on the standard CMS at $\beta_c=h_G$. Besides, we study infinite entropy shifts encoding potentials on the full shift using potentials on the renewal shift. The full shift has the largest $Y_A$-family compared to other shift spaces. Then, we find suitable conditions for the potential such that, for a large class of GCMS and any $Y_A$-family, we prove the existence of an extremal conformal measure living on $Y_A$ for low enough temperatures. In particular, the Example \ref{exa:uncountable_Y_A} presents a GCMS that admits uncountably many extremal conformal measures for $\beta$ large enough when the potential is coercive enough, see section \ref{uncountable_extremal}.  The last example we study in this section is a GCMS whose transition matrix was defined by I. Raeburn and W. Szyma\'nski in \cite{RaeSzy2004}, where the corresponding C$^*$-algebra $\mathcal{O}_A$ is not isomorphic to any graph C$^*$-algebra. In addition, this example points out a difference between KMS theory for $\mathcal{O}_A$ and for graph algebras: K. Thomsen proved in \cite{Thomsen2017} that there exists a bijection between boundary KMS weights and the infinite emitters. For Raeburn and Szyma\'nski's example, we have two infinite emitters, but only one $Y_A$-family exists, and then we have a unique KMS state living on $Y_A$ for $\beta>h_G$.

Section \ref{sec:eigenmeasures} is dedicated to the study of eigenmeasures. There we present a general theorem of the existence of eigenmeasures, a result from M. Denker and M. Yuri's paper \cite{DenYu2015}, adapted to the GCMS case, where the eigenvalue is the exponential of the pressure as in the standard theory of CMS. Later, we present an example of a potential on the renewal shift, which presents a new type of phase transition, we called it \emph{length-type phase transition}: for every $\beta > 0$ we have a unique eigenmeasure living on $X_A$. Still, there exists a critical value $\beta_c$ such that the measure lives on $\Sigma_A$ for $\beta \leq \beta_c$, and it lives on $Y_A$ for $\beta > \beta_c$. In other words, the phase transition is the change of the space where the measure lives, from $\Sigma_A$ to $Y_A$, by increasing $\beta$. We stress that this example, at the same time, deals with both classical and quantum sides of the theory: the pressure is non-zero for $\beta < \beta_c$, and although we have a unique eigenmeasure for each $\beta$, these are not KMS states. However, for $\beta \geq \beta_c$, the pressure is zero, and the unique eigenmeasure for these values of $\beta$ is a KMS state.

This paper is part of a project started during the visit of R. Exel at IME-USP in March 2018; the visit was supported by FAPESP Grant 17/26645-9. A second paper was written \cite{BisExelFrauRas2018}, where a systematic study of DLR measures in the generalized setting is done. A substantial amount of results presented in this paper are contained in the Ph.D. thesis of T. Raszeja \cite{Raszeja2020}.
\newpage
\numsec=2\numfor=1
\section{Preliminaries} \label{sec:Preliminaries}
\subsection{Thermodynamic formalism on CMS}\\

In this subsection, we introduce the basic definitions and results about countable Markov shifts, which we will call CMS. Standard references for the topic are \cite{MaulUr2001, MaulUr2003, Sarig2009, Sarig2015}. 

Consider an infinite $\{0,1\}$-matrix with no zero rows and neither zero columns. Taking as alphabet the set of natural numbers $\mathbb{N}$, the \textit{countable Markov shift} $\Sigma_A(\mathbb{N}) \equiv \Sigma_A$ associated to the matrix $A$ is given by $\Sigma_A := \left\{x \in \mathbb{N}^{\mathbb{N}_0}: A(x_j,x_{j+1}) = 1, \text{ for all } j \in \mathbb{N}_0 := \mathbb{N}\cup \{0\} \right\}.$

We may write an element of $\Sigma_A$ as $x=x_0x_1\ldots$ rather than $x=(x_0,x_1,\ldots)$ for convenience. The dynamics is given by the \textit{shift map} $\sigma: \Sigma_A \to \Sigma_A$; $x = x_0 x_1 x_2 \ldots \mapsto \sigma(x) = x_1 x_2 x_3 \ldots$.

A \textit{word} is an element $x_0\dots x_{n-1}\in \mathbb{N}^n$
$(n \in \mathbb{N})$, its length is $n$, and it is called \textit{admissible} when $[x_0\dots x_{n-1}] \neq \emptyset$. We say that the Markov shift $\Sigma_{A}$ (or the matrix $A$) is {\it{transitive}} (or {\it irreducible})  when, for every $a,b\in \mathbb{N}$, there exists an admissible word $w$, such that $aw b$ is admissible. On the other hand, $\Sigma_{A}$ is said to be {\it{topologically mixing}} when, for every $a,b\in \mathbb{N}$ there exists $N_{ab}\in\mathbb{N}$ such that, for every $n\geq N_{ab}$, there exists an admissible word $w_n$ of length $n$ such that $aw_n b$ is admissible. The topology of $\Sigma_A$ is the one generated by the \textit{cylinder sets}, defined as
$$[a_0a_1\ldots a_{n-1}]:=\{x\in \Sigma_{A}:x_i=a_i,0\leq i\leq n-1\},$$
where $a_i$ and $n$ are natural numbers. Or, equivalently, generated by the metric $d(x,y) = 2^{-\inf\{p:x_p \neq y_p\}}$.

We recall that $\Sigma_A$ is a Polish space, and the cylinder sets are clopen. Moreover, when the matrix $A$ is irreducible, then the shift is locally compact if and only if the cylinder sets are compact, which is true if and only if $A$ is row-finite, i.e., for every row in the matrix $A$, we have a finite number of 1's.

\begin{example}[Renewal shift]
An important example of topologically mixing CMS it is the so called \textit{renewal shift}, see \cite{Iommi2007,Sarig2001}. For the renewal shift, the matrix $A$ is given by $A(1,n) = A(n+1,n) = 1$ for every $n \in \mathbb{N}$ and zero for other entries of $A$. Its symbolic graph is presented in figure \ref{fig:renewal_shift_symbolic_graph}.
\end{example}
A function $F:\Sigma_{A}\to\mathbb{R}$ is called a {\it {potential}}. For every $n\geq 1$, the {\it $n$-variation of $F$} is defined by
\begin{equation*}
\Var_nF:=\sup\left\{\lvert F(x)-F(y)\rvert :~x,y\in\Sigma_{A}, x_i=y_i,0\leq i\leq n-1\right\}.
\end{equation*}
A potential $F$ is called {\it{locally H\"older continuous}} if there exists $B>0$ and $\theta\in(0,1)$ such that, for every $n\geq 2$, we have $\Var_n F\leq B\theta^n$. We say that $F$ has \textit{summable variations} if $\sum_{n\geq2}\Var_n F<\infty$, and we say that $F$ \textit{satisfies the Walters' condition} when $\sup_{n \in \mathbb{N}} \Var_{n+k}F_n < \infty \text{ for each }k \in \mathbb{N}$ and $\lim_k \sup_{n \in \mathbb{N}} \Var_{n+k}F_n = 0$. Observe that we may have $\Var_1 F =\infty$. Moreover, we say that $F$ has \emph{uniform bounded distortion} when $\sup_{n \in \mathbb{N}} \Var_{n}F_n < \infty$. We say that a potential is {\it coercive} when $\lim_{x_0 \to \infty} F(x) = + \infty$, observe that in the language of the standard literature of countable Markov shifts,
this property corresponds to say that $-F$ is coercive.

\begin{definition} Given a potential $F:\Sigma_A \to \mathbb{R}$ and a symbol $a \in \mathbb{N}$, the \emph{n-th partition function} is defined by
\begin{equation*}
    Z_n(F,[a]) = \sum_{x\in \Sigma_A:\sigma^n(x)=x}e^{ F_n(x)}\mathbbm{1}_{[a]}(x).
\end{equation*}
where $F_n(x)=\sum_{i=0}^{n-1}F(\sigma^i(x))$ is the \textit{Birkhoff (ergodic) sum} of $F$. Another important quantity is the following:
\begin{equation*}
    Z_n^*(F,[a]) := \sum_{\sigma^n x = x} e^{ F_n(x)}\mathbbm{1}_{[\varphi_a = n]}(x),
\end{equation*}
where $\varphi_a(x):= \mathbbm{1}_{[a]}(x) \inf\{n\geq 1:\sigma^nx \in [a]\}$ ($\inf \emptyset = \infty$ and $0 \cdot \infty = 0$).
The \textit{Gurevich pressure} at the symbol $a$ is defined as
\begin{equation*}
    P_G(F,[a]) :=\limsup_{n\to \infty} \frac{1}{n} \log Z_n( F,[a]).
\end{equation*}
\end{definition}
When $\Sigma_A$ is transitive and $F$ satisfies the Walters' condition, then $P_G( F,[a])$ does not depend on $a$. Since these hypotheses are the minimal assumed over the CMS and the potential, we will write $P_G(F)$. A particular but important case is to consider the constant potential $F\equiv0$, then the pressure coincides with the \emph{Gurevich entropy}, see \cite{Gu1, Gu2}:
\begin{equation*}
    h_{G}:= P_G(0) :=\limsup_{n\to \infty} \frac{1}{n} \log \sum_{x\in \Sigma_A:\sigma^n(x)=x}\mathbbm{1}_{[a]}(x).
\end{equation*}
Moreover, if $\Sigma_A$ is topologically mixing, then the pressure is a limit. In this case, the Gurevich pressure and the Gurevich entropy are given by:
\begin{equation*}
    P_G(F) =\lim_{n\to \infty} \frac{1}{n} \log Z_n( F,[a])  \quad \quad \text{and} \quad \quad h_{G} = \lim_{n\to \infty} \frac{1}{n} \log\sum_{x\in \Sigma_A:\sigma^n(x)=x}\mathbbm{1}_{[a]}(x).
\end{equation*}

\begin{definition}\label{modes-recu} Let ${\Sigma_A}$ be a topologically mixing CMS, $F:{\Sigma_A}\to\mathbb{R}$ a potential satisfying the Walters' condition and $P_G(F)< \infty$. Fix some $a\in \mathbb{N}$. We say that the potential is:
	\begin{itemize}
		\item[i)] \emph{Recurrent} if $\sum_{n\geq1}e^{-nP_G(F)}Z_n(F,a)=\infty$.
		\item[ii)] \emph{Positive recurrent} if $F$ is recurrent and $\sum_{n\geq1}ne^{-nP_G(F)}Z_n^{*}(F,a)<\infty$.
		\item[iii)] \emph{Null recurrent} if $F$ is recurrent and $\sum_{n\geq1}ne^{-nP_G(F)}Z_n^{*}(F,a)=\infty$.
		\item[iv)] \emph{Transient} if $\sum_{n\geq1}e^{-nP_G(F)}Z_n(F,a)<\infty$.
	\end{itemize}
\end{definition}

Since the ${\Sigma_A}$ is topologically mixing the modes of recurrence do not depend on the symbol $a$.

The modes of recurrence are related to the existence of eigenmeasures as follows.

\begin{definition}
Let $\Sigma_{A}$ be a CMS, $\nu$ a Borel $\sigma$-finite measure, $F:{\Sigma_A}\to\mathbb{R}$ a measurable potential and $\lambda>0$. Then, we call $\nu$ a \textit{eigenmeasure with eigenvalue $\lambda$} if
\begin{equation}\label{eq:eigenmeasure_standard}
    \int L_{F}f(x)d\nu(x)=\lambda\int f(x)d\nu(x),\quad \mbox{ for each }f\in L^1(\nu),
\end{equation}
where $L_F$ is the \textit{Ruelle's operator} defined as $L_{F}f(x):=\sum_{\sigma(y)=x}e^{F(y)}f(y)$.
\end{definition}

In definition above, we denote that a measure $\nu$ is a eigenmeasure by writing $L_{F}^{*}\nu=\lambda\nu$. We warn the reader that, although this notation could suggest a duality in the sense of the Riesz-Markov-Kakutani representation theorem, this does not hold for non-locally compact standard Markov shift spaces. See \cite{Sarig2009} regarding when the Ruelle's operator is well-defined for countable Markov shifts. Later, we will see that one of the advantages of the generalized setting is precisely the fact that the generalized Markov shifts $X_A$ are always locally compact spaces. In this case, the duality for the Ruelle's operator is used in the sense of the representation theorem, see definition \ref{def:Ruelle_transformation}.

The next theorem connects the existence of the eigenmeasures and the recurrence modes, and it is called {\it Generalized Ruelle-Perron-Frobenius Theorem}. A Borel measure  $\nu$ defined on $\Sigma_A$ is a \textit{non-singular measure} when $\nu\circ \sigma^{-1}\sim \nu$, i.e., $\nu\circ \sigma^{-1}(E)=0$ iff $\nu(E)=0$.  A $\sigma$-finite non-singular measure $\nu$ is said to be {\it conservative} if for every Borel set $W$ such that $\{\sigma^{-n}W\}_{n\geq 0}$ is pairwise disjoint we have $\nu(W)=0$. 
\begin{theorem}[Generalized Ruelle-Perron-Frobenius Theorem, \cite{Sarig1999, Sarig2009}]\label{GRPF}\\

Let ${\Sigma_A}$ be a topologically mixing CMS, $F:{\Sigma_A}\to\mathbb{R}$ a potential that satisfies the Walters' condition and $P_G(F)<\infty$. Then:
	\begin{itemize}
		\item[i)] $F$ is positive recurrent if, and only if, there exist $\lambda>0$, a positive continuous function $h$, and a conservative measure $\nu$ which is finite on cylinders, such that $L_{F}h=\lambda h$, $L_{F}^{*}\nu=\lambda\nu$, and $\int h d\nu=1$. In this case $\lambda=e^{P_G(F)}$.\vv
		\item[ii)] $F$ is null recurrent if, and only if, there exist $\lambda>0$, a positive continuous function $h$, and a conservative measure $\nu$ which is finite on cylinders, such that $L_{F}h=\lambda h$, $L_{F}^{*}\nu=\lambda\nu$, and $\int h d\nu=\infty$. In this case $\lambda=e^{P_G(F)}$.\vv
		\item[iii)] $F$ is transient if, and only if, there is no conservative measure $\nu$ which is finite on cylinders such that $L_{F}^{*}\nu=\lambda\nu$ for some $\lambda>0$.
	\end{itemize}
\end{theorem}
\begin{definition}[Induced Markov shift] Given a countable Markov shift $\Sigma_A$. Fixed $a \in \mathbb{N}$, set $S_{ind} := \{[w]: w \in \mathfrak{W}, w_i = a \iff i = 0, [wa] \neq \emptyset \}$, where $\mathfrak{W}$ is the set of admissible non-empty words. The \textit{induced Markov shift space on $a$} is the set $\Sigma_A^{ind}(a) := S_{ind}^{\mathbb{N}_0}$, where the \textit{induced shift map} $\sigma_{ind}: \Sigma_A^{ind}(a) \to \Sigma_A^{ind}(a)$ is given by
\begin{equation*}
    \sigma_{ind}(([w^0],[w^1],[w^2],\dots)) := ([w^1],[w^2],[w^3],\dots),
\end{equation*}
for every $([w^0],[w^1],[w^2],\dots) \in \Sigma_A^{ind}(a)$. Let $\pi_{ind} :\Sigma_A^{ind}(a) \to [a]$ the map given by
\begin{equation*}
    \pi_{ind}(([w^0],[w^1],[w^2],\dots)) := (w^0w^1w^2\dots).
\end{equation*}
Given a potential $F:\Sigma_A \to \mathbb{R}$, we define the \textit{induced potential on $[a]$}, $F^{ind}: \Sigma_A^{ind}(a) \to \mathbb{R}$, by
\begin{equation*}
    F^{ind} := \left(\sum_{k=0}^{\varphi_a - 1} F\circ \sigma^k \right)\circ \pi_{ind}.
\end{equation*}
The pair $(\Sigma_A^{ind}(a),F^{ind})$ is called the {\it induced system on} $[a]$.
\end{definition}

\begin{definition} Let $\Sigma_A$ be topologically mixing and let $F:\Sigma_A \to \mathbb{R}$ with summable variations and finite Gurevich pressure. Fix $a \in \mathbb{N}$ and let $(\Sigma_A^{ind}(a), F^{ind})$ be the induced system. Set $$p^*_a[F] := \sup\{p : P_G((F + p)^{ind}) < \infty\}.$$
\noindent
The \textit{$a$-discriminant} of $F$ is defined by $\Delta_a[F] := \sup\{P_G((F + p)^{ind}) : p<p^*_a[F]\} \leq \infty.$
\end{definition}
Now we are ready to state the discriminant theorem, which will be useful in our further examples. One can find the proof in \cite{Sarig2001}.
\begin{theorem}[Discriminant Theorem]
Let $\Sigma_A$ be a topologically mixing CMS,
and let $F : \Sigma_A \to \mathbb{R}$ a potential with summable variations such that $P_G(F) < \infty$. For a fixed state $a\in \mathbb{N}$,
\begin{enumerate}
   \item the equation $P_G((F + p)^{ind}) = 0$ has a unique solution $p(F)$ if $\Delta_a[F] \geq 0$, and no solution if $\Delta_a[F] < 0$. The Gurevich pressure of $F$ is given by 
   $$P_G(F) =\begin{cases} -p(F)\quad \text{if}\; \Delta_a[F]\geq 0,\\
   -p_a^*(F)\quad \text{if}\; \Delta_a[F]<0;
   \end{cases}
$$

\item  $F$ is positive recurrent if $\Delta_a[F] > 0$, and transient if $\Delta_a[F] < 0$. In the case
$\Delta_a[F] = 0$, $F$ is either positive recurrent or null recurrent.
\end{enumerate}
\end{theorem}

We also need other results related to this theorem. For instance, Proposition 3 on section 6 of \cite{Sarig2001} (see page 559), it is proved, among other results, the following:
\begin{align}
    \left|\Delta_a[F] - \log\left(\sum_{k=1}^\infty R^k Z_k^*(F,[a])\right) \right| \leq \sum_{k \geq 2} \Var_k(F), \label{eq:absol_discriminant_power_series}\\
    p_a^*(F)=-\limsup_{n\to\infty}\frac{1}{n}\log Z_k^*(F,[a])\label{eq:p_a_star_is_limsup_partition_Z_star},
\end{align}
where $R$ is the convergence radius of the power series $\sum_{k=1}^\infty y^k Z_k^*(F,[a])$.

For the case of the renewal shift, we have uniqueness (when it exists) of the critical point $\beta_c$ where the pressure $P_G(\beta F)$ is not analytic  \cite{Sarig2001}, page 561, Theorem 5. This $\beta_c$ separates two regions of $\beta$ where we have the existence and absence of the eigenmeasure. In the generalized setting, we will present an example where $\beta_c$ indicates the change of the region where is concentrated the mass of the eigenmeasure.

\begin{theorem}\label{thm:phase_transition_renewal_shift_Sarig_Gurevich_pressure}
Let $\Sigma_A$ be the renewal shift and $F:\Sigma_A\to \mathbb{R}$ a locally H\"older continuous function such that $\sup F<\infty$. Then, we have that $F^{ind}$ is locally H\"older continuous and, there exists ${0<\beta_c\leq \infty}$ such that:
\begin{itemize}
    \item[1.]$\beta F$ is positive recurrent for $0<\beta<\beta_c$ and transient for $\beta>\beta_c$.
    \item[2.]$P_G(\beta F)$ is real analytic in $(0, \beta_c)$ and linear in $(\beta_c,\infty)$. It is continuous but not
analytic at $\beta_c$ (in case $\beta_c < \infty$).
\end{itemize}
\end{theorem}

The next example will be analyzed again in the generalized setting. For now, we will apply the previous theorems to study the thermodynamic formalism when we consider the standard space $\Sigma_A$.  

\begin{example}\label{exa:renewal_potential_log}

Let $\Sigma_A$ be the renewal shift and consider the potential
\begin{equation}\label{eq:conservativity_phase_transition_potential}
    F(x) = \log (x_0) -\log(x_0+1).
\end{equation} 
Since the potential depends only on the first coordinate, we have $\sum_{k\geq 2}\Var_k(\beta F) = 0$ for every  $\beta >0$. By the inequality \eqref{eq:absol_discriminant_power_series}, and taking $a = 1$, we get
\begin{equation}\label{eq:discriminant_1_equals_series}
    \Delta_1[\beta F] = \log\sum_{k=1}^\infty R^k Z_k^*(\beta F,[1]).
\end{equation}
 Remember that $R$ is the convergence radius of the power series $\sum_{k=1}^\infty y^k Z_k^*(F,[a])$. Now, we can calculate the $1$-discriminant. We have  $Z_n^*(\beta F,[1]) = e^{\beta F_n(\overline{1,n,n-1,\dots,2})}$, therefore:
\begin{align*}
    R = \lim_{n \to \infty} \frac{Z_n^*(\beta F,[1])}{Z_{n+1}^*(\beta F,[1])} = \lim_{n \to \infty} \frac{e^{\beta F_n(\overline{1,n,n-1,\dots,2})}}{e^{\beta F_{n+1}(\overline{1,n+1,n,\dots,2})}}.
\end{align*}
Since $F_n(\overline{1,n,n-1,\dots,2}) = - \log(n+1), \forall \  n \in \mathbb{N}$, it follows that $
    Z_n^*(\beta F,[1]) = (n+1)^{-\beta} \text{and} \ R = 1.$
    
Then,
\begin{equation}\label{eq:discriminant_identity}
    \Delta_1[\beta F] = \log\left(\sum_{k=1}^\infty R^k Z_k^*(\beta F,[1])\right) =  \log\left(\sum_{k=1}^\infty  \frac{1}{(k+1)^\beta}\right) = \log\left(\zeta(\beta)-1\right),
\end{equation}
where $\zeta$ is the Riemann zeta function. For $0 < \beta \leq 1$, the series in \eqref{eq:discriminant_identity} diverges and therefore $\Delta_1[F] = \infty > 0$. Let $\beta_c >0$ be the unique solution of $\zeta(\beta_c) = 2$ ($\beta_c \approx 1.72865$), since the $\zeta$ is strictly decreasing for $\beta >1$, we have that $\zeta(\beta) - 1 > 1$ for $1 < \beta < \beta_c$, and in this case, we obtain $0 < \Delta_1 [\beta F] < \infty$. Also, $\Delta_1 [\beta_c F] =  0.$

Furthermore, observe that $\zeta(\beta) > 1$ for every $\beta > 1$, and it has $1$ as its horizontal asymptote. Then, for $\beta > \beta_c$ we have that $0 < \zeta(\beta) - 1 < 1$, and therefore $\Delta_1[\beta F] < 0$ for $\beta > \beta_c$. By the discriminant theorem, $\beta F$ is positive recurrent for $\beta < \beta_c$ and it is transient for $\beta > \beta_c$. Consequently, by the generalized RPF theorem, there exists a conservative eigenmeasure for the Ruelle's operator $L_{\beta F}$ which is finite on cylinders for $\beta \in (0,\beta_c)$, and we have the absence of such measures for $\beta \in (\beta_c,\infty)$. Note that $F$ depends only on the first coordinate and $\sup F < \infty$. 
In addition, $\Var_k F^{\induced} = 0,$ for every $k \geq 2$, and so $F^{\induced}$ is locally H\"{o}lder. Since the potential $F$ satisfies the hypotheses of Theorem \ref{thm:phase_transition_renewal_shift_Sarig_Gurevich_pressure}, then $P_G(\beta F)$ is linear on the variable $\beta$ for $\beta \geq \beta_c$. More precisely, $P_G(\beta F) = \beta p_1^*[F]$, for $\beta \geq \beta_c$. By item 2 of discriminant theorem and equation \eqref{eq:p_a_star_is_limsup_partition_Z_star}, we have that
\begin{equation} \label{eq:coefficient_pressure_zero_for_beta_greater_than_critical}
    p_1^*[F] = \limsup_n \frac{1}{n} \log Z_n^*(F,[1]) = \limsup_n \frac{1}{n} \log \left(\frac{1}{n+1}\right) = 0.
\end{equation}
This fact implies that $P_G(\beta F) = 0$, for $\beta \geq \beta_c$, and since the eigenmeasures for this interval of temperatures will be fixed points of the Ruelle's operator, see Remark \ref{remark:KMS_quasi_invariant}, we can also study the KMS$_{\beta}$ states associated with the $1$-cocycle given by the potential.
\end{example}

\subsection{Generalized Renault-Deaconu groupoids}\\

We consider the \emph{generalized Renault-Deaconu groupoid} \cite{Renault2} which is defined as follows. Let $X$ be a locally compact, Hausdorff and second countable space. Let be $U$ an open subset of $X$ and consider a local homeomorphism $\sigma: U \to X$. The \textit{generalized Renault-Deaconu groupoid} is given by
\begin{equation}\label{eq:ADR_groupoid_generalized}
    \mathcal{G}(X,\sigma) = \left\{  (x,k,y) \in X \times \mathbb{Z} \times X \text{ }\text{ }:\begin{array}{l l}
                            & \exists n,m \in \mathbb{N}_0 \text{ s.t. } k=n-m, \\
                            &  x \in \Dom (\sigma^n), y \in \Dom (\sigma^m), \sigma^n(x) = \sigma^m(y)\\
                          \end{array}\right\},
\end{equation}
where $\sigma^0:X \to X$ is the identity map, with the groupoid structure given as follows. The product is defined on the set 
$\mathcal{G}^{(2)}:= \left\{\big((x,k,z),(w,l,y)\big) \in \mathcal{G}(X,\sigma) \times \mathcal{G}(X,\sigma): z=w \right\}$,
and it is given by the rule $\big((x,k,z),(z,l,y)\big) \mapsto (x,k+l,y) \in \mathcal{G}(X,\sigma).$ The inverse map is defined on $\mathcal{G}(X,\sigma)$ by $(x,k,y) \mapsto (y,-k,x) \in \mathcal{G}(X,\sigma).$
The unit space is the set $\mathcal{G}^{(0)}:= \{(x,0,x): x \in X\}$. The range and source maps, respectively $r:\mathcal{G}(X,\sigma) \to \mathcal{G}^{(0)}$ and $s:\mathcal{G}(X,\sigma) \to \mathcal{G}^{(0)}$, are given by
\begin{align*}
    r((x,k,y)) = (x,0,x) \quad \text{and} \quad s((x,k,y)) = (y,0,y).
\end{align*}
To introduce a topology, let $n,m \in \mathbb{N}_0$ and $V_1,V_2$ be open subsets of $\Dom(\sigma^n)$ and $\Dom(\sigma^m)$, respectively. We define the sets $W(n,m,V_1,V_2) = \left\{(x,n-m,y): x \in V_1, y \in  V_2, \sigma^n(x) = \sigma^m(y) \right\}.$ They form a basis for a topology of $\mathcal{G}(X,\sigma)$ which makes it a locally compact Hausdorff second countable. Moreover, in this case $\mathcal{G}(X,\sigma)$ is an \emph{\'etale groupoid}, i.e., $r$ and $s$ are local homeomorphisms. Besides, the open sets $W(n,m,V_1,V_2)$ such that $\sigma^n\vert_{V_1}$ and $\sigma^m\vert_{V_2}$ are injective form a basis for $\mathcal{G}(X,\sigma)$.

\begin{definition}
Given an \'etale groupoid $G$, an open subset $W$ of $G$ is called an \textit{open bisection} if the maps $r$ and $s$, when restricted to $W$, are homeomorphisms onto their images.
\end{definition}

\begin{remark}
   The open sets $W(n,m,V_1,V_2)$ such that $\sigma^n\vert_{V_1}$ and $\sigma^m\vert_{V_2}$ are injective are open bisections.
\end{remark}

We identify $\mathcal{G}^{(0)}$ with $X$, via the natural homeomorphism between them. For a continuous function $F:U\to \mathbb{R}$, we think of $\mathbb{R}$ as an additive group, we define a continuous homomorphism $c_F:\mathcal{G}(X,\sigma)\to \mathbb{R}$ as 
\begin{equation}\label{eq:cocycle}
    c_F(x,n-m,y) = \begin{cases}
                        \sum_{i=0}^{n-1}F(\sigma^i(x))-\sum_{i=0}^{m-1}F(\sigma^i(y)), \quad \text{if } x,y \in U,\\
                        \sum_{i=0}^{n-1}F(\sigma^i(x)), \quad \text{if } x \in U \text{ and } y \in U^c \text{ }(m=0),\\
                        -\sum_{i=0}^{m-1}F(\sigma^i(y)), \quad \text{if } x \in U^c \text{ and } y \in U\text{ }(n=0),\\
                        0, \quad \text{otherwise};
                   \end{cases}
\end{equation}
where $(x,n-m,y) \in W(n,m,V_1,V_2)$. We may refer to $c_F$ as the \textit{1-cocycle} associated to the continuous potential $F$. We use some facts about \'etale groupoids, we refer to \cite{Deaconu1995, putnam, Renault1980, SimsSzaboWilliams2020} for this topic.\\
%
\vv

\subsection{Cuntz-Krieger algebras for infinite matrices}\\

Let $A$ a matrix of zeros and ones, and  $\widetilde{\mathcal{O}}_A$ be the unital universal $C^*$-algebra generated by a family of partial isometries $\{S_j:j \in \mathbb{N}\}$, as constructed in \cite{EL1999}, satisfying the following:

\begin{itemize}
    \item[$(EL1)$] $S_i^*S_i$ and $S_j^*S_j$ commute for every $i,j \in \mathbb{N}$;
    \item[$(EL2)$] $S_i^*S_j = 0$ whenever $i \neq j$;
    \item[$(EL3)$] $(S_i^*S_i)S_j = A(i,j)S_j$ for all $i,j \in \mathbb{N}$;
    \item[$(EL4)$] for every pair $X,Y$ of finite subsets of $\mathbb{N}$ such that the quantity
        \begin{equation*}
            A(X,Y,j):= \prod_{x \in X} A(x,j) \prod_{y \in Y} (1-A(y,j)), j \in \mathbb{N}
        \end{equation*}
    is non-zero only for a finite number of $j$'s, we have
        \begin{equation*}
            \left(\prod_{x \in X} S_x^*S_x\right) \left(\prod_{y \in Y} (1-S_y^*S_y)\right) = \sum_{j \in \mathbb{N}} A(X,Y,j)S_j S_j^*.
        \end{equation*}
\end{itemize}

Take the $C^*$-subalgebra $\mathcal{O}_A \subseteq \widetilde{\mathcal{O}}_A$, generated by the same partial isometries $S_j$, $j \in \mathbb{N}$. Note that $\mathcal{O}_A$ may coincide with $\widetilde{O}_A$ under some circunstances as proved in Proposition 8.5 of \cite{EL1999}. When these algebras do not coincide, the algebra $\widetilde{\mathcal{O}}_A$ is the canonical unitization of $\mathcal{O}_A$. The algebra $\mathcal{O}_A$ consists in a generalization for infinite transition matrices of the Cuntz-Krieger algebra \cite{CK1980}. The irreducibility of $A$ is a sufficient condition to grant the uniqueness of these algebras.

Consider the canonical basis $\{\delta_x\}_{x \in \Sigma_A}$ defined as 
\begin{equation*}
    (\delta_x)_y = \begin{cases}
                        1 \text{ if }x=y,\\
                        0 \text{ otherwise}.
                    \end{cases}
\end{equation*}
By Proposition $9.1$ of \cite{EL1999} there exists a unique representation $\pi: \widetilde{\mathcal{O}}_A \to \mathfrak{B}(\ell^2(\Sigma_A))$ s.t. each partial isometry $T_j:=\pi(S_j)$ ($j \in \mathbb{N}$) act on the canonical basis $\{\delta_x\}_{x \in \Sigma_A}$ as
\begin{equation*}
    T_s(\delta_x) = \begin{cases}
                        \delta_{sx} \text{ if } A(s,x_0)=1,\\
                        0 \text{ otherwise};
                    \end{cases} \text{with} \quad
    T_s^*(\delta_x)=\begin{cases}
            \delta_{\sigma(x)} \text{ if } x \in [s],\\
            0 \text{ otherwise}.
        \end{cases}.
\end{equation*}

We also define the projections $P_s:=T_sT_s^*$ and $Q_s:=T_s^*T_s$, given by
\begin{equation*}
    P_s(\delta_\omega)=\begin{cases}
                            \delta_\omega \text{ if } \omega \in [s], \\
                            0 \text{ otherwise;} 
                        \end{cases}  \text{and} \quad
    Q_s(\delta_\omega)=\begin{cases}
                        \delta_\omega \text{ if } \omega \in \sigma([s]),\\
                            0 \text{ otherwise.}
                        \end{cases}
\end{equation*}
The representation $\pi$ is faithful if the graph of $A$ has no terminal circuits (see Proposition 12.2 in \cite{EL1999}) and a sufficient condition for it is that $A$ be irreducible. On the section \ref{section:X_A} we will define the space $X_A$, which is the spectrum of a suitable commutative $C^*$-subalgebra of $\mathcal{O}_A$ or $\widetilde{\mathcal{O}}_A$, the main object of this paper, see next section.

\section{Some results about $\mathcal{O}_A$} \label{sec:results_O_A}

\begin{proposition}\label{prop:O_A_closure_span} $\widetilde{\mathcal{O}}_A$ is isomorphic to the closure of the linear span of the terms $T_\alpha \left(\prod_{i \in F}Q_i\right)T_\gamma^*$, where $\alpha$ and $\gamma$ are admissible finite words or the empty word, and $F \subseteq S$ is finite.
\end{proposition}

\begin{proof} We recall that $\widetilde{\mathcal{O}}_A \simeq C^*(\{T_i: i \in S\}\cup\{1\})$. First, we will prove that 
\begin{equation}\label{eq:spanO_A}
    \spann\left\{T_\alpha \left(\prod_{i \in F}Q_i\right)T_\gamma^*:F \text{ finite}; \text{ }\alpha, \gamma \text{ admissible words, including empty words} \right\}
\end{equation}
is a $*$-algebra. Indeed, the vector space properties are trivially satisfied, as well as the closeness of the involution. For the algebra product, take two generators in $\widetilde{\mathcal{O}}_A$, $T_\alpha \left(\prod_{i \in F}Q_i\right)T_\gamma^*$ and $T_{\alpha'} \left(\prod_{j \in F'}Q_i\right)T_{\gamma'}^*$ like in \eqref{eq:spanO_A}, with $\gamma = \gamma_1 \cdots \gamma_n$ and $\alpha' = \alpha_1'\cdots \alpha_m'$; $n,m \in \mathbb{N}$. We wish that the product
\begin{equation}\label{eq:prod_span}
    T_\alpha \left(\prod_{i \in F}Q_i\right)T_\gamma^*T_{\alpha'} \left(\prod_{j \in F'}Q_i\right)T_{\gamma'}^*
\end{equation}
can be written as a linear combination of terms like the generators of \eqref{eq:spanO_A} and hence we need to study the term $T_\gamma^*T_{\alpha'}$. From the axiom $(EL3)$ for the Cuntz-Krieger algebra for infinite matrices
we have that
\begin{equation} \label{eq:Q_iT_j}
    Q_i T_j = A(i,j) T_j,
\end{equation}
and consequently
\begin{equation} \label{eq:Q_iT_j_star}
    T_j^* Q_i = A(i,j) T_j^*.
\end{equation}
We have three cases to analyze as follows.

\begin{itemize}
    \item[$(a)$] If $n=m$, then by the axiom $(EL2)$ and \eqref{eq:Q_iT_j} we get
        \begin{equation*}
            T_\gamma^*T_{\alpha'} = T_{\gamma_n}^* \cdots T_{\gamma_2}^* \delta_{\gamma_1,\alpha'_1}Q_{\gamma_1} T_{\alpha'_2} \cdots T_{\alpha'_n}
            = \delta_{\gamma_1,\alpha'_1} T_{\gamma_n}^* \cdots T_{\gamma_2}^*  T_{\alpha'_2} \cdots T_{\alpha'_n} = \cdots = \delta_{\gamma,\alpha'} Q_{\gamma_n},
        \end{equation*}
    where $\delta_{\gamma,\alpha'}$ is the Kronecker delta. So,
    \begin{equation*}
    T_\alpha \left(\prod_{i \in F}Q_i\right)T_\gamma^*T_{\alpha'} \left(\prod_{j \in F'}Q_j\right)T_{\gamma'}^*= 
                            \delta_{\gamma,\alpha'} T_\alpha \left(\prod_{i \in F}Q_i\right)Q \left(\prod_{j \in F'}Q_j\right)T_{\gamma'}^*,
    \end{equation*}
    where $Q = Q_{\gamma_n}$ if $n>0$ and $Q = 1$ otherwise. We conclude that the product above belongs to \eqref{eq:spanO_A} in this case;
    \item[$(b)$] if $n>m$, by similar calculations done in the earlier case using \eqref{eq:Q_iT_j_star} instead of \eqref{eq:Q_iT_j} and defining $\overline{\gamma}:=\gamma_1\cdots \gamma_m$ we obtain $T_\gamma^*T_{\alpha'} = \delta_{\overline{\gamma},\alpha'} T_{\gamma_n}^* \cdots T_{\gamma_{m+1}}^*$. By using \eqref{eq:Q_iT_j_star} several but finite times on the term $T_{\gamma_{m+1}}^*\left(\prod_{j \in F'}Q_j\right)$, we have that
    \begin{equation*}
    T_\alpha \left(\prod_{i \in F}Q_i\right)T_\gamma^*T_{\alpha'} \left(\prod_{j \in F'}Q_j\right)T_{\gamma'}^*= \delta_{\overline{\gamma},\alpha'}\left(\prod_{j \in F'}A(j,\gamma_{m+1})\right) T_\alpha \left(\prod_{i \in F}Q_i\right) T_{\gamma'\gamma_{m+1}\cdots \gamma_n}^*.
    \end{equation*}
    We conclude that the product above also belongs to \eqref{eq:spanO_A};
    \item[$(c)$] for $n<m$ the proof is similar to the previous item by using the \eqref{eq:Q_iT_j} instead of \eqref{eq:Q_iT_j_star}.
\end{itemize}
We conclude that \eqref{eq:spanO_A} is a $*$-subalgebra of the $C^*$-algebra $\widetilde{\mathcal{O}}_A$, and hence 
\begin{equation}\label{eq:spanO_A_closed}
    B=\overline{\spann}\left\{T_\alpha \left(\prod_{i \in F}Q_i\right)T_\gamma^*:F \text{ finite}; \text{ }\alpha, \gamma \text{ finite admissible words} \right\}
\end{equation}
is a $C^*$-subalgebra of $\widetilde{\mathcal{O}}_A$. On other hand, if we take $F= \emptyset$, $\alpha = s$, $s\in S$ and $\gamma$ the empty sequence, then we conclude that $T_s \in B$ for all $s \in S$. Also, if we take $F= \emptyset$ and $\alpha = \gamma$ empty sequence, it follows that $1$ belongs to \eqref{eq:spanO_A_closed}. Since $B$ is a $C^*$-subalgebra of $\widetilde{\mathcal{O}}_A$ which contains its generators, we have that $\widetilde{\mathcal{O}}_A = B$.
\end{proof}    

\begin{remark} If $\mathcal{O}_A$ is not unital, then by similar proof as which is done for Proposition \ref{prop:O_A_closure_span} it is easy to verify that
\begin{equation*}
    \mathcal{O}_A \simeq \overline{\spann}\left\{ \begin{array}{l l}
         & F \text{ finite}; \text{ }\alpha, \gamma \text{ finite admissible words};\\
         T_\alpha \left(\prod_{i \in F}Q_i\right)T_\gamma^*: &F \neq \emptyset \text{ or } \alpha \text{ is not the empty word} \\
         &\text{or }  \gamma \text{ is not an empty word}
    \end{array}\right\} .
\end{equation*}
\end{remark}

\begin{definition} Let $\widetilde{\mathcal{D}}_A$ be the commutative unital $C^*$-subalgebra of $\widetilde{\mathcal{O}}_A$ given by
\begin{equation*}
    \widetilde{\mathcal{D}}_A:= \overline{\spann}\left\{T_\alpha \prod_{i \in F}Q_i T_\alpha^*: F \text{ finite}; \alpha \text{ finite word} \right\},
\end{equation*}
and denote by $\mathcal{D}_A$ its non-unital version when $\mathcal{O}_A$ is not unital,
\begin{equation*}
    \mathcal{D}_A:= \overline{\spann}\left\{T_\alpha \prod_{i \in F}Q_i T_\alpha^*: F \text{ finite}; \alpha \text{ finite word}; F \neq \emptyset \text{ or } \alpha \text{ is not the empty word}\right\}.
\end{equation*}
\end{definition}

The proof of the previous proposition for $\alpha = \gamma$ and $\alpha' = \gamma'$ shows that $\widetilde{\mathcal{D}}_A$ is a unital $C^*$-subalgebra of $\widetilde{\mathcal{O}}_A$. Moreover, note that $\widetilde{\mathcal{D}}_A$ is commutative.

Now, we will obtain a more suitable set of generators for $\widetilde{\mathcal{D}}_A$ which will allow us to see its spectrum as a set of configurations on the Cayley tree. 

Consider the free group $\mathbb{F}$ generated by the alphabet $\mathbb{N}$ and let the map
\begin{align*}
    T:\mathbb{F} &\to \widetilde{\mathcal{O}}_A,\\
    s &\mapsto T_s, \\
    s^{-1} &\mapsto T_{s^{-1}}:= T_s^*.
\end{align*}
Also, for any word $g$ in $\mathbb{F}$, take its reduced form $g=x_1\dots x_n$ and define that $T$ realizes the mapping
\begin{equation*}
    g \mapsto T_g:= T_{x_1} \cdots T_{x_n},
\end{equation*}
and that $T_e=1$. In order to $T$ be well-defined, we have imposed that it only acts on the reduced words. In addition, we denote by $\mathbb{F}_+$ the positive cone of $\mathbb{F}$, i.e., the unital sub-semigroup of $\mathbb{F}$ generated by $\mathbb{N}$. The map $T$ is a partial group representation which is semi-saturated and orthogonal as proved in Proposition 3.2 of \cite{EL1999}, and in particular it satisfies the property that

\begin{equation}\label{eq:prod_partial_repr}
    T_g T_h T_{h^{-1}} = T_{gh}T_{h^{-1}}, \quad g,h \in \mathbb{F}.
\end{equation}

\begin{remark}\label{lema:g_non_zero} For any $g \in \mathbb{F}$ reduced which is not in the form $\alpha\gamma^{-1}$, with $\alpha, \gamma \in \mathbb{F}_+$, it follows that $T_g = 0$.
\end{remark}

Consider the elements of $\widetilde{\mathcal{D}}_A$ given by $e_g := T_g T_g^*$, where $g$ is in the reduced form. Such elements commute each other and they are projections (see \cite{EL1999,Exel_partial_action}), and therefore they generate a commutative $C^*$-subalgebra of $\widetilde{\mathcal{O}}_A$.

\begin{proposition}\label{prop:D_A_isomorphic_e_g} $\widetilde{\mathcal{D}}_A \simeq C^*(\{e_g:g \in \mathbb{F}\})$.
\end{proposition}

\begin{proof} The main idea of the proof is to show that the faithful representation of the $C^*$-algebra $\widetilde{\mathcal{D}}_A$ in $\mathfrak{B}(\ell^2(\Sigma_A))$ coincides with the $C^*$-subalgebra $\mathfrak{U} = C^*(\{e_g:g \in \mathbb{F}\})$ contained in $\mathfrak{B}(\ell^2(\Sigma_A))$, which implies that they are isomorphic. We will show that the products $T_\alpha \left(\prod_{i \in F} Q_i\right) T_\alpha^*$ can be written as elements of $\mathfrak{U}$, and conversely $e_g$ can be written as elements of $\widetilde{\mathcal{D}}_A$. 

Let $g \in \mathbb{F}$. W.l.o.g. we may assume that $T_g \neq 0$. By Remark \ref{lema:g_non_zero} we have that $g = \alpha \gamma^{-1}$ such that $\alpha, \gamma \in \mathbb{F}_+$, with $\alpha = \alpha_0 \cdots \alpha_t$ and $\gamma = \gamma_0 \cdots \gamma_u$ for the respective cases when $\alpha$ and $\gamma$ are not $e$. Assume that $g$ is already its reduced form, i.e., $g = \alpha \gamma^{-1}, \alpha, \gamma^{-1}$ or $e$. By the axiom $(EL3)$ we have
\begin{align*}
    e_g &= T_{\alpha} T_{\gamma}^* T_{\gamma} T_{\alpha}^* = T_{\alpha} T_{\gamma_u}^*\cdots T_{\gamma_1}^* Q_{\gamma_0} T_{\gamma_1} \cdots T_{\gamma_{u}} T_{\alpha}^* 
    = T_{\alpha} T_{\gamma_{u}}^*\cdots T_{\gamma_1}^* T_{\gamma_1} T_{\gamma_{u}} T_{\alpha}^*
    = \cdots \\
    &= T_{\alpha} T_{\gamma_{u}}^* T_{\gamma_{u}} T_{\alpha}^* = T_{\alpha} Q_{\gamma_u} T_{\alpha}^* \in \widetilde{\mathcal{D}}_A,
\end{align*}
and we conclude that $\mathfrak{U} \subseteq \widetilde{\mathcal{D}}_A$. The result above is similar for $\alpha =e$ or $\gamma = e$. For the remaining inclusion, let $\alpha \in \mathbb{F}_+$ admissible or $\alpha = e$ in its reduced form, and $F \subseteq \mathbb{N}$ finite. If $\alpha = e$, we have that 
\begin{equation*}
    T_\alpha \left(\prod_{i \in F} Q_i \right) T_\alpha^* = \prod_{i \in F} Q_i = \prod_{i \in F} e_{i^{-1}} \in  \mathfrak{U}.
\end{equation*}
On other hand, if $F = \emptyset$ and $\alpha \neq e$ is an admissible word, we have:
\begin{equation*}
    T_\alpha \left(\prod_{i \in F} Q_i \right) T_\alpha^* = T_\alpha  T_\alpha^* =  e_{\alpha} \in \mathfrak{U}.
\end{equation*}
Now, suppose that $\alpha = \alpha_0 \cdots \alpha_{t} \neq e$ reduced and $F \neq \emptyset$. We will prove that 
\begin{equation}\label{eq:prod_e_alpha}
    T_\alpha \left(\prod_{j \in F} Q_j \right) T_\alpha^* = e_\alpha \prod_{\substack{j \in F \\ j \neq \alpha_{t}}} e_{\alpha j^{-1}}
\end{equation}
by induction in $|F|$. If $|F| = 1$ we have that
\begin{align*}
    T_\alpha \left(\prod_{j \in F} Q_j \right) T_\alpha^* = T_\alpha T_{i^{-1}} T_i T_{\alpha^{-1}},
\end{align*}
where $i \in F$. If $\alpha_t = i$, since $e_i$ is a projection we get
\begin{equation*}
    T_\alpha T_{i^{-1}} T_i T_{\alpha^{-1}} = T_{\alpha'} T_i T_{i^{-1}} T_i T_{i^{-1}} T_{(\alpha')^{-1}} = T_{\alpha'}  T_i T_{i^{-1}} T_{(\alpha')^{-1}} = e_{\alpha},
\end{equation*}
where $\alpha' = e$ if $t=0$ and $\alpha' = \alpha_0 \cdots \alpha_{t - 1}$ if $t>0$. On other hand, if $\alpha_t \neq i$, it follows that
\begin{equation*}
    T_\alpha T_{i^{-1}} T_i T_{\alpha^{-1}} = T_{\alpha i^{-1}}  T_{(\alpha i^{-1})^{-1}} = e_{\alpha i^{-1}},
\end{equation*}
and by using \eqref{eq:prod_partial_repr} one can verify that $e_\alpha e_{\alpha i^{-1}} = e_{\alpha i^{-1}}$. So, anyway we have that \eqref{eq:prod_e_alpha} is true for $|F|=1$. Now, suppose the validity of \eqref{eq:prod_e_alpha} for $|F| = n-1$, $n>1$. For $|F| = n$, fix $k \in F$. One can use \eqref{eq:prod_partial_repr} and the claim 1 of the proof of Proposition 3.2 of \cite{EL1999} in order to obtain
\begin{align*}
    T_\alpha \left(\prod_{i \in F} Q_i \right) T_\alpha^* &= T_\alpha T_\alpha^* T_\alpha \left(\prod_{i \in F} Q_i \right) T_\alpha^* =   T_\alpha Q_{\alpha_t} \left(\prod_{i \in F} Q_i \right) T_\alpha^* = T_\alpha Q_k Q_{\alpha_t} \left(\prod_{i \in F\setminus\{k\}} Q_i \right) T_\alpha^*\\ 
    &= T_\alpha Q_k T_\alpha^* T_\alpha \left(\prod_{i \in F\setminus\{k\}} Q_i \right) T_\alpha^* =  e_\alpha \left(\prod_{\substack{j \in \{k\}\\ j \neq \alpha_t}} e_{\alpha j^{-1}}\right) e_\alpha \left(\prod_{\substack{j \in F\setminus\{k\}\\ j \neq \alpha_t}} e_{\alpha j^{-1}}\right), 
\end{align*}
where in the last equality we used \eqref{eq:prod_e_alpha} for $|F|=1$ and the induction hypothesis. Since the $e_g$'s commute and they are projections, we conclude that
\begin{align*}
    T_\alpha \left(\prod_{i \in F} Q_i \right) T_\alpha^* &= e_\alpha \left(\prod_{\substack{j \in F\\ j \neq \alpha_{|\alpha|-1}}} e_{\alpha j^{-1}}\right),
\end{align*}
as we wished to prove. The direct consequence of the results above is that $T_\alpha \left(\prod_{i \in F} Q_i \right) T_\alpha^* \in \mathfrak{U}$, for all $\alpha$ admissible finite word and for every $F \subseteq S$ finite. Then,
\begin{equation*}
    \spann \left\{T_\alpha\left(\prod_{i \in F}Q_i\right)T_\alpha^*: \alpha\text{ admissible }, 0\leq|\alpha|< \infty ,0\leq |F|< \infty \right\} \subseteq \mathfrak{U},
\end{equation*}
and since $\mathfrak{U}$ is a $C^*$-algebra we conclude that the closure of the left hand side of the relation above is still contained in $\mathfrak{U}$, i.e., $\widetilde{\mathcal{D}}_A \subseteq \mathfrak{U}$. 
\end{proof}

\begin{remark} The elements of $\widetilde{\mathcal{D}}_A$ are diagonal operators.
\end{remark}

Now we have all the necessary background to introduce the space we call \emph{generalized countable Markov} (GCMS), which we denote by $X_A$. One of the authors and M. Laca originally proposed $X_A$ in \cite{EL1999}, this set contains the standard $\Sigma_A$, and both coincide when $\Sigma_A$ is locally compact, as we show in the next section.

\section{The set $X_A$}\label{section:X_A}

In the construction of the space $X_A$ we have that $\Sigma_A$ is dense in $X_A$, and the same is true for $Y_A=X_A \backslash \Sigma_A$ when $Y_A\neq \emptyset$. Some of the most studied CMS in the literature of the standard formalism, as the full shift and the renewal shift, are compact spaces when we consider their generalized versions $X_A$. In other words, in these cases, $X_A$ is a compactification of $\Sigma_A$. The compactness of the configuration space is a key property to obtaining many results in thermodynamic formalism. It will be one of the ingredients to prove that the Gurevich pressure $P_G(F)$ defined in $\Sigma_A$ can be used as the natural definition of pressure in the corresponding generalized space $X_A$.

\begin{definition}[Generalized countable Markov shift (GCMS)]\label{def:X_A} Given an irreducible transition matrix $A$ on the alphabet $\mathbb{N}$, \textit{the generalized countable Markov shifts} are the sets
\begin{equation*}
    X_A := \text{spec}\,\mathcal{D}_A \quad \text{and} \quad \widetilde{X}_A := \text{spec}\,\widetilde{\mathcal{D}}_A,
\end{equation*}
both endowed with the weak$^*$ topology.
\end{definition}

\begin{remark} The sets $X_A$ and $\widetilde{X}_A$ were denoted by $\Omega_A$ and $\widetilde{\Omega}_A$, respectively, in the original reference \cite{EL1999}.
\end{remark}

\begin{remark}\label{remark:unital_vs_full_row_of_ones} We remind the reader that when $\mathcal{O}_A$ is unital, then $\mathcal{D}_A$ is also unital. Consequently, $\widetilde{\mathcal{O}}_A = \mathcal{O}_A$, and then $\widetilde{\mathcal{D}}_A = \mathcal{D}_A$, which implies $\widetilde{X}_A = X_A$. Besides, $X_A$ is locally compact and $\widetilde{X}_A$ is always compact. In particular, for every matrix with a full row of $1$'s, $X_A$ is compact. In fact, if there exists a symbol $j \in \mathbb{N}$ such that $A(j,n) = 1$ for every $n \in \mathbb{N}$, then $Q_j = T_j^* T_j = 1 \in \mathcal{D}_A$ and therefore $X_A$ is compact. 
\end{remark}

\begin{remark}\label{remark:subalgebra_character} Given a commutative algebra $B$ and $J$ a closed self-adjoint two-sided ideal of $B$, then the set of characters of $J$ is given by $\widehat{J} = \{\varphi \in \widehat{B}: \varphi\vert_J \neq 0\}$. Therefore,
\begin{equation*}
    \widehat{\mathcal{D}}_A = \{\varphi\vert_{\mathcal{D}_A}:\varphi \in \widehat{\widetilde{D}}_A, \quad \varphi\vert_{\mathcal{D}_A} \neq 0\} = X_A.
\end{equation*}
It is straightforward to conclude that $X_A \subseteq \widetilde{X}_A$.
\end{remark}

\begin{proposition}\label{prop:X_A_tilde_X_A} $X_A = \widetilde{X}_A \setminus \{\varphi_0\}$, where $\varphi_0$ is the character in $\widetilde{X}_A$ given by
\begin{equation*}
    \varphi_0(e_g) := \begin{cases}
                        1, \quad \text{if } g=e;\\
                        0, \quad \text{otherwise.}
                      \end{cases}
\end{equation*}
\end{proposition}

\begin{proof} Suppose that $\varphi_0 \in X_A$. For $g \neq e$ we have $\varphi_0(T_gT_g^*) = \varphi_0(e_g) = 0$, and then $\varphi_0\vert_{\mathcal{D}_A} = 0$, which is not a character for the algebra $\mathcal{D}_A$. By Remark \ref{remark:subalgebra_character} we have that $X_A \subseteq \widetilde{X}_A\setminus \{\varphi_0\}$. Conversely, by Proposition \ref{prop:D_A_isomorphic_e_g}, for given $\varphi \in \widetilde{X}_A$ such that $\varphi \neq \varphi_0$, there exists $g \neq e$ such that $\varphi(e_g)=1$, and hence $\varphi\vert_{\mathcal{D}_A} \neq 0$. Then  $\varphi\vert_{\mathcal{D}_A}$ is a character of $\mathcal{D}_A$, and again by Remark \ref{remark:subalgebra_character} we have that $\widetilde{X}_A\setminus \{\varphi_0\} \subseteq X_A$. 
\end{proof}

\begin{remark} Observe that if $\mathcal{D}_A$ is not unital, then $\widetilde{X}_A$ is the Alexandrov compactification of $X_A$, where the compactification point of $X_A$ is the character $\varphi_0$. 
\end{remark}

The next result connects equivalent conditions to $\mathcal{O}_A$ be unital.

\begin{theorem}[Theorem 8.5 of \cite{EL1999}]\label{thm:O_A_unital} The following are equivalent:
\begin{itemize}
    \item[$(i)$] $\mathcal{O}_A=\mathcal{\widetilde{O}}_A$;
    \item[$(ii)$] $\mathcal{O_A}$ is unital;
    \item[$(iii)$] $\varphi_0 \notin \widetilde{X}_A$;
    \item[$(iv)$] On the space $\{0,1\}^S$ (column space of the matrix $A$, endowed with the product topology), the null vector is not a limit point of the columns of $A$;
    \item[$(v)$] There is $Y\subseteq S$ finite such that $A(\emptyset,Y,j)$ has finite support on $j$.
\end{itemize}
\end{theorem}

We will introduce two new notions related to compactness, which will be used to prove that the Denker-Yuri pressure coincides with the Gurevich pressure for a class of GCMS which includes many of the most studied cases of the literature, as the full shift and the renewal shift.  In the next definition, for each $i \in \mathbb{N}$, we set $r_A(i):=\{j \in \mathbb{N}: A(i,j) = 1\}$.

\begin{definition} We say that the irreducible matrix $A$ is \emph{$r$-compact} when, for every $i \in \mathbb{N}$, there exists $m = m(i) \in \mathbb{N}$ and a finite set $\{j_1,...,j_m\} \subseteq \mathbb{N}$, such that
\begin{equation*}
    r_A(i)^c \subseteq \bigcup_{k=1}^m r_A(j_k).
\end{equation*}
\end{definition}

\begin{lemma} \label{lemma:equivalence_compactness_X_A_and_graph} Let $A$ be an irreducible matrix $A$. Then, $X_A$ is compact if and only if $A$ is $r$-compact.
\end{lemma}

\begin{proof} Suppose that $X_A$ is compact, then by Theorem \ref{thm:O_A_unital} $(v)$, this is equivalent to state that there exists $Y \subseteq \mathbb{N}$, finite set, such that $A(\emptyset, Y, j) = 0$ except for $j \in J_Y \subseteq \mathbb{N}$, a finite set. Observe that
\begin{equation*}
    A(\emptyset,Y,j) = \prod_{y \in Y}(1-A(y,j)),
\end{equation*}
and then $A(\emptyset,Y,j) = 1$ iff $A(y,j) = 0, \text{ for every } y \in Y$.
By definition, we have
\begin{equation*}
    J_Y = \{j \in \mathbb{N}: A(y,j) = 0, \text{ for every } y \in Y\} = \bigcap_{y \in Y} r_A(y)^c.
\end{equation*}
For every $i \in \mathbb{N}$, we may write $r_A(i)^c = (r_A(i)^c \cap J_Y^c) \sqcup (r_A(i)^c \cap J_Y)$.
Now, if $j \in r_A(i)^c \cap J_Y$, since $A$ is irreducible, there exists $k_j \in \mathbb{N}$ such that $A(k_j,j) = 1$, so we fix a $k_j$ for each $j \in J_Y$. We have that
\begin{equation*}
    r_A(i)^c \subseteq J_Y^c \cup \left( \bigcup_{j \in J_Y}r_A(k_j)\right)  = \left(\bigcup_{y \in Y}r_A(y) \right) \cup \left( \bigcup_{j \in J_Y}r_A(k_j)\right).
\end{equation*}
Therefore $A$ is $r$-compact. Conversely, suppose now that $A$ is $r$-compact. Then, for every $i \in \mathbb{N}$ there exists $m = m(i) \in \mathbb{N}$ and a finite set $\{j_1,...,j_m\} \subset \mathbb{N}$, $j_k = j_k(i)$ such that
\begin{equation*}
    r_A(i)^c \subseteq \bigcup_{k=1}^m r_A(j_k).
\end{equation*}
For every $p \in \mathbb{N}$ we have either $p \in r_A(i)$ or $p \in r_A(i)^c$, and then $\mathbb{N} = r_A(i) \cup \bigcup_{k=1}^m r_A(j_k)$. So consider the set $Y = \{i\} \cup \{j_1,...,j_m\}$.

For every $j \in \mathbb{N}$, we have the following:
\begin{equation*}
    A(\emptyset,Y,j) = \prod_{y \in Y}(1 - A(y,j)) = (1-A(i,j)) \prod_{k =1}^m(1 - A(j_k,j)),
\end{equation*}
If $j \in r_A(i)$, then $A(i,j) = 1$. When $j \in r_A(i)^c$ we have that $j \in r_A(j_k)$ for some $k$, that is, $A(j_k, j) = 1$. Therefore $A(\emptyset,Y,j) = 0$ for every $j \in \mathbb{N}$, that is, $X_A$ is compact.
\end{proof}

\begin{definition}\label{def:i_1} Given $A$ an irreducible matrix and its corresponding CMS $\Sigma_A$, we define the map $i_1: \Sigma_A \xhookrightarrow{} X_A$, $\Sigma_A \ni \omega \mapsto \varphi_{\omega}\vert_{\mathcal{D}_A} \in X_A,$ where $\varphi_{\omega}$ is the evaluation map
\begin{equation}\label{eq:varphi_omega_definition}
    \varphi_\omega(R)= (R\delta_\omega,\delta_\omega), \quad R \in \mathfrak{B}(\ell^2(\Sigma_A)).
\end{equation}
We will be denoted by $i_1$ and the context will let evident if the codomain is $X_A$ or $\widetilde{X}_A$.
\end{definition}

The next lemma is straightforward consequence of the Gelfand's representation theorem.

\begin{lemma}\label{lemma:dense_character_general} Given a commutative $C^*$-algebra $B$, let $Y \subseteq \widehat{B}$ such that for any $a \in B$ it follows that
\begin{equation*}
    (\varphi(a) = 0 \quad \forall \ \varphi \in Y) \implies a=0,
\end{equation*}
i.e., $Y$ separates points in $B$. Then $Y$ is dense in $\widehat{B}$ (weak$^*$ topology).
\end{lemma}

\begin{lemma}\label{lemma:i_1_inj_cont} The map $i_1$ is a topological embedding.
\end{lemma}

\begin{proof} The injectivity is straightforward: let $\omega,\eta \in \Sigma_A$ such that $\varphi_\omega = \varphi_\eta$. If $\omega \neq \eta$, then there exists $n \in \mathbb{N}$ such that $\omega_0\dots\omega_{n-1} \neq \eta_0\dots\eta_{n-1}$. By taking $T = T_{\omega_0\dots\omega_{n-1}}T_{\omega_0\dots\omega_{n-1}}^* \in \mathcal{D}_A$ we have that
\begin{align*}
    \varphi_\omega(T_{\omega_0\dots\omega_{n-1}}T_{\omega_0\dots\omega_{n-1}}^*) &= (T_{\omega_0\dots\omega_{n-1}}T_{\omega_0\dots\omega_{n-1}}^* \delta_\omega, \delta_\omega) = 1 \text{ and}\\
    \varphi_\eta(T_{\omega_0\dots\omega_{n-1}}T_{\omega_0\dots\omega_{n-1}}^*) &= (T_{\omega_0\dots\omega_{n-1}}T_{\omega_0\dots\omega_{n-1}}^* \delta_\eta, \delta_\eta) = 0,
\end{align*}
and hence we have that $\varphi_\omega \neq \varphi_\eta$, a contradiction, and we conclude that $i_1$ is injective. Now we prove its continuity. It is sufficient to prove that $\varphi_{\omega^n}(T)\to \varphi_\omega(T)$ for every $T$ in the form $T=T_\alpha(\prod_{i\in F}Q_i)T_\alpha^*$, since these are generators of $\mathcal{D}_A$ and the elements of the spectrum are $*$-homomorphisms. Let $(\omega^n)_{n \in \mathbb{N}}$ be a sequence in $\Sigma_A$ converging to $\omega \in \Sigma_A$. 

Take $T=T_\alpha(\prod_{i\in F}Q_i)T_\alpha^*$ an arbitrary generator of $\mathcal{D}_A$. We have that
\begin{align*}
    \varphi_{\omega}(T)&=
    \begin{cases}
        1,\quad \text{if $\omega\in [\alpha]$ and $A(i,\omega_{|\alpha|})=1,\text{ }\forall i\in F$}, \\
        0,\quad \text{otherwise}.
    \end{cases}
\end{align*}
For large enough $m \in \mathbb{N}$, we have that $\omega_0^n\dots\omega^n_{|\alpha|}=\omega_0\dots\omega_{|\alpha|}$ for every $n > m$, and then we have $\varphi_{\omega^n}(T)=\varphi_{\omega}(T)$. Since $F$ and $\alpha$ are arbitrary, we have $\varphi_{\omega^n}(T)\to \varphi_\omega (T)$ for every generator $T \in \mathcal{D}_A$, that is,  $\varphi_{\omega^n}\stackrel{w^*}{\to} \varphi_\omega$. Therefore $i_1$ is continuous. The proof of the continuity of $i_1^{-1}$ is similar.
\end{proof}

The lemma above shows that $\Sigma_A$ is topologically copied in $X_A$. From now on, we omit the notations of restriction to $\widetilde{\mathcal{D}}_A$ and $\mathcal{D}_A$ on $\varphi_\omega$. In the next result, we prove that $\Sigma_A$, seen as a subset of $X_A$ (or $\widetilde{X}_A$), is weak$^*$-dense.

\begin{proposition}\label{prop:i_1_Sigma_A_dense} $i_1(\Sigma_A)$ is dense in $X_A$ (in $\widetilde{X}_A$).
\end{proposition}

\begin{proof} We claim that the elements of $i_1(\Sigma_A)$ separate points in $\mathcal{D}_A$. Suppose that there exists $T \in \mathcal{D}_A$ such that $\varphi_\omega (T) = 0$ for every $\omega \in \Sigma_A$. We recall that $T$ is a diagonal operator. Suppose that $\varphi_\omega(T) = 0$ for every $\omega \in \Sigma_A$, then
\begin{equation*}
    \varphi_\omega(T) = (T \delta_\omega, \delta_\omega) = 0,
\end{equation*}
then $\spann\{\delta_\omega:\omega \in \Sigma_A\} \subseteq \ker T$, and therefore $\overline{\spann}\{\delta_\omega:\omega \in \Sigma_A\} \subseteq \ker T$ since $\ker T$ is closed. We conclude that $T =0$ and the claim is proved. By Lemma \ref{lemma:dense_character_general} it follows that $i_1(\Sigma_A) = \{\varphi_\omega\}_{\omega \in \Sigma_A}$ is dense in $X_A$. Since $X_A$ is dense in $\widetilde{X}_A$, we have that $i_1(\Sigma_A)$ is also dense in $\widetilde{X}_A$. 
\end{proof}

Now, we prove that if $A$ is row-finite, then $\Sigma_A = X_A$, that is, $i_1$ is a surjection. 

\begin{proposition}\label{prop:Sigma_A_locally_compact_coincides_with_X_A} If $\Sigma_A$ is locally compact, then $i_1$ is surjective. 
\end{proposition}

\begin{proof} Let $\varphi \in X_A$. By Proposition \ref{prop:i_1_Sigma_A_dense}, there exists a sequence $(\varphi_{\omega^n})_{n\in \mathbb{N}}$ in $i_1(\Sigma_A)$, where $\omega^n \in \Sigma_A$ for every $n \in \mathbb{N}$, and such that $\varphi_{\omega^n}\stackrel{w^*}{\to} \varphi$. Since $\varphi$ is a character, we have necessarily that $\varphi \neq 0$, and by Proposition \ref{prop:D_A_isomorphic_e_g}, there exists $g \in \mathbb{F}$ such that $\varphi(e_g) \neq 0$. We assume $g = \alpha$ a positive admissible word, the other cases are similar. Since the elements $e_g$ are projections, we obtain $\varphi(e_\alpha) = \varphi(e_\alpha^2)$ and hence $\varphi(e_\alpha) = 1$. Then, there exists $N \in \mathbb{N}$ such that, for every $n > N$, we have
\begin{equation*}
    \varphi_{\omega^n}(e_\alpha) = (T_\alpha T_\alpha^* \delta_{\omega^n},\delta_{\omega^n} ) = 1,
\end{equation*}
and hence for every $n >N$ we have that $\omega^n \in [\alpha]$. On the other hand, since $\Sigma_A$ is locally compact, we have that $[\alpha]$ is compact, and so the sequence $(\omega^n)_{n\in \mathbb{N}}$ has a convergent subsequence $(\omega^{n_p})_{p\in \mathbb{N}}$ in $[\alpha]$, converging to some $\omega \in [\alpha] \subset \Sigma_A$. By continuity of $i_1$ proved in Lemma \ref{lemma:i_1_inj_cont} we have that $\lim_p \varphi_{\omega^{n_p}} = \varphi_\omega$, then
\begin{equation*}
    \varphi = \lim_n \varphi_{\omega^n} = \lim_p \varphi_{\omega^{n_p}} = \varphi_\omega,
\end{equation*}
and therefore $i_1$ is surjective. 
\end{proof}

\begin{remark} When $\Sigma_A$ is compact, $i_1$ is simply a homeomorphism between compact metric spaces. If $\Sigma_A$ is non-compact, but it is locally compact, we have $X_A=\Sigma_A$, differently from O. Shwartz's paper \cite{Shwartz2019}, where is presented a construction of the Martin boundary for locally compact CMS. The approaches are different since the Martin boundary adds extra points to the space $\Sigma_A$. In addition, if $\Sigma_A$ is not locally compact and $X_A$ compact, $i_1$ is a compactification of $\Sigma_A$, and indeed it is not the Stone-C\v{e}ch compactification since both $\Sigma_A$ and $X_A$ are metric spaces.
\end{remark}

Proposition \ref{prop:D_A_isomorphic_e_g} gives us a easier way to see $X_A$ (respect. $\widetilde{X}_A$). Given $\varphi \in X_A$ (or $\widetilde{X}_A$). We can determine its image completely simply by taking its values on the generators $(e_g)_{g \in \mathbb{F}}$. Since $e_g$ is idempotent for any $g$, it follows that $\varphi(e_g) \in \{0,1\}$. By endowing $\{0,1\}^{\mathbb{F}}$ with the product topology when we consider the discrete topology in $\{0,1\}$, we introduce now the mapping that realizes the characters of the generalized Markov shift spaces as configurations on the Cayley tree generated by $\mathbb{F}$.

\begin{definition}\label{projection} Let $A$ be a transition matrix and consider its respective generalized Markov shift space $X_A$. Define the map $i_2: X_A \to \{0,1\}^{\mathbb{F}}$ (repect. for $\widetilde{X}_A$) given by $i_2(\varphi) := \xi$, where
\begin{align*}
    \xi_g := \pi_g(\xi) = \varphi(e_g), \quad g \in \mathbb{F},
\end{align*}
where $\pi_g: \{0,1\}^\mathbb{F}\to \{0,1\}$ is the canonical projection. An element $\xi \in \{0,1\}^\mathbb{F}$ is called a \textit{configuration}. We say that a configuration $\xi$ \textit{is filled in $g \in \mathbb{F}$} when $\xi_g = 1$.  
\end{definition}
The following proposition shows that the topological properties of $X_A$ are preserved by $i_2$.

\begin{proposition}\label{prop:i_2_top_embedding} The inclusion $i_2$ is a topological embedding.
\end{proposition}

\begin{proof} First we show that $i_2$ is injective and continuous. The injectivity is straightforward: given $\varphi, \psi \in \widetilde{X}_A$ s.t. $i_2(\varphi) = i_2(\psi)$, it follows that $\varphi(e_g)=\psi(e_g)$ for all $g \in \mathbb{F}$. Since $\{e_g:g \in \mathbb{F}\}$ generates $\widetilde{\mathcal{D}}_A$, there exists an unique *-homomorphism which extends the function $e_g \mapsto \varphi(e_g), (g \in \mathbb{F})$ and such uniqueness implies that $\varphi = \psi$, i.e., $i_2$ is injective. For the continuity, let $(\varphi_n)_{n \in \mathbb{N}}$ be a sequence in $\widetilde{X}_A$ such that $\varphi_n \stackrel{w^*}{\to} \varphi \in \widetilde{X}_A$. The topology in $\widetilde{X}_A$ is the weak$^*$ topology and it is metrizable. We have the following equivalences:
\begin{equation*}
    \varphi_n \stackrel{w^*}{\to} \varphi \iff \varphi_n(e_g) \to \varphi(e_g), \quad \forall g \in \mathbb{F} \iff \{\varphi_n(e_g)\}_{g \in \mathbb{F}} \to \{\varphi(e_g)\}_{g \in \mathbb{F}}, 
\end{equation*}
where the last convergence above is the precisely the one in the product topology. Observe that every injective continuous map from a compact space to a Hausdorff space is a topological embedding. If $X_A$ is not compact, the proof follows by taking the restriction on $i_2$ for $\widetilde{X}_A$.  
\end{proof}

 Now we can see the characters in $X_A$ and $\widetilde{X}_A$ as configurations in the Cayley graph generated by $\mathbb{F}$, where the words $g$ are the vertices, and the oriented edges multiply by the right the word in the source vertex by a letter $a$, leading to the range vertex. The inverse way of the edge represents the multiplication by the inverse of the correspondent letter $a$.
 
\begin{figure}[H]
\begin{center}
		\begin{tikzpicture}[scale=1.5,decoration={markings, mark=at position 0.5 with {\arrow{>}}}]
		\node[circle, draw=black, fill=black, inner sep=1pt,minimum size=1pt] (0) at (0,0) {};
		\node[circle, draw=black, fill=black, inner sep=1pt,minimum size=1pt] (1) at (3,0) {};
    	\draw[postaction={decorate}, >=stealth] (0)  to (1);
   	    \node[above] at (0,0) {$g$};
   	    \node[above] at (1.5,0) {$a$};
        \node[above] at (3,0) {$ga$};
		\end{tikzpicture}
	\end{center}
\end{figure}
 
The next corollary is straightforward.

\begin{corollary}\label{cor:i_2_i_1_Sigma_A_dense} $i_2 \circ i_1(\Sigma_A)$ is dense in $i_2(X_A)$. Moreover, if $\mathcal{O}_A$ is not unital, then $i_2 \circ i_1(\Sigma_A)$ is dense in $i_2(\widetilde{X}_A)$.
\end{corollary}

From now we will describe $X_A$ (respect. $\widetilde{X}_A$) by its copy $i_2(X_A)$ (respect. $i_2(\widetilde{X}_A)$) contained in $\{0,1\}^\mathbb{F}$, except when the map $i_2$ is explicitly needed.
\begin{definition} \label{def:convex_configurations} A configuration $\xi$ is said to be \textit{connected} when for any two filled elements $a,b \in \mathbb{F}$, we have that $\xi$ is filled in the whole shortest path in the Cayley tree between $a$ and $b$. 
\end{definition}
\begin{remark} A configuration $\xi$ filled in $e$ is connected if and only if, for every $g \in \mathbb{F}$ filled in $\xi$, $\xi$ is also filled in the subwords of $g$.
\end{remark}
Now, we present some properties of the configurations in $X_A$ (respect $\widetilde{X}_A$), and for such goal, we define the following set as it is done in \cite{EL1999}.
\begin{definition}\label{def:set_Omega_A_tau} For a given transition matrix $A$, we define the set $\Omega_A^{\tau} \subset \{0,1\}^\mathbb{F}$, where a configuration $\xi$ belongs to $\Omega_A^{\tau}$ if and only if it satisfies the rules below:
\begin{itemize}
    \item[$(R1)$] $\xi_e = 1$;
    \item[$(R2)$] $\xi$ is connected;
    \item[$(R3)$] for every $g \in \mathbb{F}$, if $\xi_g=1$, then there exists at most one $i \in \mathbb{N}$ s.t. $\xi_{g i}=1$;
    \item[$(R4)$] for every $g \in \mathbb{F}$, if $\xi_g= \xi_{g j}=1$, where $j \in \mathbb{N}$, then for every $i \in \mathbb{N}$, it follows that
    \begin{equation*}
        \xi_{g i^{-1}}=1 \iff A(i,j)=1.
    \end{equation*}
\end{itemize}
\end{definition}

The figures \ref{fig:Rule_3} and \ref{fig:Rule_4} represents the rules $(R3)$ and $(R4)$, respectively. The black dots represents that the configuration $\xi$ is filled in the respective word, while the white dot represents that $\xi$ is not filled in the respective word.

\begin{figure}[H]
    \begin{center}
		\begin{tikzpicture}[scale=1.5,decoration={markings, mark=at position 0.5 with {\arrow{>}}}]
		\node[circle, draw=black, fill=black, inner sep=1pt,minimum size=5pt] (0) at (-3.5,0) {0};
		\node[circle, draw=black, fill=black, inner sep=1pt,minimum size=5pt] (1) at (-3.5,1) {1};
		\node[circle, draw=black, fill=black, inner sep=1pt,minimum size=5pt] (2) at (-2.5,0) {2};
    	\draw[postaction={decorate}, >=stealth] (0)  to (1);
    	\draw[postaction={decorate}, >=stealth] (0)  to (2);
   		\node[left] at (-3.5,0.5) {$j$};
   		\node[left] at (-3.6,0) {$g$};
   		\node[right] at (-2.4,0) {$g i$};
   	    \node[below] at (-3.0,0) {$i$};
   	    \node[above] at (-3.5,1.1) {$g j$};
        \node[below] at (-3.0,-0.5) {Forbidden filling};
        \node[circle, draw=black, fill=black, inner sep=1pt,minimum size=5pt] (3) at (-0.6,0) {3};
		\node[circle, draw=black, fill=white, inner sep=1pt,minimum size=10pt] (4) at (-0.6,1) {\textcolor{white}{4}};
		\node[circle, draw=black, fill=black, inner sep=1pt,minimum size=5pt] (5) at (0.4,0) {5};
    	\draw[postaction={decorate}, >=stealth] (3)  to (4);
    	\draw[postaction={decorate}, >=stealth] (3)  to (5);
   		\node[left] at (-0.6,0.5) {$j$};
   		\node[left] at (-0.7,0) {$g$};
   		\node[right] at (0.5,0) {$g i$};
   	    \node[below] at (0,0) {$i$};
   	    \node[above] at (-0.6,1.1) {$g j$};
        \node[circle, draw=black, fill=black, inner sep=1pt,minimum size=5pt] (6) at (1.5,0) {6};
		\node[circle, draw=black, fill=white, inner sep=1pt,minimum size=10pt] (7) at (1.5,1) {\textcolor{white}{7}};
		\node[circle, draw=black, fill=white, inner sep=1pt,minimum size=10pt] (8) at (2.5,0) {\textcolor{white}{7}};
    	\draw[postaction={decorate},>=stealth] (6)  to (7);
    	\draw[postaction={decorate}, >=stealth] (6)  to (8);
   		\node[left] at (1.5,0.5) {$j$};
   		\node[left] at (1.4,0) {$g$};
   		\node[right] at (2.6,0) {$g i$};
   	    \node[below] at (2,0) {$i$};
   	    \node[above] at (1.5,1.1) {$g j$};
   	    \node[below] at (1,-0.5) {Allowed fillings};
		\end{tikzpicture}
	\end{center}
	\caption{Representation of the rule $(R3)$.\label{fig:Rule_3}}
\end{figure}

\begin{figure}[H]
\begin{center}
		\begin{tikzpicture}[scale=1.5,decoration={markings, mark=at position 0.5 with {\arrow{>}}}]
		\node[circle, draw=black, fill=black, inner sep=1pt,minimum size=5pt] (0) at (0,0) {0};
		\node[circle, draw=black, fill=black, inner sep=1pt,minimum size=5pt] (1) at (0,1) {1};
		\node[circle, draw=black, fill=black, inner sep=1pt,minimum size=5pt] (2) at (1,1) {2};
    	\draw[postaction={decorate}, >=stealth] (0)  to (1);
    	\draw[postaction={decorate}, >=stealth] (1)  to (2);
   		\node[left] at (0,0.5) {$i$};
   	    \node[above] at (0.5,1) {$j$};
   	    \node[above] at (0,1.1) {$g$};
        \node[above] at (1,1.1) {$gj$};
        \node[above] at (0.5,1.5) {$A(i,j)=1$};
       \node[circle, draw=black, fill=white, inner sep=1pt,minimum size=10pt] (3) at (3,0) {};
		\node[circle, draw=black, fill=black, inner sep=1pt,minimum size=5pt] (4) at (3,1) {4};
		\node[circle, draw=black, fill=black, inner sep=1pt,minimum size=5pt] (5) at (4,1) {5};
    	\draw[postaction={decorate},>=stealth] (3)  to (4);
    	\draw[postaction={decorate}, >=stealth] (4)  to (5);
   		\node[left] at (3,0.5) {$i$};
   	    \node[above] at (3.5,1) {$j$};
   	    \node[above] at (3,1.1) {$g$};
        \node[above] at (4,1.1) {$gj$};
        \node[above] at (3.5,1.5) {$A(i,j)=0$};
		\end{tikzpicture}
	\end{center}
	\caption{ Representation of the rule $(R4)$. \label{fig:Rule_4}}
\end{figure}

\begin{remark} The set $\Omega_A^\tau$ is a compact subset of $\{0,1\}^{\mathbb{F}}$ that contains $\widetilde{X}_A$ (see \cite{EL1999,Raszeja2020}).
\end{remark}

Now, we introduce the definition of stem and root of a configuration as in \cite{EL1999}.

\begin{definition} By a \textit{positive word} in $\mathbb{F}$ we mean any finite or infinite sequence $\omega = \omega_0 \omega_1 \cdots$ satisfying $\omega_j \in \mathbb{F}_+$ for every $j$, including the empty word $e$. As well as it is defined for the classical Markov shift spaces, a positive word $\omega$ in $\mathbb{F}$ is said to be admissible when $A(\omega_j,\omega_{j+1}) = 1$ for every $j$. Given an either finite or infinite positive word $\omega = \omega_0 \omega_1 \cdots$, define the set
\begin{equation*}
    \llbracket \omega \rrbracket := \{e,\omega_0,\omega_0 \omega_1,\omega_0\omega_1\omega_2,\cdots\}
\end{equation*}
of the subwords of $\omega$. 
\end{definition}

\begin{remark} Observe that if $\omega$ is an infinite positive word, then $\omega \notin \llbracket \omega \rrbracket$.
\end{remark}

\begin{definition}[stems and roots] Let $\xi \in \Omega_A^\tau$. The \textit{stem} of $\xi$, denoted by $\kappa(\xi)$, is the positive admissible word such that
\begin{equation*}
    \{g \in \mathbb{F}: \xi_g = 1\} \cap \mathbb{F}_+ = \llbracket \omega \rrbracket.
\end{equation*}
We say that a configuration $\xi \in \Omega_A^\tau$ is a \textit{bounded element} if its stem has finite length. If $\xi$ is not bounded we call it \textit{unbounded}.

Given $g \in \mathbb{F}$ s.t. $\xi_g = 1$, the \textit{root} of $g$ relative to $\xi$, denoted by $R_\xi(g)$, is defined by
\begin{equation*}
    R_\xi(g) := \{x \in \mathbb{N}: \xi_{g x^{-1}} = 1\}.
\end{equation*}
\end{definition}

\begin{remark} By Proposition 5.4 of \cite{EL1999}, the stem of a configuration in $\Omega_A^\tau$ is unique. Also, observe that if $\kappa(\xi) \neq e$, then $\kappa_{|\kappa|-1} \in R_\xi(\kappa(\xi))$. 
\end{remark}

For a configuration $\xi \in \Omega_A^\tau$ and a given element of $g \in \mathbb{F}$ such that $\xi_g = 1$, we can define the future and the past of $g$ in $\xi$. The \textit{future of }$g$ in $\xi$ is the set $\{g\omega \in \mathbb{F}: \omega \in \mathbb{F}_+, \xi_{g\omega} = 1\}$, which corresponds to the unique filled path from $g$ in the same orientation of the edges of the Cayley tree. The \textit{past of }$g$ in $\xi$ is the set $\{g\omega^{-1} \in \mathbb{F}: \omega \in \mathbb{F}_+, \xi_{g\omega} = 1\}$, corresponding to every filled path that ends in $g$ by following the same orientation of the edges of the Cayley tree.

The stem of a configuration corresponds to the longest positive word such that every of its subwords are filled in the configuration. These subwords are precisely the future of $e$, the identity of $\mathbb{F}$, and $e$ itself. On the other hand, the root of an element of $g \in \mathbb{F}$ on a configuration corresponds to the edges of the immediate past of $g$, that is, the edges that arrives in $g$ such that the words in the source of these edges are filled. By section 5 of \cite{EL1999} the inclusion $i_2 \circ i_1 : \Sigma_A \to \{0,1\}^{\mathbb{F}}$ defines a bijection between the elements of $\Sigma_A$ and the unbounded configurations of $\Omega_A^\tau$, when restricted to its image. Furthermore, there is a bijection between the set of bounded configurations of $\Omega_A^\tau$ and the set of pairs $(\kappa,R)$, where $\kappa$ is a finite positive admissible word and $R \subseteq \mathbb{N}$ satisfying $\kappa_{|\kappa|-1} \in R$. The Corollary 7.7 of \cite{EL1999} characterizes the bounded elements of $X_A$, which we denote by $Y_A$, that is,
\begin{equation*}
    Y_A:= X_A \setminus \Sigma_A,
\end{equation*}
where we omitted the inclusions $i_1$ and $i_2$. We state this result next.

\begin{definition} Given a transition matrix $A$, we say that a natural number $i \in \mathbb{N}$ is an \emph{infinite emitter} when the set $r_A(i)=\{ j \in  \mathbb{N}; A(i,j)=1\}$ is infinite. 
\end{definition}

From now we shall use the following notation: we denote by $\mathfrak{C}(A)$ the set of columns of $A$, also the elements of $\mathfrak{C}(A)$ are denoted as follows: given $j \in \mathbb{N}$ we denote by $\mathfrak{c}(j)$ the $j$-th column of the matrix $A$, that is,
\begin{equation*}
    \mathfrak{c}(j)_i = \begin{cases}
                1, \quad \text{if } A(i,j) = 1;\\
                0, \quad \text{otherwise.}
             \end{cases}
\end{equation*}

\begin{proposition}\label{prop:characterization_of_elements_of_Y_A} Let $\xi \in \Omega_A^\tau$. Then $\xi \in Y_A$ if and only if $\xi \neq \varphi_0$ and it satisfies simultaneously the two following conditions:
\begin{itemize}
    \item[$(a)$] $\kappa(\xi) = e$ or $\kappa(\xi)$ ends in an infinite emitter; 
    \item[$(b)$] If $\kappa(\xi)\neq e$, $R_\xi(\kappa(\xi))$ is a limit point of the sequence $(\mathfrak{c}(i))_{A(\kappa(\xi)_{|\kappa(\xi)|-1},i)=1}$, where the root $R_\xi(\kappa(\xi))$ can be seen as an element of $\{0,1\}^\mathbb{N}$ as follows: $R_\xi(\kappa(\xi))=(x_i)_{i \in \mathbb{N}} \in \{0,1\}^\mathbb{N}$ where $x_i=1$ iff $x_i \in R_\xi(\kappa(\xi))$. If $\kappa(\xi)= e$, $R_\xi(\kappa(\xi))$ is a limit point of the sequence $(\mathfrak{c}(i))_{i \in \mathbb{N}}$.
\end{itemize}
\end{proposition}

\begin{proof} It follows from Corollary 7.5 of \cite{EL1999} and Proposition \ref{prop:i_1_Sigma_A_dense}.
\end{proof}
In order to provide more intuition to the reader, we discuss how the convergence of elements of $\Sigma_A$ to elements of $Y_A$ occurs. Let $\xi \in Y_A$. There must exist a sequence $(\xi^n)_{n \in \mathbb{N}}$ in $\Sigma_A$ that converges to $\xi$. Let $\omega$ be the stem of $\xi$. We must have necessarily that $\xi^n_{\omega j} \to 0$ for every $j \in \mathbb{N}$. In particular, we only need to consider those $j$ s.t. $\omega j$ is admissible. Then we must have for each one of these particular $j$'s that $\xi^n_{\omega j} = 1$ for a finite quantity of $n$'s. So w.l.o.g. we may suppose that $\xi^n_{\omega j_n}$ for each $n$ and that $j_n \neq j_m$ if $n \neq m$. So, we must have an infinite different possible $j_k$'s s.t. $\omega j_k$ is admissible, that is, $\omega = e$ or $\omega$ ends in an infinite emitter. The figure \ref{fig:sequence_of_Sigma_A_converging_to_Y_A} illustrates this discussion.
\begin{figure}[H]
\centering
\begin{tikzpicture}
\node[circle,fill=black,inner sep=0pt,minimum size=8pt,label=below:{$\omega$}]      (maintopic)                              {};
\node[circle,fill=black,minimum size=8pt, label=right:{$\omega j_1$}]        (j1)       [above right= 1.5cm and 1cm of maintopic] {};
\node[circle,fill=white,draw,minimum size=8pt,label=right:{$\omega j_2$}]        (j2)       [above right= 1cm and 1.1cm of maintopic] {};
\node[circle,fill=white,draw,minimum size=8pt,label=right:{$\omega j_3$}]        (j3)       [above right= 0.5cm and 1.2cm of maintopic] {};
\node[circle,fill=white,draw,minimum size=8pt,label=right:{$\omega j_4$}]        (j4)       [above right= 0cm and 1.3cm of maintopic] {};

\draw[->,shorten >=0.1cm](maintopic.north) edge[bend left=45] (j1.west);
\draw[->,shorten >=0.1cm](maintopic.north) edge[bend left=45] (j2.west);
\draw[->,shorten >=0.1cm](maintopic.north) edge[bend left=45] (j3.west);
\draw[->,shorten <=0.7cm,ultra thick](j3.east) -- ($(j3.west)+(0:2.2)$);
\draw[->,shorten >=0.1cm](maintopic.north) edge[bend left=45] (j4.west);

\node[circle,fill=black,inner sep=0pt,minimum size=8pt,label=below:{$\omega$}]      (maintopic2) [right= 3.5cm of maintopic]                              {};
\node[circle,fill=white,draw,minimum size=8pt, label=right:{$\omega j_1$}]        (j12)       [above right= 1.5cm and 1cm of maintopic2] {};
\node[circle,fill=black,draw,minimum size=8pt,label=right:{$\omega j_2$}]        (j22)       [above right= 1cm and 1.1cm of maintopic2] {};
\node[circle,fill=white,draw,minimum size=8pt,label=right:{$\omega j_3$}]        (j32)       [above right= 0.5cm and 1.2cm of maintopic2] {};
\node[circle,fill=white,draw,minimum size=8pt,label=right:{$\omega j_4$}]        (j42)       [above right= 0cm and 1.3cm of maintopic2] {};

\draw[->,shorten >=0.1cm](maintopic2.north) edge[bend left=45] (j12.west);
\draw[->,shorten >=0.1cm](maintopic2.north) edge[bend left=45] (j22.west);
\draw[->,shorten >=0.1cm](maintopic2.north) edge[bend left=45] (j32.west);
\draw[->,shorten <=0.7cm,ultra thick](j32.east) -- ($(j32.west)+(0:2.2)$);
\draw[->,shorten >=0.1cm](maintopic2.north) edge[bend left=45] (j42.west);

\node[circle,fill=black,inner sep=0pt,minimum size=8pt,label=below:{$\omega$}]      (maintopic3) [right= 3.5cm of maintopic2]                              {};
\node[circle,fill=white,draw,minimum size=8pt, label=right:{$\omega j_1$}]        (j13)       [above right= 1.5cm and 1cm of maintopic3] {};
\node[circle,fill=white,draw,minimum size=8pt,label=right:{$\omega j_2$}]        (j23)       [above right= 1cm and 1.1cm of maintopic3] {};
\node[circle,fill=black,draw,minimum size=8pt,label=right:{$\omega j_3$}]        (j33)       [above right= 0.5cm and 1.2cm of maintopic3] {};
\node[circle,fill=white,draw,minimum size=8pt,label=right:{$\omega j_4$}]        (j43)       [above right= 0cm and 1.3cm of maintopic3] {};

\draw[->,shorten >=0.1cm](maintopic3.north) edge[bend left=45] (j13.west);
\draw[->,shorten >=0.1cm](maintopic3.north) edge[bend left=45] (j23.west);
\draw[->,shorten >=0.1cm](maintopic3.north) edge[bend left=45] (j33.west);
\draw[->,shorten <=0.7cm,ultra thick](j33.east) -- ($(j33.west)+(0:2.2)$);

\draw[->,shorten >=0.1cm](maintopic3.north) edge[bend left=45] (j43.west);


\node[circle,fill=black,inner sep=0pt,minimum size=8pt,label=below:{$\omega$}]      (maintopic4) [right= 3.5cm of maintopic3]                              {};
\node[circle,fill=white,draw,minimum size=8pt, label=right:{$\omega j_1$}]        (j14)       [above right= 1.5cm and 1cm of maintopic4] {};
\node[circle,fill=white,draw,minimum size=8pt,label=right:{$\omega j_2$}]        (j24)       [above right= 1cm and 1.1cm of maintopic4] {};
\node[circle,fill=white,draw,minimum size=8pt,label=right:{$\omega j_3$}]        (j34)       [above right= 0.5cm and 1.2cm of maintopic4] {};
\node[circle,fill=black,draw,minimum size=8pt,label=right:{$\omega j_4$}]        (j44)       [above right= 0cm and 1.3cm of maintopic4] {};
\node     (dots)       [right= 2.0cm of j34] {$\bullet\bullet\bullet$};

\draw[->,shorten >=0.1cm](maintopic4.north) edge[bend left=45] (j14.west);
\draw[->,shorten >=0.1cm](maintopic4.north) edge[bend left=45] (j24.west);
\draw[->,shorten >=0.1cm](maintopic4.north) edge[bend left=45] (j34.west);
\draw[->,shorten <=0.7cm,ultra thick](j34.east) -- ($(j34.west)+(0:2.2)$);
\draw[->,shorten >=0.1cm](maintopic4.north) edge[bend left=45] (j44.west);
\end{tikzpicture}
\caption{A sequence of unbounded elements of $X_A$ converging to an element of $Y_A$ viewed close to the stem $\omega$ of the limit configuration. The black dots represents the filled vertices and the white dots represents the non-filled vertices. The oriented edge from the vertex $\omega$ to $\omega j_k$, $k \in \mathbb{N}$ represents the multiplication of $\omega$ by $j_k$. \label{fig:sequence_of_Sigma_A_converging_to_Y_A}}

\end{figure}

Also, we have the density of $Y_A$ on $X_A$, as proved next.

\begin{proposition}\label{prop:Y_A_dense} If $Y_A$ is non-empty, then it is dense in $X_A$.
\end{proposition}

\begin{proof} We recall that  there exists a symbol $i \in \mathbb{N}$ such that $|\{j \in \mathbb{N}: A(i,j) = 1\}| = \infty$. Let $\xi \in \Sigma_A$, that is, $\kappa(\xi) = x = x_0 x_1 x_2 \cdots \in \Sigma_A$. By transitivity of $\Sigma_A$, there exists an admissible word $w_k$ such that $x_k w_k i$ is also admissible. Since $A(i,j) = 1$ is satisfied for infinite values of $j$, by the density of elements of $\Sigma_A$, we may construct for each $n \in \mathbb{N}_0$ at least one $\xi^n \in Y_A$ satisfying $\kappa(\xi^n) = x_0 \cdots x_n w^n i$. For each $n$ choose one of the possible $\xi^n$ as previously. We claim that $\xi^n \to \xi$. Indeed, for every $p \in \mathbb{N}_0$ the sequences $((\xi^n)_{x_0\cdots x_p})_{n \in \mathbb{N}_0}$ are constant for $n>p$ and therefore convergent, and moreover $(\xi^n)_{x_0\cdots x_p} \to 1 = \xi_{x_0\cdots x_p}.$

Consequently for every $g \in \mathbb{F}$ in the form $g = \alpha \gamma^{-1}$ or $\gamma^{-1}$, where $\alpha$ and $\gamma$ are admissible words, we have that $(\xi^n)_g$ is constant for $n > |\alpha|$ (we consider $|\alpha| = 0$ for $g = \gamma^{-1}$). So  $(\xi^n)_{g} \to \xi_g,$ and therefore $\xi^n \to \xi$. 
\end{proof}
We warn the reader that, although $ Y_A $ is countable in many cases, this is not guaranteed. We can have $Y_A$ uncountable even under the transitivity hypothesis, as we show in the next example. Furthermore, the next example shows a transition matrix $A$ s.t. $\Sigma_A$ is topologically mixing, and it is not locally compact and, also, $X_A$ is not compact. 

\begin{figure}[H]
 \centering
 \includegraphics[scale = .3]{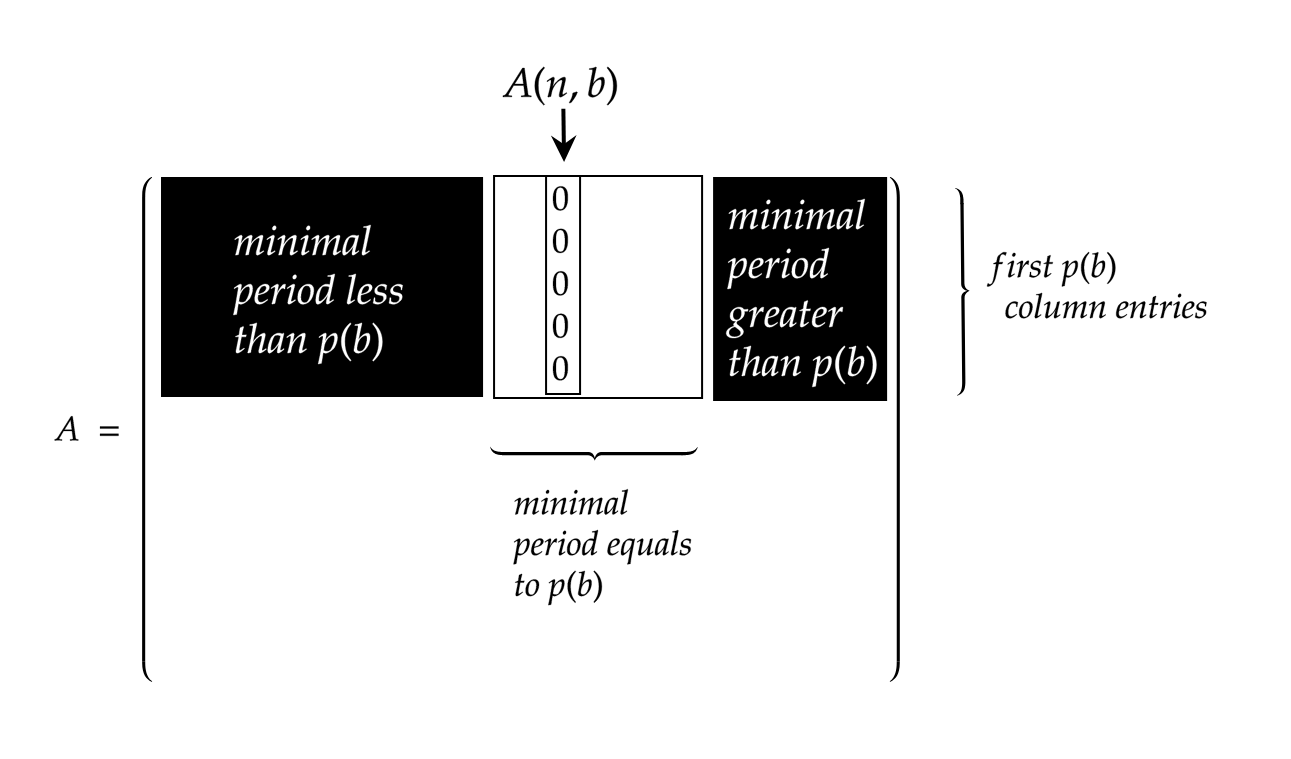}
 \caption{A contradiction on the matrix constructed in this example if, for some $b >2$, we would have $A(c,b) = 0$ for every $c \leq p(b)$. In this case, the first $p(b)$ terms of $A(n,b)$ must be zero, and by periodicity, we conclude that $A(n,b) = 0$, for every $n \in \mathbb{N}$.\label{fig:uncountable_Y_A}}
\end{figure}

\begin{example}[Uncountable $Y_A$] \label{exa:uncountable_Y_A} Consider the transition matrix $A$ as follows: each of its columns is a periodic sequences in $\{0,1\}^\mathbb{N}$, excluding the zero sequence. Also, we order the columns by increasing minimal period. Set
\begin{align*}
    A(n,1)= \begin{pmatrix}
             1\\
             1\\
             1\\
             1\\
             1\\
             1\\
             \vdots
            \end{pmatrix}, \quad
    A(n,2)= \begin{pmatrix}
             1\\
             0\\
             1\\
             0\\
             1\\
             0\\
             \vdots
            \end{pmatrix}
    \quad \text{ and } \quad
    A(n,3)= \begin{pmatrix}
             0\\
             1\\
             0\\
             1\\
             0\\
             1\\
             \vdots
            \end{pmatrix}.
\end{align*}
And for the remaining columns with same minimal period, take any ordering. Denote by $p(n)$, $n \in \mathbb{N}$, the minimal period of the sequence of the $n$-th column. It is straightforward that $p(n) < n$ for every $n > 2$. We prove the following claims.

\textit{Claim: for $b \in \mathbb{N}$, $b > 2$, there exists $c \in \mathbb{N}$ satisfying $A(c,b) = 1$ and $c \leq p(b)$.} In fact, suppose that for every $c \leq p(b)$ we have $A(c,b) = 0$. Then we have that the $b$-th column has only zeros in its first $p(b)$-entries as it is shown in Figure \ref{fig:uncountable_Y_A}, and since $p(b)$ is the period of such sequence, it follows that $A(n,b) = 0$ for every $n \in \mathbb{N}$, a contradiction because every column of $A$ is not the zero column.

\textit{Claim: $A$ is irreducible.} Indeed, let $a,b\in \mathbb{N}$. The case $b=1$ and $b=2$ are clear: since for $b=1$, we have $A(a,b)=1$ for every $a$ because $A(n,1)$ is the column where all of its entries are equal $1$. On the other hand if $b=2$, then $A(1,b)=1$ and $A(a,1)=1$, and so $a1b$ is admissible. 

Now, suppose $b>2$. By the previous claim, there exists $c_1 \in \mathbb{N}$ $c_1 \leq p(b) < b$ s.t. $A(c_1,b) = 1$. If $c_1\leq 2$, then $a1c_1b$ is admissible. Otherwise, then we restart the process, by taking $c_2<p(c_1)<c_1$ such that $A(c_2,c_1)=1$ and verifying if $c_2 \leq 2$ or not. Since the process must terminate, after $m \in \mathbb{N}$ iterations, we find $c_1,\dots,c_m \in \mathbb{N}$, with $c_m\leq 2$, and s.t. $a1c_nc_{n-1}\cdots c_1b$ is admissible.

\textit{Claim: $A$ is topologically mixing.} It is straightforward from the last claim and the fact that $A(1,1) = 1$. In fact, by the proof above, for every $a,b \in \mathbb{N}$, if $b = 1$, then, for every $n \in \mathbb{N}$ we have the admissible word $a1^n 1$, where
\begin{equation*}
    1^n = \underbrace{1\cdots 1}_{n \text{ times}}.
\end{equation*}
Now, if $b = 2$, then $a1b$ is admissible. Again, we have that $a1^nb$ is admissible for every $n \in \mathbb{N}$. If $b >2$, by the proof of last claim, there exists $m \in \mathbb{N}$ s.t. $a1^nc_m \cdots c_1 b$ is admissible for every $n \in \mathbb{N}$, and we conclude that $A$ is topologically mixing.

\textit{Claim: the limit points of the sequence $(\mathfrak{c}(n))_{\mathbb{N}}$, where $\mathfrak{c}(n)$ is the $n$-th column of $A$ is a uncountable set.} This is straightforward from the fact that, for every finite sequence $a_1 \cdots a_n$ of elements in $\{0,1\}$, there is an infinite number of columns of $A$ s.t. that start with the sequence $a_1 \cdots a_n$.

\textit{Claim: $\Sigma_A$ and $X_A$ are not compact.} For $\Sigma_A$ is straightforward, observe for instance, that $A(n,1) = 1$ for infinite quantity of $n$'s. By the last claim, every element in $\{0,1\}^\mathbb{N}$ is a limit point of $(\mathfrak{c}(n))_{\mathbb{N}}$, so in particular the null vector is it, and by Theorem \ref{thm:O_A_unital} $(iv)$, we have that $\mathcal{O}_A$ is not unital, then $\mathcal{D}_A$ is not unital, and therefore $X_A$ is not compact. 
\end{example}


Now, we prove the compatibility between the classical and GCMS spaces in terms of measure theory.

\begin{proposition}\label{prop:Sigma_A_Borel_subset_in_X_A}
$i_2\circ i_1(\Sigma_A)$ is a measurable set in the Borel $\sigma$-algebra of $i_2(X_A)$.
\end{proposition}
\begin{proof} Using the fact that $i_2$ is a homeomorphism,
\begin{equation*}
  i_2(X_A\setminus i_1(\Sigma_A)) = i_2(X_A)\setminus i_2\circ i_1(\Sigma_A) =\{\xi\in i_2(X_A):|\kappa(\xi)|<\infty\}=\bigcup_{\alpha\in \mathfrak{W}_0}\{\xi\in i_2(X_A):\kappa(\xi)=\alpha\},  
\end{equation*}
where $\mathfrak{W}_0$ is the countable set of all admissible finite words, including the identity $e$. Note that
\begin{align*}
    H(\alpha):&=\{\xi\in i_2(X_A):\kappa(\xi) =\alpha\} = \{\xi\in i_2(X_A):\xi_\alpha=1 \text{ and } \xi_{\alpha s}=0,\text{ }\forall s\in \mathbb{N}\} \\
    &= \{\xi\in i_2(X_A):\xi_\alpha=1\} \cap \{\xi\in i_2(X_A): \xi_{\alpha s}=0,\text{ }\forall s\in \mathbb{N} \}.
\end{align*}
Since $H(\alpha)$ is an intersection of two closed sets, this means that $H(\alpha)$ is closed in $i_2(X_A)$.  As $i_2(X_A)\setminus i_2\circ i_1(\Sigma_A)$ is an countable union of those sets, we conclude $i_2(X_A)\setminus i_2\circ i_1(\Sigma_A)$ is a $F_\sigma$, hence $i_2\circ i_1(\Sigma_A)$ is a $G_\delta$, a Borel set. 
\end{proof}

\begin{proposition}\label{prop:Borel_sets_preserved_from_Sigma_A_to_X_A}For every Borel set $B \subseteq \Sigma_A$, $i_1(B)$ is a Borel set in $X_A$.
\end{proposition} 

\begin{proof} It is equivalent to prove that $i_2 \circ i_1(B)$ is a Borel set. Also, it is sufficient to prove the result for the cylinders in $\Sigma_A$ because they do form an countable basis of the topology. Given a cylinder set $[\alpha]\subseteq \Sigma_A$, we have that $i_2 \circ i_1([\alpha]) = \{\xi \in i_2(X_A): \kappa(\xi) \text{ is infinite and $\xi_\alpha=1$}\}$.

Denoting by $\pi_g$ the canonical projection, see Definition \ref{projection}, it follows that
\begin{equation*}
    i_2 \circ i_1([\alpha]) = \left(\bigcap_{\gamma \in \llbracket\alpha\rrbracket}\pi_{\gamma}^{-1}(\{1\})\right)\cap i_2\circ i_1(\Sigma_A)
\end{equation*}
is a Borel set in $i_2(X_A)$. 
\end{proof}

Since the Borel sets of $\Sigma_A$ are preserved in the sense that they are also Borel sets of $X_A$ we are able to see Borel measures of $\Sigma_A$ as measures of $X_A$ that gives mass only for $\Sigma_A$. Conversely, when we restrict Borel measures on $X_A$ to $\Sigma_A$ we obtain a Borel measure on $\Sigma_A$. This compatibility of Borel $\sigma$-algebras and measures is crucial for a stronger result: the compatibility of conformal measures on classical and generalized Markov shift spaces. This will be proved in subsection \ref{sec:TF_Generalized}.

We now introduce the generalized shift map for $X_A$, inherited from $\Sigma_A$.

\begin{definition}\label{generalized_shift_map}[Generalized shift map] Consider the generalized Markov shift space. The \textit{generalized shift map} is the function $\sigma :X_A \setminus \{\xi \in X_A:\kappa(\xi) = e\} \to X_A$ defined by
\begin{equation*}
    (\sigma(\xi))_g = \xi_{x_0^{-1} g}, \quad g \in \mathbb{F},
\end{equation*}
where $x_0 = \kappa(\xi)_0$, that is, $x_0$ is the first symbol of the stem of $\xi$. 
\end{definition}

Observe that $\sigma$ is simply a translation on the configurations of the Cayley tree of generated by $\mathbb{N}$. It is straightforward that such map is well-defined surjective continuous map. More than that, $\sigma$ is a local homeomorphism.

\subsection{$Y_A$-families} \label{subsec:Y_A_families}\\

Next, we define the $Y_A$-families for $X_A$, a notion that will be used later to characterize extremal conformal measures on $X_A$ that vanish on $\Sigma_A$, in the next section.

\begin{notation}
From now on we denote by $\mathscr{E}$ the set of empty-stem configurations of $X_A$.
\end{notation}

\begin{definition}[$Y_A$-family] Let $\{\xi^{0,\mathfrak{e}}\}_{\mathfrak{e} \in \mathscr{E}}$, the collection of all configurations on $X_A$ (or $\widetilde{X}_A$) that have empty stem. For each $\xi^{(0,\mathfrak{e})}$, we define the \textit{$Y_A$-family} of $\xi^{(0,\mathfrak{e})}$ as the set
\begin{equation}
    Y_A(\xi^{0,\mathfrak{e}}):= \bigsqcup_{n \in \mathbb{N}_0} \sigma^{-n}(\xi^{0,\mathfrak{e}}).
\end{equation}
\end{definition}


We prove some properties of the $Y_A$-families.

\begin{proposition}[Properties of the $Y_A$-families] \label{prop:properties_of_Y_A_families} Consider the family $\{\xi^{0,\mathfrak{e}}\}_{\mathfrak{e} \in \mathscr{E}}$ of all distinct configurations in $Y_A$ with empty stem. The following properties hold:
\begin{itemize}
    \item[$(i)$] $Y_A(\xi^{0,\mathfrak{e}}) \cap Y_A(\xi^{0,\mathfrak{e}'}) = \emptyset$ whenever $\mathfrak{e} \neq \mathfrak{e}'$;
    \item[$(ii)$] $Y_A = \bigsqcup_{\mathfrak{e} \in \mathscr{E}}Y_A(\xi^{0,\mathfrak{e}})$;
    \item[$(iii)$] $Y_A(\xi^{0,\mathfrak{e}})$ is $\sigma$-invariant for all $\mathfrak{e} \in \mathscr{E}$, in the sense that
    \begin{equation*}
        \sigma(Y_A(\xi^{0,\mathfrak{e}}) \cap \Dom\, \sigma) = Y_A(\xi^{0,\mathfrak{e}});
    \end{equation*}
    \item[$(iv)$] fixed $\mathfrak{e} \in \mathscr{E}$, $R_\xi(\kappa(\xi)) = R_{\xi^{0,\mathfrak{e}}}(e)$ for every $\xi \in Y_A(\xi^{0,\mathfrak{e}})$;
    \item[$(v)$] $Y_A(\xi^{0,\mathfrak{e}})$ is countable for every $\mathfrak{e} \in \mathscr{E}$;
    \item[$(vi)$] $Y_A(\xi^{0,\mathfrak{e}})$ is a Borel subset of $X_A$ for every $\mathfrak{e} \in \mathscr{E}$.
\end{itemize}
\end{proposition}

\begin{proof}\\
\begin{itemize}
    \item[$(i)$] Let $\xi \in Y_A(\xi^{0,\mathfrak{e}}) \cap Y_A(\xi^{0,\mathfrak{e}'})$, where $\mathfrak{e} \neq \mathfrak{e}'$. Then $\xi^{0,\mathfrak{e}} = \sigma^n(\xi) = \xi^{0,\mathfrak{e}'} $, where $n = |\kappa(\xi)|$, a contradiction because $\xi^{0,\mathfrak{e}} \neq \xi^{0,\mathfrak{e}'}$.
    \item[$(ii)$] The inclusion
    \begin{equation*}
        Y_A \supseteq \bigsqcup_{\mathfrak{e} \in \mathscr{E}}Y_A(\xi^{0,\mathfrak{e}})
    \end{equation*}
    is straightforward. For the another inclusion, take $\xi \in Y_A$. If $\kappa(\xi) = e$, there is nothing to be proven, so suppose $\kappa(\xi) = \kappa \neq e$. Then $\sigma^{\kappa}(\xi) = \xi^{0,\mathfrak{e}}$ for some $\mathfrak{e} \in \mathscr{E}$ and hence $\xi \in Y_A(\xi^{0,\mathfrak{e}})$.
    \item[$(iii)$] Given $\xi \in Y_A(\xi^{0,\mathfrak{e}}) \cap \Dom \sigma$ and since the translation $\sigma$ preserves the root of the stem, we have that $\sigma(\xi) \in Y_A(\xi^{0,\mathfrak{e}})$.
    \item[$(iv)$] it is straightforward from definition of $Y_A(\xi^{0,\mathfrak{e}})$.
    \item[$(v)$] $Y_A(\xi^{0,\mathfrak{e}})$ is at most a countable union of countable sets and therefore it is countable.
    \item[$(vi)$] Since $X_A$ is Hausdorff, every singleton is closed and then it is a Borel subset of $X_A$. By $(v)$, each $Y_A(\xi^{0,\mathfrak{e}})$ is a countable union of Borel sets, and therefore the $Y_A$-families are Borel subsets of $X_A$. 
\end{itemize}
\end{proof}

By items $(i)$ and $(ii)$ of Proposition \ref{prop:properties_of_Y_A_families}, the collection of $Y_A$-families forms a partition of $Y_A$. Item $(iv)$ states that $Y_A$-families are one-to-one with the set of all roots for elements of $Y_A$, and by consequence, each $Y_A$-family satisfies a \emph{stem uniqueness property}: a configuration $\xi$ belonging to a $Y_A$-family is the unique element in that family whose its stem is the word $\kappa(\xi)$. These facts essentially highlight the bijection between the set $Y_A$ and the set of pairs stem-root $(\kappa,R)$: each  $\xi \in Y_A$ corresponds to a unique pair $(\kappa(\xi),R_\xi(\kappa(\xi)))$, where $\kappa(\xi)$ is the stem and $R_\xi(\kappa(\xi))$ the root of the element $\xi$. Considering this fact, we define the following.


\begin{definition} Given a transition matrix s.t. its respective set of empty-stem configurations $\{\xi^{0,\mathfrak{e}}\}_{\mathfrak{e} \in \mathscr{E}}$ in $Y_A$. We define the sets
\begin{equation*}
    \mathfrak{W}_\mathfrak{e}:= \{\omega: \kappa(\xi) = \omega \text{ for some } \xi \in Y_A(\xi^{0,\mathfrak{e}})\},
\end{equation*}
for each $\mathfrak{e} \in \mathscr{E}$. If $\mathscr{E} = \{\mathfrak{e}\}$, that is, $\mathscr{E}$ is a singleton, we simply write $\mathfrak{W}_\mathfrak{e} = \mathfrak{W}$. 
\end{definition}

We now focus on some particularly special examples. They will be recalled further when we construct the thermodynamic formalism on generalized Markov shift spaces.

\subsection{Renewal shift}
\label{subsec:Generalized_Renewal_shift}

We first characterize the \emph{generalized renewal shift} space. By taking the alphabet $\mathbb{N}$, the renewal transition matrix is given by $A(1,n) = A(n+1,n) = 1$ for every $n \in \mathbb{N}$ and zero in the remaining entries. Its symbolic graph is the following:
\begin{figure}[H]
\[
\begin{tikzcd}
\circled{1}\arrow[loop left]\arrow[r,bend left]\arrow[rr,bend left]\arrow[rrr,bend left]\arrow[rrrr, bend left]\arrow[rrrrr, bend left]\arrow[rrrrrr, bend left]&\circled{2}\arrow[l]&\circled{3}\arrow[l]&\circled{4}\arrow[l]&\circled{5}\arrow[l]&\circled{6}\arrow[l]&\arrow[l]\cdots
\end{tikzcd}
\]
\caption{The renewal shift symbolic graph.\label{fig:renewal_shift_symbolic_graph}}
\end{figure}

\begin{figure}[H]
    \centering
\scalebox{0.6}{

\tikzset{every picture/.style={line width=0.75pt}} 



}
\caption{The empty-stem configuration of the renewal shift. Only the filled vertices are shown. The red vertices are some explicit examples of filled elements of $\mathbb{F}$.\label{figure.emptyconfigurationrenewal}}
\end{figure}

First we observe that the first row of $A$ is full of $1$'s, and then, by Remark \ref{remark:unital_vs_full_row_of_ones}, $\mathcal{O}_A$ is unital and therefore $X_A$ is compact. On the other hand, we observe that $1$ is the unique infinite emitter in the symbolic graph. Also,
\begin{equation*}
    \mathfrak{c}(j)_i = \begin{cases}
                1, \quad \text{if } i \in \{1,j+1\};\\
                0, \quad \text{otherwise.}
             \end{cases}
\end{equation*}
Then, the unique limit point of the set of matrix columns of $A$ is
\begin{equation*}
    c = \begin{pmatrix} 1 \\ 0 \\0 \\ 0 \\\vdots
    \end{pmatrix}.
\end{equation*}
We conclude that there exists a unique $Y_A$-family, where 
\begin{equation*}
    \mathfrak{W} = \{e\} \sqcup \{\omega \text{ positive admissible word}, |\omega| \geq 1, \omega_{|\omega|-1}= 1\}.
\end{equation*}
We denote by $\xi^0$ the unique configuration in $Y_A$ with empty stem, and it is represented in the figure \ref{figure.emptyconfigurationrenewal}.

The configuration of the figure \ref{fig:321_renewal} presents the unique configuration relative to the stem $321$. Its uniqueness is a straightforward consequence of Proposition \ref{prop:properties_of_Y_A_families}.

\begin{figure}
    \centering

\scalebox{0.7}{

\tikzset{every picture/.style={line width=0.75pt}} 



}
\caption{321 stem configuration for the renewal shift. The directed graph denoted by $\xi^0\setminus \{e_1\}$ is 
obtained removing the root $e$ and the edge $e_1$, labeled with $1$, from the directed graph associated to the 
empty configuration $\xi^0$ in Figure \ref{figure.emptyconfigurationrenewal}. \label{fig:321_renewal}}
\end{figure}

Since this shift space is a simple concrete and explicit example, we show explicitly how the shift map acts on it. Figure \ref{fig:321_renewal_shift_action} shows the shift action on the configuration with stem $321$ of Figure \ref{fig:321_renewal}, and the Figure \ref{fig:321_renewal_shift_action_detailed} presents the same shift action in same configuration with focus on the translation of some branches of that configuration. In both pictures, $\eta^0$, jointly with the edge `$1$' and the vertex $e$, is a copy of the empty-stem configuration.

\begin{figure}[H]
    \centering

\scalebox{0.6}{

\tikzset{every picture/.style={line width=0.75pt}} 



}
\caption{ \label{fig:321_renewal_shift_action}}
\end{figure}

\begin{figure}[H]
    \centering

\scalebox{0.5}{

\tikzset{every picture/.style={line width=0.75pt}} 

%

}
\caption{\label{fig:321_renewal_shift_action_detailed}}

\end{figure}

\newpage

\subsection{Pair renewal shift}
\label{subsec:Pair_Renewal}

This shift space corresponds to a slight modification of the renewal shift, by adding another infinite emitter to the symbolic graph, and we described it as follows. Consider the transition matrix satisfying
\begin{equation*}
    A(1,n)=A(2,2n)=A(n+1,n) = 1, \, n\in \mathbb{N},
\end{equation*}
and zero otherwise. This shift will be called from now by \emph{pair renewal shift} and its explicit associated transition matrix is
\begin{equation*}
    A = \begin{pmatrix}
    1 & 1 & 1 & 1 & 1 & 1 & 1 & 1 & 1 & \cdots\\
    1 & 1 & 0 & 1 & 0 & 1 & 0 & 1 & 0 & \cdots\\
    0 & 1 & 0 & 0 & 0 & 0 & 0 & 0 & 0 & \cdots\\
    0 & 0 & 1 & 0 & 0 & 0 & 0 & 0 & 0 & \cdots\\
    0 & 0 & 0 & 1 & 0 & 0 & 0 & 0 & 0 &\cdots\\
    0 & 0 & 0 & 0 & 1 & 0 & 0 & 0 & 0 &\cdots\\
    0 & 0 & 0 & 0 & 0 & 1 & 0 & 0 & 0 &\cdots\\
    0 & 0 & 0 & 0 & 0 & 0 & 1 & 0 & 0 &\cdots\\
    0 & 0 & 0 & 0 & 0 & 0 & 0 & 1 & 0 &\cdots\\
    \vdots & \vdots & \vdots & \vdots & \vdots & \vdots & \vdots & \vdots & \vdots & \ddots
    \end{pmatrix}.
\end{equation*}
Its symbolic graph is the following:
\[
\begin{tikzcd}
\circled{1}\arrow[loop left]\arrow[r,bend left]\arrow[rr,bend left]\arrow[rrr,bend left]\arrow[rrrr, bend left]\arrow[rrrrr, bend left]\arrow[rrrrrr, bend left]&\circled{2}\arrow[loop below]\arrow[l]\arrow[rr,bend right]\arrow[rrrr,bend right]\arrow[rrrrr,bend right]&\circled{3}\arrow[l]&\circled{4}\arrow[l]&\circled{5}\arrow[l]&\circled{6}\arrow[l]&\arrow[l]\cdots
\end{tikzcd}
\]
The current transition matrix $A$ has a full line of $1$`s in the first row and then $\mathcal{O}_A$ is unital as the renewal shift and therefore $X_A$ is compact. The only two possibilities for limit points on the sequence of columns of $A$ are
\begin{equation}\label{eq:limit_columns_pair_renewal_shift}
    c_1 = 
.
\end{equation}
By Proposition \ref{prop:characterization_of_elements_of_Y_A}, there are only two distinct emtpty stem configurations, $\xi^{0,1}$ and $\xi^{0,1}$, and they are presented in the figure \ref{fig:emptyconfigurationpairrenewal}.

\begin{figure}[H]
 \scalebox{.7}{

\tikzset{every picture/.style={line width=0.75pt}} 
{

\tikzset{every picture/.style={line width=0.75pt}} 

%

}}
\caption{The two empty-stem configurations of the pair renewal shift. The left configuration $\xi^{0,1}$ satisfies $R_{\xi^{0,1}}(e) = c_1$, while the right one $\xi^{0,2}$ corresponds to $R_{\xi^{0,2}}(e) = c_2$, where $c_1$ and $c_2$ are the column vectors of \eqref{eq:limit_columns_pair_renewal_shift}. \label{fig:emptyconfigurationpairrenewal}}
\end{figure}

For this GCMS, there are only two $Y_A$-families. Observe that the unique infinite emitters the pair renewal shift symbolic graph are the symbols $1$ and $2$. By definition of $Y_A$-family, the stems of $Y_A(\xi^{0,1})$ and $Y_A(\xi^{0,2})$ are the sets $\mathfrak{W}_1$ and $\mathfrak{W}_2$, respectively, and they given by
\begin{align*}
    \mathfrak{W}_1 &= \{e\} \sqcup \{\omega \text{ positive admissible word}, |\omega| \geq 1, \omega_{|\omega|-1} \in \{1,2\}\},\\
    \mathfrak{W}_2 &= \{e\} \sqcup \{\omega \text{ positive admissible word}, |\omega| \geq 1, \omega_{|\omega|-1} = 1\}.
\end{align*}

\subsection{Prime renewal shift}
\label{subsec:Prime_Renewal_shift}

This example is one step further on generalizing the renewal shift. Take matrix $A$ as follows: for each $p$ prime number and  $n \in \mathbb{N}$ we have
\begin{equation*}
A(n+1,n)=A(1,n) = A(p,p^n) = 1,
\end{equation*}
and zero for the other entries of $A$. From now on, we will call this shift \emph{prime renewal shift}.
The transition matrix is
\begin{equation*}
    A = \begin{pmatrix}
    1 & 1 & 1 & 1 & 1 & 1 & 1 & 1 & 1 & \cdots\\
    1 & 1 & 0 & 1 & 0 & 0 & 0 & 1 & 0 & \cdots\\
    0 & 1 & 1 & 0 & 0 & 0 & 0 & 0 & 1 & \cdots\\
    0 & 0 & 1 & 0 & 0 & 0 & 0 & 0 & 0 & \cdots\\
    0 & 0 & 0 & 1 & 1 & 0 & 0 & 0 & 0 &\cdots\\
    0 & 0 & 0 & 0 & 1 & 0 & 0 & 0 & 0 &\cdots\\
    0 & 0 & 0 & 0 & 0 & 1 & 1 & 0 & 0 &\cdots\\
    0 & 0 & 0 & 0 & 0 & 0 & 1 & 0 & 0 &\cdots\\
    0 & 0 & 0 & 0 & 0 & 0 & 0 & 1 & 0 &\cdots\\
    \vdots & \vdots & \vdots & \vdots & \vdots & \vdots & \vdots & \vdots & \vdots & \ddots
    \end{pmatrix}.
\end{equation*}
Now, consider the space of the columns of $A$ endowed with the product topology. All the possible limit points of sequences of elements of $\mathfrak{C}(A)$ are
\begin{align*}
    c(p) = \begin{pmatrix}
               1\\
               0\\
               0\\
               \vdots\\
               0\\
               1 \quad(p\text{-th coordinate})\\
               0\\
                  \vdots
           \end{pmatrix}
\quad \text{or} \quad
    c(1)  = \begin{pmatrix}
               1\\
               0\\
               0\\
               0\\
               \vdots
             \end{pmatrix}.
\end{align*}
where $p$ is a prime number. 
The prime renewal shift has a transition matrix with a full line of $1$'s in the first row and then $\mathcal{O}_A$ is unital, and hence its respective GCMS $X_A$ is compact. For each $p$ prime number or $1$, let $\xi^{0,p}$ be the empty-stem configuration such that $\xi^{0,p}_{1^{-1}}=\xi^{0,p}_{p^{-1}}=1$. Also, let $\mathfrak{W}_p$ be the set of stems of configurations in $Y_A(\xi^{0,p})$. We have that
\begin{align*}
    \mathfrak{W}_1 &= \{e\} \sqcup \{\omega \text{ positive admissible word}, |\omega| \geq 1, \omega_{|\omega|-1} = 1\}, \text{ for }p =1;\\
    \mathfrak{W}_p &= \{e\} \sqcup \{\omega \text{ positive admissible word}, |\omega| \geq 1, \omega_{|\omega|-1} \in \{1,p\}\}, \text{ for }  p \text{ is a prime number}.
\end{align*}

\section{Topological aspects of GCMS}\label{section_cylinder}

Further in this paper, we construct the thermodynamic formalism for the set $X_A$, and one of the important aspects of this new theory is its connections with the standard one for $\Sigma_A$. For instance, we present examples in sections \ref{new_measures} and \ref{sec:eigenmeasures} where, for a given potential, there exists a critical value of inverse temperature where there are conformal measures or eigenmeasures, probabilities in both cases, that give mass zero for $\Sigma_A$ when the temperature is above such critical point. At the critical value or below, these measures, when they exist, they live on $\Sigma_A$. And the important fact here is that, when you increase the temperature to the critical point, the measures living out of $\Sigma_A$ converges weakly$^*$ to the measure at the critical point, connecting the both standard and generalized thermodynamic formalisms. And these results are consequences of the topological features of the GCMS: we present a basis for the separable metric space $X_A$ that is closed under intersections. The topology of $X_A$ is realized by a generalization of the notion of cylinder set of $\Sigma_A$. The aforementioned properties of the basis for $X_A$ grant that a net of probability Borel measures $\{\mu_\beta\}_{\beta}$ converges another Borel probability $\mu$ when
$\lim_\beta \mu_\beta(\mathcal{B}) = \mu(\mathcal{B})$, for every $\mathcal{B}$ element of the basis (Theorem 8.2.17 of \cite{Bogachev2007}).

There are some differences and similarities between the generalized cylinders and the standard ones. The straightforward similarity, as pointed out above, is that the basis for $X_A$ plays a similar role as the convergence of probability measures on $\Sigma_A$, since a net of probabilities on $\Sigma_A$ converges to another probability on $\Sigma_A$ if, and only if, it converges on the cylinders sets. On the other hand, and unlike the standard cylinders, the generalized cylinders are indexed by elements of $\mathbb{F}$ instead of $\mathbb{F}_+$. However, as we prove in this section, the elements of the basis of $X_A$ always can be written in terms of cylinders on the positive words, jointly with a subset of $Y_A$ that is written in terms of positive words as well.

We recall that the topology on $X_A$ is the subspace topology of the product topology on $\{0,1\}^\mathbb{F}$, this topology is generated by the subbasis $\{\pi_g^{-1}(B)\cap X_A: g \in \mathbb{F}, B \subseteq \{0,1\}\}$, which makes $X_A$ a second countable topological space. This let us introduce the notion of generalized cylinder set as follows. 

\begin{definition}\label{def:generalized_cylinders} The \textit{generalized cylinder sets} are the sets
\begin{equation*}
    C_g := \pi^{-1}_g(\{1\})\cap X_A = \{\xi \in X_A: \xi_g =1\},
\end{equation*}
where $g \in \mathbb{F}$.
\end{definition}

Observe that, for $g \in \mathbb{F}$, we have
\begin{equation*}
    \pi_g^{-1}(B)\cap X_A = \begin{cases}
                                    \emptyset, \quad \text{if }B = \emptyset, \\
                                    C_g, \quad \text{if }B = \{1\}, \\
                                    C_g^c, \quad \text{if }B = \{0\}, \\
                                    X_A, \quad \text{if }B = \{0,1\}.
                                 \end{cases}
\end{equation*}
Then, a topological basis $\mathscr{B}$ for $X_A$ is the collection of finite intersections of elements in
\begin{equation}\label{eq:subbasis_X_A}
   \{C_g:g=\alpha \gamma^{-1}\}\cup\{C_g^c:g=\alpha \gamma^{-1}\}, 
\end{equation}
and it is straightforward that $\mathscr{B}$ is closed under intersections and it is countable. Hence, $\mathscr{B}$ fulfills the topological hypotheses of the Theorem 8.2.17 of \cite{Bogachev2007}. Note that in \eqref{eq:subbasis_X_A}, we removed trivial cylinders: for every $g \in \mathbb{F}$ which is not in the form $\alpha \gamma^{-1}$ for all $\alpha, \gamma$ finite positive admissible words, we have that $\xi_g = 0$ for all $\xi \in X_A$, then $\pi_g^{-1}(B) \in \{\emptyset,X_A\}$ and such terms may be excluded from the subbasis without changes in the topology. In this section, we prove results that simplify the basis even further. First, we show that generalized cylinders of positive words are compact sets.

\begin{theorem}\label{thm:cylinders_are_compact} Let $\alpha$ be a positive admissible word. Then the set $C_\alpha$ is compact on $X_A$.
\end{theorem}

\begin{proof} The case when $X_A$ is compact is straightforward, since $C_\alpha$ is closed, so suppose that $X_A$ is not compact. In this case, it is equivalent to prove that $C_\alpha \subseteq X_A$ is closed on $\widetilde{X}_A$. Since the topology on $\widetilde{X}_A$ is the product topology on $\{0,1\}^\mathbb{F}$ we have that any sequence $(\xi^n)_\mathbb{N}$ in $\widetilde{X}_A$ converges to an element $\xi \in \widetilde{X}_A$ if, and only if, it converges coordinate-wise. Let $(\xi^n)_\mathbb{N}$ be a sequence in $C_\alpha$ converging to $\xi \in \widetilde{X}_A$. Then $\xi_\alpha = \lim_{n \to \infty} (\xi^n)_\alpha = \lim_{n \to \infty} 1 = 1$. Therefore $\xi \in C_\alpha$, and then $C_\alpha$ is closed in $\widetilde{X}_A$, that is, $C_\alpha$ is a compact subset of $X_A$. 
\end{proof}

\begin{proposition}\label{prop:cylinder_reducing_inverse} For any $\alpha, \gamma$ positive admissible words with $\gamma= \gamma_0 \cdots \gamma_{|\gamma|-1}$, $|\gamma|>1$, we have that $C_{\alpha \gamma^{-1}}=C_{\alpha \gamma_{|\gamma|-1}^{-1}}$.
\end{proposition}

\begin{proof} The inclusion $C_{\alpha \gamma^{-1}}\subseteq C_{\alpha \gamma_{|\gamma|-1}^{-1}}$ is a direct consequence from the connectedness property and the fact that $\xi_e = 1$ for every configuration in $X_A$, because they imply that, if $g \in \mathbb{F}$ is filled in $\xi$, then every subword of $g$ is also filled. The opposite inclusion holds by rule $(R4)$ in Definition \ref{def:set_Omega_A_tau}, by taking $g = \alpha \gamma_{|\gamma|-1}^{-1}$, $j = \gamma_{|\gamma|-1}^{-1}$, and $i = \gamma_{|\gamma|-2}$. The rest of the proof follows by using $(R4)$ for the remaining vertices: take $g = \alpha (\gamma_{p} \cdots \gamma_{|\gamma|-1})^{-1}$, $j = \gamma_{p}$ and $i = \gamma_{p-1}$ for $p \in \{1,\dots, |\gamma|-2\}$ as in Definition \ref{def:set_Omega_A_tau}.
\end{proof}

In order to proceed with more simplifications of the basis, we define the following.

\begin{definition} Given positive admissible words $\alpha$ and $\gamma$ we define the following finite sets:
\begin{align*}
    F_\alpha &:= \{\xi \in X_A : \kappa(\xi) \in  \llbracket \alpha \rrbracket\setminus\{\alpha\}\},\\
    F_\alpha^* &:= \{\xi \in X_A : \kappa(\xi) \in  \llbracket \alpha \rrbracket\},\\
    F_\gamma^\alpha &:= \{\xi \in F_\gamma: \alpha \in \llbracket \kappa(\xi) \rrbracket\}.
\end{align*}
In addition, given $H,I \subseteq \mathbb{N}$, we define the sets
\begin{align*}
    G(\alpha,H) &:=\{\xi \in X_A : \kappa(\xi)=\alpha, \xi_{\alpha j^{-1}}=1, \forall j \in H\},\\
    K(\alpha,H) &:=\{\xi \in X_A : \kappa(\xi)=\alpha, \xi_{\alpha j^{-1}}=0, \forall j \in H\},\\
    GK(\alpha,H,I) &:= G(\alpha,H) \cap K(\alpha,I),
\end{align*}
with the particular notations $G(\alpha,j):=G(\alpha,\{j\})$, $K(\alpha,j):= K(\alpha,\{j\})$, and $GK(\alpha,j,k) := GK(\alpha,\{j\},\{k\})$, where $j,k \in \mathbb{N}$.
Finally, for $m \in \{0,1,..., |\alpha|-1\}$, define the words
\begin{equation*}
    \delta^m(\alpha) = \begin{cases}
                        e, \quad \text{if } m = 0;\\
                        \alpha_0 \cdots \alpha_{m-1}, \quad \text{otherwise}.
                       \end{cases}   
\end{equation*}
\end{definition}

\begin{remark}\label{rmk:F_is_empty} Note that $F_\gamma^\gamma = F_e = \emptyset$ and that $F_\gamma^\alpha = \emptyset$ if $\alpha \notin \llbracket\gamma\rrbracket$. Also, if $H \cap I \neq \emptyset$, then $GK(\alpha,H,I) = \emptyset$.
\end{remark}

\begin{proposition}\label{prop:general_C_alpha_inverse_j} For every $\alpha$ positive admissible word and for every $j \in \mathbb{N}$ s.t. $j \neq \alpha_{|\alpha|-1}$ we have that
\begin{equation*}
    C_{\alpha j^{-1}} = G(\alpha,j)\sqcup \bigsqcup_{k:A(j,k)=1}C_{\alpha k}.
\end{equation*}
\end{proposition}

\begin{proof} The inclusion
\begin{equation*}
    C_{\alpha j^{-1}} \supseteq G(\alpha,j)\sqcup \bigsqcup_{k:A(j,k)=1}C_{\alpha k}
\end{equation*}
is straightforward. For the opposite inclusion, let $\xi \in C_{\alpha j^{-1}}$. If $\kappa(\xi)= \alpha$, then $\xi \in G(\alpha,j)$. Suppose then $\kappa(\xi)\neq \alpha$. By connectedness and $\xi_e=1$ we get $\alpha \in \llbracket\kappa(\xi)\rrbracket\setminus\{\kappa(\xi)\}$, and hence $\alpha k \in \llbracket\kappa(\xi)\rrbracket$ for some $k \in \mathbb{N}$. Since $\xi_{\alpha j^{-1}}=1$, we have necessarily that $A(j,k)=1$. 
\end{proof}


\begin{proposition}\label{prop:general_C_alpha_complement} Let $\alpha$ be a non-empty positive admissible word. Then,
\begin{equation*}
    C_\alpha^c = F_\alpha\sqcup\bigsqcup_{m=0}^{|\alpha|-1}\bigsqcup_{k\neq \alpha_m} C_{\delta^m(\alpha)k}.
\end{equation*}
\end{proposition}
\begin{proof} The inclusion 
\begin{equation*}
    C_\alpha^c \supseteq F_\alpha\sqcup\bigsqcup_{m=0}^{|\alpha|-1}\bigsqcup_{k\neq \alpha_m} C_{\delta^m(\alpha)k}
\end{equation*}
is straightforward. Let $\xi \in C_\alpha^c$. If $\kappa(\xi) \in \llbracket \alpha \rrbracket$ then we necessarily have that $\xi \in F_\alpha$. Suppose that $\kappa(\xi) \notin \llbracket \alpha \rrbracket$. Then there exists $m \in \{0,\dots,|\alpha|-1\}$ and $k \in \mathbb{N}\setminus\{\alpha_m\}$ such that $\delta^m(\alpha)k \in \llbracket\kappa(\xi)\rrbracket$ and therefore $\xi \in C_{\delta^m(\alpha)k}$. 
\end{proof}

\begin{proposition}\label{prop:general_C_alpha_j_inverse_complement} Let $\alpha$ be a non-empty positive admissible word and $j \neq \alpha_{|\alpha|-1}$. Then,
\begin{equation*}
    C_{\alpha j^{-1}}^c = K(\alpha,j) \sqcup F_\alpha \sqcup \left( \bigsqcup_{m=0}^{|\alpha|-1}\bigsqcup_{p\neq \alpha_m}  C_{\delta^m(\alpha )p}\right)\sqcup \bigsqcup_{\substack{p: A(j,p)=0}}  C_{\alpha p}.
\end{equation*}
\end{proposition}
\begin{proof} By Propositions \ref{prop:general_C_alpha_inverse_j} and \ref{prop:general_C_alpha_complement} we obtain
\begin{align*}
    C_{\alpha j^{-1}}^c =  \bigcap_{k:A(j,k)=1} \left(\left(G(\alpha,j)^c\cap F_{\alpha k}\right)\sqcup\bigsqcup_{m=0}^{|\alpha|}\bigsqcup_{p\neq (\alpha k)_m} G(\alpha,j)^c\cap C_{\delta^m(\alpha k)p}\right).
\end{align*}
In order to characterize the sets $G(\alpha,j)^c\cap F_{\alpha k}$ we use the following identity,
\begin{align*}
    G(\alpha,j)^c &= \Sigma_A \sqcup G_1\sqcup G_2 \sqcup K(\alpha,j),\text{ where }G_1 = C_\alpha^c \cap Y_A \quad \text{and}\\
    G_2 &= \left\{\xi \in Y_A:\alpha \in \llbracket\kappa(\xi)\rrbracket\setminus\{\kappa(\xi)\}\right\}.
\end{align*}
It is straightforward to verify that
\begin{equation*}
    F_{\alpha k} \cap \Sigma_A = F_{\alpha k} \cap G_2 = \emptyset, \quad  F_{\alpha k} \cap G_1 = F_\alpha, \quad F_{\alpha k} \cap K(\alpha,j) = K(\alpha,j).
\end{equation*}
Now we use the decomposition
\begin{equation*}
    G(\alpha,j)^c = G_3\sqcup K(\alpha,j),\text{ where  }  G_3 =\left\{\xi \in X_A:\kappa(\xi) \neq \alpha \right\},
\end{equation*}
in order to obtain for $m \in \{0,\dots,|\alpha|\}$ and $p \neq (\alpha k)_m$ the identities
\begin{equation*}
    C_{\delta^m(\alpha k)p}\cap G_3 = C_{\delta^m(\alpha k)p} \quad \text{and} \quad C_{\delta^m(\alpha k)p}\cap K(\alpha,j) = \emptyset.
\end{equation*}
Thus we have \vspace{-3mm}
\begin{align*}
    C_{\alpha j^{-1}}^c &= K(\alpha,j) \sqcup F_\alpha \sqcup \bigcap_{k:A(j,k)=1}\left( \bigsqcup_{m=0}^{|\alpha|-1}\bigsqcup_{p\neq \alpha_m}  C_{\delta^m(\alpha )p}\right)\sqcup \bigcap_{k:A(j,k)=1}\bigsqcup_{p\neq k}  C_{\alpha p}\\
    &= K(\alpha,j) \sqcup F_\alpha \sqcup \left( \bigsqcup_{m=0}^{|\alpha|-1}\bigsqcup_{p\neq \alpha_m}  C_{\delta^m(\alpha )p}\right)\sqcup \bigsqcup_{\substack{p: A(j,p)=0}}  C_{\alpha p}.
\end{align*} 
\end{proof}

\begin{remark} Proposition \ref{prop:general_C_alpha_inverse_j} and Proposition \ref{prop:general_C_alpha_j_inverse_complement} are also valid for $\alpha = e$, by removing the hypothesis $j \neq \alpha_{|\alpha|-1}$. In particular, for Proposition \ref{prop:general_C_alpha_j_inverse_complement}, we have the following change:
\begin{align*}
    C_{j^{-1}}^c = K(e,j) \sqcup \bigsqcup_{\substack{p: A(j,p)=0}}  C_{\alpha p}.
\end{align*} 
The proof is essentially the same as presented above.
\end{remark}

We moved some auxiliary lemmas to the appendix in section \ref{ape:proof_giant_theorem_generalized_cylinders} to avoid technicalities. The lemmas are used to prove the next theorem and identities involving intersections of generalized cylinders, the most obvious is the following:
\begin{equation}\label{eq:C_cap_C}
    C_\alpha \cap C_\gamma = \begin{cases}
                                C_\alpha, \quad \text{if } \gamma \in \llbracket \alpha \rrbracket,\\
                                C_\gamma, \quad \text{if } \alpha \in \llbracket \gamma \rrbracket,\\
                                \emptyset, \quad \text{otherwise};
                            \end{cases} 
\end{equation}
which holds for every $\alpha$ and $\gamma$ positive admissible words.

Summarizing the results, Proposition \ref{prop:cylinder_reducing_inverse} grants that we only need to consider the basis of the finite intersections of generalized cylinders in the form $C_\alpha$ and $C_{\alpha j^{-1}}$, $\alpha$ positive admissible word and $j \in \mathbb{N}$, and their complements. However, Propositions \ref{prop:general_C_alpha_inverse_j}, \ref{prop:general_C_alpha_complement} and \ref{prop:general_C_alpha_j_inverse_complement} show that $C_{\alpha j^{-1}}$ as much as $C_\alpha$ and $C_{\alpha j^{-1}}$ can be written as a disjoint union of generalized cylinders together with a subset of $Y_A$. It is remarkable that for any finite intersection of these elements we obtain again a disjoint union between a subset of $Y_A$ and a disjoint union of cylinders. We emphasize that all these open sets are also closed. 

The next result presents all the possible intersections between two elements of the subbasis.

\begin{theorem}\label{thm:huge_generators_intersections} Given $\alpha,\gamma$ positive admissible words and $j,k \in \mathbb{N}$ we have that
\begin{align}
    C_\alpha \cap C_\gamma^c &= 
    \begin{cases}
        F_{\gamma}^{\alpha}\sqcup \bigsqcup_{n=|\alpha|}^{|\gamma|-1} \bigsqcup_{p \neq \gamma_n} C_{\delta^{n}(\gamma)p},\quad \alpha \in \llbracket \gamma \rrbracket\setminus\{\gamma\},\\
        C_\alpha,\quad \alpha \notin \llbracket \gamma \rrbracket \text{ and } \gamma \notin \llbracket \alpha \rrbracket, \\
        \emptyset, \quad \text{otherwise};
    \end{cases} \label{eq:C_cap_C_comp}\\
    C_\alpha^c \cap C_\gamma^c &= \begin{cases}
                                    C_\alpha^c, \quad \alpha \in \llbracket \gamma \rrbracket,\\
                                    C_\gamma^c, \quad \gamma \in \llbracket \alpha \rrbracket,\\
                                    F_{\alpha'}^* \sqcup F_\alpha^{\alpha'\alpha_{|\alpha'|}} \sqcup F_\gamma^{\alpha'\gamma_{|\alpha'|}} \sqcup  \left(\bigsqcup_{m = 0}^{|\alpha|-1} \bigsqcup_{q\neq \alpha_m} C_{\delta^m(\alpha)q}\right) \\
                                    \qquad \qquad \qquad \sqcup \left(\bigsqcup_{|\alpha'| < n < |\gamma|-1} \bigsqcup_{p\neq \gamma_n} C_{\delta^n(\gamma )p}\right), \quad \text{otherwise};
                                  \end{cases}\label{eq:C_comp_cap_C_comp}\\
    C_\alpha \cap C_{\gamma j^{-1}} &= \begin{cases}
                                    C_{\gamma j^{-1}},\quad \alpha \in \llbracket \gamma \rrbracket,\\
                                    C_\alpha,\quad \gamma \in \llbracket \alpha \rrbracket \setminus \{\alpha\} \text{ and } A(j,\alpha_{|\gamma|})=1, \\
                                    \emptyset, \quad \text{otherwise};
                                    \end{cases}\label{eq:C_cap_C_inverse}\\
    C_\alpha \cap C_{\gamma j^{-1}}^c &= \begin{cases}
                                    K(\gamma,j) \sqcup F_\gamma^\alpha \sqcup \left( \bigsqcup_{n=|\alpha|}^{|\gamma|-1}\bigsqcup_{p\neq \gamma_m} C_{\delta^n(\gamma)p}\right)\sqcup \bigsqcup_{\substack{p: A(j,p)=0}} C_{\gamma p},\quad \alpha \in \llbracket \gamma \rrbracket,\\
                                    C_\alpha,\quad \gamma \in \llbracket \alpha \rrbracket \setminus \{\alpha\} \text{ and } A(j,\alpha_{|\gamma|})=0, \text{ or } \alpha \notin \llbracket \gamma \rrbracket \text{ and }\gamma \notin \llbracket \alpha \rrbracket,\\
                                    \emptyset, \quad \text{otherwise};
                                    \end{cases}\label{eq:C_cap_C_inverse_comp}\\
    C_\alpha^c \cap C_{\gamma j^{-1}} &= \begin{cases}
                                        \emptyset,\quad \alpha \in \llbracket \gamma \rrbracket,\\
                                        G(\gamma,j) \sqcup F_\alpha^{\gamma \alpha_{|\gamma|}} \sqcup  \mathfrak{C}[\alpha,\gamma,j], \quad \gamma \in \llbracket \alpha \rrbracket \setminus \{\alpha\} \text{ and }A(j,\alpha_{|\gamma|}) = 1,\\
                                        G(\gamma,j) \sqcup  \mathfrak{C}[\alpha,\gamma,j], \quad \gamma \in \llbracket \alpha \rrbracket \setminus \{\alpha\} \text{ and }A(j,\alpha_{|\gamma|}) = 0,\\
                                        C_{\gamma j^{-1}}, \quad \text{otherwise};
                                        \end{cases}\label{eq:C_comp_cap_C_inverse}
\end{align}

\begin{align}
    C_\alpha^c \cap C_{\gamma j^{-1}}^c &= \begin{cases}
                                C_\alpha^c,\quad \alpha \in \llbracket \gamma \rrbracket,\\
                                K(\gamma,j) \sqcup F_\gamma \sqcup F_\alpha^{\gamma \alpha_{|\gamma|}} \sqcup \left(\bigsqcup_{n=0}^{|\gamma|-1}\bigsqcup_{p\neq \gamma_n}  C_{\delta^n(\gamma )p}\right) \\
                                \qquad \qquad \qquad \sqcup \mathfrak{D}[\alpha,\gamma,j], \quad \gamma \in \llbracket \alpha \rrbracket \setminus \{\alpha\} \text{ and }A(j,\alpha_{|\gamma|}) = 0,\\
                                K(\gamma,j) \sqcup F_\gamma \sqcup \left(\bigsqcup_{n=0}^{|\gamma|-1}\bigsqcup_{p\neq \gamma_n}  C_{\delta^n(\gamma )p} \right)\\
                                \qquad \qquad \qquad \sqcup \mathfrak{D}[\alpha,\gamma,j], \quad \gamma \in \llbracket \alpha \rrbracket \setminus \{\alpha\} \text{ and }A(j,\alpha_{|\gamma|}) = 1,\\
                                K(\gamma,j) \sqcup F_{\alpha'}^* \sqcup F_\alpha^{\alpha' \alpha_{|\gamma|}} \sqcup F_\gamma^{\alpha'\gamma_{|\alpha'|}} \sqcup \left(\bigsqcup_{m = 0}^{|\alpha|-1} \bigsqcup_{q\neq \alpha_m} C_{\delta^m(\alpha)q}\right) \\\qquad \qquad \qquad \sqcup \left(\bigsqcup_{|\alpha'| < n \leq |\gamma|-1} \bigsqcup_{p\neq \gamma_n} C_{\delta^n(\gamma )p}\right)\\
                                \qquad \qquad \qquad \sqcup \bigsqcup_{\substack{p: A(j,p)=0}} C_{\gamma p} , \quad \text{otherwise}; \end{cases}\label{eq:C_comp_cap_C_inverse_comp}
\end{align}

\begin{align}
    C_{\alpha j^{-1}} \cap C_{\gamma k^{-1}} &= \begin{cases}
                                G(\alpha,\{j,k\}) \sqcup \bigsqcup_{q:A(j,q)A(k,q)=1}C_{\alpha q},\quad \alpha = \gamma,\\
                                G(\gamma,k) \sqcup  \bigsqcup_{m:A(k,m)=1}  C_{\gamma m}, \quad \alpha \in \llbracket \gamma \rrbracket \setminus \{\gamma\} \text{ and }A(j,\gamma_{|\alpha|}) = 1,\\
                                 G(\alpha,j) \sqcup  \bigsqcup_{q:A(j,q)=1}  C_{\alpha q}, \quad \gamma \in \llbracket \alpha \rrbracket \setminus \{\alpha\} \text{ and }A(k,\alpha_{|\gamma|}) = 1,\\
                            \emptyset, \quad \text{otherwise}; \end{cases}\label{eq:C_inverse_cap_C_inverse}
\end{align}

\begin{align}
    C_{\alpha j^{-1}} \cap C_{\gamma k^{-1}}^c &= \begin{cases}
                                GK(\alpha,j,k) \sqcup \bigsqcup_{q:A(j,q) = 1,A(k,q)=0}C_{\alpha q},\quad \alpha = \gamma,\\
                                K(\gamma,k) \sqcup G(\alpha,j) \sqcup F_{\gamma}^{\alpha \gamma_{|\alpha|}} \sqcup \mathfrak{C}[\gamma,\alpha,j]\\
                                \qquad \qquad \qquad \sqcup  \bigsqcup_{m:A(k,m)=0}  C_{\gamma m}, \quad \alpha \in \llbracket \gamma \rrbracket \setminus \{\gamma\} \text{ and }A(j,\gamma_{|\alpha|}) = 1,\\
                                G(\alpha,j) \sqcup \mathfrak{C} [\gamma,\alpha,j], \quad \alpha \in \llbracket \gamma \rrbracket \setminus \{\gamma\} \text{ and }A(j,\gamma_{|\alpha|}) = 0,\\    
                                G(\alpha,j) \sqcup  \bigsqcup_{q:A(j,q)=1}  C_{\alpha q}, \quad \gamma \in \llbracket \alpha \rrbracket \setminus \{\alpha\} \text{ and }A(k,\alpha_{|\gamma|}) = 0, \\
                                \qquad \qquad \qquad \text{ or if } \alpha \notin \llbracket \gamma \rrbracket \text{ and } \gamma \notin \llbracket  \alpha \rrbracket\\
                                \emptyset, \quad \text{otherwise}; \end{cases}\label{eq:C_inverse_cap_C_inverse_comp}
\end{align}

\begin{align} 
    C_{\alpha j^{-1}}^c \cap C_{\gamma k^{-1}}^c &= \begin{cases}
                                K(\alpha,\{j,k\}) \sqcup F_\alpha \sqcup \left(\bigsqcup_{q: A(j,q) = A(k,q) = 0} C_{\alpha q}\right)\\
                                \qquad \qquad \qquad \sqcup \left(\bigsqcup_{m = 0}^{|\alpha| - 1} \bigsqcup_{q\neq \alpha_m} C_{\delta^m(\alpha)q}\right),\quad \alpha = \gamma,\\
                                K(\alpha,j) \sqcup K(\gamma,k) \sqcup F_\alpha \sqcup F_\gamma^{\alpha \gamma_{|\alpha|}} \sqcup \mathfrak{D}[\gamma,\alpha,k] \sqcup \left( \bigsqcup_{m = 0}^{|\alpha| - 1} \bigsqcup_{q\neq \alpha_m} C_{\delta^m(\alpha)q} \right) \\
                                \qquad \qquad \qquad \sqcup  \left(\bigsqcup_{p:A(k,p)=0}  C_{\gamma p}\right), \quad \alpha \in \llbracket \gamma \rrbracket \setminus \{\gamma\} \text{ and }A(j,\gamma_{|\alpha|}) = 0,\\
                                K(\alpha,j) \sqcup F_\alpha \sqcup \mathfrak{D}[\gamma,\alpha,k]\\
                                \qquad \qquad \qquad \sqcup \left( \bigsqcup_{m = 0}^{|\alpha| - 1} \bigsqcup_{q\neq \alpha_m} C_{\delta^m(\alpha)q} \right), \quad \alpha \in \llbracket \gamma \rrbracket \setminus \{\gamma\} \text{ and }A(j,\gamma_{|\alpha|}) = 1,\\
                                K(\alpha,j) \sqcup K(\gamma,k) \sqcup F_\gamma \sqcup F_\alpha^{\gamma \alpha_{|\gamma|}} \sqcup \mathfrak{D}[\alpha,\gamma,j] \sqcup \left( \bigsqcup_{n = 0}^{|\gamma| - 1} \bigsqcup_{p\neq \gamma_n} C_{\delta^n(\gamma)p} \right)\\
                                \qquad \qquad \qquad \sqcup  \left(\bigsqcup_{q:A(j,q)=0}  C_{\alpha q} \right), \quad \gamma \in \llbracket \alpha \rrbracket \setminus \{\alpha\} \text{ and }A(k,\alpha_{|\gamma|}) = 0,\\
                                K(\gamma,k) \sqcup F_\gamma \sqcup \mathfrak{D}[\alpha,\gamma,j]\\
                                \qquad \qquad \qquad \sqcup \left( \bigsqcup_{n = 0}^{|\gamma| - 1} \bigsqcup_{p\neq \gamma_n} C_{\delta^n(\gamma)p} \right), \quad \gamma \in \llbracket \alpha \rrbracket \setminus \{\alpha\} \text{ and }A(k,\alpha_{|\gamma|}) = 1,\\
                                K(\alpha,j) \sqcup K(\gamma,k) \sqcup F_{\alpha'} \sqcup F_\alpha^{\alpha'\alpha_{|\alpha'|}} \sqcup F_\gamma^{\alpha'\gamma_{|\alpha'|}} \sqcup  \left(\bigsqcup_{q:A(j,q)=0}  C_{\alpha q} \right) \\\qquad \qquad \qquad \sqcup \left( \bigsqcup_{p:A(k,p)=0}  C_{\gamma p}\right)
                                \sqcup \left(\bigsqcup_{m = 0}^{|\alpha|-1} \bigsqcup_{q\neq \alpha_m} C_{\delta^m(\alpha)q}\right)\\
                                \qquad \qquad \qquad \sqcup \left(\bigsqcup_{|\alpha'| < n \leq |\gamma|-1} \bigsqcup_{p\neq \gamma_n} C_{\delta^n(\gamma )p}\right), \quad \text{otherwise}; \end{cases}\label{eq:C_inverse_comp_cap_C_inverse_comp}
\end{align}
where
\begin{equation*}
    \mathfrak{C}[\alpha,\gamma,j] := \begin{cases}
    \bigsqcup_{\substack{p:A(j,p)=1,\\p \neq \alpha_{|\gamma|}}} C_{\gamma p}, \quad A(j,\alpha_{|\gamma|}) = 0,\\
    \left(\bigsqcup_{|\gamma| < m \leq |\alpha|-1} \bigsqcup_{q \neq \alpha_m} C_{\delta^m(\alpha)q}\right) \sqcup \bigsqcup_{\substack{p:A(j,p)=1,\\p \neq \alpha_{|\gamma|}}} C_{\gamma p}, \quad A(j,\alpha_{|\gamma|}) = 1,
    \end{cases}
\end{equation*}
and
\begin{equation*}
    \mathfrak{D}[\alpha,\gamma,j] := \begin{cases}
    \bigsqcup_{\substack{p:A(j,p)=0,\\p \neq \alpha_{|\gamma|}}} C_{\gamma p}, \quad A(j,\alpha_{|\gamma|}) = 1,\\
    \left(\bigsqcup_{|\gamma| < m \leq |\alpha|-1} \bigsqcup_{q \neq \alpha_m} C_{\delta^m(\alpha)q}\right) \sqcup \bigsqcup_{\substack{p:A(j,p)=0,\\p \neq \alpha_{|\gamma|}}} C_{\gamma p}, \quad A(j,\alpha_{|\gamma|}) = 0.
    \end{cases}
\end{equation*}
Also, $\alpha'$ is the longest word in $\llbracket \alpha \rrbracket \cap \llbracket \gamma \rrbracket$.
\end{theorem}

\begin{proof} See appendix in section \ref{ape:proof_giant_theorem_generalized_cylinders}. 
\end{proof}

\begin{proposition}\label{prop:X_A_decomposition} Consider the set $X_A$. Then,
\begin{equation*}
    X_A = \left(\bigcup_{j \in \mathbb{N}}G(e,j)\right) \sqcup \left(\bigsqcup_{j \in \mathbb{N}}C_j\right).
\end{equation*}
\end{proposition}

\begin{proof} The inclusion
\begin{equation*}
    X_A \supseteq \left(\bigcup_{j \in \mathbb{N}}G(e,j)\right) \sqcup \left(\bigsqcup_{j \in \mathbb{N}}C_j\right)
\end{equation*}
is straightforward. For the opposite inclusion, let $\xi \in X_A$. We have two possibilities, namely $\kappa(\xi) \neq e$ or $\kappa(\xi) = e$. In the first case, we have necessarily that $\xi \in C_j$ for some $j \in \mathbb{N}$. Now if $\kappa(\xi) = e$, we have necessarily that $\xi \in G(e,j)$ for some $j \in \mathbb{N}$, otherwise $\xi = \varphi_0$, the configuration filled only in $e$. This is not possible because $\varphi_0 \notin X_A$. 
\end{proof}

\begin{corollary} It is true that
\begin{equation*}
    X_A = \bigcup_{j \in \mathbb{N}}\sigma(C_j).
\end{equation*}
\end{corollary}

\subsection{Cylinder topology: renewal shift}
\label{subsec:cylinder_topology_Renewal}

In the particular case of the renewal shift, we have a simple characterization and we present it now. Observe that, for every generalized cylinder $C_\alpha$ on a positive word $\alpha$, if it does not end with $1$, we have $C_{\alpha} = C_{\gamma}$, where $\gamma$ is the unique shortest word that it starts with $\alpha$ and ends in $1$. Here we represent the elements of $Y_A$ by their stems in $\mathfrak{W}$.

\begin{lemma}\label{lemma:G_K_renewal}
For the renewal shift, we have
\begin{equation*}
    G(\alpha,F) = \begin{cases}
        \{\alpha \}, \quad \text{if } \alpha_{|\alpha| - 1} = 1 \text{ and } F = \{1\},\\
        \emptyset, \quad \text{otherwise};
    \end{cases}
\end{equation*}
and
\begin{equation*}
    K(\alpha,F) = \begin{cases}
        \{\alpha \}, \quad \text{if } \alpha_{|\alpha| - 1} = 1 \text{ and } F \neq \{1\},\\
        \emptyset, \quad \text{otherwise};
    \end{cases}
\end{equation*}
for $\alpha$ non-empty finite admissible word.
\end{lemma}

\begin{proof} From the renewal shift we have that there exists $\xi \in Y_A$ such that $\kappa(\xi) = \alpha$ if and only if $\alpha_{|\alpha|-1} = 1$. We complete the proof by recalling that $R_{\xi}(\alpha) = \{1\}$ for $\alpha$ ending with $1$. 
\end{proof}

\begin{remark} The lemma above is also valid for $\alpha = e$, just by removing the condition $\alpha_{|\alpha| - 1} = 1$.
\end{remark}

\begin{corollary}\label{cor:C_alpha_j_inv_renewal} Consider the set $X_A$ of the renewal Markov shift. For every $\alpha$ positive admissible word, $|\alpha|\neq 0$, and every $j \in \mathbb{N}$, we have that $C_{\alpha j^{-1}}$ is the empty set or $C_{\alpha j^{-1}} = C_\gamma$, where $\gamma$ is a positive admissible word.
\end{corollary}

\begin{proof} If $j = \alpha_{|\alpha|-1}$ then the statement is obvious. So consider $j \neq \alpha_{|\alpha|-1}$. By Lemma \ref{lemma:G_K_renewal} we have $G(\alpha,F) = \emptyset$, so Proposition \ref{prop:general_C_alpha_inverse_j} gives
\begin{equation*}
    C_{\alpha j^{-1}} = \bigsqcup_{k:A(j,k)=1} C_{\alpha k} = \begin{cases}
        C_{\alpha (\alpha_{|\alpha|-1} -1)}, \quad \text{if } \alpha_{|\alpha|-1} > 1 \text{ and } j = 1,\\
        C_{\alpha (j-1)}, \quad \text{if } \alpha_{|\alpha|-1} = 1,\\
        \emptyset, \quad \text{otherwise}.
    \end{cases}
\end{equation*}
\end{proof}
\begin{remark} If $\alpha = e$, then
\begin{equation*}
    C_{j^{-1}} = \begin{cases}
        X_A, \quad \text{if } j = 1,\\
        \emptyset, \quad \text{otherwise};
    \end{cases}
\end{equation*}
which is straightforward from the renewal shift.
\end{remark}

\begin{remark} For the renewal shift, for given positive admissible words $\alpha$ and $\gamma$, we have that
\begin{align*}
    F_\alpha &= \{\omega \in \mathfrak{W}: \omega_{|\omega|-1} = 1, \omega \in \llbracket \alpha \rrbracket \setminus \{\alpha\}\}, \quad F_\alpha^* = \{\omega \in \mathfrak{W}: \omega_{|\omega|-1} = 1, \omega \in \llbracket \alpha \rrbracket\},
\end{align*}
are finite sets, and $F_\gamma^\alpha$ is finite as well.
\end{remark}

\begin{corollary}\label{cor:C_alpha_j_inv_comp_renewal} Let $\alpha$ be a non-empty positive admissible word and $j \neq \alpha_{|\alpha|-1}$. Then,
\begin{equation*}
    C_{\alpha j^{-1}}^c = \begin{cases}
        C_{\alpha (\alpha_{|\alpha|-1} -1)}^c, \quad \text{if } \alpha_{|\alpha|-1} > 1 \text{ and } j = 1,\\
        C_{\alpha (j-1)}^c, \quad \text{if } \alpha_{|\alpha|-1} = 1,\\
        X_A, \quad \text{otherwise}
    \end{cases}
\end{equation*}
\end{corollary}

\begin{proof} It is straightforward from the proof of Corollary \ref{cor:C_alpha_j_inv_renewal}. 
\end{proof}

\begin{remark} Explicitly, by Proposition \ref{prop:general_C_alpha_complement}, we have
\begin{equation}\label{eq:c_alpha_inverse_comp_renewal}
    C_{\alpha j^{-1}}^c = \begin{cases}
        F_{\alpha (\alpha_{|\alpha|-1} -1)}\sqcup\bigsqcup_{m=0}^{|\alpha|}\bigsqcup_{k\neq ({\alpha (\alpha_{|\alpha|-1} -1)})_m} C_{\delta^m({\alpha (\alpha_{|\alpha|-1} -1)})k}, \quad \text{if } \alpha_{|\alpha|-1} > 1 \text{ and } j = 1,\\
        F_{\alpha(j-1)}\sqcup\bigsqcup_{m=0}^{|\alpha|} \bigsqcup_{k\neq (\alpha (j-1))_m} C_{\delta^m(\alpha (j-1))k}, \quad \text{if } \alpha_{|\alpha|-1} = 1,\\
        X_A, \quad \text{otherwise}.
    \end{cases}
\end{equation}
\end{remark}

The previous results in this subsection allow us to reduce the subbasis of the renewal shift to generalized cylinders on positive words and their complements. Moreover, for these cylinders, we only need to consider the words ending in `$1$'. We present the pairwise intersections of the elements of the subbasis for the generalized renewal shift in the next corollary.

\begin{corollary} For any positive admissible words $\alpha$ and $\gamma$ we have that
\begin{align*}
    C_\alpha \cap C_\gamma &= \begin{cases}
                                C_\alpha, \quad \text{if } \gamma \in \llbracket \alpha \rrbracket,\\
                                C_\gamma, \quad \text{if } \alpha \in \llbracket \gamma \rrbracket,\\
                                \emptyset, \quad \text{otherwise};
                            \end{cases} \\
    C_\alpha \cap C_\gamma^c &= \begin{cases}
                                F_{\gamma}^{\alpha}\sqcup \bigsqcup_{n=|\alpha|}^{|\gamma|-1} \bigsqcup_{j\neq \gamma_n} C_{\delta^{n}(\gamma)j},\quad \text{if } \alpha \in \llbracket \gamma \rrbracket,\\
                                C_\alpha,\quad \text{if } \alpha \notin \llbracket \gamma \rrbracket \text{ and } \gamma \notin \llbracket \alpha \rrbracket, \\
                                \emptyset, \quad \text{otherwise};
                               \end{cases}\\
    C_\alpha^c \cap C_\gamma^c &= \begin{cases}
                                    C_\alpha^c, \quad \text{if } \alpha \in \llbracket \gamma \rrbracket,\\
                                    C_\gamma^c, \quad \text{if } \gamma \in \llbracket \alpha \rrbracket,\\
                                    F_{\alpha'}^* \sqcup F_\alpha^{\alpha'\alpha_{|\alpha'|}} \sqcup F_\gamma^{\alpha'\gamma_{|\alpha'|}} \sqcup  \left(\bigsqcup_{m = 0}^{|\alpha|-1} \bigsqcup_{k\neq \alpha_m} C_{\delta^m(\alpha)k}\right)\\
                                    \qquad \qquad \qquad \sqcup \left(\bigsqcup_{|\alpha'| < n < |\gamma|-1} \bigsqcup_{p\neq \gamma_n} C_{\delta^n(\gamma )p}\right), \quad \text{otherwise}. \end{cases}
\end{align*}
\end{corollary}

\begin{proof} The first identity is straightforward, while the rest of them are precisely the equations \eqref{eq:C_cap_C_comp} and \eqref{eq:C_comp_cap_C_comp} from Theorem \ref{thm:huge_generators_intersections}. 
\end{proof}

From the corollary above, it is straightforward that every finite intersection among the elements of the subbasis is in the form
\begin{equation}\label{eq:basis_renewal_is_finite_set_union_countable_cylinders}
    F \sqcup \bigsqcup_{n \in \mathbb{N}} C_{w(n)},
\end{equation}
where $F \subseteq Y_A$ is finite and $w(n)$ is a positive admissible word, for each $n \in \mathbb{N}$.

\begin{remark}
    The fact that elements of a subbasis for the topology of $X_A$ have the form \ref{eq:basis_renewal_is_finite_set_union_countable_cylinders} allows us to control the weak$^*$ convergence of probability measures on GCMS. Even for the weak$^*$ convergence of probability measures on the standard CMS, we have recent results \cite{IommiVelozo2019}.
\end{remark}

\subsection{Cylinder topology: the pair renewal shift}
\label{subsec:cylinder_topology_Pair}

As well as it was done in subsection \ref{subsec:cylinder_topology_Renewal}, we also characterize the cylinders and intersections for the pair renewal shift. Particularly for this section, we use the following notation: for $\alpha$ positive admissible word, we represent configuration $\xi  \in Y_A(\xi^{0,k})$ with stem $\alpha$ by $\alpha^k \in \mathfrak{W}_k$, $k \in \{1,2\}$.

First, we notice that, for every generalized cylinder $C_\alpha$ on a positive word $\alpha$, if it does not end with $1$ or $2$, then $C_{\alpha} = C_{\gamma}$, where $\gamma$ is the unique shortest word that it starts with $\alpha$ and ends in $2$.

\begin{lemma}\label{lemma:G_K_Pair_renewal}
For the pair renewal shift we have
\begin{equation*}
    G(\alpha,F) = \begin{cases}
        \{\alpha^1,\alpha^2\}, \quad \text{if } \alpha_{|\alpha| - 1} = 1 \text{ and } F = \{1\};\\
        \{\alpha^1\}, \quad \text{if } \alpha_{|\alpha| - 1} = 1 \text{ and } F \in \{\{1,2\}, \{2\}\};\\
        \{\alpha^1\}, \quad \text{if } \alpha_{|\alpha| - 1} = 2 \text{ and } F \in \{\{1\},\{1,2\},\{2\}\} ;\\
        \emptyset, \quad \text{otherwise},
    \end{cases}
\end{equation*}
and
\begin{equation*}
    K(\alpha,F) = \begin{cases}
        \{\alpha^1,\alpha^2\}, \quad \text{if } \alpha_{|\alpha| - 1} = 1 \text{ and } \{1,2\} \not\subseteq F;\\
        \{\alpha^2\}, \quad \text{if } \alpha_{|\alpha| - 1} = 1,\text{ }  1 \notin F \text{ and }  2 \in F;\\
        \emptyset, \quad \text{otherwise},
    \end{cases}
\end{equation*}
for $\alpha$ non-empty finite admissible word.
\end{lemma}

\begin{proof} From the pair renewal shift we have that there exists $\xi \in Y_A$ such that $\kappa(\xi) = \alpha$, if and only if $\alpha \in \mathfrak{W}_1 \cup \mathfrak{W}_2$. We complete the proof by using that
\begin{itemize}
    \item $\xi = \alpha^1$ if and only if $\alpha_{|\alpha|-1} \in \{1,2\}$ and $R_\xi(\alpha) = \{1,2\}$;
    \item $\xi = \alpha^2$ if and only if $\alpha_{|\alpha|-1} = 1$ and $R_\xi(\alpha) = \{1\}$. 
\end{itemize}
\end{proof}

\begin{remark} The lemma above is also valid for $\alpha = e$, just by removing the conditions for $\alpha_{|\alpha| - 1}$.
\end{remark}

\begin{corollary}\label{cor:C_alpha_j_inv_Pair_renewal} Consider the set $X_A$ of the pair renewal Markov shift. For any $\alpha$ positive admissible word, $|\alpha|\neq 0$, and any $j \in \mathbb{N}$, we have for $j \neq \alpha_{|\alpha|-1}$ that
\begin{equation*}
    C_{\alpha j^{-1}} = \begin{cases}
                            C_{\alpha (\alpha_{|\alpha|-1} - 1)}, \quad \text{if } \alpha_{|\alpha|-1}>2 \text{ and } j=1;\\
                            C_{\alpha (\alpha_{|\alpha|-1} - 1)}, \quad \text{if } \alpha_{|\alpha|-1}>2 \text{ and }, \text{ } j=2 \text{ and } \alpha_{|\alpha|-1} \text{ is even};\\
                            \{\alpha^1\}\sqcup C_{\alpha 1} \sqcup \bigsqcup_{k \in \mathbb{N}} C_{\alpha (2k)}, \quad  \text{if } \alpha_{|\alpha|-1}=2 \text{ and } j=1;\\
                            \{\alpha^1\}\sqcup C_{\alpha 1} \sqcup \bigsqcup_{k \in \mathbb{N}} C_{\alpha (2k)}, \quad  \text{if } \alpha_{|\alpha|-1}=1 \text{ and } j=2;\\
                            \emptyset, \quad \text{otherwise};
                        \end{cases}
\end{equation*}
\end{corollary}

\begin{proof} By Lemma \ref{lemma:G_K_Pair_renewal} we have $G(\alpha,j) = \emptyset$ if $\alpha_{|\alpha|-1} > 2$, and $G(\alpha,j) = \{\alpha^1\}$ either when $\alpha_{|\alpha|-1} = 2$ and $j=1$, or when $\alpha_{|\alpha|-1} = 1$ and $j=2$. The remaining cases are empty. By Proposition \ref{prop:general_C_alpha_inverse_j}, it remains to analyze the unions in the form
\begin{equation*}
    \bigsqcup_{k:A(j,k)=1} C_{\alpha k}.
\end{equation*}
It is straightforward from the pair renewal transition matrix that
\begin{equation*}
    \bigsqcup_{k:A(j,k)=1} C_{\alpha k} = \begin{cases}
                                            C_{\alpha (\alpha_{|\alpha|-1} - 1)}, \quad \text{if } \alpha_{|\alpha|-1}>2 \text{ and } j=1;\\
                                            C_{\alpha (\alpha_{|\alpha|-1} - 1)}, \quad \text{if } \alpha_{|\alpha|-1}>2 \text{ and }, \text{ } j=2 \text{ and } \alpha_{|\alpha|-1} \text{ is even};\\
                                            C_{\alpha 1} \sqcup \bigsqcup_{k \in \mathbb{N}} C_{\alpha (2k)}, \quad  \text{if } \alpha_{|\alpha|-1}=2 \text{ and } j=1;\\
                                            C_{\alpha 1} \sqcup \bigsqcup_{k \in \mathbb{N}} C_{\alpha (2k)}, \quad  \text{if } \alpha_{|\alpha|-1}=1 \text{ and } j=2;\\
                                            \emptyset, \quad \text{otherwise};
                                          \end{cases}.  
\end{equation*}
\end{proof}

\begin{remark} When $j = \alpha_{|\alpha|-1}$, then $C_{\alpha j^{-1}}$ is a cylinder or $X_A$. If $\alpha = e$, then 
\begin{equation*}
    C_{j^{-1}} = \begin{cases}
                    X_A, \quad j = 1,\\
                    \{e^1\}\sqcup C_1 \sqcup \bigsqcup_{k \in \mathbb{N}} C_{(2k)},\quad j = 2,\\        
                    \emptyset, \quad \text{otherwise};
                 \end{cases}
\end{equation*}
which is straightforward from the pair renewal shift.
\end{remark}

\begin{remark} In the pair renewal shift, for given positive admissible words $\alpha$ and $\gamma$, we have that
\begin{align*}
    F_\alpha &= \{\omega \in \mathfrak{W}_1 \cup \mathfrak{W}_2: \omega \in \llbracket \alpha \rrbracket \setminus \{\alpha\}\}, \quad F_\alpha^* = \{\omega \in \mathfrak{W}_1 \cup \mathfrak{W}_2: \omega \in \llbracket \alpha \rrbracket\},
\end{align*}
are finite sets, and $F_\gamma^\alpha$ is finite as well.
\end{remark}

\begin{corollary}\label{cor:C_alpha_j_inv_comp_Pair_renewal} Let $\alpha$ be a non-empty positive admissible word and $j \neq \alpha_{|\alpha|-1}$. Then,
\begin{equation*}
    C_{\alpha j^{-1}}^c = \begin{cases}
                            C_{\alpha (\alpha_{|\alpha|-1} - 1)}^c, \quad \text{if } \alpha_{|\alpha|-1}>2 \text{ and } j=1;\\
                            C_{\alpha (\alpha_{|\alpha|-1} - 1)}^c, \quad \text{if } \alpha_{|\alpha|-1}>2 \text{ and }, \text{ } j=2 \text{ and } \alpha_{|\alpha|-1} \text{ is even};\\
                            F_\alpha \sqcup \left( \bigsqcup_{m=0}^{|\alpha|-1}\bigsqcup_{p\neq \alpha_m}  C_{\delta^m(\alpha )p}\right), \quad  \text{if } \alpha_{|\alpha|-1}=2 \text{ and } j=1;\\
                            \{\alpha^2\} \sqcup F_\alpha \sqcup \left( \bigsqcup_{m=0}^{|\alpha|-1}\bigsqcup_{p\neq \alpha_m}  C_{\delta^m(\alpha )p}\right)\sqcup \bigsqcup_{\substack{p\in \mathbb{N}}}  C_{\alpha (2p+1)}, \quad  \text{if } \alpha_{|\alpha|-1}=1 \text{ and } j=2;\\
                            X_A, \quad \text{otherwise};
                        \end{cases}
\end{equation*}
\end{corollary}

\begin{proof} It is straightforward from Proposition \ref{prop:general_C_alpha_j_inverse_complement}, Lemma \ref{lemma:G_K_Pair_renewal} and Corollary \ref{cor:C_alpha_j_inv_Pair_renewal}. 
\end{proof}

Unlike the renewal shift case in subsection \ref{subsec:cylinder_topology_Renewal}, in general, a cylinder in the form $C_{\alpha j^{-1}}$ cannot be written as a positive one. However, they are pairwise disjoint unions of positive generalized cylinders, jointly with finite subsets of $Y_A$. And the same occurs with their complements. On the other hand, we only need to consider when $\alpha$ ends in a symbol of $\{1,2\}$. Similarly to the renewal shift case, from Theorem \ref{thm:huge_generators_intersections}, it is straightforward that every finite intersection among the elements of the subbasis is in the form
\begin{equation}\label{eq:basis_pair_renewal_is_finite_set_union_countable_cylinders}
    F \sqcup \bigsqcup_{n \in \mathbb{N}} C_{w(n)},
\end{equation}
where $F \subseteq Y_A$ is finite and $w(n)$ is a positive admissible word for each $n \in \mathbb{N}$.

\section{Thermodynamic formalism}\label{sec:TF_Generalized}

One of the main contributions of this paper is to develop the thermodynamic formalism for GCMS, showing its connections with
the standard theory and pointing out the differences. Our results recover the usual thermodynamic formalism as a particular case. However, some of the theorems can be stated even for more general spaces,
namely for a locally compact Hausdorff second countable space $X$, which is not necessarily a symbolic space (generalized or not). For a reference where the authors develop the thermodynamic formalism for Cuntz-Krieger C$^*$-algebras on finite alphabet case, see \cite{KerrPinz2002}.

In this section, we prove that all the following notions coincide: conformal measure defined by M. Denker and M. Urba\'nski, conformal measure defined by O. Sarig, the notion of fixed point associated to the Ruelle's transformation, and the definition of quasi-invariant measure on Renault-Deaconu groupoids. These equivalences are known but not on this generality. After this theorem, we dedicate the rest of the paper to study the thermodynamic formalism on GCMS spaces.

As we highlighted before, in general, we will assume two fundamental hypotheses, which essentially say that both dynamics and potential (the necessary objects to define a thermodynamic formalism over a space) are \emph{partially defined}. In most of the results, we assume the following: \newline

There exists $U \subseteq X$, an open set of $X$ such that:\\

\textbf{A.} The \emph{dynamics} is given by a local homeomorphism $\sigma: U \to X$.\\

\textbf{B.} The \emph{potential} $F:U\to \mathbb{R}$ is a continuous function.\newline

With these two assumptions and adding some regularity to the potential, we will recover the results of the standard theory of CMS on $\Sigma_A$. For many important examples, the GCMS $X_A$ is a compactification of $\Sigma_A$. Then, suppose we define the potential in the entire space $X_A$. In that case, this assumption will impose that any continuous potential has to be bounded. If we want finite Gurevich pressure and consider only bounded potentials, many shifts are excluded. On the other hand, by construction, defining $U$ as the set of elements of $X_A$ which the stem is not empty, we have $\Sigma_A \subseteq U$. So, by defining $U$ as the domain of the potential, we will recover the standard formalism by taking the restriction $F\vert_{\Sigma_A}$ even for unbounded potentials, and we can cover the same class of regularity used in the formalism in the usual symbolic spaces.

\subsection{Conformal measures on locally compact Hausdorff second countable spaces}\\

Fix a locally compact, Hausdorff and second countable topological space $X$ endowed with a local homeomorphism  $\sigma: U \to X$, where $U$ is an open subset of $X$. In addition, consider its respective Renault-Deaconu groupoid $\mathcal{G}(X,\sigma)$, and a continuous potential $F:U \to \mathbb{R}$. Since we are interested in phase transition phenomena, we also consider in the statements the inverse of the temperature $\beta >0$, which will be a factor multiplying $F$. The Ruelle's transformation, which is the generalized version of the Ruelle's operator, is defined as follows:

\begin{definition}\label{def:Ruelle_transformation} The \textit{Ruelle's transformation} is the linear transformation
\begin{align}
    L_{\beta F}:C_c(U) &\to C_c(X) \nonumber\\
    f  & \mapsto L_{\beta F}(f)(x):= \sum_{\sigma(y)=x}e^{\beta F(y)}f(y).\label{eq:general_Ruelle}
\end{align}
\end{definition}


\begin{definition}[Eigenmeasure associated to the Ruelle Transformation] Consider the Borel $\sigma$-algebra $\mathcal{B}_X$. A measure $\mu$ on $\mathcal{B}_X$ is said to be an \textit{eigenmeasure} with eigenvalue $\lambda$ for the Ruelle transformation $L_{\beta F}$ when
\begin{equation}\label{eq:conformal_eigenmeasure_functions_X_A}
    \int_{X} L_{\beta F}(f)(x)d\mu(x) = \lambda \int_{U} f(x)d\mu(x),
\end{equation}
for all $f \in C_c(U)$.
\end{definition}
In other words, the equation \eqref{eq:conformal_eigenmeasure_functions_X_A} can be rewritten by using \eqref{eq:general_Ruelle} as
\begin{equation}\label{eq:conformal_eigenmeasure_functions_X_A_2}
    \int_{X} \sum_{\sigma(y) = x} e^{\beta F(y)}f(y) d\mu(x) = \lambda \int_{U} f(x)d\mu(x),
\end{equation}
for all $f \in C_c(U)$. As in the standard theory of countable Markov shifts, when a measure $\mu$ satisfies the equation \eqref{eq:conformal_eigenmeasure_functions_X_A_2} we write $L_{\beta F}^*\mu = \lambda\mu$.

Now we introduce the notions of conformal measure in the senses of Denker-Urba\'nski and Sarig in the generalized setting. We do not need continuity assumption on the potential for these two notions, and the following definition of conformal measure makes sense for measurable spaces, a more general setting.

\begin{definition} Let $(X,\mathcal{F})$ be a measurable space and $\sigma: U \to X$ a measurable endomorphism. A set $B\subseteq U$ is called \textit{special} when $B \in \mathcal{F}$, $\sigma(B)\in \mathcal{F}$ and $\sigma_B:=\sigma\vert_B: B\to \sigma(B)$ is injective.
\end{definition}

\begin{definition}[Conformal measure - Denker-Urba\'nski] Let $(X,\mathcal{F})$ be a measurable space and $D:U \to [0,\infty)$ also measurable. A measure $\mu$ in $X$ is said to be \emph{$D$-conformal} in the sense of Denker-Urba\'nski if
\begin{equation}\label{eq:conformal_Ur_sets_potential}
    \mu(\sigma(B)) = \int_B D d\mu,
\end{equation}
for every special set $B \subseteq U$.
\end{definition}

Although the previous definition is very general, we will always consider $\mathcal{F} = \mathcal{B}_X$.

\begin{definition}
Given a Borel measure $\mu$ on $X$ we define the measure $\mu\odot\sigma$ on $U$ by
$$\mu\odot\sigma(E):=\sum_{i\in \mathbb{N}}\mu(\sigma(E_i)).$$
For every measurable set $E\subseteq U$, where the $\{E_i\}_{i \in \mathbb{N}}$ is a family of pairwise disjoint special sets such that $E =\sqcup_i E_i$. 
\end{definition}

\begin{remark}
We show that $\mu\odot\sigma$ is well-defined. First we prove the existence of at least one countable family $\{E_i\}$, as above. Indeed, if $E\subseteq U$, since $\sigma$ is a local homeomorphism, for each $x\in$ E there is an open subset $H_x\ni x$ such that $\sigma$ is injective, and we have $E\subseteq \cup_{x\in E} H_x$. For each of those $H_x$, there exists an open basic set $U_x$ such that $x\in U_x$, and we can enumerate $\{U_x\}=\{U_1,U_2,\dots\}$ because the topology basis is countable. Observe that $\sigma$ is injective on each $U_i$. Take $E_1:=E\cap U_1$, $E_n:=E\cap U_n\setminus\bigsqcup_{i=1}^{n-1}E_i$ and we have what we claimed.

Now we shall see that the definition does not depend on the decomposition of $E$. Let $E=\bigsqcup E_i=\bigsqcup F_j$, then $E=\bigsqcup_{i,j}E_i\cap F_j$. Therefore,
$$\sum_i\mu(\sigma(E_i))=\sum_i\mu(\sigma(\sqcup_j E_i\cap F_j))=\sum_i\mu(\sqcup_j\sigma( E_i\cap F_j))=\sum_{i,j}\mu(\sigma(E_i\cap F_j))$$
Doing analogously for $\{F_j\}$ instead of $\{E_i\}$ we conclude the proof since $\sum_i\mu(\sigma(E_i))=\sum_j\mu(\sigma(F_j))$.
\end{remark}

\begin{definition}[Conformal measure - Sarig]
A Borel measure $\mu$ in $X$ is called \textit{$(\beta F, \lambda)$-conformal} in the sense of Sarig if there exists $\lambda >0$ such that
\begin{equation*}
\dfrac{d\mu\odot\sigma}{d\mu}(x)= \lambda e^{-\beta F(x)}\quad \mu - a.e. \; x\in U.
\end{equation*}
\end{definition}


The next theorem states the equivalence among the generalized versions of conformal measures, eigenmeasures and quasi-invariant measures.  

\begin{theorem}\label{thm:equivalences_conformal_measures_generalized_Markov_shift} Let $X$ be locally compact, Hausdorff and second countable space, $U \subseteq X$ open and $\sigma: U \to X$ a local homeomorphism. Let $\mu$ be a Borel measure that is finite on compacts. For a given continuous potential $F:U\to \mathbb{R}$, the following are equivalent.

\begin{itemize}
    \item[$(i)$] $\mu$ is $e^{\beta F}$-conformal measure in the sense of Denker-Urba\'nski;
    \item[$(ii)$] $\mu$ is a eigenmeasure associated with eigenvalue 1 for the Ruelle Transformation $L_{-\beta F}$, that is
    \begin{equation*}
        \int_{X} \sum_{\sigma(y)=x} f(y) e^{-\beta F(y)} d\mu(x) = \int_{U} f(x) d\mu(x),
    \end{equation*}
    for all $f \in C_c(U)$;
    \item[$(iii)$] $\mu$ is  $e^{-\beta c_F}$-quasi-invariant on $\mathcal{G}(X,\sigma)$, i.e
    \begin{equation}\label{eq:quasi-invariant-measure}
\int_{X} \sum_{r(\gamma)=x} e^{\beta c_F(\gamma)} f(\gamma)  d\mu(x) = \int_{X}\sum_{s(\gamma)=x}f(\gamma)  d\mu(x).
    \end{equation}
    for all $f \in C_c(\mathcal{G}(X,\sigma))$;
    \item[$(iv)$] $\mu$ is $(-\beta F, 1)$-conformal in the sense of Sarig.
\end{itemize}
\end{theorem}

\begin{proof} 

$(iii)\implies (ii)$ is analogous to Proposition 4.2 in \cite{Renault2003}, but we repeat the proof. For any $f\in C_c(U)$, consider:
$$\int_X \sum_{\sigma(y)=x}f(y)e^{-\beta F(y)}d\mu(x)=\int_X\sum_{s(\gamma)=x}(fe^{-\beta F}\circ r)(\gamma)\mathbbm{1}_{S}(\gamma)d\mu(x). $$
where $S$ is the set $\{(x,1,\sigma(x)): x\in U\}$. Now, we use that $\mu$ is quasi-invariant to conclude that
$$\int_X\sum_{s(\gamma)=x}(fe^{-\beta F}\circ r)(\gamma)\mathbbm{1}_{S}(\gamma)d\mu(x)=\int_X\sum_{r(\gamma)=x}(fe^{-\beta F})\circ r(\gamma)\mathbbm{1}_{S}(\gamma)e^{\beta c_F(\gamma)}d\mu(x)=\int_U f d\mu. $$
Proving the implication we were interested.
 
For $(ii) \implies (i)$ let $V$ be an open subset of $U$ such that $\sigma\vert_V$ is injective, and let $W = \sigma(V)$. Also denote by $\tau : W \to V$ the inverse of the restriction of $\sigma$ to $V$. We then have two measures of interest on $V$, namely
\begin{equation*}
\tau^*(\mu\vert_W)\quad \text{and} \quad e^{\beta F}\mu\vert_V.    
\end{equation*}
We claim that the above measures on $V$ are equal. By the uniqueness part of the Riesz-Markov theorem, it is enough to prove that
\begin{equation}\label{eq:urbanski_integrated}
    \int_V gd\tau^*(\mu\vert_W)=\int_V ge^{\beta F}d\mu\vert_V,
\end{equation}
for every $g$ in $C_c(V)$. Given such a $g$, we consider its extension to the whole of $U$ by setting it to be zero on $U \setminus V$ . The extended function is then in $C_c(U)$. Defining $f = ge^{\beta F}$ , we then have that

\begin{align*}
    \int_V g e^{\beta F} d\mu\vert_V&=\int_U f d\mu \stackrel{(\ref{thm:equivalences_conformal_measures_generalized_Markov_shift}.ii)}{=}\int_X \sum_{\sigma(y)=x}f(y)e^{-\beta F(y)}d\mu(y)\\
    &=\int_W f(\tau(x))e^{-\beta F(\tau(x))}d\mu(x)=\int_W g(\tau(x))d\mu(x)=\int_V g d\tau^*(\mu).
\end{align*}
This proves \eqref{eq:urbanski_integrated}, and hence also that $\tau^*(\mu\vert_W)=e^{\beta F}\mu\vert_V$. It follows that, for every measurable set $E\subseteq V$,
\begin{equation*}
    \mu(\sigma(E))=\mu(\tau^{-1}(E))=\tau^*(\mu\vert_W)(E)=\int_E e^{\beta F}d\mu.
\end{equation*}
Now, suppose $E\subseteq U$ is a special set. Since $\sigma$ is a local homeomorphism and $X$ is second countable, there exists a countable collection of open sets $\{V_i\}_{i\in \mathbb{N}}$ such that $\sigma\vert_{V_i}$ is injective and $E\subseteq \bigcup_{i\in \mathbb{N}}V_i$. Then we have a countable collection of measurable sets $\{E_i\}_{i\in\mathbb{N}}$, pairwise disjoint, such that $E_i\subseteq V_i$ and $E=\sqcup_{i\in \mathbb{N}}E_i$. We conclude, using that $E$ is special, that
\begin{equation*}
    \mu(\sigma(E))=\sum_{i\in \mathbb{N}}\mu(\sigma(E_i))=\sum_{i\in \mathbb{N}}\int_{E_i}e^{\beta F} d\mu=\int_E e^{\beta F}.
\end{equation*}

$(i)\implies (iii)$. We consider the open bisections defined in the preliminaries $W(n,m,C,B)$. W.l.o.g we can consider $\sigma^n(C)=\sigma^m(B)$, since if not we could take open sets $C'\subseteq C$ and $B'\subseteq B$ such that $\sigma^n(C')=\sigma^n(C)\cap \sigma^m(B)=\sigma^m(B')$. Also, we can suppose that $\sigma^n$ is injective when restricted $C$, similarly for $\sigma^m$ and $B$. In this setting, we can define the map $\sigma^{n-m}_{CB}:=\sigma^{-m}_B\circ\sigma^n_C$ and similarly $\sigma^{m-n}_{BC}:=\sigma^{-n}_C\circ \sigma^{m}_B$.
\[
\begin{tikzcd}
 & \sigma^n(C)=\sigma^m(B)  \\
C \arrow{ur}{\sigma^n_C} \arrow[rr,dashed,"\sigma^{n-m}_{CB}"] && B \arrow[ul,"\sigma^m_B"'] \arrow[ll,dashed,"\sigma^{m-n}_{BC}", bend left]
\end{tikzcd}
\]
Let $f \in C_c(\mathcal{G}(X,\sigma)$ s.t. $supp(f)\subseteq W(n,m,C,B)$. Let us see how the equation \eqref{eq:quasi-invariant-measure} on item $(iii)$ simplifies in this case. Observe first the left hand side. If $x \notin C$, clearly there is no $\gamma\in W(n,m,C,B)$ such that $r(\gamma)=x$, so the integration can be done in $C$. Now, for $x\in C$, consider $\gamma_1,\gamma_2\in W(n,m,C,B)$ such that $r(\gamma_1)=r(\gamma_2)=x$. Since the range map is injective in such set we have $\gamma_1=\gamma_2$ and we conclude the summation on the left hand side of equation \eqref{eq:quasi-invariant-measure} have at most one non-zero term for each $x\in C$. Denoting this term by $\gamma_x$, we see this term is written as $\gamma_x=(x,n-m,\sigma_{CB}^{n-m}(x))$. So, the left hand side of equation \eqref{eq:quasi-invariant-measure} is
\begin{equation}\label{eq:theorem3,LHS}
\int_C e^{\beta c_F(\gamma_x)} f(\gamma_x)  d\mu(x)=\int_C e^{\beta c_F((x,n-m,\sigma_{CB}^{n-m}(x)))} f(x,n-m,\sigma_{CB}^{n-m}(x))  d\mu(x).    
\end{equation}
Calculation on the right hand side of equation \eqref{eq:quasi-invariant-measure} is done in a similar fashion, we have:
\begin{equation}\label{eq:theorem3,RHS}
    \int_B f(\sigma_{BC}^{m-n}(y),n-m,y)d\mu(y).
\end{equation}
Now let $g:C\to \mathbb{C}$ defined by $g(x)=f(x,n-m,\sigma_{CB}^{n-m}(x)).$ Observe that $g(\sigma^{m-n}_{BC}(y))=f(\sigma^{m-n}_{BC}(y),n-m,y)$, which is the function in the equation \eqref{eq:theorem3,RHS}. We rewrite the quasi-invariant condition with the considerations from above
\begin{equation}\label{eq:QIC}
    \int_C e^{\beta c_F(x,n-m,\sigma^{n-m}(x))}g(x)d\mu(x)=\int_B g(\sigma^{m-n}_{BC}(y))d\mu(y),
\end{equation}
for $g\in C_c(C)$. We just need to prove that item $(i)$ implies equation \eqref{eq:QIC}. First, observe that item $(i)$ implies for all $C$ open subset of $U$, $\sigma\vert_C$ injective that
\begin{equation}\label{eq:urbanski_before_change_of_wariables}
\int_C g(x) e^{\beta F(x)} d\mu(x) = \int_{\sigma(C)} g(\sigma^{-1}(x)) d\mu(x),    
\end{equation}
for every $g\in C_c(C)$. Note that equation \eqref{eq:urbanski_before_change_of_wariables} is equation \eqref{eq:QIC} when $m=0$ and $n=1$. To prove \eqref{eq:QIC} we proceed by induction on $n+m$. If $n+m=0$ we have $C=B$ , $\sigma_{BC}^{n-m}=Id$ and $c_F(x,0,x)=0$, so equation (\ref{eq:QIC}) is satisfied. Take $n\neq 0$.
\begin{figure}[H]
    \centering
    \[
    \begin{tikzcd}[row sep=5em]
    & \mathbb{C} &  \\[-2em]
    C\arrow[ur,"g"]\arrow[rr,"\sigma"]\arrow[d,"\sigma^{n-m}_{CB}"']& & \sigma(C)=C'\arrow[ul,"g'"']\arrow[d,"\sigma^{n-1}"]\\
    B\arrow[urr,"\sigma^{m-(n-1)}_{BC'}",dashed]\arrow[rr,"\sigma^m_B"']& & \sigma^n(C)=\sigma^m(B)
    \end{tikzcd}
    \]
    \caption{\label{fig:diagram_maps_proof_equivalences}}
\end{figure}
 \noindent
Let $g':C'=\sigma(C)\to  \mathbb{C}$ defined by $g'(x)=g(\sigma^{-1}(x))$, $g'\in C_c(C')$. By induction hypothesis,
$$\int_B g'(\sigma_{BC'}^{m-(n-1)}(y))d\mu(y)=\int_{C'}e^{\beta c_F(x,n-1-m,\sigma^{(n-1)-m}(x))}g'(x)d\mu(x).$$
On the other hand, using as reference the figure \ref{fig:diagram_maps_proof_equivalences},
$$\int_B g'(\sigma_{BC'}^{m-(n-1)}(y))d\mu(y)=\int_B g(\sigma^{-1}\sigma_{BC'}^{m-(n-1)}(y))d\mu(y)=\int_B g(\sigma_{BC}^{m-n}(y))d\mu(y),$$
which is the right hand side of equation \eqref{eq:QIC}. Then,
\begin{equation}\label{eq:g_2}
\int_{C'}e^{\beta c_F(x,n-1-m,\sigma^{(n-1)-m}(x))}g'(x)d\mu=\int_{\sigma(C)}\underbrace{e^{\beta c_F(x,n-1-m,\sigma^{(n-1)-m}(x))}g(\sigma^{-1}(x))}_{g_2(x)}d\mu.    
\end{equation}
The equation \eqref{eq:urbanski_before_change_of_wariables}, by a change of variables, can be written as
$$\int_C g_2(\sigma(x))e^{\beta F(x)} d\mu(x)=\int_{\sigma(C)}g_2(x)d\mu\quad \forall g_2\in C_c(\sigma(C)). $$
Applying the identity above to \eqref{eq:g_2}, we obtain
\begin{align*}
\int_{C'}e^{\beta c_F(x,n-1-m,\sigma^{(n-1)-m}(x))}g'(x)d\mu(x)&=\int_C g_2(\sigma(x))e^{\beta F(x)} d\mu(x)\\
&= \int_C e^{\beta c_F(\sigma(x),n-1-m,\sigma^{(n-1)-m}(\sigma(x)))}g(x)e^{\beta F(x)}d\mu(x).
\end{align*}
By left hand side of equation \eqref{eq:QIC}, it is enough to verify that
$$c_F(\sigma(x),n-1-m,\sigma^{n-1-m}(\sigma(x)))+F(x)=c_F(x,n-m,\sigma^{n-m}(x)).$$
It is true by the cocycle property of $c_F$ and the fact that $F(x)=c_F(x,1,\sigma(x))$ along with the observation that $$(x,1,\sigma(x))(\sigma(x),n-1-m,\sigma^{n-m}(x))=(x,n-m,\sigma(x)).$$
The implication $(i)\implies (iii)$ is proved for $f$ supported on the open bisection $W(n,m,C,B)$, and therefore it is also proved for every $f\in C_c(\mathcal{G}(X,\sigma)).$
 \noindent

Now, we prove $(i) \iff (iv)$. Suppose that $\dfrac{d\mu\odot\sigma}{d\mu}=e^{\beta F}\quad \mu - a.e.\; x\in U$. Take $E\subseteq U$ such that $\sigma|_E$ is injective. Then,
$$\mu(\sigma(E))=\mu\odot\sigma(E)=\int_U \mathbbm{1}_{E}\,d\mu\odot\sigma=\int_U \mathbbm{1}_E e^{\beta F(x)}d\mu(x)=\int_E e^{\beta F(x)} d\mu(x)$$
and we have proved item $(i)$.
 \noindent
Now the converse. Let $E\subseteq$ U and $\{E_i\}_{i\in \mathbb{N}}$ be its decomposition. Hence,
\begin{align*}
\mu\odot\sigma(E)= \sum_{i\in \mathbb{N}}\mu(\sigma(E_i))=\sum_{i\in \mathbb{N}}  \int_X \mathbbm{1}_{E_i} e^{\beta F(x)}\mu(x)=\int_E e^{\beta F(x)} d\mu(x). \nonumber
\end{align*}
Since this is true for every measurable set $E\subseteq U$, we have
$$\dfrac{d\mu\odot\sigma}{d\mu}(x)=e^{\beta F(x)}\quad \mu - a.e. \; x\in U.$$
\end{proof}

\begin{remark}\label{remark:KMS_quasi_invariant} The above theorem is of particular interest, since it is known that if a measure $\mu$ is $e^{-\beta c_F}$-quasi-invariant on $\mathcal{G}(X,\sigma)$, then the state defined by 
\begin{equation}
\varphi_\mu(f)=\int_X f(x,0,x)d\mu(x),\quad f\in C_c(\mathcal{G}(X,\sigma))
\end{equation}
is a KMS$_\beta$ state of the full groupoid $C^*$-algebra $C^*(\mathcal{G}(X,\sigma))$ for the one parameter group of automorphisms $(\eta_t)$ defined as $\eta_t(f)(\gamma)=e^{itc_F(\gamma)}f(\gamma),$ for $f\in C_c(\mathcal{G}(X,\sigma)), \gamma\in \mathcal{G}(X,\sigma)$.

When the \textit{isotropy} subgroupoid  $\text{Iso}(\mathcal{G}(X,\sigma)):=\{\gamma\in \mathcal{G}(X,\sigma): r(\gamma)=s(\gamma) \}$ coincides with $\mathcal{G}^{(0)}$, we say that a groupoid is \textit{principal}.  Observe that the units of $c_F^{-1}(0)$ are $X$.

If the subgroupoid $c_F^{-1}(0)$ is principal, then every KMS$_\beta$ state of $C^*(\mathcal{G}(X,\sigma))$ for $\eta$ have the form $\varphi_\mu$, for a $\mu$ $e^{-\beta c_F}$-quasi-invariant measure, see Theorem 3.3.12 of \cite{Renault2009} for these assertions. For the same result on the finite alphabet case, but considering eigenmeasures instead quasi-invariant measures, see \cite{KesStadStrat2007}.

As an example, suppose $F>0$. Then for $\gamma=(x,n-m,x)\in \text{Iso}(c_F^{-1}(0))$, we have
$$\sum_{i=0}^{n-1}F(\sigma^i(x))=\sum_{i=0}^{m-1}F(\sigma^i(x)),$$
which can only happens if $n=m$, thus $\gamma=(x,0,x)\in X$, so $c_F^{-1}(0)$ is principal.
\end{remark}

From now on, we will develop the thermodynamic formalism for GCMS, then we fix:
\newline

\textbf{(1)} The space $X$ is the GCMS $X_A$ constructed on section \ref{section:X_A}, for an irreducible matrix $A$.\\

\textbf{(2)} The open set $U \subseteq X$ is set of \emph{shiftable} points of $X_A$, see Definition \ref{generalized_shift_map}, so $U$ coincides with the union of generalized cylinders (Definition \ref{def:generalized_cylinders}), that is, $U = \displaystyle\cup_{n\in \mathbb{N}} C_n$. Note that $\Sigma_A \subseteq U$.\\

\textbf{(3)} The local homeomorphism is the shift map $\sigma:U\to X_A$, see definition \ref{generalized_shift_map}.
\newline

A natural question is if the conformal measures in $X_A$ when restricted to $\mathcal{B}_{\Sigma_A}$ are also conformal measures for the standard CMS $\Sigma_A$, the next results answer this type of question. We omit the proofs since they only use standard arguments of the thermodynamic formalism on countable Markov shifts. For the details, see \cite{Raszeja2020}.

\begin{proposition}\label{prop:eigen_measure_nonsingular_Generalized_Markov_shift} Every eigenmeasure on $X_A$ is non-singular in $U$.
\end{proposition}

\begin{proposition}\label{thm:restriction_eigenmeasures} For $X_A$, let $F:U \to \mathbb{R}$ be a mensurable potential and $\lambda > 0$ such that there exists a $(F,\lambda)$-conformal measure $\mu$ on $X_A$. Let $\mu_{\Sigma_A}$ the restriction of $\mu$ to $\mathcal{B}_{\Sigma_A}$, defined by $\mu_{\Sigma_A}(E) := \mu(E\cap \Sigma_A), E \in \mathcal{B}_{X_A}$.
Then, $\mu_{\Sigma_A}$ is also a $(F,\lambda)$-conformal measure. 
Moreover, $\mu_{\Sigma_A}\vert_{\mathcal{B}_{\Sigma_A}}$ is a $(F\vert_{\Sigma_A},\lambda)$-conformal measure. The analogous result for $\mu_{Y_A}$, the restriction of $\mu$ to $\mathcal{B}_{Y_A}$, is also true.
\end{proposition}

\begin{corollary} \label{cor:decomposition_eigenmeasures} Consider the space $X_A$, a mensurable potential $F:U \to \mathbb{R}$, and $\lambda > 0$. Then every $(F,\lambda)$-conformal measure can be decomposed into the linear combination of the two $(F,\lambda)$-conformal measures $\mu = \mu_{\Sigma_A} + \mu_{Y_A}$ as defined in Proposition \ref{thm:restriction_eigenmeasures}.
\end{corollary}

\begin{proposition} \label{prop:extension_eigenmeasures}
Let $F:\Sigma_A \to \mathbb{R}$ be a mensurable potential and $\lambda > 0$ such that there exists a $(F,\lambda)$-conformal measure $\mu$ on $\Sigma_A$. Let $\widetilde{F}$ an extension of $F$ to $U$. Consider its natural extension $\mu_{\text{ext},\Sigma_A}$ on $X_A$, defined by $\mu_{\text{ext},\Sigma_A}(E) := \mu(E\cap \Sigma_A), E \in \mathcal{B}_{X_A}.$
    Then, $\mu_{\text{ext},\Sigma_A}$ is a $(\widetilde{F},\lambda)$-conformal measure on $X_A$. The analogous result for a potential $F:Y_A \to \mathbb{R}$ is also true.
\end{proposition}

\subsection{Pressures on GCMS}\label{sec:pressures}\\

The pressure function is one of the central objects in thermodynamic formalism. Since the Gurevich pressure is the most common definition in the literature of countable Markov shifts, given a potential $F: U \to \mathbb{R}$, it is desirable that the pressure for $F$ be an extension of the Gurevich pressure for $F\vert_{\Sigma_A}:\Sigma_A \to \mathbb{R}$. The main result of this section is to present a proof of this fact, which is valid for the same class of regularity of potentials as we have well-established results for thermodynamic formalism on CMS, and for a class of irreducible matrices which includes the main of the examples as the full shift, renewal shift and others.

We can not use the definition of Gurevich pressure for $F: U \to \mathbb{R}$ directly since the definition uses periodic orbits, and we can have finite words on $Y_A$, points on which we can apply the shift map only a finite number of times. To make the bridge between the standard and generalized settings, we use the concept of pressure at a point $x \in X_A$, denoted by $P(F,x)$, defined by M. Denker and M. Yuri in the context of \textit{Iterated Function System} (IFS), see \cite{DenYu2015}. We compare $P(F,x)$ with the Gurevich pressure. The pressure $P(F,x)$ is introduced for IFS defined in a Polish space $X$ and a family $\mathcal{V}$ of homeomorphisms $v:D(v)\to v(D(v))\subset X$, where each $D(v)\subset X$ is a closed nonempty subset. For us, the Polish space $X$ is the space $X_A$ and the family of homeomorphisms are chosen to be $\{(\sigma\vert_{C_i})^{-1}: i\in \mathbb{N}\}$, the inverses of the shift map when restricted to the generalized  cylinder sets $C_i := \{\xi \in X_A: \xi_i =1\}$.


\begin{definition} Given $\beta > 0$ and a potential $F:U \to \mathbb{R}$ and $x \in X_A$, the \emph{n-th patition function at the point }$x$ is defined by
\begin{equation*}
    Z_n(\beta F,x) = \sum_{\sigma^n(y)=x}e^{\beta F_n(y)},
\end{equation*}
where $F_n$ is the Birkhoff sum of $F$. The \emph{pressure at the point }$x$ is
\begin{equation*}
    P(\beta F,x) :=\limsup_{n\to \infty} \frac{1}{n} \log Z_n(\beta F,x).
\end{equation*}
\end{definition}

Now, we present the notion of $s$-compactness on $X_A$, which is a sufficient condition for the pressure at each point to coincide with the Gurevich pressure for a wide class of potentials when $X_A$ is compact, including the Walters potentials with the finite first variation.

\begin{definition} Suppose $X_A \neq \Sigma_A$. We say that $X_A$ is \textit{$s$-compact} when, for every set $\mathfrak{R} \subseteq \mathbb{N}$ s.t. $\mathfrak{R} = R_{\xi^0}(e)$, for some empty-stem configuration $\xi^0$, there exists a finite set $I_{\mathfrak{R}}$ s.t. $\mathfrak{R} \subseteq \bigcup_{i \in I_{\mathfrak{R}}} s_A(i)$.
\end{definition}

\begin{figure}[H]
    \centering
\scalebox{0.6}{
\tikzset{every picture/.style={line width=0.75pt}} 

\begin{tikzpicture}[x=0.75pt,y=0.75pt,yscale=-1,xscale=1]

\draw  [color={rgb, 255:red, 0; green, 0; blue, 0 }  ,draw opacity=1 ][line width=1.5]  (114.5,147) .. controls (132.5,140) and (238.5,120) .. (291.5,179) .. controls (344.5,238) and (488.5,281) .. (424.5,347) .. controls (360.5,413) and (197.5,425) .. (132.5,385) .. controls (67.5,345) and (40.5,213) .. (52.5,183) .. controls (64.5,153) and (96.5,154) .. (114.5,147) -- cycle ;
\draw  [color={rgb, 255:red, 0; green, 0; blue, 0 }  ,draw opacity=1 ][line width=1.5]  (327.5,30) .. controls (429.5,53) and (593.5,55) .. (605.5,143) .. controls (617.5,231) and (538.5,324) .. (450.5,372) .. controls (362.5,420) and (259.5,503) .. (120.5,407) .. controls (-18.5,311) and (-25.5,99) .. (80.5,47) .. controls (186.5,-5) and (225.5,7) .. (327.5,30) -- cycle ;
\draw  [color={rgb, 255:red, 0; green, 0; blue, 0 }  ,draw opacity=1 ][fill={rgb, 255:red, 0; green, 0; blue, 0 }  ,fill opacity=1 ] (293,267.5) .. controls (293,269.99) and (295.01,272) .. (297.5,272) .. controls (299.99,272) and (302,269.99) .. (302,267.5) .. controls (302,265.01) and (299.99,263) .. (297.5,263) .. controls (295.01,263) and (293,265.01) .. (293,267.5) -- cycle ;
\draw [color={rgb, 255:red, 0; green, 0; blue, 0 }  ,draw opacity=1 ][line width=1.5]    (416.81,129.24) .. controls (370.66,117.11) and (293.62,152.42) .. (297.5,263) ;
\draw [shift={(421,130.5)}, rotate = 198.81] [fill={rgb, 255:red, 0; green, 0; blue, 0 }  ,fill opacity=1 ][line width=0.08]  [draw opacity=0] (13.4,-6.43) -- (0,0) -- (13.4,6.44) -- (8.9,0) -- cycle    ;
\draw  [color={rgb, 255:red, 0; green, 0; blue, 0 }  ,draw opacity=1 ][line width=1.5]  (151.5,172) .. controls (169.5,165) and (265.5,156) .. (283.5,193) .. controls (301.5,230) and (282.5,282) .. (243.5,275) .. controls (204.5,268) and (150.5,253) .. (115.5,230) .. controls (80.5,207) and (133.5,179) .. (151.5,172) -- cycle ;
\draw  [color={rgb, 255:red, 0; green, 0; blue, 0 }  ,draw opacity=1 ][fill={rgb, 255:red, 0; green, 0; blue, 0 }  ,fill opacity=1 ] (380,332.5) .. controls (380,334.99) and (382.01,337) .. (384.5,337) .. controls (386.99,337) and (389,334.99) .. (389,332.5) .. controls (389,330.01) and (386.99,328) .. (384.5,328) .. controls (382.01,328) and (380,330.01) .. (380,332.5) -- cycle ;
\draw  [color={rgb, 255:red, 0; green, 0; blue, 0 }  ,draw opacity=1 ][fill={rgb, 255:red, 0; green, 0; blue, 0 }  ,fill opacity=1 ] (201.37,196.01) .. controls (200.3,198.25) and (201.24,200.94) .. (203.48,202.02) .. controls (205.72,203.09) and (208.41,202.15) .. (209.49,199.91) .. controls (210.56,197.67) and (209.62,194.98) .. (207.38,193.9) .. controls (205.14,192.83) and (202.45,193.77) .. (201.37,196.01) -- cycle ;
\draw [color={rgb, 255:red, 0; green, 0; blue, 0 }  ,draw opacity=1 ][line width=1.5]    (409.78,99.83) .. controls (359.94,72.66) and (252.31,86.77) .. (207.38,193.9) ;
\draw [shift={(413.5,102)}, rotate = 212.01] [fill={rgb, 255:red, 0; green, 0; blue, 0 }  ,fill opacity=1 ][line width=0.08]  [draw opacity=0] (13.4,-6.43) -- (0,0) -- (13.4,6.44) -- (8.9,0) -- cycle    ;
\draw  [color={rgb, 255:red, 0; green, 0; blue, 0 }  ,draw opacity=1 ][fill={rgb, 255:red, 0; green, 0; blue, 0 }  ,fill opacity=1 ] (229,236.5) .. controls (229,238.99) and (231.01,241) .. (233.5,241) .. controls (235.99,241) and (238,238.99) .. (238,236.5) .. controls (238,234.01) and (235.99,232) .. (233.5,232) .. controls (231.01,232) and (229,234.01) .. (229,236.5) -- cycle ;
\draw [color={rgb, 255:red, 0; green, 0; blue, 0 }  ,draw opacity=1 ][line width=1.5]    (382.31,112.31) .. controls (336.62,116.15) and (274.24,135.97) .. (233.5,232) ;
\draw [shift={(386.5,112)}, rotate = 176.27] [fill={rgb, 255:red, 0; green, 0; blue, 0 }  ,fill opacity=1 ][line width=0.08]  [draw opacity=0] (13.4,-6.43) -- (0,0) -- (13.4,6.44) -- (8.9,0) -- cycle    ;
\draw  [color={rgb, 255:red, 0; green, 0; blue, 0 }  ,draw opacity=1 ][fill={rgb, 255:red, 0; green, 0; blue, 0 }  ,fill opacity=1 ] (322,318.5) .. controls (322,320.99) and (324.01,323) .. (326.5,323) .. controls (328.99,323) and (331,320.99) .. (331,318.5) .. controls (331,316.01) and (328.99,314) .. (326.5,314) .. controls (324.01,314) and (322,316.01) .. (322,318.5) -- cycle ;
\draw [color={rgb, 255:red, 0; green, 0; blue, 0 }  ,draw opacity=1 ][line width=1.5]    (384.5,328) .. controls (360.98,240.78) and (376.84,183.33) .. (429.25,149.07) ;
\draw [shift={(432.5,147)}, rotate = 508.28] [fill={rgb, 255:red, 0; green, 0; blue, 0 }  ,fill opacity=1 ][line width=0.08]  [draw opacity=0] (13.4,-6.43) -- (0,0) -- (13.4,6.44) -- (8.9,0) -- cycle    ;
\draw  [color={rgb, 255:red, 0; green, 0; blue, 0 }  ,draw opacity=1 ][line width=1.5]  (387.5,86) .. controls (405.5,79) and (501.5,70) .. (519.5,107) .. controls (537.5,144) and (518.5,196) .. (479.5,189) .. controls (440.5,182) and (386.5,167) .. (351.5,144) .. controls (316.5,121) and (369.5,93) .. (387.5,86) -- cycle ;
\draw [color={rgb, 255:red, 0; green, 0; blue, 0 }  ,draw opacity=1 ][line width=1.5]    (326.5,314) .. controls (302.98,226.78) and (336.12,179.9) .. (389.22,146.06) ;
\draw [shift={(392.5,144)}, rotate = 508.28] [fill={rgb, 255:red, 0; green, 0; blue, 0 }  ,fill opacity=1 ][line width=0.08]  [draw opacity=0] (13.4,-6.43) -- (0,0) -- (13.4,6.44) -- (8.9,0) -- cycle    ;

\draw (142,186.4) node [anchor=north west][inner sep=0.75pt]  [font=\Huge,color={rgb, 255:red, 0; green, 0; blue, 0 }  ,opacity=1 ]  {$\mathfrak{R}$};
\draw (200,28.4) node [anchor=north west][inner sep=0.75pt]  [font=\Huge,color={rgb, 255:red, 0; green, 0; blue, 0 }  ,opacity=1 ]  {$\mathbb{N}$};
\draw (448,94.4) node [anchor=north west][inner sep=0.75pt]  [font=\Huge,color={rgb, 255:red, 0; green, 0; blue, 0 }  ,opacity=1 ]  {$I_{\mathfrak{R}}$};
\draw (438,150) node [anchor=north west][inner sep=0.75pt]  [font=\large,color={rgb, 255:red, 255; green, 255; blue, 255 }  ,opacity=1 ] [align=left] {finite set};
\draw (139,301.4) node [anchor=north west][inner sep=0.75pt]  [font=\Huge,color={rgb, 255:red, 0; green, 0; blue, 0 }  ,opacity=1 ]  {$\bigcup _{i\in I_{\mathfrak{R}}} s_{A}( i)$};

\end{tikzpicture}
}
\caption{The $s$-compactness property. In the picture $\mathfrak{R} = R_{\xi^0}(e)$, for some empty-stem configuration $\xi^0$.\label{fig:s_compactness}}
\end{figure}

\begin{remark}
As in the definition of compactness ($r$-compactness), we have an easy and sufficient condition for the $s$-compactness property: if the matrix $A$ has a column of 1's, then the space $X_A$ is $s$-compact.
\end{remark}

\begin{remark}
In the rest of the paper, given a potential $F:U \to \mathbb{R}$ we observe that all the notions respect to the regularity of $F$ can be extend from $\Sigma_A$ to $X_A$. In fact, we use the definitions of the $n$-variations and regularity of the potential using the generalized cylinders $C_{\alpha}$ in $X_A$ instead of standard cylinders $[\alpha]$ in $\Sigma_A$, where $\alpha$ is a positive word. Then, summable variations, Walters' etc, also make sense for for $X_A$, and we recover the usual definition with the restriction of $F$ to $\Sigma_A$.
\end{remark}

\begin{theorem}\label{thm:equality_pressures_general} Let $A$ be irreducible, such that $X_A$ is compact and $s$-compact. Consider $F: U \to \mathbb{R}$ a potential with uniform bounded distortion. Then, for every $x \in X_A$, we have
\begin{equation*}
    P(\beta F,x) = P_G(\beta F\vert_{\Sigma_A}).
\end{equation*}
\end{theorem}

\begin{proof} Since $F$ has uniform bounded distortion, take $\sup_{n \in \mathbb{N}} \Var_{n}F_n = K< \infty$.

First, we prove that $P(\beta F,x) \leq P_G(\beta F\vert_{\Sigma_A})$. We divide the proof in 2 cases.
\begin{itemize}
    \item[\textbf{Case 1:}] Let $x \in U$, we have $
        Z_n(\beta F,x) = \displaystyle\sum_{\sigma^n y = x} e^{\beta F_n(y)} = \sum_{\substack{|\omega| =n,\\ A(\omega_{n-1},x_0) = 1}} e^{\beta F_n(\omega x)}.$\\
        
    Since $X_A$ is compact, by Lemma \ref{lemma:equivalence_compactness_X_A_and_graph} we have $r$-compactness, then there exists a finite collection $\{i_1,...,i_m\}\subset \mathbb{N}$ s.t. $\mathbb{N} = \bigcup_{k = 1}^m r_A(i_k)$, so $\omega_0 \in r_A(i_k)$ for some $k$. On the other hand, by transitivity, there exists a word $\eta^{k,x_0}$ s.t. $x_0 \eta^{k,x_0} i_k$ is also admissible. By the uniform bounded distortion of $F$, we have that
    \begin{equation*}
        F_n(\omega x) \leq F_n(\overline{\omega x_0 \eta^{k,x_0} i_k}) + K,
    \end{equation*}
    where $\overline{\omega x_0 \eta^{k,x_0} i_k}$ is the periodic configuration that repeats the word $\omega x_0 \eta^{k,x_0} i_k$. Hence,
    \begin{align}\label{upper_partition}
        Z_n(\beta F,x) &\leq \sum_{k=1}^m \sum_{\substack{|\omega| =n,\\ A(\omega_{n-1},x_0) = 1, \\ A(i_k, \omega_0) = 1}} e^{\beta F_n(\overline{\omega x_0 \eta^{k,x_0} i_k}) + \beta K}.
    \end{align}
By the identity $F_n(\overline{\omega x_0 \eta^{k,x_0} i_k}) = F_{n+2+|\eta^{k,x_0}|}(\overline{\omega x_0 \eta^{k,x_0} i_k}) - F_{2+|\eta^{k,x_0}|}(\overline{x_0 \eta^{k,x_0} i_k \omega})$, and taking $y^{(k,x_0)} \in C_{x_0 \eta^{k,x_0} i_k}$ we have
    \begin{equation}\label{lower_periodic}
         F_{2+|\eta^{k,x_0}|}(y^{(k,x_0)}) -  K \leq F_{2+|\eta^{k,x_0}|}( \overline{x_0 \eta^{k,x_0} i_k \omega}).
    \end{equation}
  Therefore, using (\ref{upper_partition}) and (\ref{lower_periodic}), we obtain
    \begin{align*}
        Z_n(\beta F,x) &\leq \sum_{k=1}^m \sum_{\substack{|\omega| =n,\\ A(\omega_{n-1},x_0) = 1, \\ A(i_k, \omega_0) = 1}} \exp \left(\beta F_{n+2+|\eta^{k,x_0}|}(\overline{\omega x_0 \eta^{k,x_0} i_k}) + 2\beta K - \beta F_{2+|\eta^{k,x_0}|}(y^{(k,x_0)})\right)\\
        &\leq \sum_{k=1}^m e^{ 2\beta K - \beta F_{2+|\eta^{k,x_0}|}(y^{(k,x_0)})} Z_{n+2+|\eta^{k,x_0}|}(\beta F\vert_{\Sigma_A},[i_k]).
    \end{align*}
    Then, we have
    \begin{align*}
        \frac{1}{n} \log Z_n(\beta F,x) &\leq \frac{1}{n} \log \left(\sum_{k=1}^m e^{ 2\beta K - \beta F_{2+|\eta^{k,x_0}|}(y^{(k,x_0)})} Z_{n+2+|\eta^{k,x_0}|}(\beta F,[i_k]) \right)\\
        &\leq \frac{2\beta K}{n} + \frac{1}{n} \log \left[m \max_{1 \leq k \leq m}\left( e^{ - \beta F_{2+|\eta^{k,x_0}|}(y^{(k,x_0)})}\right) \max_{1 \leq k \leq m}\left(Z_{n+2+|\eta^{k,x_0}|}(\beta F,[i_k])\right) \right]\\
    \end{align*}
    By taking $\limsup$ in $n$ we conclude that $P(\beta F,x) \leq P_G(\beta F\vert_{\Sigma_A})$.
    \item[\textbf{Case 2:}] $x \in X_A\backslash U$. In this case $x$ is an empty-stem configuration $x = \xi^0$, so let $\mathfrak{R} = R_{\xi^0}(e)$. Then, by $s$-compactness of $X_A$, we have
    \begin{align*}
        Z_n(\beta F,x) &= \sum_{\sigma^ny = x} e^{\beta F_n(y)} = \sum_{\substack{|\omega| =n,\\ \omega_{n-1} \in \mathfrak{R}}} e^{\beta F_n(\omega)}\\ &\leq \sum_{j \in I_{\mathfrak{R}}} \sum_{\substack{|\omega| =n,\\ \omega_{n-1} \in \mathfrak{R} \\ A(\omega_{n-1},j) = 1}} e^{\beta F_n(\omega)}\\ &\leq
        \sum_{j \in I_{\mathfrak{R}}}  \sum_{\ell=1}^{m}  \sum_{\substack{|\omega| =n,\\ A(i_{\ell},\omega_0) = 1 \\ A(\omega_{n-1},j) = 1}} e^{\beta F_n(\omega)}.
    \end{align*}

  Since $X_A$ is compact, in last inequality we used the lemma \ref{lemma:equivalence_compactness_X_A_and_graph} and $r$-compactness of the space. From this point, we use the same arguments as in the previous case to bound the sum above.
\end{itemize}
Now we prove that $P(\beta F,x) \geq P_G(\beta F)$. Again we divide the proof into two cases:
\begin{itemize}
    \item[\textbf{Case 1:}] $x \in U$. We have
    \begin{align*}
        Z_n(\beta F,x) &= \sum_{\sigma^n y = x} e^{\beta F_n(y)} = \sum_{\substack{|\omega| =n,\\ A(\omega_{n-1},x_0) = 1}} e^{\beta F_n(\omega x)} \geq \sum_{\substack{|\omega| =n,\\ A(\omega_{n-1},x_0) = 1 \\ A(x_0,\omega_0) =1 }} e^{\beta F_n(\omega x)}.
    \end{align*}
    By regularity of $F$, $F_n(\omega x) \geq F_n(\overline{\omega x_0}) - K$, and then
    \begin{align*}
        Z_n(\beta F,x) &\geq \sum_{\substack{|\omega| =n,\\ A(\omega_{n-1},x_0) = 1 \\ A(x_0,\omega_0) =1 }} e^{\beta F_n(\overline{\omega x_0}) - \beta K} = \sum_{\substack{|\omega| =n,\\ A(\omega_{n-1},x_0) = 1 \\ A(x_0,\omega_0) =1 }} e^{\beta F_{n+1}(\overline{\omega x_0}) - \beta F(\overline{x_0\omega})- \beta K}.
    \end{align*}
    Fixed $y^{x_0}\in C_{x_0}$ we have $F(x_0\overline{\omega x_0}) \leq F(y^{x_0}) + K$, and hence
    \begin{align*}
        Z_n(\beta F,x) &\geq e^{- \beta F(y^{x_0})- 2\beta K}\sum_{\substack{|\omega| =n,\\ A(\omega_{n-1},x_0) = 1 \\ A(x_0,\omega_0) =1 }} e^{\beta F_{n+1}(\overline{\omega x_0}) } \\
        &= e^{- \beta F(y^{x_0})- 2\beta K} Z_{n+1}(\beta F,[x_0]).
    \end{align*}
    
    By taking the $\frac{1}{n}\log$ and the $\limsup$ in $n$ we obtain $P(\beta F,x) \geq P_G(\beta F\vert_{\Sigma_A})$.
    \vspace{.2cm}
    
    \item[\textbf{Case 2:}] Let $x \in X_A\backslash U$. Again, $x = \xi^0$ is an empty-stem configuration, and since $X_A$ is compact, we have that $\mathfrak{R} = R_{\xi^0}(e) \neq \emptyset$. Fix $i \in \mathfrak{R}$, we have
    \begin{align*}
        Z_n(\beta F,x) &= \sum_{\sigma^n y = x} e^{\beta F_n(y)} =  \sum_{\substack{|\omega| =n,\\ \omega_{n-1}\in \mathfrak{R}}} e^{\beta F_n(\omega )} \geq  \sum_{\substack{|\omega| =n-1,\\ A(\omega_{n-2},i)=1}} e^{\beta F_n(\omega i)},
    \end{align*}
     The rest of the proof follows similarly as in the previous case. 
\end{itemize}
\end{proof}

\section{Conformal measures on GCMS}\label{new_measures}

In this section we will study conformal and eigenmeasures in GCMS $X_A = \Sigma_A \cup Y_A$. The main goal is to study the conformal probability measures which give mass to the new part of the space $Y_A$, and highlight the connection of these measures with the standard conformal measures. In order to do this, we define the notion of a measure \emph{to live} on a set: 

\begin{definition} Given $(X,\mathcal{F},\mu)$ a measure space, where $\mu$ is a positive measure, we say that $\mu$ \textit{lives} on $B \in \mathcal{F}$ if $\mu(B) > 0$ and $\mu(B^c) = 0$.
\end{definition}

Now we characterize measures that live on $Y_A$ by using $Y_A$-families.  Every non-zero positive Borel measure $\mu$ on $X_A$ that lives on $Y_A(\xi^{0,\mathfrak{e}})$ for some $\mathfrak{e}$ is necessarily an atomic measure. In other words, for each measurable set $E \subseteq X_A$ we can write  $\mu(E) = \sum_{\omega \in \mathfrak{W}_\mathfrak{e}} \mathbbm{1}_{E}(\omega) c_\omega$, where
\begin{equation}\label{eq:coefficients_Y_A_confomal}
    c_\omega := \mu(\{\omega\}), \quad \omega \in \mathfrak{W}_\mathfrak{e},
\end{equation}
and again we identify each configuration in $Y_A(\xi^{0,\mathfrak{e}})$ with its stem. This identification is possible since there exists a bijection between $Y_A(\xi^{0,\mathfrak{e}})$ and $\mathfrak{W}_\mathfrak{e}$. From now on, the idea is to consider the family of variables $\{c_\omega\}_{\omega \in \mathfrak{W}_\mathfrak{e}}$. 

Given a potential $F: U \to \mathbb{R}$, the Denker-Urba\'nski conformality condition for measures living on $Y_A$ will be written using the following function $D:U\cap Y_A \to \mathbb{R}$
\begin{equation*}
    D(\omega) := e^{F(\omega)}, \quad \omega \neq e,
\end{equation*}
and we get a general formulation for the conformal measures living on $Y_A$-families as follows.

\begin{theorem}\label{theorem.coeficientes} A measure $\mu$ which lives on a $Y_A$-family $Y_A(\xi^{0,\mathfrak{e}})$, where $\xi^{0,\mathfrak{e}}$ is its respective empty-stem configuration, satisfies the Denker-Urba\'nski conformality condition \eqref{eq:conformal_Ur_sets_potential} if and only if the coefficients $c_\omega$ in \eqref{eq:coefficients_Y_A_confomal} satisfy
\begin{equation*}
    c_\omega D(\omega) = c_{\sigma(\omega)}, \quad \omega \in \mathfrak{W}_\mathfrak{e}\setminus \{e\}.
\end{equation*}
\end{theorem}

\begin{proof}  It is straightforward from the Denker-Urba\'nski conformality condition for characteristic functions on the special set $\{\omega\}$, hence the condition above is necessary. The converse is clear because, for every special set $E$, we have that $e \notin E$ and $\sum_{\omega\in E}D(\omega)c_\omega=\sum_{\omega\in E}c_{\sigma(\omega)}$ implies the Denker-Urba\'nski corformality condition.  
\end{proof}

From the theorem above, every non-zero conformal measure living on an $Y_A$-family necessarily gives mass to every point of this family. Observe that the identity $c_{\omega} D(\omega) = c_{\sigma(\omega)}$, $\omega \neq e$, implies
\begin{equation*}
c_{\omega} \prod_{i=0}^{|\omega|-1}D(\sigma^i(\omega)) = c_{e}, \quad \omega \in \mathfrak{W}_\mathfrak{e}\setminus\{e\},
\end{equation*}
where $c_e := \mu(\{\xi^{0,\mathfrak{e}}\})$. The equation above can be rewritten as
\begin{equation}\label{eq:c_omega_new}
c_{\omega} e^{ F_{|\omega|}(\omega)} = c_{e}, \quad \omega \in \mathfrak{W}_\mathfrak{e} \setminus\{e\},
\end{equation}
where $F_n$ is the Birkhoff's sum. 

For any potential $F$, if $\mu$ is a $e^F$-conformal probability measure in a $Y_A$-family we must have $c_e > 0$, otherwise all other $c_\omega$'s are zero by \eqref{eq:c_omega_new}. That is equivalent to impose $c_\omega > 0$ for all $\omega \in \mathfrak{W}_\mathfrak{e}$, since it is a necessary condition to obtain $c_e > 0$. At the same time, our focus at this point is probability measures living on a $Y_A$-family. Then, we also should have
\begin{equation}\label{eq:c_omega_probability}
    \sum_{\omega \in \mathfrak{W}_\mathfrak{e}}c_\omega  = c_e\left( 1+\sum_{n \in \mathbb{N}}\sum_{\substack{\omega \in \mathfrak{W}_\mathfrak{e}\\ |\omega| = n}}e^{-\beta F_n(\omega)} \right)= 1.
\end{equation}
The following result is a consequence of this characterization.

\begin{theorem}\label{thm:uniqueness_conformal_probabilities_Y_A_families} Fixed a potential $F:U \to \mathbb{R}$ and a $Y_A$-family $Y(\xi^{0,\mathfrak{e}})$, there exists at most one $e^F$-conformal probability living on such family. 
\end{theorem}

\begin{proof} Suppose there exist two collections of strictly positive numbers $\{c_{\omega}\}_{\omega \in \mathfrak{W}_\mathfrak{e}}$ and $\{d_{\omega}\}_{\omega \in \mathfrak{W}_\mathfrak{e}}$ satisfying
\begin{align*}
    \begin{cases}
        c_{\omega} = e^{-F_{|\omega|}(\omega)} c_{e}, \quad \omega \in \mathfrak{W}_\mathfrak{e},\\
        \sum_{\omega \in \mathfrak{W}_\mathfrak{e}}c_\omega = 1;
    \end{cases}
    \quad \text{and} \quad 
    \begin{cases}
        d_{\omega} = e^{-F_{|\omega|}(\omega)} d_{e}, \quad \omega \in \mathfrak{W}_\mathfrak{e},\\
        \sum_{\omega \in \mathfrak{W}_\mathfrak{e}}d_\omega = 1.
    \end{cases}
\end{align*}
Since $c_e$ and $d_e$ are positive numbers, there exists $\lambda \in \mathbb{R}_{+}^*$ s.t. $d_e = \lambda c_e$. Then, for every $\omega \in \mathfrak{W}_\mathfrak{e}$,
\begin{equation*}
    d_{\omega} = e^{-F_{|\omega|}(\omega)} d_{e} = e^{-F_{|\omega|}(\omega)} \lambda c_e = \lambda d_{\omega}.
\end{equation*}
On the other hand,
\begin{equation*}
    1 = \sum_{\omega \in \mathfrak{W}_\mathfrak{e}}d_\omega = \sum_{\omega \in \mathfrak{W}_\mathfrak{e}} \lambda c_\omega = \lambda \sum_{\omega \in \mathfrak{W}_\mathfrak{e}} c_\omega = \lambda,
\end{equation*}
and hence $\lambda = 1$. Therefore, $d_\omega = c_\omega$ for every $\omega \in \mathfrak{W}_\mathfrak{e}$. 
\end{proof}

\begin{corollary}\label{cor:eigenmeasures_dimension_1} Given a potential $F:U \to \mathbb{R}$, the space of eigenmeasures for an associated eigenvalue $\lambda$ living on an $Y_A$ family has dimension at most $1$.
\end{corollary}

\begin{proof} It is straightforward from Theorems \ref{thm:equivalences_conformal_measures_generalized_Markov_shift} and \ref{thm:uniqueness_conformal_probabilities_Y_A_families}. 
\end{proof}

From now on we will study and caractherize the extremal conformal measures living on $Y_A$.

\begin{lemma}\label{lemma:conformal_Y_A_minus_Y_A_family} Given a potential $F:U \to \mathbb{R}$ and $Y_A^1$ a $Y_A$-family, suppose that there exists a $e^F$-conformal measure $\mu$ given mass to $Y_A^1$. Then the restriction of $\mu$ to $(Y_A^1)^c$, given by 
\begin{equation*}
    \mu_{|(Y_A^1)^c}(B) := \mu(B \cap (Y_A^1)^c), \quad B \in \mathcal{B}_{X_A}
\end{equation*}
is a $e^F$-conformal measure as well. 
\end{lemma}

\begin{proof} Similar to Proposition \ref{thm:restriction_eigenmeasures}. 
\end{proof}

\begin{lemma}\label{lemma:extremal_conforma_Y_A_characterization} Let $F: U \to \mathbb{R}$ be a potential. If $\mu$ is a $e^{F}$-conformal probability measure that gives mass to a $Y_A$-family and its complement, then $\mu$ is not extremal. 
\end{lemma}

\begin{proof} Let $\mu$ be an $e^{F}$-conformal probability on $Y_A$ and suppose that $\mu$ gives positive mass to $Y_A^1$ (a $Y_A$-family) and to its complement, in other words, $\mu(Y_A^1)>0$ and $\mu((Y_A^1)^c) >0$. By Lemma \ref{lemma:conformal_Y_A_minus_Y_A_family}, the measures given by
\begin{align*}
    \mu_1 (B):= \frac{\mu(B \cap Y_A^1)}{\mu(Y_A^1)} \quad \text{and} \quad \mu_2 (B):= \frac{\mu(B \cap (Y_A^1)^c)}{\mu((Y_A^1)^c)},
\end{align*}
defined for every $B \in \mathcal{B}_{X_A}$ are $e^F$-conformal probabilities on $X_A$. Then $\mu = \mu(Y_A^1) \mu_1 + \mu((Y_A^1)^c) \mu_2$, and we conclude that $\mu$ is not extremal. 


\end{proof}

\begin{remark}\label{rmk:extremal_measures} For generalized Markov shifts with a countable number of empty-stem elements, the lemma above and Theorem \ref{thm:uniqueness_conformal_probabilities_Y_A_families} shows that a $e^F$-conformal probability is extremal if, and only if, it lives on a unique $Y_A$-family. For the uncountable empty-stem case it is also true, as we prove in Theorem \ref{thm:extremal_conformal_measures_general} below.
\end{remark}

\begin{lemma}\label{lemma:properties_empty_stem} Consider the space $Y_A \subset X_A$, $Y_A \neq \emptyset$ for an irreducible matrix $A$. Let $\mathscr{E}$ be the set of empty-stem configurations of $Y_A$. Then,

\begin{enumerate}
    \item $\mathscr{E}$ is a closed set in both $X_A$ and $Y_A$ (subspace topology);
    \item given a potential $F:U\to \mathbb{R}$ and a $e^F$-conformal measure $\mu$ on $X_A$, $\mu(\mathscr{E})=0$ iff $\mu(Y_A) = 0$.
\end{enumerate}
\end{lemma}

\begin{proof} Item $(1)$ is straightforward because $\mathscr{E} = U^c$, where $U$ is the domain of $\sigma$, an open set. Consequently, $\mathscr{E}$ is a Borel set for both spaces and therefore $\mathcal{B}_\mathscr{E} \subset \mathcal{B}_{Y_A} \subset \mathcal{B}_{X_A}$. It is clear that $\mu(Y_A) =0$ implies $\mu(\mathscr{E}) = 0$. We prove the converse. 

We claim that, if $\mu(\mathscr{E}) = 0$, then $\mu(\sigma^{-1}(\mathscr{E}))=0$. In fact, suppose that $\mu(\mathscr{E}) = 0$. Since $\sigma(\sigma^{-1}(\mathscr{E})\cap C_i) \subset \mathscr{E}$, we have
\begin{align*}
    \mu \odot \sigma (\sigma^{-1}(\mathscr{E})) = \sum_{i \in \mathbb{N}} \mu(\sigma(\sigma^{-1}(\mathscr{E})\cap C_i)) = 0.
\end{align*}
By Theorem \ref{thm:equivalences_conformal_measures_generalized_Markov_shift}, a $e^F$-conformal measure is a $(-F,1)$-conformal measure, note that the equivalence between (i) and (iv) on Theorem \ref{thm:equivalences_conformal_measures_generalized_Markov_shift} does not require continuity of the potential. Then,
\begin{equation*}
    \frac{d \mu \odot \sigma}{d \mu} = e^{F} \quad \mu-\text{a.e. on }U,
\end{equation*}
and so
\begin{equation*}
    \mu(\sigma^{-1}(\mathscr{E})) = \int_{\sigma^{-1}(\mathscr{E})} e^{-F} d\mu\odot\sigma = 0.
\end{equation*}
The claim is proved. By Proposition \ref{prop:eigen_measure_nonsingular_Generalized_Markov_shift}, $\mu$ is non-singular on $U$, then $\mu(\sigma^{-k}(\mathscr{E})) =0$ for every $k \in \mathbb{N}$. Since $Y_A = \bigsqcup_{k \in \mathbb{N}_0}\sigma^{-k}(\mathscr{E})$, we conclude that $\mu(Y_A) = 0$.
\end{proof}

\begin{theorem}\label{thm:extremal_conformal_measures_general} Let $F:U \to \mathbb{R}$ be a potential. The extremal $e^{F}$-conformal probabilities living on $Y_A$ are precisely the ones living on a unique $Y_A$-family. 
\end{theorem}

\begin{proof} Let $\mu$ be a $e^F$-conformal probability living on $Y_A$. If $\mu$ gives positive mass to some $Y_A$-family, then the result follows from Theorem \ref{thm:uniqueness_conformal_probabilities_Y_A_families} and Lemma \ref{lemma:extremal_conforma_Y_A_characterization}. So, suppose that $\mu$ gives mass zero to every $Y_A$-family, or equivalently, $\mu(\{\xi^0\}) = 0$ for every $\xi^0 \in \mathscr{E}$. By Lemma \ref{lemma:properties_empty_stem} $(2)$, $\mu(Y_A) >0$ if and only if $\mu(\mathscr{E}) >0$. Then the restriction $\nu:=\mu\vert_{\mathcal{B}_{\mathscr{E}}}$ to the Borel $\sigma$-algebra $\mathcal{B}_{\mathscr{E}} \subset \mathcal{B}_{Y_A}$ is a positive measure. Since $X_A$ is Hausdorff and second countable, so is $\mathscr{E}$. Also, $\nu$ gives mass zero to every singleton of $\mathscr{E}$, hence $\nu$ is an atomless measure (see Lemma 12.18 of \cite{AliBorder2007}). Since $\mu(\mathscr{E}) >0$ and atomless, we have that $Y_A$ is uncountable. Furthermore, there exists a Borel subset $\mathscr{E}^1 \subset \mathscr{E}$ satisfying $\nu(\mathscr{E}^1) = \frac{\nu(\mathscr{E})}{2} > 0$ (Corollary 1.12.10 of \cite{Bogachev2007_1}). Take $\mathscr{E}^2:= \mathscr{E}\setminus \mathscr{E}^1$, so $\nu(\mathscr{E}^2) = \nu(\mathscr{E}^1)$ and, $k \in \{1,2\}$, we define
\begin{equation*}
    Y_k := \bigsqcup_{\xi^0 \in \mathscr{E}^k} Y_A\left(\xi^0\right) = \bigsqcup_{n\in \mathbb{N}_0} \sigma^{-n}(\mathscr{E}^k), \quad k \in \{1,2\}.
\end{equation*}
So $\{Y_1,Y_2\}$ is a Borel partition of $Y_A$ s.t. $\mu(Y_1)>0$ and $\mu(Y_2)>0$. Now, we claim that the probabilities $\mu_1$ and $\mu_2$, defined by
\begin{equation*}
    \mu_k(B) = \frac{\mu(B\cap Y_k)}{\mu(Y_k)}, \quad k \in \{1,2\},
\end{equation*}
are distinct $e^F$-conformal probabilities. Let $B \subset U$ be a special set. Then, for each $k \in \{1,2\}$
\begin{align*}
    \mu(Y_k) \mu_k(\sigma(B)) &=  \mu(\sigma(B)\cap Y_k) = \mu\left(\sigma(B)\cap \bigsqcup_{\xi^0 \in \mathscr{E}^k} Y_A\left(\xi^0\right)\right)\\
    &\stackrel{(\dagger)}{=} \mu\left(\sigma(B)\cap \bigsqcup_{\xi^0 \in  \mathscr{E}^k}\sigma\left(\sigma^{-1}\left(Y_A\left(\xi^0\right)\right)\right)\right) \\
    &= \mu\left(\bigsqcup_{\xi^0 \in \mathscr{E}^k}\sigma(B)\cap\sigma\left(\sigma^{-1}\left(Y_A\left(\xi^0\right)\right)\right)\right) \\
    &\stackrel{(\ddagger)}{=} \mu\left(\bigsqcup_{\xi^0 \in  \mathscr{E}^k}\sigma\left(B\cap\sigma^{-1}\left(Y_A\left(\xi^0\right)\right)\right)\right)\\
    &= \mu\left(\sigma\left(B\cap\bigsqcup_{\xi^0 \in  \mathscr{E}^k}\sigma^{-1}\left(Y_A\left(\xi^0\right)\right)\right)\right)\\
    &= \mu\left(\sigma\left(B\cap\sigma^{-1}( Y_k)\right)\right) \stackrel{(\bullet)}{=} \int_{B\cap\sigma^{-1}( Y_k)} e^F d\mu
\end{align*}
where in $(\dagger)$ we used the dynamic invariance of the $Y_A$-families, and in $(\bullet)$ we used the fact that $B\cap\sigma^{-1}( Y_k) \subset B$ and hence it is a special set. The identity $(\ddagger)$ is true since we have that $B$ is special and $\sigma(B)\cap\sigma(\sigma^{-1}\left(Y_A\left(\xi^0\right)\right) = \sigma\left(B\cap\sigma^{-1}\left(Y_A\left(\xi^0\right)\right)\right)$. Now, for $k \in \{1,2\}$, we have
\begin{align*}
    \mu_k(\sigma(B)) &= \frac{1}{\mu(Y_k)} \int_{B\cap\sigma^{-1}( Y_k)} e^F d\mu = \int_{B} e^F d\mu_k,
\end{align*}
And hence each $\mu_k$ is a conformal measure. It is straightforward see that for $\lambda = \mu(Y_1) \in (0,1)$ we have
$\mu = \mu(Y_1) \mu_1 + \mu(Y_2) \mu_2 = \lambda \mu_1 + (1-\lambda) \mu_2$ and, therefore $\mu$ is not extremal.
\end{proof}

\begin{theorem}\label{thm:existence_conformal_pressures} Let $\Sigma_A$ be transitive and s.t. $X_A$ is compact and $s$-compact. Also, let $F:U \to \mathbb{R}$ be a potential which has uniform bounded distortion. We have the following:
\begin{itemize}
    \item[$(a)$] If $P_G(-\beta F) < 0$, then there exists an extremal $e^{\beta F}$-conformal probability living on each $Y_A$-family.
    \item[$(b)$] If $P_G(-\beta F) > 0$, then there are no $e^{\beta F}$-conformal probabilities living on $Y_A$.
\end{itemize}
\end{theorem}

\begin{proof} Given any $Y_A$-family generated by an empty-stem configuration $\xi^0$, where $\mathfrak{R} = R_{\xi^0}(e)$, there exists a conformal probability living on $Y_A(\xi^0)$ if, and only if, the series
\begin{align*}
    \sum_{n \in \mathbb{N}}\sum_{\substack{|\omega| = n\\ \omega_{n-1} \in \mathfrak{R}}}e^{-\beta F_n(\omega)} = \sum_{n \in \mathbb{N}} Z_n(-\beta F, \xi^0)
\end{align*}
converges. Since $\limsup_n \sqrt[n]{Z_n(-\beta F, \xi^0)} = e^{P(-\beta F,\xi^0)} = e^{P_G(-\beta F)}$, by the root test we have that when $P_G(-\beta F) < 0$ there exists a $e^{\beta F}$-conformal probability living on each $Y_A(\xi^0)$.
\end{proof}

\begin{proposition}\label{prop:phase_transition_entropy} If $\Sigma_A$ is transitive, $X_A$ is compact and $s$-compact and let $F:U \to \mathbb{R}$ be a potential with uniform bounded variations. We have the following:
\begin{itemize}
    \item[$(i)$] If $0 < \inf F $, for $\frac{h_G}{\inf F}< \beta$, there exists a unique extremal $e^{\beta F}$-conformal probability measure living on each $Y_A$-family.
    \item[$(ii)$] If $0 \leq \sup F < \infty$, for $\beta < \frac{h_G}{\sup F}$, there are not $e^{\beta F}$-conformal probability measures living on $Y_A$.
\end{itemize}
\end{proposition}

\begin{proof} Fix an empty configuration $\xi^0$. Let $Y_A(\xi^0)$ the $Y_A$-family associated with $\xi^0$ and its set of stems $\mathfrak{W}$. By the characterization of conformal measures living on $Y_A$, see equation \eqref{eq:c_omega_probability}, there exists a $e^{\beta F}$-conformal measure on $Y_A(\xi^0)$ if and only if 
\begin{equation*}
    \sum_{n \in \mathbb{N}} \sum_{\substack{\omega \in \mathfrak{W}\\|\omega| = n}} e^{-\beta F_n(\omega)} < \infty.
\end{equation*}
 For $0<\inf F$, we have the following bound:
\begin{align}\label{eq:convergence_series_conformal_inf}
    \sum_{n \in \mathbb{N}} \sum_{\substack{\omega \in \mathfrak{W}\\|\omega| = n}} e^{-\beta F_n(\omega)} &\leq \sum_{n \in \mathbb{N}} \sum_{\substack{\omega \in \mathfrak{W}\\|\omega| = n}} e^{-\beta n \inf F}.
\end{align}
Also, for $0 \leq \sup F < \infty$, we get:
\begin{align}\label{eq:divergence_series_conformal_sup}
    \sum_{n \in \mathbb{N}} \sum_{\substack{\omega \in \mathfrak{W}\\|\omega| = n}} e^{-\beta n \sup F} &\leq \sum_{n \in \mathbb{N}} \sum_{\substack{\omega \in \mathfrak{W}\\|\omega| = n}} e^{-\beta F_n(\omega)}.
\end{align}
By the RHS of \eqref{eq:convergence_series_conformal_inf} and LHS \eqref{eq:divergence_series_conformal_sup}, we consider the constant potentials $\inf F$ and $\sup F$, for each respective series. Observe that, for every constant $c$, we have $P_G(c)=c + h_G$. 

Then, $P_G(-\beta \inf F) < 0$ iff $\beta > \frac{h_G}{\inf F}$, and $P_G(-\beta \sup F) > 0$ iff $\beta < \frac{h_G}{\sup F}$. 

By Theorem \ref{thm:existence_conformal_pressures}, the RHS of \eqref{eq:convergence_series_conformal_inf} converges for $\beta > \frac{h_G}{\inf F}$ and the LHS of \eqref{eq:divergence_series_conformal_sup} diverges for $\beta < \frac{h_G}{\sup F}$. 
\end{proof}

We summarize the results from the previous proposition in the next figure:

\begin{figure}[H]
    \centering
\scalebox{0.8}{
\tikzset{every picture/.style={line width=0.75pt}} 

\begin{tikzpicture}[x=0.75pt,y=0.75pt,yscale=-1,xscale=1]

\draw [color={rgb, 255:red, 0; green, 0; blue, 0 }  ,draw opacity=1 ][fill={rgb, 255:red, 0; green, 0; blue, 0 }  ,fill opacity=1 ][line width=2.25]  [dash pattern={on 6.75pt off 4.5pt}]  (17,178.02) -- (177.94,176.38) ;
\draw [shift={(177.94,176.38)}, rotate = 539.4200000000001] [color={rgb, 255:red, 0; green, 0; blue, 0 }  ,draw opacity=1 ][line width=2.25]      (7.83,-7.83) .. controls (3.51,-7.83) and (0,-4.32) .. (0,0) .. controls (0,4.32) and (3.51,7.83) .. (7.83,7.83) ;
\draw [color={rgb, 255:red, 0; green, 0; blue, 0 }  ,draw opacity=1 ][fill={rgb, 255:red, 0; green, 0; blue, 0 }  ,fill opacity=1 ][line width=2.25]    (174.08,105.87) -- (174.5,143.4) ;
\draw [shift={(174.55,148.4)}, rotate = 269.36] [fill={rgb, 255:red, 0; green, 0; blue, 0 }  ,fill opacity=1 ][line width=0.08]  [draw opacity=0] (16.07,-7.72) -- (0,0) -- (16.07,7.72) -- (10.67,0) -- cycle    ;
\draw [color={rgb, 255:red, 0; green, 0; blue, 0 }  ,draw opacity=1 ][fill={rgb, 255:red, 0; green, 0; blue, 0 }  ,fill opacity=1 ][line width=2.25]    (177.94,176.38) .. controls (179.59,174.7) and (181.26,174.69) .. (182.94,176.34) .. controls (184.61,177.99) and (186.28,177.98) .. (187.94,176.31) .. controls (189.59,174.63) and (191.26,174.62) .. (192.94,176.27) .. controls (194.62,177.92) and (196.29,177.91) .. (197.94,176.23) .. controls (199.59,174.55) and (201.26,174.54) .. (202.94,176.19) .. controls (204.61,177.84) and (206.28,177.83) .. (207.94,176.16) .. controls (209.59,174.48) and (211.26,174.47) .. (212.94,176.12) .. controls (214.62,177.77) and (216.29,177.76) .. (217.94,176.08) .. controls (219.59,174.4) and (221.26,174.39) .. (222.94,176.04) .. controls (224.61,177.69) and (226.28,177.68) .. (227.94,176.01) .. controls (229.59,174.33) and (231.26,174.32) .. (232.94,175.97) .. controls (234.62,177.62) and (236.29,177.61) .. (237.94,175.93) .. controls (239.6,174.26) and (241.27,174.25) .. (242.94,175.9) .. controls (244.62,177.55) and (246.29,177.54) .. (247.94,175.86) .. controls (249.59,174.18) and (251.26,174.17) .. (252.94,175.82) .. controls (254.62,177.47) and (256.29,177.46) .. (257.94,175.78) .. controls (259.6,174.11) and (261.27,174.1) .. (262.94,175.75) .. controls (264.62,177.4) and (266.29,177.39) .. (267.94,175.71) .. controls (269.59,174.03) and (271.26,174.02) .. (272.94,175.67) .. controls (274.62,177.32) and (276.29,177.31) .. (277.94,175.63) .. controls (279.6,173.96) and (281.27,173.95) .. (282.94,175.6) .. controls (284.62,177.25) and (286.29,177.24) .. (287.94,175.56) .. controls (289.59,173.88) and (291.26,173.87) .. (292.94,175.52) .. controls (294.62,177.17) and (296.29,177.16) .. (297.94,175.48) .. controls (299.6,173.81) and (301.27,173.8) .. (302.94,175.45) .. controls (304.62,177.1) and (306.29,177.09) .. (307.94,175.41) .. controls (309.59,173.73) and (311.26,173.72) .. (312.94,175.37) .. controls (314.62,177.02) and (316.29,177.01) .. (317.94,175.33) .. controls (319.6,173.66) and (321.27,173.65) .. (322.94,175.3) .. controls (324.62,176.95) and (326.29,176.94) .. (327.94,175.26) .. controls (329.59,173.58) and (331.26,173.57) .. (332.94,175.22) .. controls (334.62,176.87) and (336.29,176.86) .. (337.94,175.18) .. controls (339.6,173.51) and (341.27,173.5) .. (342.94,175.15) .. controls (344.62,176.8) and (346.29,176.79) .. (347.94,175.11) .. controls (349.59,173.43) and (351.26,173.42) .. (352.94,175.07) .. controls (354.62,176.72) and (356.29,176.71) .. (357.94,175.03) .. controls (359.6,173.36) and (361.27,173.35) .. (362.94,175) .. controls (364.62,176.65) and (366.29,176.64) .. (367.94,174.96) .. controls (369.59,173.28) and (371.26,173.27) .. (372.94,174.92) .. controls (374.62,176.57) and (376.29,176.56) .. (377.94,174.88) .. controls (379.6,173.21) and (381.27,173.2) .. (382.94,174.85) .. controls (384.62,176.5) and (386.29,176.49) .. (387.94,174.81) .. controls (389.59,173.13) and (391.26,173.12) .. (392.94,174.77) -- (396.83,174.74) -- (396.83,174.74) ;
\draw [color={rgb, 255:red, 0; green, 0; blue, 0 }  ,draw opacity=1 ][fill={rgb, 255:red, 0; green, 0; blue, 0 }  ,fill opacity=1 ][line width=2.25]    (396.83,174.74) -- (612,174.74) ;
\draw [shift={(617,174.74)}, rotate = 180] [fill={rgb, 255:red, 0; green, 0; blue, 0 }  ,fill opacity=1 ][line width=0.08]  [draw opacity=0] (16.07,-7.72) -- (0,0) -- (16.07,7.72) -- (10.67,0) -- cycle    ;
\draw [shift={(396.83,174.74)}, rotate = 0] [color={rgb, 255:red, 0; green, 0; blue, 0 }  ,draw opacity=1 ][line width=2.25]      (7.83,-7.83) .. controls (3.51,-7.83) and (0,-4.32) .. (0,0) .. controls (0,4.32) and (3.51,7.83) .. (7.83,7.83) ;
\draw [color={rgb, 255:red, 0; green, 0; blue, 0 }  ,draw opacity=1 ][fill={rgb, 255:red, 0; green, 0; blue, 0 }  ,fill opacity=1 ][line width=2.25]    (399.4,104.23) -- (399.82,145.04) ;
\draw [shift={(399.88,150.04)}, rotate = 269.40999999999997] [fill={rgb, 255:red, 0; green, 0; blue, 0 }  ,fill opacity=1 ][line width=0.08]  [draw opacity=0] (16.07,-7.72) -- (0,0) -- (16.07,7.72) -- (10.67,0) -- cycle    ;

\draw (604.76,184.58) node [anchor=north west][inner sep=0.75pt]  [font=\LARGE,color={rgb, 255:red, 0; green, 0; blue, 0 }  ,opacity=1 ]  {$\beta $};
\draw (48.62,192.94) node [anchor=north west][inner sep=0.75pt]  [font=\LARGE,color={rgb, 255:red, 0; green, 0; blue, 0 }  ,opacity=1 ] [align=left] {absence};
\draw (442.64,188.02) node [anchor=north west][inner sep=0.75pt]  [font=\LARGE,color={rgb, 255:red, 0; green, 0; blue, 0 }  ,opacity=1 ] [align=left] {existence};
\draw (378.33,51.08) node [anchor=north west][inner sep=0.75pt]  [font=\huge,color={rgb, 255:red, 0; green, 0; blue, 0 }  ,opacity=1 ]  {$\frac{h_{G}}{\inf F}$};
\draw (150.59,51.08) node [anchor=north west][inner sep=0.75pt]  [font=\huge,color={rgb, 255:red, 0; green, 0; blue, 0 }  ,opacity=1 ]  {$\frac{h_{G}}{\sup F}$};
\draw (269,191.4) node [anchor=north west][inner sep=0.75pt]  [font=\LARGE]  {$?????$};

\end{tikzpicture}
}
\caption{\label{fig:conformal_general_entropy}}
\end{figure}

The next result uses the fact that $e^\beta$-conformal probability measures are the KMS$_{\beta}$ states for gauge action when we consider the constant potential $F\equiv 1$. Similar results for KMS$_{\beta}$ weights on graph $C^{*}$-algebras were obtained by K. Thomsen in \cite{Tho3}, where the Gurevich entropy $h_G$ is also a critical value for the existence of conformal measures.

\begin{corollary}(Gauge potential)\label{cor:phase_transition_entropy_gauge} Let $A$ be an irreducible matrix, with $X_A$ compact and $s$-compact, and let $F:U \to \mathbb{R}$ be the constant potential $F\equiv 1$. We have the following:
\begin{itemize}
    \item[(i)] For $\beta>h_G$ we there exists a unique $e^\beta$-conformal probability measure living on each $Y_A$-family;
    \item[(ii)] If $h_G < \infty$, for $\beta_c = h_G$ there exists a $e^{\beta_c}$-conformal probability measure living on $X_A$.
      \item[(iii)] For $0< \beta < h_G$ there are not $e^\beta$-conformal probability measures living on $Y_A$.
\end{itemize}
\end{corollary}

For $\beta_c = h_G$, since for the potential $F\equiv 1$ the grupoid is principal, see Remark \ref{remark:KMS_quasi_invariant}, the $e^\beta$-conformal probability measures are the KMS$_{\beta}$ states, see \cite{Renault2009}. Then, the existence comes from the fact that by the item (i), we know that we have the existence of KMS$_{\beta}$ for every $\beta>h_G$ and, the Theorem 5.3.25 of \cite{Bratteli1997} guarantees that the set of $\beta$'s for which we have the existence of a KMS$_{\beta}$ states is closed. Then it must do exist a $e^\beta$-conformal probability measure at $\beta_c = h_G$. The question is to know where this measure lives, and we will study this problem in the next section in a couple of examples.

Now we will exhibit two concrete examples on which there exists a unique $e^{\beta_c}$-conformal probabi-lity measure at $\beta_c = h_G$, and this measure lives on $\Sigma_A$. Besides, we can prove that when $\beta > \beta_c$ the $e^\beta$-conformal $\nu_{\beta}$ (which lives on $Y_A$) converge weakly$^*$ to $\nu_{\beta_c}$ when $\beta$ converges to $\beta_c$, a new feature with respect to the standard formalism on countable Markov shifts.  Every GCMS presented in this section is compact and $s$-compact, and since the gauge potential has uniform bounded variation, we will use Corollary \ref{cor:phase_transition_entropy_gauge} extensively for all examples except the full shift, where we encode it to the renewal shift.

\subsection{Renewal shift}\\

In this section we study the conformal measures on $X_A$ in the renewal shift. We connect these results to the standard theory on $\Sigma_A$. In particular, we present the first example of a \emph{length-type phase transition}, which consists in a change of space where the conformal measure lives, from $Y_A$ to $\Sigma_A$, when we take the limit $\beta \to \beta_c$. 

We recall that the generalized renewal shift has only one $Y_A$-family, which is $Y_A$ itself.

\begin{proposition}\label{prop:phase_transition_conformal_potential_1}
Consider the renewal shift and $F\equiv1$. Then, $h_G=\log 2$, and we have the following:
\begin{itemize}
    \item[(i)] For $\beta>\log 2$ we have a unique $e^\beta$-conformal probability measure living on $Y_A$.
    \item[(ii)] For $\beta_c = \log 2$ there is a unique $e^{\beta_c}$-conformal probability measure living on $\Sigma_A$.
    \item[(iii)] For $\beta < \log 2$ there are not $e^\beta$-conformal probability measures.
\end{itemize}
\end{proposition}
\begin{proof}
The existence of a unique $e^\beta$-conformal probability measure living $Y_A$ for $\beta > \log 2$ is a consequence of the Corollary \ref{cor:phase_transition_entropy_gauge} combining with the uniqueness of the $Y_A$-family for the renewal GCMS. At the critical point $\beta_c = \log 2$, the following straightforward calculation shows that the series \eqref{eq:c_omega_probability} diverges. In fact, for each $n \geq 1$, there are exactly $2^{n-1}$ elements in $Y_A$ whose stem has length $n$, and then
\begin{equation*}
    \sum_{n \in \mathbb{N}} \sum_{\substack{\omega \in \mathfrak{W}\\|\omega| = n}} e^{-\beta_c F_n(\omega)} = \frac{1}{2}\sum_{n \in \mathbb{N}} 1 = \infty.
\end{equation*}
Then, there is no $e^{\beta_c}$-conformal probability measure living on $Y_A$, since we know that there exist conformal measures at $\beta_c$, these measures must live on $\Sigma_A$. In fact, we will prove the uniqueness of the conformal measure at $\beta_c$. Also, due to Corollary \ref{cor:phase_transition_entropy_gauge} we have the absence of $e^\beta$-conformal probability measures living on $Y_A$ for $\beta < \beta_c$.

It remains to prove that there is no $e^\beta$-conformal probability measures living on $\Sigma_A$ for $\beta < \beta_c$ and the uniqueness at $\beta_c$. In this case, we may restrict the measure to the standard theory of the thermodynamic formalism on CMS $\Sigma_A$, and remembering that, by Theorem \ref{thm:equivalences_conformal_measures_generalized_Markov_shift}, we have to change the sign of the potential to work with eigenmeasures for the Ruelle operator on CMS. Then, from now we will work with the potential $- F\vert_{\Sigma_A}\equiv -1$. It is easy to verify that this potential $- \beta F$ is positive recurrent for every $\beta > 0$, and also $P_{G}(-\beta) =  \log 2 - \beta$. 

For each $\beta >0$, to find an eigenmeasure living on $\Sigma_A$ associated to the potential $-\beta F\vert_{\Sigma_A}\equiv - \beta$ is equivalent to prove that there exist $\lambda_{\beta} >0$ and a probability measure $\nu:\mathcal{B}_{\Sigma_A} \to [0,1]$ that solves the equation
\begin{equation*}
     \int f d\nu = \lambda_{\beta} \int L_{-\beta} f d\nu, \forall \\ \  f \in L^1(\nu).
\end{equation*}

By the generalized RPF theorem, since $-\beta F\vert_{\Sigma_A}$ is positive recurrent for every $\beta>0$ we have the existence of such probability measure and, in addition, $\lambda_{\beta}$ has to be equal to $e^{P_G(-\beta)} = 2 e^{-\beta}$. Here we are looking for conformal probability measures, so we have to impose $\lambda_{\beta}\equiv 1$, and then the unique value of $\beta >0$ for which we have the existence of a $e^\beta$-conformal probability measure living on $\Sigma_A$ is $\beta_c = \log 2$. The uniqueness comes from the fact that the space of eigenmeasures associated with each eigenvalue has dimension 1 for positive recurrent potentials, see \cite{Daon2013}. Eigenmeasures for general eigenvalues are useful for DLR measures in the standard and generalized setting; see our second paper \cite{BisExelFrauRas2018}.

The determine the measure $\nu$, due the structure of the renewal matrix, it is enough to evaluate $\nu$ on cylinders $[\alpha]$, for every admissible (positive) word $\alpha$ that ends with the symbol 1. With standard arguments one can show that:
\begin{equation}\label{eq:Sarig_measure_potential_1}
    \nu([\alpha]) = 2^{-|\alpha|}.
\end{equation}

For detailed calculations of this example, see Corollary 5.29 of \cite{Raszeja2020}.
\end{proof}

The figure \ref{fig:renewal_F_equals_1} compares the standard formalism with the generalized one for the renewal shift and potential $F\equiv1$, as in Proposition \ref{prop:phase_transition_conformal_potential_1}.

\begin{figure}[h!]
  \hspace{.1cm}
 \includegraphics[scale=.3]{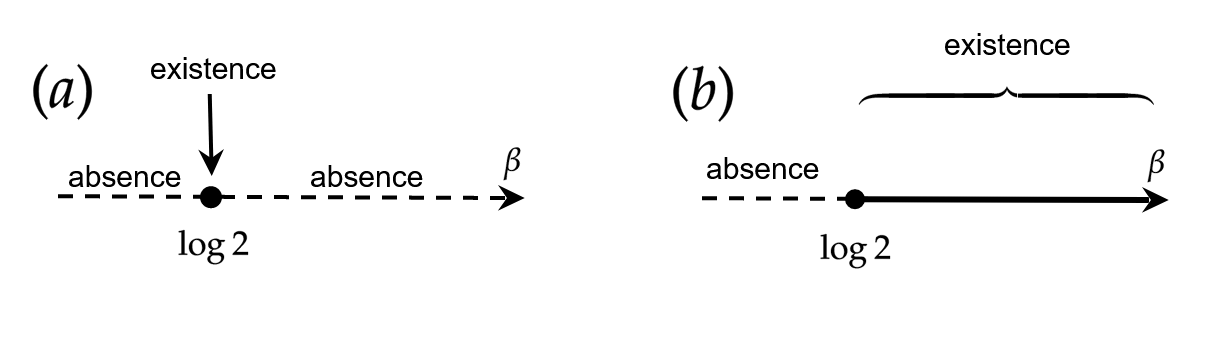}
 
 \caption{The phase transitions on different thermodynamic formalisms for conformal probabilities in the renewal shift space with potential $F\equiv1$. The picture $(a)$ represents the standard formalism on $\Sigma_A$, where we have a unique $e^\beta$-conformal probability for a unique possible inverse of temperature, namely $\beta = \log 2$ (black dot). For the finite alphabet we also have the situation where we have a unique $\beta$ for which does exist KMS$_{\beta}$ states, see \cite{BratJorgOstro2004, Exel2004, KerrPinz2002}. The picture $(b)$ represents the generalized formalism on $X_A$ and, unlike in $(a)$, we are able to detect not only the previous standard $e^\beta$-conformal probability, but also, for each $\beta > \log 2$, we have a unique $e^\beta$-conformal probability (continuous line), in this case the measures live on $Y_A$.  \label{fig:renewal_F_equals_1}}
\end{figure}

Let $\mu_\beta$, $\beta > \log 2$, be the conformal measure of Proposition \ref{prop:phase_transition_conformal_potential_1} for the inverse of temperature $\beta$. We show next that the net $\{\mu_\beta\}_{\beta > \log 2}$ converges to the $2$-conformal measure when $\beta \to \beta_c$. We need the following lemma.

\begin{lemma}\label{lemma:subwords_of_a_word_renewal_counting} Let $\omega \in \mathfrak{W}$ s.t. $\alpha \in \llbracket \omega \rrbracket$. Write $n=|\omega| = |\alpha| + p$, $p \in \mathbb{N}_0$. If $p=0$, then there exists a unique word in $Y_A$ s.t. $\alpha \in \llbracket \omega \rrbracket$. If $p \in \mathbb{N}$, then there are $2^{p-1}$ words in $Y_A$ which satisfies $\alpha \in \llbracket \omega \rrbracket$.
\end{lemma}

\begin{proof} If $p=0$, the result is straightforward, then suppose that $p \in \mathbb{N}$, that is, $|\omega|>|\alpha|$. Then, $\omega = \alpha \gamma$, where both $\alpha$ and $\gamma$ end with `$1$' and $|\gamma|=p$ and therefore the number of possibilities for $\gamma$, and consequently $\omega$, is $2^{p-1}$. 
\end{proof}

\begin{theorem}\label{thm:colapsing_renewal} The net $\{\mu_\beta\}_{\beta> \log 2}$ converges to $\nu$ in the weak$^*$ topology when we take the limit $\beta \to \log 2$.
\end{theorem}
\begin{proof}
The measures $\mu_\beta$, $\beta > \log 2$ are explicitly obtained by explicit calculation of the series in \eqref{eq:c_omega_probability} combined with \eqref{eq:c_omega_new}. By denoting $\mu_\beta(\{\omega \}):= c_{\omega}^\beta$, $\omega \in \mathfrak{W}$, we obtain
\begin{equation}\label{eq:conformal_measure_Y_A_explicit}
     c_\omega^\beta := \frac{e^{-\beta |\omega|}}{1+\frac{1}{2}\sum_{n \in \mathbb{N}}e^{n(\log 2 - \beta)}} = \frac{e^{-|\omega|\beta}(e^\beta-2)}{e^\beta-1}.
\end{equation}
By subsection \ref{subsec:cylinder_topology_Renewal}, every generalized cylinder on a positive admissible word is written by choosing this word ending in `$1$'. Hence,
\begin{equation*}
    \mu_\beta(C_\alpha) = \sum_{\omega \in \mathfrak{W}: \alpha \in \llbracket\omega\rrbracket} c_\omega^\beta,
\end{equation*}
and then,
\begin{equation*}
    \mu_\beta(C_\alpha) = \sum_{n \in \mathbb{N}}\sum_{\substack{\omega \in \mathfrak{W}: \alpha \in \llbracket\omega\rrbracket \\ |\omega| = n}} c_n^\beta, \quad \text{where} \ \ c_n^\beta := \frac{e^\beta-2}{e^{n \beta }(e^\beta-1)}.
\end{equation*}
However, if $\alpha \in \llbracket \omega \rrbracket$, then $|\omega|\geq |\alpha|$ and hence $\mu_\beta(C_\alpha) = \sum_{n = |\alpha|}^\infty\sum_{\substack{\omega \in Y_A: \alpha \in \llbracket\omega\rrbracket \\ |\omega| = n}} c_n^\beta.$

By Lemma \ref{lemma:subwords_of_a_word_renewal_counting}, we have that
\begin{equation*}
    \mu_\beta(C_\alpha) = c_{|\alpha|}^\beta + \sum_{n=|\alpha|+1}^\infty 2^{n-|\alpha|-1}c_n^\beta = c_{|\alpha|}^\beta + \sum_{k=1}^\infty 2^{k-1}c_{k + |\alpha|}^\beta,
\end{equation*}
where in the last equality we used the change of index $k = n-|\alpha|$ in the sum. It follows that
\begin{align*}
    \mu_\beta(C_\alpha) &= \frac{e^\beta-2}{e^\beta-1}\left[\frac{1}{e^{|\alpha|\beta}} + \sum_{k=1}^\infty \frac{2^{k-1}}{e^{(k + |\alpha|)\beta}}\right] =  e^{-|\alpha|\beta}.
\end{align*}
By taking the limit, we get $\lim_{\beta \to \log 2}\mu_\beta(C_\alpha) = 2^{-|\alpha|}.$ On the other hand, for every $\alpha \in \mathfrak{W}$, we have
\begin{align*}
   \lim_{\beta \to \log 2}\mu_\beta(\{\alpha\})=\lim_{\beta \to \log 2} \frac{e^\beta-2}{e^{|\alpha|\beta}(e^\beta-1)}=0.
\end{align*}
Then, for every finite set of words $F \subset \mathfrak{W}$, we have $\lim_{\beta \to \log 2}\mu_\beta(F)=0$. 

Since every basic set of $X_A$ as in subsection \ref{subsec:cylinder_topology_Renewal} has the form
\begin{equation*}
    F\sqcup \bigsqcup_{p \in \mathbb{N}} C_{w(p)}, 
\end{equation*}
where $F \subseteq Y_A$ is finite, and $w(p)$ is some positive admissible word ending in `$1$', we obtain
\begin{align}\label{eq:limit_mu_beta_complement}
    \lim_{\beta \to \log 2}\mu_\beta\left(F\sqcup \bigsqcup_{p \in \mathbb{N}} C_{w(p)}\right) &= \lim_{\beta \to \log 2}\mu_\beta \left(\bigsqcup_{p \in \mathbb{N}} C_{w(p)}\right)\nonumber \\
    &= \lim_{\beta \to \log 2}\sum_{p\in \mathbb{N}}\mu_\beta \left(C_{w(p)}\right) = \sum_{p\in \mathbb{N}}\lim_{\beta \to \log 2}\mu_\beta \left(C_{w(p)}\right).
\end{align}
Now, we will compare this result with the measure $\nu$ of \eqref{eq:Sarig_measure_potential_1} from Proposition \ref{prop:phase_transition_conformal_potential_1}. Given $\alpha$ positive admissible word ending in `$1$', we have
\begin{align*}
    2^{|\alpha|}\nu(C_\alpha) &= 2^{|\alpha|}\nu([\alpha]) =  \sum_{a \in \mathbb{N}}\int_{[a]} 1 d\nu(x) = \nu(\Sigma_A) = 1.
\end{align*}
Therefore, $\nu(C_\alpha) = 2^{-|\alpha|}$, and finally $\nu(C_\alpha) = \lim_{\beta \to \log 2}\mu_\beta(C_\alpha)$.

By the last equality, the $\sigma$-additivity of measures and the identity \eqref{eq:limit_mu_beta_complement}, we obtain for every basic open set $B$, which is in the form $F\sqcup \bigsqcup_{p \in \mathbb{N}} C_{w(p)}$ (see subsection \ref{subsec:cylinder_topology_Renewal}) that
\begin{equation*}
    \nu\left(F\sqcup \bigsqcup_{p \in \mathbb{N}} C_{w(p)}\right) = \lim_{\beta \to \log 2}\mu_\beta\left(F\sqcup \bigsqcup_{p \in \mathbb{N}} C_{w(p)}\right),
\end{equation*}
that is,
\begin{equation*}
    \nu(B) = \lim_{\beta \to \log 2}\mu_\beta(B).
\end{equation*}
Our basis is closed under finite intersections, and the net $\{\mu_\beta\}_{\beta > \log 2}$ converges to $\nu$ on the basic sets as $\beta \to \log 2$. Besides, $\mu_\beta$, for $\beta > \beta_c$, and $\nu$ are defined on the Borel $\sigma$-algebra of a metric space. Therefore, by Theorem 8.2.17 of \cite{Bogachev2007}, we have that $\mu_\beta$ converges on the weak$^*$ topology as $\beta$ goes to $\beta_c$.
\end{proof}


\subsection{Pair renewal shift}\\

We recall that the pair renewal shift has two $Y_A$-families, and they were described in subsection \ref{subsec:Pair_Renewal}. The next lemma gives, for each $k\in\{1,2\}$, the quantity of configurations of the $Y_A$-family $Y_A(\xi^{0,k})$ with stem of same length. 

\begin{lemma}\label{thm:counting_configurations_pair_renewal_shift}
 
Consider the pair renewal shift and $n \in \mathbb{N}_0$. Then,
\begin{align*}
    |\sigma^{-n}(\xi^{0,1})| = \frac{1}{4}\left[(1-\sqrt{2})^{n} + (1-\sqrt{2})^{n+1} + (1+\sqrt{2})^{n} + (1+\sqrt{2})^{n+1}\right]
\end{align*}
and
\begin{align*}
|\sigma^{-n}(\xi^{0,2})| = \begin{cases}
                                 1, \quad \text{if $n=0$};\\
                                 \frac{1}{4}\left[(1-\sqrt{2})^{n-1} + (1-\sqrt{2})^{n} + (1+\sqrt{2})^{n-1} + (1+\sqrt{2})^{n}\right], \quad \text{otherwise}.
                              \end{cases}
\end{align*}
\end{lemma}

\begin{proof} See Theorem 4.74 of \cite{Raszeja2020}.
\end{proof}

For the pair renewal shift and potential $F \equiv 1$ we have the following result.

\begin{proposition}\label{thm:existence_conformal_measures_Pair_renewal} For the pair renewal shift and constant potential $F \equiv 1$, there exists a critical value $\beta_c = h_G = \log(1+\sqrt{2})$ s.t.
\begin{itemize}
    \item[$(i)$] for $\beta > \beta_c$ there exist two extremal $e^\beta$-conformal probabilities living on $Y_A$, each one living on a distinct $Y_A$-family;
    \item[$(ii)$] for $\beta = \beta_c$ there exists a unique $e^\beta$-conformal probability living on $\Sigma_A$;
    \item[$(iii)$] for $\beta < \beta_c$, there are no $e^\beta$-conformal probabilities.
\end{itemize}
\end{proposition}

\begin{proof} By Corollary \ref{cor:phase_transition_entropy_gauge}, there exist two extremal $e^\beta$-conformal probabilities living on $Y_A$ for $\beta > \log(1+\sqrt{2})$, and for $\beta < \log(1+\sqrt{2})$ we have the absence of these measures. At the critical point $\beta_c = \log(1+\sqrt{2})$, the series \eqref{eq:c_omega_probability} diverges for each $Y_A$-family, and it can be verified in a similar fashion as in the renewal shift case, by using Lemma \ref{thm:counting_configurations_pair_renewal_shift}. Now, we analyze the existence of these probabilities, but living on $\Sigma_A$. For every $n \in \mathbb{N}$, any $e^\beta$-conformal probability measure that lives on $\Sigma_A$ necessarily satisfies
\begin{equation}\label{eq:conformality_Pair_renewal}
 \mu_\beta(\sigma([n]))=\int_{[n]} e^\beta d\mu_\beta = e^\beta \mu_\beta([n]), \quad n \in \mathbb{N}.
\end{equation}
We also have that,
\begin{equation*}
    \sigma([1])=\Sigma_A, \quad \sigma([2])=[1]\sqcup \bigsqcup_{n \in \mathbb{N}} [2n] \quad \text{and} \quad \sigma([i])=[i-1], \quad i \neq 1.
\end{equation*}
By \eqref{eq:conformality_Pair_renewal}, we have $\mu_\beta([1])=e^{-\beta}$. By induction we obtain $\mu_\beta([n])=e^{-\beta(n-2)}\mu_\beta([2])$, $n > 2.$

Now, for the cylinder $[2]$, we have that
\begin{equation}\label{eq:eigenmeasure_Sigma_A_Pair_Renewal_cylinder_2}
    e^\beta \mu_\beta([2]) = e^{-\beta}+\mu_\beta([2])+\sum_{n=2}^\infty e^{-2\beta (n-1)}\mu_\beta([2]). 
\end{equation}
Then,
\begin{equation}\label{eq:eigenmeasure_Sigma_A_Renewal_cylinder_2_semifinal}
    \left(e^\beta-\sum_{n=0}^\infty e^{-2\beta n}\right)\mu_\beta([2]) = \frac{2 \sinh(\beta) - 1}{1-e^{-2\beta}} \mu_\beta([2])=  e^{-\beta}
\end{equation}
The number multiplying $\mu_\beta([2])$ in the identity above is zero if, and only if, $\sinh(\beta) = \frac{1}{2}$, and for such $\beta$ it is straightforward that there is not a probability conformal measure, and that $\beta \neq \log(1+\sqrt{2})$. Consider from now on the remaining case that $\sinh(\beta) \neq \frac{1}{2}$. 
By imposing that $\mu_\beta$ is a probability, that is,
\begin{equation*}
    e^{-\beta}+\mu_\beta([2])+ \mu_\beta([2]) \sum_{n=1}^\infty e^{-n\beta} = \frac{1-e^{-\beta}}{e^{\beta}+e^{-2\beta}-2} = 1,  
\end{equation*}
one can solve the RHS equation above, and conclude that there exists a unique positive value $\beta$, namely $\beta = \log(1+\sqrt{2}) = \beta_c$, such that $\mu_\beta$ is a conformal probability measure. In this case, the measure $\mu_{\beta_c}$ satisfies
\begin{align*}
   \mu_{\beta_c}([1]) &= e^{-\beta_c},\\
   \mu_{\beta_c}([2]) &=  e^{-\beta_c}\frac{1-e^{-2\beta_c}}{2 \sinh(\beta_c) - 1},\\
   \mu_{\beta_c}([n]) &= e^{-\beta_c(n-1)}\frac{1-e^{-2\beta_c}}{2 \sinh(\beta_c) - 1}, \quad n>2.
\end{align*}
Its uniqueness is granted because, for every cylinder $[\alpha]$, $|\alpha|>1$, we may apply the conformality equation \eqref{eq:conformality_Pair_renewal} finite times, and therefore the measure for every cylinder $[\alpha]$ depends only on the cylinders of length one. 

Summarizing the results, there exist two conformal probability measures living on $Y_A$ if and only if $\beta > \beta_c$, and there exists a conformal probability measure living on $\Sigma_A$ if and only if $\beta = \beta_c$, and this measure is unique. By Corollary \ref{cor:decomposition_eigenmeasures}, every probability conformal measure on $X_A$ can be written as a sum of its normalized restrictions on $\Sigma_A$ and $Y_A$, so there are no conformal probability measures on $X_A$ for $\beta < \beta_c$. 
\end{proof}

\begin{remark}\label{remark:probability_eigenmeasures_Y_A_pair_renewal_on_cylinders} We observe that similar to what happens for conformal probability measures living on $\Sigma_A$, it is possible to determine the conformal probability measure on each $Y_A$-family by its values on the cylinders of length one, with the difference that we must consider the measure on the singletons on the respective empty stems. We have
\begin{equation*}
    \sigma(C_n \cap Y_A(\xi^{0,k})) = \sigma(C_n) \cap Y_A(\xi^{0,k}), \quad n \in \mathbb{N}, \quad k \in \{1,2\},
\end{equation*}
and then,
\begin{equation*}
    \sigma(C_n)\cap Y_A(\xi^{0,k}) = \begin{cases}
                                        Y_A(\xi^{0,k}), \quad \text{if } n = 1;\\
                                        Y_A(\xi^{0,k}) \cap \left(C_1 \sqcup \bigsqcup_{n \in \mathbb{N}} C_{2n}\right), \quad \text{if } n = 2;\\
                                        Y_A(\xi^{0,k}) \cap C_{n-1}, \quad \text{otherwise},
                                     \end{cases}
\end{equation*}
where $k \in \{1,2\}$. The conformality equation \eqref{eq:conformality_Pair_renewal} is analogous for generalized cylinders. We denote by $\mu_{\beta,k}$ the extremal conformal probability living on  $Y_A(\xi^{0,k})$, $k \in \{1,2\}$. We have
\begin{equation}\label{limit_C_2}
\begin{cases}
    \mu_{\beta,k}(C_1) &= e^{-\beta},\\
    \mu_{\beta,k}(C_n) &= e^{-\beta(n-2)}\mu_{\beta,k}(C_2), \quad n>2.
\end{cases}
\end{equation}
Since these measures are probabilities we must have
\begin{align*}
    e^\beta \mu_{\beta,k}(C_2)=e^{-\beta}+\mu_{\beta,k}(C_2)+\mu_{\beta,k}(\{\xi^{0,1}\})+\sum_{n=2}^\infty e^{-2\beta (n-1)}\mu_{\beta,k}(C_2).
\end{align*}
We obtain from Lemma \ref{thm:counting_configurations_pair_renewal_shift} and identity \eqref{eq:c_omega_probability} that
\begin{equation*}
    \mu_{\beta,k}(C_2) = \begin{cases}
                           \left(1 + \frac{4 e^\beta}{\mathfrak{L}(\beta)}\right) e^{-\beta}\frac{1-e^{-2\beta}}{2 \sinh(\beta) - 1}, \quad \text{if } k=1;\\
                           \left(1 + \frac{4 e^\beta}{4 + e^{-\beta} \mathfrak{L}(\beta)}\right) e^{-\beta}\frac{1-e^{-2\beta}}{2 \sinh(\beta) - 1}, \quad \text{if } k=2,
                         \end{cases}
\end{equation*}
where
\begin{equation*}
    \mathfrak{L}(\beta) := \sum_{n \in \mathbb{N}_0}\left[(1-\sqrt{2})^{n} + (1-\sqrt{2})^{n+1} + (1+\sqrt{2})^{n} + (1+\sqrt{2})^{n+1}\right] e^{-\beta n}.
\end{equation*}
\end{remark}

The figure \ref{fig:phase_transition_pair_renewal_potential_1} compares the standard formalism of $\Sigma_A$ to the generalized one of $X_A$ for the pair renewal shift and potential $F \equiv 1$.

\begin{figure}[h!]
  \hspace{-.5cm}
 \includegraphics[scale=.3]{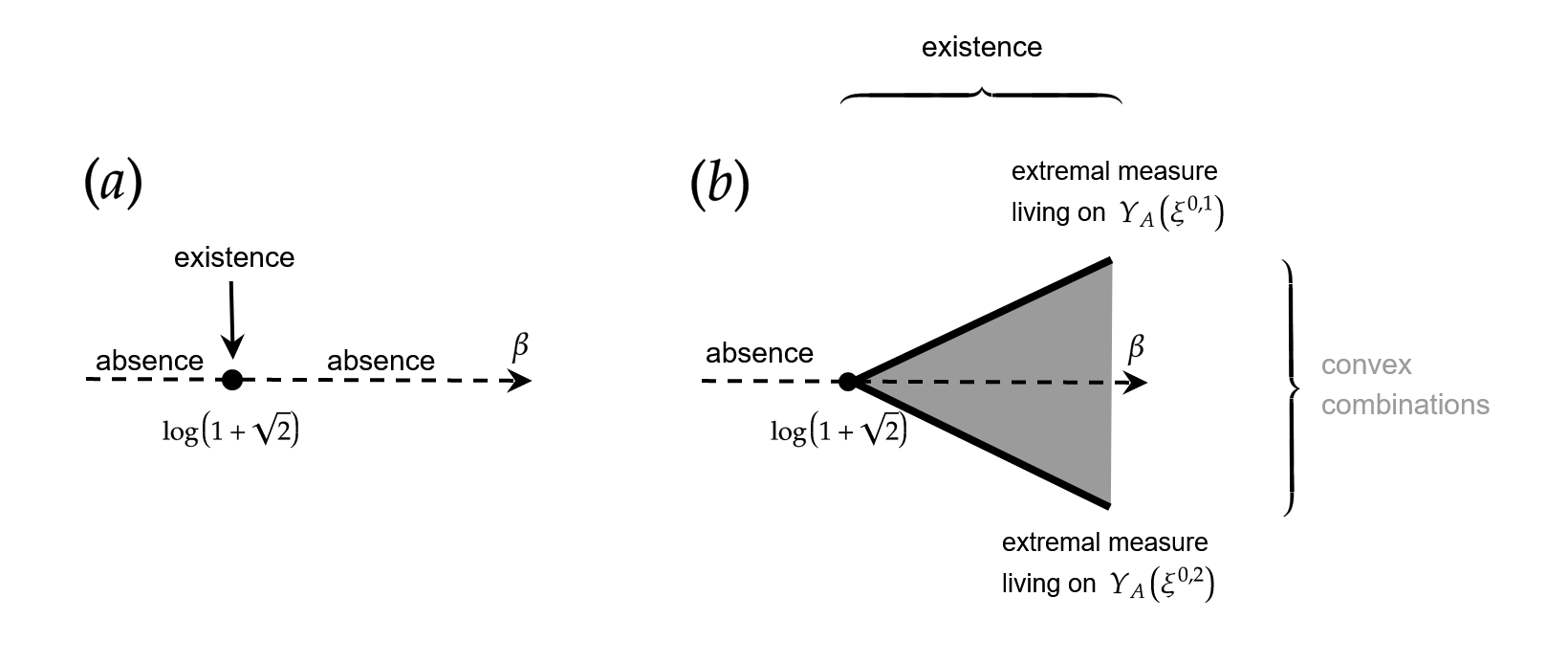}
 
 \caption{The phase transitions on different thermodynamic formalisms for conformal probabilities in the pair renewal shift space with potential $F\equiv 1$. The picture $(a)$ represents the standard formalism on $\Sigma_A$, where we have a unique conformal probability for a unique possible inverse of temperature, namely $\beta = \log(1+ \sqrt{2})$ (black dot). The picture $(b)$ represents the generalized formalism on $X_A$ and, unlike in $(a)$, we are able to detect not only the previous standard conformal probability, but also, for $\beta > \log(1+ \sqrt{2})$, two extremal conformal probabilities (black lines), each one living on a different $Y_A$-family. In this case, by convex combinations (gray triangle) we have infinitely many conformal measures for each $\beta > \log(1+ \sqrt{2})$.  \label{fig:phase_transition_pair_renewal_potential_1}}
\end{figure}

Still for the case where $F\equiv 1$, we shall connect the $e^\beta$-conformal measures living on $Y_A$ ($\beta > \log(1+\sqrt{2}$) to the unique one living on $\Sigma_A$ ($\beta_c = \log(1+\sqrt{2})$). 

Consider the set $\Phi_\beta$ of $e^\beta$-conformal probabilities living on $Y_A$ for the pair renewal shift and potential $F \equiv 1$ as above, which is the set of convex combination of the extremal probabilities $\mu_{\beta,1}$ and $\mu_{\beta,2}$. Also we consider the Hausdorff distance between $\Phi_\beta$ and $\mu_{\beta_c}$ on the set of probability measures on $X_A$, which is given by
\begin{equation*}
    d_H(\mu_{\beta_c},\Phi_\beta)=\max \left\{\,\inf _{y\in \Phi_\beta}d(\mu_{\beta_c},y),\,\sup _{y\in \Phi_\beta}d(\mu_{\beta_c},y)\,\right\}= \sup _{y\in \Phi_\beta}d(\mu_{\beta_c},y),
\end{equation*}
where $d$ is a metric compatible with the weak$^*$-topology.

\begin{theorem}\label{thm:pair_renewal_colapsing} $d_H(\mu_{\beta_c},\Phi_\beta) \to 0$ for $\beta \to \beta_c$ by above.
\end{theorem}

\begin{proof} Since for $k \in \{0,1\}$ we have
\begin{equation*}
    \mu_{\beta,k}(\{\xi^{0,k}\}) = \mu_\beta(\{\xi^{0,1}\})=\frac{4}{\mathcal{L}(\beta)}\to 0
\end{equation*}
when $\beta \to \beta_c$ by above, then for every finite set $F \subseteq Y_A$ it holds that $\mu_{\beta,k}(F) \to 0$.

Also, it follows from remark \ref{remark:probability_eigenmeasures_Y_A_pair_renewal_on_cylinders} that
\begin{align*}
    \lim_{\beta \to \beta^c} \mu_{\beta,1}(C_2) &=  \lim_{\beta \to \beta^c} \left(1 + \frac{4 e^\beta}{\mathfrak{L}(\beta)}\right) e^{-\beta}\frac{1-e^{-2\beta}}{2 \sinh(\beta) - 1} = e^{-\beta_c}\frac{1-e^{-2\beta_c}}{2 \sinh(\beta_c) - 1} = \mu_{\beta_c}([2]),\\
    \lim_{\beta \to \beta^c} \mu_{\beta,2}(C_2) &=  \lim_{\beta \to \beta^c} \left(1 + \frac{4 e^\beta}{4 + e^{-\beta} \mathfrak{L}(\beta)}\right)e^{-\beta}\frac{1-e^{-2\beta}}{2 \sinh(\beta) - 1} = \mu_{\beta_c}([2]).
\end{align*}
That is, $\lim_{\beta \to \beta^c} \mu_{\beta,k}(C_2) = \mu_{\beta_c}([2])$, $k \in\{1,2\}$. Consequently, by equation \eqref{limit_C_2}, we have
\begin{equation*}
    \lim_{\beta \to \beta^c}\mu_{\beta,k}(C_n) = \mu_{\beta_c}([n]), \quad n \geq 2, 
\end{equation*}
and for $C_1$ the limit above is straightforward.

By \eqref{eq:conformality_Pair_renewal} the limit above also holds for every generalized cylinder set on positive words i.e.,
\begin{equation*}
    \lim_{\beta \to \beta^c}\mu_{\beta,k}(C_{\alpha}) = \mu_{\beta_c}([\alpha]), \quad \alpha \text{ positive admissible word}. 
\end{equation*}
We recall that every basic set generated by the subbasis of cylinders and their complements of the generalized pair renewal shift $X_A$ is a countable union of disjoint positive generalized cylinders jointly with a finite subset of $Y_A$ as in \eqref{eq:basis_pair_renewal_is_finite_set_union_countable_cylinders}. Also, the basis of finite intersections of these cylinders and complements of cylinders is, by definition, closed under finite intersections. So, by Theorem 8.2.17 of \cite{Bogachev2007}, both extremal measures $\mu_{\beta,1}$ and $\mu_{\beta,2}$ converges to $\mu_{\beta_c}$ on the weak$^*$ topology. 

We claim that $\Phi_\beta$ is closed. In fact, let $\{\eta_n\}_{\mathbb{N}}$ be a sequence in $\Phi_\beta$ converging weakly$^*$ to a Borel probability $\nu$. Since $U$ is open, $C_c(U)$ can be seen as a subspace of $C_c(X_A)$, so by Theorem \ref{thm:equivalences_conformal_measures_generalized_Markov_shift} we have
\begin{equation*}
    \int_U fd\eta=\lim \int_U f d\eta_n=\lim \int_{X_A} L_{-\beta F}(f)d\eta_n=\int_{X_A}L_{-\beta F}(f) d\eta,
\end{equation*}
and hence $\eta \in \Phi_\beta$, and the claim is proved. Since $X_A$ is a compact metric space, we have that the space of probability Borel measures is compact on the weak$^*$ topology and therefore $\Phi_\beta$ is compact, since it is closed. On the other hand, the function $\nu \mapsto d(\mu_{\beta_c},\nu)$ is continuous, and by compactness of $\Phi_\beta$, there exists $\nu_\beta\in  \Phi_\beta $ satisfying
\begin{equation*}
    \sup _{\nu \in \Phi_\beta}d(\mu_{\beta_c},\nu)=d(\mu_{\beta_c},\nu_\beta ).
\end{equation*}
Since $\nu_\beta \in \Phi_\beta$, then $\nu_\beta=\lambda \mu_{\beta,1}+ (1-\lambda) \mu_{\beta,2} $, for some $\lambda \in [0,1]$. 

We prove now that $\lim_{\beta\to \beta_c}d(\mu_{\beta_c},\eta_\beta)=0$. Indeed, given $\epsilon>0$ there exists $\delta$, s.t. $\beta-\beta_c<\delta$ implies
\begin{align*}
 \left| \int f d\mu_{\beta,1} -\int f d\mu_{\beta_c}\right|<\epsilon \quad \text{and} \quad  \left| \int f d\mu_{\beta,2} -\int f d\mu_{\beta_c}\right|<\epsilon 
\end{align*}
for every $f\in C(X_A)$. Then
\begin{equation*}
    \left|\int f d \nu_\beta -\int f d\mu_{\beta_c}\right|<\epsilon,
\end{equation*}
and so $\lim_{\beta\to \beta_c}d(\mu_{\beta_c},\nu_\beta)=0$, which implies $\lim_{\beta\to \beta_c}d_H(\mu_{\beta_c},\Phi_\beta)=0$. 
\end{proof}

\subsection{Prime renewal shift: infinitely many extremal conformal measures}\\

Now we study the conformal measures on the prime renewal shift. We will see that the existence of infinitely many empty-stem configurations can lead to the existence of infinitely many extremal conformal measures living on $Y_A$. First, we characterize the extremal conformal measures living on $Y_A$ as in the standard renewal shift. Take $p \in \{1\}\cup\{q \in \mathbb{N}: q \text{ is a prime number}\}$, and define
\begin{equation}\label{eq:coefficients_Y_A_confomal_prime_renewal_shift}
    c_{\omega,p}:= \mu(\{\omega^p\}), \quad \omega^p \in \mathfrak{W}_p
\end{equation}
where $\xi^{0,p}$ is the empty-stem configuration such that and $R_{\xi^{0,p}}(e) = \{1,p\}$, as it is in subsection \ref{subsec:Prime_Renewal_shift}. Consider a potential $F:U \to \mathbb{R}$ and define $D(\omega^p) = e^{F(\omega^p)}$.

By Theorem \ref{theorem.coeficientes}, identity \eqref{eq:c_omega_probability} and Theorem \ref{thm:extremal_conformal_measures_general}, each extremal conformal probability living on $Y_A$ lives in some $Y_A$-family $Y_A(\xi^{0,p})$, and it satisfies
\begin{align}\label{eq:system_conformal_prime_renewal}
    \begin{cases}
        c_{\omega,p} e^{ F_{|\omega|}(\omega^p)} = c_{e^p}, \quad \omega^p \in \mathfrak{W}_p \setminus\{e\};\\
        \sum_{\omega \in \mathfrak{W}_p}c_{\omega,p} = 1.    
    \end{cases}
\end{align}

\begin{lemma}\label{prop:control_number_configurations} For the prime renewal shift we have that
\begin{equation*}
    2^{n-1}\leq |\{\xi \in Y_A(\xi^{0,p}): |\kappa(\xi)|=n\}|\leq 3^n, \quad \forall n \in \mathbb{N},
\end{equation*}
for every $p \in \{1\}\cup\{q \in \mathbb{N}: q \text{ is a prime number}\}$.
\end{lemma}

\begin{proof} See Proposition 4.76 of \cite{Raszeja2020}.
\end{proof}

\begin{corollary} \label{prop:entropy_prime_renewal_bound} For the prime renewal shift we have $h_G \in [\log 2, \log 3]$.
\end{corollary}

\begin{corollary}
Let $F\equiv 1$. Then, we have the following results.
\begin{itemize}
    \item[$(i)$] For $\beta> h_G$, there exists infinitely countable many extremal $e^{\beta}$-conformal probability measures that vanishes on $\Sigma_A$, each one of them living on a unique $Y_A$-family.
    \item[$(ii)$] For $\beta < h_G$ there are not $e^\beta$- conformal probability measure living on $Y_A$.
\end{itemize}
\end{corollary}
\begin{proof}
It is a straightforward consequence of Corollary \ref{cor:phase_transition_entropy_gauge}. 
\end{proof}

\subsection{Infinite entropy and uncountably many extremal conformal measures}\label{uncountable_extremal}\\

Many of the previous examples in this section have finite entropy. Still, the generalized formalism can be used on countable Markov shifts with infinite Gurevich entropy, including the full shift and many others. In this subsection, we use an idea from \cite{JaeKessLamei2014}, where the authors introduced loop spaces. We will see that for any irreducible matrix $A$ such that $Y_A \neq \emptyset$, if $\varphi(x)\geq \log (x_0+1)$, then $e^{\beta \varphi}$-conformal probability measures living on $Y_A$ always do exist for $\beta > \beta_c$, with $\beta_c$ being the positive solution for $\zeta(\beta) = 2$, where $\zeta$ is the Riemann zeta function. In particular, since the full shift has the largest set of stems comparing with $Y_A$-families associated with any irreducible matrix $A$, we can consider the matrix of the Example \ref{exa:uncountable_Y_A}, where we have uncountably many $Y_A$-families. This example will admit uncountably many extremal $e^{\beta \varphi}$-conformal measures when $\beta$ is large enough.

We denote by $\Sigma_{\text{full}}$ and $X_{\text{full}}$ the respective CMS and GCMS of the full shift. In $X_{\text{full}}$, there exists a unique $Y_A$-family, denoted by $Y_{\text{full}}$, associated to the empty-stem configuration $\xi^0$, where $R_{\xi^0}(e) = \mathbb{N}$. We consider the renewal shift $\Sigma_{\text{ren}}$ and its GCMS $X_{\text{ren}}$. Moreover, on this section, $\Sigma_{\text{ren},1}^f$ will denote the non-empty finite admissible words of the renewal shift that ends in $1$ and the symbol $1$ does not appear in any other position of the word:
$$\Sigma_{\text{ren},1}^f = \{1, 21, 321, 4321, 54321, \dots\}.$$
We define the map $\psi:\Sigma_{\text{full}} \to \Sigma_{\text{ren}}$, given by $\psi(x) = \psi(x_0x_1x_3\cdots) = \Theta(x_0) \Theta(x_1) \Theta(x_2) \cdots$, where $\Theta: \mathbb{N} \to \Sigma_{\text{ren},1}^f$ is given by
\begin{equation*}
    \Theta(x_p) = \begin{cases}
                    1, \quad \text{if } x_p = 1,\\
                    x_p (x_p -1) (x_p -2) \cdots 1, \quad \text{otherwise}.
                  \end{cases}
\end{equation*}
The function $\Theta$ is a bijection, whose inverse map is given by $\Theta^{-1}(n,n-1,\dots, 1) = n$.

Observe that every element $x \in \Sigma_{\text{ren}}$ can be written in a unique way as $x = B_0 B_1 B_2 \cdots$, where $B_k \in \Sigma_{\text{ren},1}^f$ for every $k$. Moreover, every infinite sequence of elements of $\Sigma_{\text{ren},1}^f$ is an element of $\Sigma_{\text{ren}}$. Consequently, $\psi$ is a bijection, and $\psi^{-1}(\omega) = \Theta^{-1}(B_0) \Theta^{-1}(B_1) \Theta^{-1}(B_2) \cdots,$
when $\omega = B_0 B_1 B_2 \cdots$, $B_k \in \Sigma_{\text{ren},1}^f$. The functions $\psi$ and $\psi^{-1}$ can be extended by the same rules for $X_{\text{full}}$ and $X_{\text{ren}}$ respectively, and $\psi$ is still a bijection, whose inverse is $\psi^{-1}$, including the empty words.


\begin{lemma}\label{lemma:algebrically_conjugated_potentials_renewal_and_full} Consider a potential on the GCMS of the fullshift $F: U_{\text{full}} \to \mathbb{R}$ satisfying
\begin{equation*}
    F(x) = G_{x_0}(\psi(x)) :=\sum_{k=0}^{x_0-1}G\circ\sigma^{k}(\psi(x)),
\end{equation*}
where $G:U_{\text{ren}} \to \mathbb{R}$ is a potential on the GCMS of the renewal shift, where without ambiguity we associate their configurations of finite stem to their respective stems. Then,
\begin{equation*}
    F_n(x) = G_{\sum_{i=0}^{n-1} x_i}(\psi(x)) := \sum_{k=0}^{(\sum_{i=0}^{n-1} x_i) - 1}G\circ\sigma^{k}(\psi(x)),
\end{equation*}
for every $n \leq |x|$ if $x$ is finite, and for every $n \in \mathbb{N}$ if $x$ is infinite.
\end{lemma}

\begin{proof} We prove the lemma for $x \in \Sigma_{\text{full}}$ and the proof for the finite stem configurations is similar. We write $\psi(x) = B_0 B_1 B_2 \cdots$, $B_j \in \Sigma_{\text{ren},1}^f$ for every $j$. Observe that $|B_j| = x_j$, and then we write $B_j = b_{j,0} b_{j,1} \cdots b_{j,x_j - 1}$, where $b_{j,p} \in \mathbb{N}$. Then,
\begin{align*}
    F_n(x) &= \sum_{k=0}^{n-1} F(\sigma^k(x)) = \sum_{k=0}^{n-1} G_{(\sigma^k(x))_0}(\psi(\sigma^k(x))) = \sum_{k=0}^{n-1} G_{x_k}(B_k B_{k+1} \cdots).
\end{align*}
Now, note that
\begin{equation*}
    G_{x_k}(B_k B_{k+1} \cdots) = \sum_{j = 0}^{x_k-1} G\circ \sigma^j(B_k B_{k+1} \cdots) = \sum_{j = 0}^{x_k-1} G(b_{k,j}b_{k,j+1} \cdots b_{k,x_k-1} B_{k+1} \cdots),
\end{equation*}
and therefore
\begin{align*}
    F_n(x) &= \sum_{k=0}^{n-1} G_{x_k}(B_k B_{k+1} \cdots) = \sum_{k=0}^{n-1}\sum_{j = 0}^{x_k-1} G(b_{k,j}b_{k,j+1} \cdots b_{k,x_k-1} B_{k+1} \cdots)= G_{\sum_{i=0}^{n-1} x_i}(\psi(x)). 
\end{align*} 
\end{proof}


\begin{theorem}\label{thm:conformal_measures_full_shift} Consider a potential on the GCMS of the fullshift $F: U_{\text{full}} \to \mathbb{R}$ satisfying
\begin{equation*}
    F(x) = \sum_{k=0}^{x_0-1}G\circ\sigma^{k}(\psi(x)),
\end{equation*}
where $G:U_{\text{ren}} \to \mathbb{R}$ is a potential on the GCMS of the renewal shift. Then,
\begin{equation*}
    \sum_{x \in Y_{\text{full}}\setminus\{e\}} e^{-\beta F_{|x|}(x)} = \sum_{\omega \in Y_{\text{ren}}\setminus\{e\}} e^{-\beta G_{|\omega|}(\omega)}.
\end{equation*}
In particular, there exists a $e^{\beta F}$-conformal probability measure living on $Y_{\text{full}}$ if and only if there exists a $e^{\beta G}$-conformal probability measure living on $Y_\text{ren}$.
\end{theorem}

\begin{proof} By Lemma \ref{lemma:algebrically_conjugated_potentials_renewal_and_full} we have that
\begin{align*}
    \sum_{x \in Y_{\text{full}}\setminus\{e\}} e^{-\beta F_{|x|}(x)} &= \sum_{x \in Y_{\text{full}}\setminus\{e\}} \exp\left(-\beta G_{\sum_{i=0}^{|x|-1} x_i}(\psi(x))\right) = \sum_{x \in Y_{\text{full}}\setminus\{e\}} \exp\left(-\beta G_{|\psi(x)|}(\psi(x))\right) \\
    &= \sum_{\omega \in Y_\text{ren}\setminus\{e\}} \exp\left(-\beta G_{|\omega|}(\omega)\right),
\end{align*}
where we used that $\sum_{i=0}^{|x|-1} x_i = |\psi(x)|$, and that $\psi$ is a bijection. The rest of the statement of the theorem is straightforward from the characterization of the conformal probability measures living on $Y_A$-families. 
\end{proof}

Next result presents a class of potentials on the full shift that are encoded by potentials on the renewal shift as in Theorem \ref{thm:conformal_measures_full_shift}.

\begin{proposition} \label{prop:class_potentials_on_full_shift_encoded_on_the_renewal_shift} Consider a potential $F:U_{\text{full}} \to \mathbb{R}$ in the form $F(x) = f(x_0 + 1) - f(1)$, where $x_0$ is the first coordinate of $x$, and $f:\mathbb{N} \to \mathbb{R}$ a function. Then,
\begin{equation}\label{eq:encoding_full_renewal}
    F(x) = \sum_{k=0}^{x_0-1}G\circ\sigma^{k}(\psi(x)),
\end{equation}
for the potential $G: U_{\text{ren}} \to \mathbb{R}$ given by $G(x) = f(x_0+1) - f(x_0)$.
\end{proposition}

\begin{proof} For every $x \in U_{\text{full}}$ and $0 \leq k \leq x_0$, we have that $\left(\sigma^k(\psi(x))\right)_0 = x_0 - k$. Then,
\begin{align*}
    \sum_{k = 0}^{x_0 - 1} G \circ \sigma^k(\psi(x)) &= \sum_{k = 0}^{x_0 - 1} \left[f(x_0 - k +1) - f(x_0 -k)\right] = f(x_0 + 1) - f(1) = F(x).
\end{align*}
\end{proof}

\begin{example}\label{exa:phase_transition_full_shift} By choosing $G \equiv 1$ on the renewal shift, we obtain the potential $F(x) = x_0$ on the full shift. Since the Gurevich entropy for the renewal shift is $\log 2$, by Corollary \ref{cor:phase_transition_entropy_gauge} and Theorem \ref{thm:conformal_measures_full_shift}, there exists a (unique) $e^{\beta F}$-conformal probability measure living on $Y_{\text{full}}$ for $\beta > \log 2$, and for $\beta \leq \log 2$ there are no $e^{\beta F}$-conformal probability measures that vanishes on $\Sigma_{\text{full}}$. At the value $\beta = \log 2$ the existence is due to the fact that these conformal measures are KMS$_{\beta}$ states.
\end{example}

\begin{figure}[H]
    \centering
\scalebox{1}{
\tikzset{every picture/.style={line width=0.75pt}} 

\begin{tikzpicture}[x=0.75pt,y=0.75pt,yscale=-1,xscale=1]

\draw [line width=1.5]  [dash pattern={on 5.63pt off 4.5pt}]  (14,99) -- (243.37,99.92) ;
\draw [shift={(247.37,99.93)}, rotate = 180.23] [fill={rgb, 255:red, 0; green, 0; blue, 0 }  ][line width=0.08]  [draw opacity=0] (13.4,-6.43) -- (0,0) -- (13.4,6.44) -- (8.9,0) -- cycle    ;
\draw  [color={rgb, 255:red, 0; green, 0; blue, 0 }  ,draw opacity=1 ][fill={rgb, 255:red, 0; green, 0; blue, 0 }  ,fill opacity=1 ] (86,99.5) .. controls (86,97.01) and (88.01,95) .. (90.5,95) .. controls (92.99,95) and (95,97.01) .. (95,99.5) .. controls (95,101.99) and (92.99,104) .. (90.5,104) .. controls (88.01,104) and (86,101.99) .. (86,99.5) -- cycle ;
\draw [color={rgb, 255:red, 0; green, 0; blue, 0 }  ,draw opacity=1 ][fill={rgb, 255:red, 0; green, 0; blue, 0 }  ,fill opacity=1 ][line width=2.25]    (90.5,99.5) -- (237.37,99.93) ;

\draw (72,114.4) node [anchor=north west][inner sep=0.75pt]    {$\log 2$};
\draw (236,71.4) node [anchor=north west][inner sep=0.75pt]    {$\beta $};
\draw (91,44.4) node [anchor=north west][inner sep=0.75pt]    {$\overbrace{\ \ \ \ \ \ \ \ \ \ \ \ \ \ \ \ \ \ \ \ \ \ \ \ \ \ \ \ \ \ \ \ \ \ \ }$};
\draw (134,14) node [anchor=north west][inner sep=0.75pt]   [align=left] {existence};

\end{tikzpicture}
}
\caption{\label{fig:phase_transition_full_shift}}
\end{figure}

\begin{theorem}
Let $A$ be an irreducible matrix such that $Y_A\neq \emptyset$ and $\varphi:U \to \mathbb{R}$ such that $\varphi(x) \geq \log (x_0+1)$, for every $x \in U$, where $x_0$ is the first coordinate of $x$. Then, there exist $e^{\beta \varphi}$-conformal probability measures living on $Y_A$ when $\beta$ is large enough. In particular, one $e^{\beta \varphi}$-conformal probability measure for each $Y_A$-family.
\end{theorem}
\begin{proof} By taking $f(x) = \log x_0$ in Proposition \ref{prop:class_potentials_on_full_shift_encoded_on_the_renewal_shift}, we generate the potential $F(x) = \log(x_0 + 1)$ on $U_{\text{full}}$, and $G(x) = \log(x_0 + 1) - \log(x_0)$ on $U_{\text{ren}}$ as in the Theorem \ref{thm:conformal_measures_full_shift}. For each $\beta>0$, there exists a $e^{\beta F}$-conformal measure living in $Y_{\text{full}}$ if, and only if,  there exists a $e^{\beta G}$-conformal measure living on  $Y_\text{ren}$. By Theorem \ref{thm:equivalences_conformal_measures_generalized_Markov_shift}, $e^{\beta G}$-conformal measures on the renewal shift are precisely the fixed point measures of $L_{-\beta G}$. By Theorem \ref{thm:complete_characterization_eigenmeasures_renewal}, for $\beta_c$ being the positive solution for $\zeta(\beta) = 2$, where $\zeta$ is the Riemann zeta function, we have that the eigenvalue is one for $\beta \geq \beta_c$. Hence, by the same theorem, there exists $e^{\beta G}$-conformal measures living on $Y_\text{ren}$ for $\beta > \beta_c$, and we have the absence of these measures for $\beta = \beta_c$. Since
\begin{equation*}
    \infty = \sum_{x \in Y_\text{ren}\setminus\{e\}} e^{-\beta_c G_{|x|}(x)} \leq  \sum_{x \in Y_\text{ren}\setminus\{e\}} e^{-\beta G_{|x|}(x)}, \text{ $\beta < \beta_c$},
\end{equation*}
we also have the absence of $e^{\beta G}$-conformal measures for $\beta < \beta_c$. For the potential $F(x)=\log(x_0+1)$, the interval on $\beta$ where there are $e^{\beta F}$-conformal measures is $(\beta_c,\infty)$. Now, the unique $Y_A$-family in $X_{\text{full}}$ is the largest $Y_A$-family possible among all transition matrices since its stems are the set of all finite words, including the empty word. Consequently, for every transitive GCMS with $Y_A\neq \emptyset$, and for any potential $\varphi(x) \geq \log(x_0+1)$, we have
\begin{equation*}
    \sum_{x \in Y_{A}(\xi^0)\setminus\{e\}} e^{-\beta \varphi_{|x|}(x)} \leq \sum_{x \in Y_{\text{full}}\setminus\{e\}} e^{-\beta F_{|x|}(x)},
\end{equation*}
where $A$ is a transition matrix and $\xi^0$ is any empty-stem configuration of $X_{A}$. Then, for $\beta > \beta_c$, we have an extremal $e^{\beta \varphi}$-conformal probability measure living on each $Y_A$-family, for any arbitrary irreducible transition matrix.
\end{proof}

\begin{corollary}[uncountably many extremal conformal measures]
If $Y_A$ admits uncountably many different $Y_A$-families, then we have the existence of uncountably many extremal $e^{\beta \varphi}$-conformal probability measures living on $Y_A$ when $\varphi(x) \geq \log (x_0+1)$ and $\beta$ is large enough. We have this situation in Example \ref{exa:uncountable_Y_A}.
\end{corollary}

\subsection{Infinite emitters, $Y_A$-families and extremal conformal measures}\\

It is well known that the algebras $\mathcal{O}_A$ and graph algebras are not isomorphic in many cases, see, for instance, \cite{KatsuraSimsTomforde2009, Tomforde2002}. The most important cases for us are matrices $A$ which are not row finite. For this class of matrices, none of these notions of algebras include each other. In the previous examples of this section, as the renewal shift, pair renewal shift, and prime renewal shift, there exists a bijection between infinite emitters and the extremal $e^{\beta}$-conformal probability measures. We know that $e^{\beta}$-conformal probability measures are the KMS$_{\beta}$ states for the gauge potential. Recently, K. Thomsen \cite{Thomsen2017} proved the existence of KMS$_{\beta}$ weights for $\beta > h_{G}$ in the case of graph algebras, these weights live on a set of finite paths, in some sense something analogous in what happened in our examples with the conformal measures living on $Y_A$. Moreover, as in our previous examples, he proved that there is a bijection between infinite emitters and extremal KMS$_\beta$ weights.

These results may suggest that the existence of a bijection between infinite emitters and the extremal $e^{\beta}$-conformal probability measures is also a general fact for the algebras $\mathcal{O}_A$. However, that is not the case. The right object to codify the extremal conformal measures for these algebras are the $Y_A$-families, defined in the subsection \ref{subsec:Y_A_families}. A good example to see this different behaviour is to consider a Cuntz-Krieger algebra with infinite matrix which is not a graph algebra as follows.

Consider the example 4.2 of I. Raeburn and W. Szyma\'nski \cite{RaeSzy2004}, where the transition matrix is given by
\begin{equation*}
    A(i,j) = \begin{cases}
                1, \quad \text{if } i=j;\\
                1, \quad \text{if } i = j + 2;\\
                1, \quad \text{if } i \in \{1,2\} \text{ and } j \geq 3;\\
                0, \quad \text{Otherwise.}
             \end{cases}
\end{equation*}
In this case, $\Sigma_A$ is topologically mixing and $X_A$ is compact, since $S_1^*S_1 + S_2 S_2^* = 1$, and then $\mathcal{O}_A$ is unital. The symbols $1$ and $2$ are infinite emitters, so $Y_A \neq \emptyset$, however there exists only one $Y_A$-family, since the unique limit point of the columns of $A$ is the column
\begin{equation*}
    \begin{pmatrix}
        1 \\ 1 \\ 0 \\ 0 \\ 0 \\ \vdots 
    \end{pmatrix}.
\end{equation*}
Also, $X_A$ is $s$-compact since the column matrix above has a finite quantity of $1$'s. Moreover, this example has finite Gurevich entropy. By Corollary \ref{cor:phase_transition_entropy_gauge}, for the gauge potential, there exists a unique extremal $e^{\beta}$-conformal probability measure living on $Y_A$ for $\beta > h_G$, and there are not $e^{\beta}$-conformal probability measures living on $Y_A$ for $\beta \leq \beta_c$.

\section{Eigenmeasures} \label{sec:eigenmeasures}

After the study of the conformal measures, the next natural problem concerns eigenmeasures associated with more general eigenvalues $\lambda >0$, $L_{\beta F}^*\mu = \lambda\mu$. Recently, in \cite{Shwartz2019}, O. Shwartz studied the problem of finding eigenmeasures on CMS $\Sigma_A$ for transient potentials (see the definition \ref{modes-recu} and the Generalized RPF Theorem \ref{GRPF}), he also uses a compactification, namely the Martin boundary. However, in addition to the transience of the potential, he assumes that $\Sigma_A$ is locally compact.

By the construction of $X_A$, we only assume transitivity, and as we will see, we can deal with potentials with different recurrence modes for different values of $\beta$. In addition, since we know that for compact and $s$-compact GCMS the Gurevich pressure is a natural definition, for potentials which restricted to $\Sigma_A$ are recurrent it is important to study the problem of the eigenmeasure for the eigenvalue $\lambda_{\beta}=e^{P(\beta F)}$ since it is the only possible eigenvalue. More precisely, we present in the next section one example where the compactness of $X_A$ guarantees the existence of the eigenmeasure associated with the value $\lambda_{\beta}=e^{P(\beta F)}$ for every $\beta>0$, such that for $\beta$ small enough the potential is positive recurrent, and for large $\beta$ it is transient.

We found conditions for the existence of (new) conformal measures in the previous sections. One should ask if it is always possible to grant the existence of an eigenmeasure associated with the value $\lambda_{\beta}=e^{P(\beta F)}$.  

Given an irreducible matrix $A$, the following condition assures the existence of an eigenmeasure (living on $\Sigma_A$) associated with the eigenvalue $e^{P_{G}(F)}$ for a potential $F$ with summable variations:
\begin{equation}
    \displaystyle\sum_{n \in \mathbb{N}} e^{\sup F\vert_{C_n}} < \infty.
\end{equation}
See R. D. Mauldin and M. Urba\'nski in \cite{MaulUr2001} for the existence of the eigenmeasure, and R. Freire and V. Vargas for the existence of the eigenfunction \cite{FreireVargas2018}.

\begin{example} As an example of potential with such regularity one could choose $F(x) = -x_0$.
\end{example}

Since $\Sigma_A \subseteq X_A$, and in many cases $X_A$ is compact, it is very natural to expect to find more eigenmeasures for the Ruelle transformation on $X_A$ than in the standard symbolic space $\Sigma_A$. Our result for existence is a direct consequence of Theorem 3.6 on \cite{DenYu2015}, the original argument is due to S. J. Patterson \cite{Patterson1976}. We repeat the proof below.

\begin{theorem}[Denker-Yuri]\label{thm:existence_eigenmeasures_Denker_Yuri}
Let $X_A$ be compact, suppose there exists a $x\in X_A$ such that $P(F,x)$ is finite and $F:U\to \mathbb{R}$ continuous, then there exists an eigenmeasure $m$ for $L_F$ with eigenvalue $e^{P(F,x)}$.
\end{theorem}
\begin{proof}
We shall use Lemma 3.1 of \cite{DenUr1991} to construct our measure. Given a sequence of real numbers $(a_n)_{n\in \mathbb{N}}$, there exists a sequence of positive numbers $(b_n)_{n\in \mathbb{N}}$ such that 
\begin{equation}\label{eq:Denker_Urbanski_lemma}
    \sum_{n\in \mathbb{N}}b_n\exp{(a_n-ns)}=\left\{\begin{array}{lr}
    <\infty,     & s>c; \\
     \infty,    & s\leq c;
    \end{array}\right.
\end{equation}
and $\lim_n b_n/b_{n+1}=1$, where $c:=\limsup_n a_n/n$. As the pressure $P(F,x)$ is finite, $Z_n(F,x)$ is finite for $n$ large enough. Then, without loss of generality, we assume that $Z_n(F,x)$ is finite for every $n \in \mathbb{N}$. By taking $a_n=\log
Z_n(F,x)$, we have $c=P(F,x)$, and we may define
\begin{equation*}
    M(p,x)=\sum_{n\in\mathbb{N}}b_ne^{-np}Z_n(F,x), \quad p > c,
\end{equation*}
and the measures
\begin{equation}
    m(p,x)=M(p,x)^{-1}\sum_{n\in \mathbb{N}}b_n e^{-np}\sum_{\sigma^n(y)=x}e^{F_n(y)}\delta_y,
\end{equation}
where $\delta_y$ is the Dirac measure on $y$. From \eqref{eq:Denker_Urbanski_lemma} we claim that ${\lim_{p\downarrow P(F,x)}M(p,x)=\infty}$. Indeed, for every sequence $(p_k)_{k \in \mathbb{N}}$ such that $p_k\downarrow c$, let $S(N,k)=\sum_{n=1}^{N}b_n\exp(-np_k)Z_n(F,x)$. Such sequence of two variables is increasing on both of them, by Monotone Convergence Theorem, we have $\lim_N\lim_kS(N,k)=\lim_k\lim_NS(N,k)$, and the claim is proved because $\lim_N\lim_kS(N,k) = \infty$. Now, fix a sequence $p_k\downarrow P(F,x)$ such that $m(p_k,x)\to m$ weakly$^*$. We have
\begin{align*}
    \int L_Fg(y) m(p_k,x)(dy)&= M(p_k,x)^{-1}\sum_{n\in \mathbb{N}} b_n e^{-np_k}\sum_{\sigma^n(y)=x}e^{F_n(y)}L_Fg(y)\\
    &=M(p_k,x)^{-1}e^{p_k}\sum_{n\in \mathbb{N}} b_n e^{-(n+1)p_k}\sum_{\sigma^{n+1}(w)=x}e^{F_{n+1}(w)}g(w),
\end{align*}
for every $g\in C_c(U)$. Define
\begin{equation}\label{eq:notation_Ruelle_Operator}
    L_F^{n+1} g (x) := \sum_{\sigma^{n+1}(w)=x}e^{F_{n+1}(w)}g(w).
\end{equation}
By using that $b_n = b_n+b_{n+1}-b_{n+1}$, one gets
\begin{align*}
    \int L_Fg(y) m(p_k,x)(dy)&= e^{p_k}M(p_k,x)^{-1}\sum_{n\in \mathbb{N}} (b_n/b_{n+1}-1)b_{n+1}e^{-(n+1)p_k}L^{n+1}_Fg(x) \\
    &+e^{p_k}\int g(y) m(p_k,x)(dy)-b_1 M(p_k,x)^{-1}L_Fg(x).
\end{align*}
Now, by taking $p_k\downarrow P(F,x)$, the RHS of the above expression becomes $e^{P(F,x)}\int g(y)dm(y)$. In fact, note that
\begin{align*}
    b_1M(p_k,x)^{-1}L_Fg(x) &\to 0,\\
    e^{p_k}\int g(y) m(p_k,x)(dy) &\to e^{P(F,x)}\int g(y)dm(y).
\end{align*}
We claim that
\begin{equation*}
    e^{p_k}M(p_k,x)^{-1}\sum_{n\in \mathbb{N}} (b_n/b_{n+1}-1)b_{n+1}e^{-(n+1)p_k}L^{n+1}_Fg(x) \to 0.
\end{equation*}
Indeed, the term above has the upper bound
\begin{equation}\label{eq:upper_bound_Denker_Yuri_proof}
e^{p_k}M(p_k,x)^{-1}\|g\|_{\infty}\sum_{n\in \mathbb{N}} |b_n/b_{n+1}-1|b_{n+1}e^{-(n+1)p_k}Z_{n+1}(F,x).
\end{equation}
Now, for every $\epsilon>0$, there exists $N\in\mathbb{N}$ s.t. $|b_n/b_{n+1}-1|<\epsilon$ for every $n>N$, then \eqref{eq:upper_bound_Denker_Yuri_proof} is bounded above by
\[e^{p_k}M(p_k,x)^{-1}\|g\|_{\infty}\left(\sum_{n=1}^{N}|b_n/b_{n+1}-1|b_{n+1}e^{-(n+1)p_k}Z_{n+1}(F,x)+\epsilon \sum_{n=N+1}^{\infty}b_{n+1}e^{-(n+1)p_k}Z_{n+1}(F,x)\right).\]
It is straightforward that the first term in the last expression vanishes when $k$ goes to infinity, since it is a finite sum. The remaining term in same expression is less or equal to $\epsilon$ because
\begin{equation*}
    \sum_{n=N+1}^{\infty}b_{n+1}e^{-(n+1)p_k}Z_{n+1}(F,x) \leq M(p_k,x).
\end{equation*}
Since $\epsilon$ is arbitrary the expression goes to zero as $p_k\downarrow P(F,x)$, proving the claim.
\end{proof}

\begin{remark} We emphasize to the reader that \eqref{eq:notation_Ruelle_Operator} was used just as a notation. We did not define it as a transformation on $C_c(U)$. 
\end{remark}

\begin{remark} The continuity of the potential is crucial for the proof of Theorem \ref{thm:existence_eigenmeasures_Denker_Yuri} as we will see now. If $F$ is continuous $L_Fg \in C_c(X_A)$, and $m(p_k,x) \to m$ weakly$^*$, then
\begin{equation*}
    \int L_Fg(y) m(p_k,x)(dy) \to \int L_Fg(y) m(dy),
\end{equation*}
for arbitrary $g \in C_c(U)$. The definition of eigenmeasure can be written for measurable potentials as linear transformation between positive measurable functions on $U$ to positive measurable functions on $X_A$ (see \cite{BelBisEndo2020, BisExelFrauRas2018}). For non-continuous potentials, all equivalences in Theorem \ref{thm:equivalences_conformal_measures_generalized_Markov_shift} but the notion of quasi-invariant measure still hold if we change the space where $L_F$ is defined considering the positive and mensurable functions defined on $U$ instead of $C_c(U)$, see \cite{BisExelFrauRas2018}. 

However, the potential has to be continuous for the proof of existence. For instance, consider the Example 2 in \cite{Sarig2001}, we may extend the potential $F$ from the standard renewal shift to $U \subset X_A$ as 
\begin{equation*}
    F(x) = \sum_{m\in \mathbb{N}} \mathbbm{1}_{C_{1m}}(x) \log f_m,
\end{equation*}
where $\log f_m = o(m)$. Such potential is discontinuous at $x = 1$. Also, in this case it holds that $P_G(\beta F\vert_{\Sigma_A}) = P(\beta F,x)$ for every $x \in X_A$. We know that for this example all the eigenmeasures (when they do exist) are finite and live on $\Sigma_A$. For $\beta \geq 1$ we have that $\beta F$ is transient, so $X_A$ is compact but for $\beta \geq 1$ do not exist eigenmeasures for $\beta F$.
\end{remark}

\begin{corollary}
Let $X_A$ be compact and $s$-compact. Suppose that $P_{G}(F)$ is finite, $F$ is continuous with uniform bounded distortion. Then, there exists an eigenmeasure $\mu_{\beta}$ for $L_{\beta F}$ associated with the eigenvalue $e^{P_{G}(\beta F)}$ for each $\beta \geq 1$.
\end{corollary}

\begin{corollary}\label{generalized_renewal}
Let $X_A$ the generalized renewal shift and $F:U \to \mathbb{R}$ bounded above with uniform bounded distortion. Then, for each $\beta >0$, we have that $P_{G}(\beta F)$ is finite, and there exists an eigenmeasure $\mu_{\beta}$ for $L_{\beta F}$ associated with the eigenvalue $e^{P_{G}(\beta F)}$.
\end{corollary}

\begin{remark} It is important to highlight that the thermodynamic formalism drastically changes in the generalized renewal shift with respect to the standard theory of CMS. To see this, compare the corollary above with the Theorem \ref{thm:phase_transition_renewal_shift_Sarig_Gurevich_pressure}.
\end{remark}

\begin{theorem} Let $\Sigma_A$ is topologically mixing, $X_A$ compact and $s$-compact. Let $F:U \to \mathbb{R}$ be a potential satisfying the Walters' condition with $\Var_1 F < \infty$ s.t. $P_G(F)< \infty$. Suppose that  $F\vert_{\Sigma_A}$ is positive recurrent and that the Denker-Yuri's probability eigenmeasure $\mu$ from Theorem \ref{thm:existence_eigenmeasures_Denker_Yuri} lives in $\Sigma_A$. Then, $\mu$ is a Sarig's eigenmeasure from Generalized RPF (Theorem \ref{GRPF}) for the same eigenvalue $\lambda = e^{P_G(F)}$.
\end{theorem}

\begin{proof} The restriction of $\mu$ to $\mathcal{B}_{\Sigma_A}$ is an eigenmeasure associated to the same eigenvalue as in the statement due to Proposition \ref{thm:restriction_eigenmeasures}. Since $F\vert_{\Sigma_A}$ satisfies Walters' condition we have, by positive recurrence and the generalized RPF theorem, the existence of the Sarig's eigenmeasure. The space of the eigenmeasures has dimension $1$, see \cite{Daon2013}, then we necessarily have that $\mu$ is a Sarig's eigenmeasure. 
\end{proof}

\subsection{An example of length-type phase transition}\\

At this subsection, we will always assume that $A$ is the matrix of the renewal shift, with $X_A$ and $
\Sigma_A$ its respective GCMS and CMS. We will compare in a concrete example the thermodynamic formalism 
for a same potential $F:U \to \mathbb{R}$, and  $F\vert_{\Sigma_A}:\Sigma_A \to \mathbb{R}$.

From now on, fix the potential $F:U \to \mathbb{R}$ given by $F(x) := \log(x_0) - \log(x_0+1)$,
where $x_0$ is the first coordinate of the stem of $\xi$.

\begin{lemma}[characterization of the eigenmeasures] \label{lemma:condition_for_eigenmeasures_renewal_potential_log}If a probability measure $\mu_\beta$ is an eigenmeasure for the Ruelle transformation $L_{\beta F}$, with associated eigenvalue $\lambda > 0$, then
\begin{equation}\label{eq:condition_for_eigenmeasures_renewal_potential_log}
    \mu_\beta(C_n) = \frac{1}{\lambda^n}\frac{1}{(n+1)^\beta}
\end{equation}
for every $n \in \mathbb{N}$. In particular,
\begin{equation}\label{eq:eigenmeasure_probability_sum_renewal_potential_log}
    1 = \mu_\beta(\{\xi^0\}) + \sum_{n \in \mathbb{N}} \frac{1}{\lambda^n}\frac{1}{(n+1)^\beta}.
\end{equation} 
\end{lemma}

\begin{proof} By Theorem \ref{thm:equivalences_conformal_measures_generalized_Markov_shift}, $\mu_\beta$ is an eigenmeasure as in the statement of this lemma if and only if it is a $\lambda e^{-\beta F}$-conformal measure in the sense of Denker-Urba\'nski, and then
\begin{equation}\label{eq:condition_conformality_renewal_potential_log_lemma_proof}
    \mu_\beta(\sigma(C_n)) = \int_{C_n} \lambda e^{-\beta F} d\mu_\beta,
\end{equation}
for every $n \in \mathbb{N}$. For $n =1$, we have $\sigma(C_1) = X_A$, and since $\mu_\beta$ is a probability, equation \eqref{eq:condition_conformality_renewal_potential_log_lemma_proof} gives
\begin{equation*}
    1 = \mu_\beta(\sigma(C_1)) = \int_{C_1} \lambda e^{-\beta F} d\mu_\beta = \lambda 2^\beta \mu_\beta(C_1),
\end{equation*}
that is,
\begin{equation}\label{eq:measure_on_C_1_renewal_potential_log_lemma_proof}
    \mu_\beta(C_1) = \frac{1}{\lambda}\frac{1}{2^\beta}.
\end{equation}
Now, for $n \neq 1$, then $\sigma(C_n) = C_{n-1}$ and then \eqref{eq:condition_conformality_renewal_potential_log_lemma_proof} implies
\begin{equation*}
    \mu_\beta(\sigma(C_n)) = \mu_\beta(C_{n-1}) = \int_{C_n} \lambda e^{-\beta F} d\mu_\beta = \lambda \left(\frac{n}{n+1}\right)^{-\beta} \mu_\beta(C_n),
\end{equation*}
i.e.,
\begin{equation}\label{eq:measure_on_C_n_renewal_potential_log_lemma_proof}
    \mu_\beta(C_n) = \frac{1}{\lambda} \left(\frac{n}{n+1}\right)^{\beta} \mu_\beta(C_{n-1}).
\end{equation}
Now we prove the validity of \eqref{eq:condition_for_eigenmeasures_renewal_potential_log}. The result is straightforward for $n = 1$ because of \eqref{eq:measure_on_C_1_renewal_potential_log_lemma_proof}. Now, suppose that \eqref{eq:condition_for_eigenmeasures_renewal_potential_log} holds for $C_n$. By equation \eqref{eq:measure_on_C_n_renewal_potential_log_lemma_proof} we have
\begin{align*}
    \mu_\beta(C_{n+1}) = \frac{1}{\lambda} \left(\frac{n+1}{n+2}\right)^{\beta} \mu_\beta(C_n) = \frac{1}{\lambda} \left(\frac{n+1}{n+2}\right)^{\beta} \frac{1}{\lambda^n}\frac{1}{(n+1)^\beta} = \frac{1}{\lambda^{n+1}}\frac{1}{(n+2)^\beta},
\end{align*}
and the equation \eqref{eq:condition_for_eigenmeasures_renewal_potential_log} holds by induction.
\end{proof}

\begin{remark} Observe that the existence of an probability eigenmeasure $\mu_\beta$ as in the statement of the lemma above imposes that the series in RHS of \eqref{eq:eigenmeasure_probability_sum_renewal_potential_log} converges. 
\end{remark}

\begin{lemma}\label{lemma:eigenmeasure_C_alpha_formula} Let $\beta >0$ and $\mu_\beta$ be an eigenmeasure for the Ruelle transformation $L_{\beta F}$, with associated eigenvalue $\lambda > 0$. Then, for $\alpha = \alpha_0 \cdots \alpha_{n-1}$, $n > 1$, positive admissible word, we have
\begin{equation}\label{eq:eigenmeasure_C_alpha_formula}
    \mu_\beta(C_\alpha) = \frac{e^{\beta\sum_{k=0}^{n-2}F(\alpha_k)}}{\lambda^{\alpha_{n-1}+(n-1)}} \frac{1}{(\alpha_{n-1}+1)^\beta},
\end{equation}
where $F(m) := F\vert_{C_{m}} \equiv \log(m) - \log(m+1)$, $m \in \mathbb{N}$.
\end{lemma}

\begin{proof} We recall that, if $\mu_\beta$ is an eigenmeasure as in the statement above, then
\begin{equation}\label{eq:conformality_DU_eigenmeasure_recall}
    \mu_\beta(\sigma(C_\alpha)) = \int_{C_\alpha} \lambda e^{-\beta F} d\mu_\beta,
\end{equation}
for every $\alpha$ positive admissible word s.t. $|\alpha| > 1$.
We prove the lemma by induction. For $n = 2$, the identity above becomes
\begin{equation*}
    \mu_\beta(C_{\alpha_1}) = \mu_\beta(\sigma(C_{\alpha_0 \alpha_1})) = \int_{C_{\alpha_0 \alpha_1}} \lambda e^{-\beta F} d\mu_\beta = \lambda e^{-\beta F(\alpha_0)}\mu_\beta(C_{\alpha_0\alpha_1}),
\end{equation*}
that is,
\begin{equation*}
    \mu_\beta(C_{\alpha_0\alpha_1}) = \frac{e^{\beta F(\alpha_0)}}{\lambda} \mu_\beta(C_{\alpha_1}) = \frac{e^{\beta F(\alpha_0)}}{\lambda^{\alpha_1 + 1}} \frac{1}{(\alpha_1+1)^\beta},
\end{equation*}
where in the last equality we used Lemma \ref{lemma:condition_for_eigenmeasures_renewal_potential_log}. Now, suppose that \eqref{eq:eigenmeasure_C_alpha_formula} holds for some $n > 2$ and let $\alpha = \alpha_0 \cdots \alpha_n$ be a positive admissible word. By \eqref{eq:conformality_DU_eigenmeasure_recall} and the inductive step for the word $\alpha_1 \cdots \alpha_n$, we have
\begin{equation*}
    m_\beta(C_\alpha) = \frac{e^{\beta F(\alpha_0)}}{\lambda} m_\beta(C_{\sigma(\alpha)}) = \frac{e^{\beta\sum_{k=0}^{n-1}F(\alpha_k)}}{\lambda^{\alpha_{n}+n}} \frac{1}{(\alpha_n+1)^\beta}. 
\end{equation*}
\end{proof}

\begin{theorem}[Lenght-type phase transition]\label{thm:complete_characterization_eigenmeasures_renewal}\\ 

For every $\beta > 0$, there exists a unique probability eigenmeasure associated to the eigenvalue $\lambda_{\beta} = e^{P_G(\beta F)}$. Moreover, there is critical value $\beta_c$, which is the (real) solution for $\zeta(\beta_c) = 2$ such that
\begin{itemize}
    \item[$(i)$] if $\beta >\beta_c$, then the probability eigenmeasure lives on $Y_A$;
    \item[$(ii)$] if $\beta \leq \beta_c$, then the probability eigenmeasure lives on $\Sigma_A$.
\end{itemize}
\end{theorem}

\begin{proof} The proof is a summarization of some results we developed and proved in this paper.
\begin{itemize}
    \item[(1)] The potential $F$ has uniformly bounded variations, and the renewal shift is both compact and $s$-compact. Hence, by Theorem \ref{thm:equality_pressures_general}, we conclude that $P(\beta F,x) = P_G(\beta F)$ for every $x \in X_A$. Moreover, since $\sup F< \infty$, we have by direct calculations for the renewal shift space $\Sigma_A$ that $P_G(\beta F) \leq \log 2 + \beta \sup F < \infty$, for every $\beta >0$.
    \item[(2)] Since $X_A$ is compact (see subsection \ref{subsec:Generalized_Renewal_shift}) we have by (1) and Theorem \ref{thm:existence_eigenmeasures_Denker_Yuri} that there exists a probability eigenmeasure $\mu_\beta$ of the Ruelle operator for every $\beta >0$, and the associated eigenvalue is $\lambda_\beta = e^{P_G(\beta F)}$. Therefore, the first claim of the statement of this theorem is proved, and by Lemma \ref{lemma:condition_for_eigenmeasures_renewal_potential_log} we have
    \begin{equation}\label{eq:Denker_eigenmeasure_probability_sum}
        1 = \mu_\beta(\{\xi^0\}) + \sum_{n\in \mathbb{N}} \frac{1}{\lambda_\beta^n}\frac{1}{(1+n)^\beta}.
    \end{equation}
    \item[(3)] Again by Example \ref{exa:renewal_potential_log}, there exists a critical value $\beta_c$, which is precisely the positive solution of $\zeta(\beta_c) = 2$ such that the potential is positive recurrent for $0< \beta < \beta_c$ and transient for $\beta > \beta_c$. We study each case as follows.
    \begin{itemize}
        \item[(3.a)] \textbf{Case $0 < \beta < \beta_c$.} By the generalized RPF Theorem, since the potential is positive recurrent, there exists an eigenmeasure $m_\beta$ living in $\Sigma_A$ and it is necessarily associated to the eigenvalue $e^{P(\beta F)}$. Moreover, the space of these eigenmeasures for fixed $\beta$ has dimension $1$. Accordingly to Theorem 30 of \cite{BelBisEndo2020}, since the potential satisfies $\Var_1 F <\infty$, the eigenmeasure is finite, so consider that $m_\beta$ is normalized, that is, it is a probability. Then, $m_\beta$ is the unique probability eigenmeasure living on $\Sigma_A$. Proposition \ref{prop:extension_eigenmeasures} gives that $m_\beta$ can be seen as the unique probability eigenmeasure on $X_A$ which lives on $\Sigma_A$. Again by Lemma \ref{lemma:condition_for_eigenmeasures_renewal_potential_log}, we obtain
        \begin{equation}\label{eq:Sarig_eigenmeasure_probability_sum}
            1 = m_\beta(\{\xi^0\}) + \sum_{n\in \mathbb{N}} \frac{1}{\lambda_\beta^n}\frac{1}{(1+n)^\beta} = \sum_{n\in \mathbb{N}} \frac{1}{\lambda_\beta^n}\frac{1}{(1+n)^\beta},
        \end{equation}
        and we conclude that
        \begin{equation*}
            \sum_{n\in \mathbb{N}} \frac{1}{\lambda_{\beta}^n}\frac{1}{(1+n)^\beta} = 1.
        \end{equation*}
        In other words, the series above does not depend on the measure, and by \eqref{eq:Denker_eigenmeasure_probability_sum} we conclude that $m_\beta(\xi^0) = 0$, and then $m_\beta(Y_A) = 0$, so $m_\beta$ is an probability eigenmeasure that lives on $\Sigma_A$ and then $m_\beta = \mu_\beta$. Therefore, given $0 < \beta < \beta_c$, $\mu_\beta$ is the unique probability eigenmeasure associated to the eigenvalue $\lambda_{\beta} = e^{P_G(\beta F)}$, and it lives on $\Sigma_A$.
        \item[(3.b)] \textbf{Case $\beta = \beta_c$.} By Example \ref{exa:renewal_potential_log} we have that $P_G(\beta_c F) = 0$, that is, $\lambda_{\beta_c} = 1$, and then
        \begin{equation}\label{eq:Denker_eigenmeasure_probability_sum_critical_beta}
            1 = \mu_{\beta_c}(\{\xi^0\}) + \sum_{n\in \mathbb{N}} \frac{1}{(1+n)^{\beta_c}} = \mu_{\beta_c}(\{\xi^0\}) -1  + \sum_{n\in \mathbb{N}} \frac{1}{n^{\beta_c}} = \mu_{\beta_c}(\{\xi^0\}) -1  + \zeta(\beta_c).
        \end{equation}
        Since $\zeta(\beta_c) = 2$ we conclude that $\mu_{\beta_c}(\{\xi^0\}) = 0$. So $\mu_{\beta_c}$ lives on $\Sigma_A$. Moreover this measure is unique, the Lemma \ref{lemma:eigenmeasure_C_alpha_formula} gives us the value $\mu_{\beta_c}(C_\alpha)$ for every generalized cylinder on positive words with $|\alpha| \geq 2$, namely
        \begin{align*}
            \mu_{\beta_c}([\alpha]) = \mu_{\beta_c}(C_\alpha) = e^{\beta_c \sum_{i=0}^{n-2}F(\alpha_i)}\frac{1}{(\alpha_{n-1}+1)^{\beta_c}},
        \end{align*}
        and then we can extend the measure to the whole space uniquely.
        \item[(3.c)] \textbf{Case $\beta > \beta_c$.} We have that $P_G(\beta F) = 0$ for every $\beta > \beta_c$, i.e., $\lambda_{\beta} = 1$ and then, as in (\ref{eq:Denker_eigenmeasure_probability_sum_critical_beta}), we have $\mu_\beta(\{\xi^0\}) = 2 - \zeta(\beta)$. Since $1<\zeta(\beta) < 2$ for $\beta > \beta_c$, we have necessarily $\mu_\beta(Y_A) > 0$. We claim that $\mu_\beta(\Sigma_A) = 0$. In fact, if $\mu_\beta(\Sigma_A) > 0$, let $\nu_\beta$ be the restriction of $\mu_\beta$ to $\mathcal{B}_{\Sigma_A}$. We observe that $\nu_\beta$ is a non-zero eigenmeasure on $\Sigma_A$ with associated eigenvalue $1$, due to Proposition \ref{thm:restriction_eigenmeasures}, and it is finite because $\mu_\beta$ is a probability. So we may take $\nu_\beta$ normalized in order to be a probability, and by Lemma \ref{lemma:condition_for_eigenmeasures_renewal_potential_log} we have
        \begin{equation*}
            1 = \nu_\beta(\{\xi^0\}) + \sum_{n\in \mathbb{N}} \frac{1}{(1+n)^{\beta}} =  -1  + \zeta(\beta),
        \end{equation*}
        and we obtain $\zeta(\beta) = 2$, a contradiction. Therefore $\mu_\beta$ lives on $Y_A$ and it is unique because $\xi^0$ is the unique element of $X_A$ with empty stem. 
    \end{itemize}
\end{itemize}
\end{proof}

\begin{remark} Since for $\beta > \beta_c$ the probability eigenmeasure $\mu_\beta$ in Theorem \ref{thm:complete_characterization_eigenmeasures_renewal} is a $e^{-\beta F}$-conformal measure, and it is explicitly determined by \eqref{eq:c_omega_new}, that is,  
\begin{equation*}
    \mu_\beta(\omega) \equiv c_\omega^\beta = e^{\beta F_{|\omega|}(\omega)} c_e = e^{\beta F_{|\omega|}(\omega)} (2-\zeta(\beta)), \quad \omega \in \mathfrak{W}.
\end{equation*}
\end{remark}

The next two results are straightforward corollaries of Theorem \ref{thm:complete_characterization_eigenmeasures_renewal} above.

\begin{corollary}\label{cor:KMS_renewal_potential_log}  Consider the C$^*$-dynamical system $(C^*(\mathcal{G}(X_A, \sigma)),\tau)$, where $\tau = \{\tau_t\}_{t \in \mathbb{R}}$ is the one-parameter group of automorphisms given by $\tau_t(f)(\gamma) = e^{-it \beta c_F(\gamma)}f(\gamma)$
for $f\in C_c(\mathcal{G}(X_A,\sigma))$, and (uniquely) extended to $C^*(\mathcal{G}(X_A, \sigma))$. Then, we have the following:
\begin{itemize}
    \item[$(i)$] for $\beta \geq \beta_c$ there exists a unique KMS$_\beta$ state on $C^*(\mathcal{G}(X_A,\sigma))$;
    \item[$(ii)$] for $\beta < \beta_c$ there are not KMS$_\beta$ states on $C^*(\mathcal{G}(X_A,\sigma))$.
\end{itemize}
\end{corollary}

\begin{figure}[H]
 \includegraphics[scale=.3]{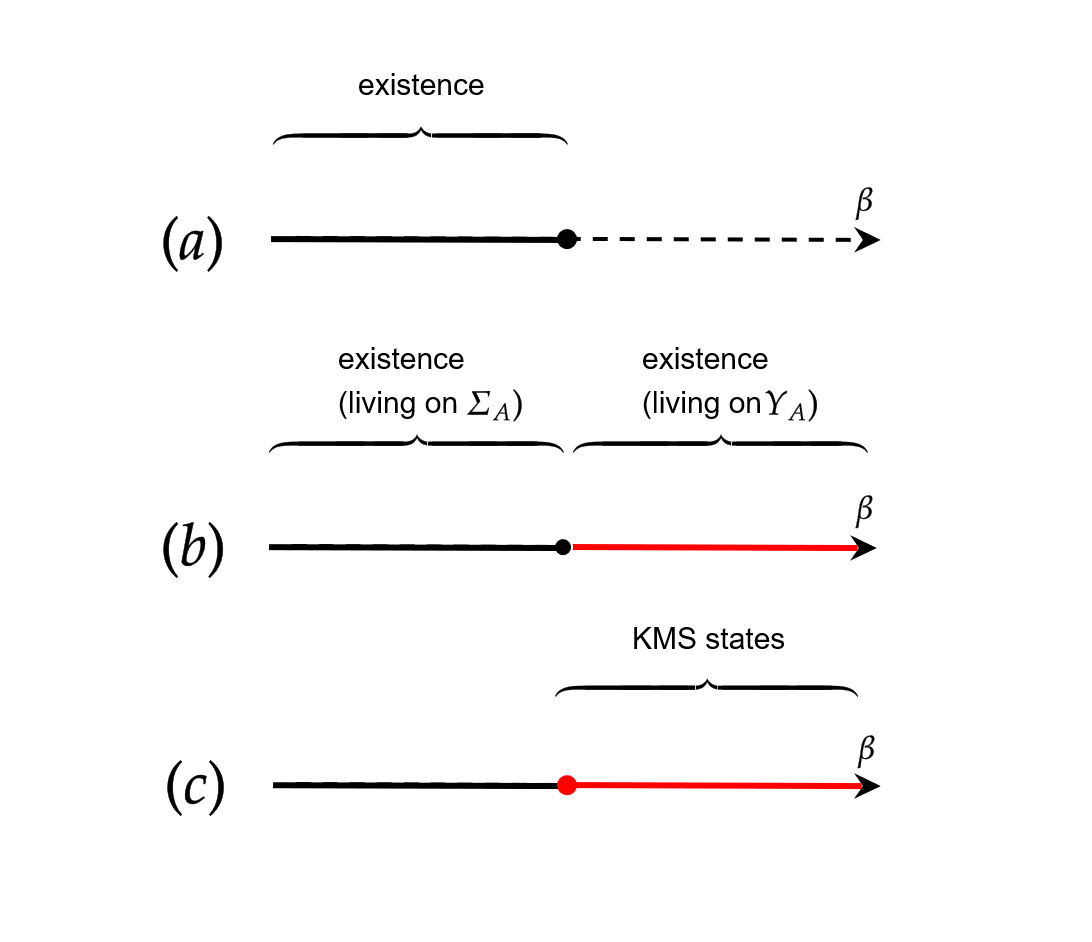}
 
 \caption{The phase transitions on different thermodynamic formalisms for probability eigenmeasures in the renewal shift space with potential as in Theorem \ref{thm:complete_characterization_eigenmeasures_renewal}. In this figure, $\beta_c$ is the unique real positive solution for the equation $\zeta(\beta_c) = 2$. The picture $(a)$ represents the standard formalism on $\Sigma_A$, where we have a unique eigenmeasure for each $\beta \leq \beta_c$ (black line). In this case it is not possible to detect any probability eigenmeasure for $\beta> \beta_c$. The picture $(b)$ represents the generalized formalism on $X_A$ and, unlike in $(a)$, we can see more than the standard probability eigenmeasures; we can detect a unique eigenmeasure for each $\beta > \log 2$ (red line), and this measure vanishes on $\Sigma_A$. Picture $(c)$ shows what are the eigenmeasures that correspond to KMS$_{\beta}$-states.  \label{fig:complete_characterization_eigenmeasures_renewal}}
\end{figure}

\begin{proof} Since $P_G(\beta F) = 0$ for $\beta \geq \beta_c$, by Theorem \ref{thm:complete_characterization_eigenmeasures_renewal} $(i)$ there exists a unique probability eigenmeasure $\mu_\beta$, and in this case the eigenvalue is $1$. Then, by Theorem \ref{thm:equivalences_conformal_measures_generalized_Markov_shift}, $\mu_\beta$ is a $e^{\beta c_F}$-quasi-invariant probability measure, and therefore by Remark \ref{remark:KMS_quasi_invariant} there exists a unique KMS$_\beta$-state, given by
\begin{equation*}
    \varphi_{\mu_\beta}(f)=\int_{X_A} f(x,0,x)d\mu_\beta(x),\quad f\in C_c(\mathcal{G}(X_A,\sigma)).
\end{equation*}
For $\beta < \beta_c$, by Theorem \ref{thm:complete_characterization_eigenmeasures_renewal} $(ii)$, there exists a unique probability eigenmeasure living on $\Sigma_A$, and in this case the eigenvalue is strictly greater than $1$, so there are no $e^{\beta c_F}$-quasi-invariant probability measures, and therefore there are no KMS$_\beta$ states.
\end{proof}

\begin{remark}
    The existence of eigenmeasures for transient potentials already was obtained in the standard thermodynamic formalism for countable Markov shifts by V. Cyr \cite{Cyr2010}, see also \cite{Shwartz2019}, assuming that $\Sigma_A$ is locally compact. In our example, $\Sigma_A$ is not locally compact. For $\beta$ s.t. $\beta F$ is transient, we can prove that the eigenmeasures always do exist (but they live on $Y_A$) due to the compactness of $X_A$ using Denker-Yuri's \cite{DenYu2015} existence theorem. Moreover, the phase transition in this setting is not in the sense of existence or absence of eigenmeasures but means that the eigenmeasure changes the space for which it gives mass, from $\Sigma_A$ to $Y_A$. Figure \ref{fig:complete_characterization_eigenmeasures_renewal} compares the phase transitions between standard and generalized formalisms.
\end{remark}

\begin{theorem}\label{thm:weak_convergence_eigenmeasure_renewal} For each $\beta>0$, let $\mu_{\beta}$ be the unique probability eigenmeasure as in theorem \ref{thm:complete_characterization_eigenmeasures_renewal}. Then, the net $(\mu_\beta)_{\beta > \beta_c}$ converges on the weak$^*$ topology to $\mu_{\beta_c}$ as $\beta$ goes to $\beta_c$.
\end{theorem}

\begin{proof} Observe that $\mu_\beta(\{e\}) = c_e^\beta = 2 - \zeta(\beta) \to 0$ as $\beta \to \beta_c$ by above. Then, $\lim_{\beta \to \beta^c}\mu_\beta(F) \to 0$, by above for every $F \subseteq Y_A$ finite set. For $\beta \geq \beta_c$, the associated eigenvalue is $1$, and by Lemmas \ref{lemma:condition_for_eigenmeasures_renewal_potential_log} and \ref{lemma:eigenmeasure_C_alpha_formula}, we have for every $\alpha$ positive admissible word that $\mu_\beta(C_\alpha)$ is continuous on $\beta$ and then $\mu_\beta(C_\alpha) \to \mu_{\beta_c}(C_\alpha) = \mu_{\beta_c}([\alpha])$, as $\beta$ decreases to $\beta_c$. Then, for a typical element of the basis in the form
\begin{equation*}
    F \sqcup \bigsqcup_{n \in \mathbb{N}}C_{w(n)},
\end{equation*}
where $F \subseteq Y_A$ is finite, and $w(n)$ is a positive admissible word. As in the proof of the Theorem \ref{thm:colapsing_renewal}, we have
\begin{equation*}
     \lim_{\beta \to \beta_c} \mu_\beta\left(F \sqcup \bigsqcup_{n \in \mathbb{N}}C_{w(n)}\right) = \mu_{\beta_c}\left(F \sqcup \bigsqcup_{n \in \mathbb{N}}C_{w(n)}\right). 
\end{equation*}
We conclude the proof by theorem 8.2.17 of \cite{Bogachev2007}. 
\end{proof}

\section{Concluding remarks}

This paper introduced the thermodynamic formalism for generalized countable Markov shifts $X_A = \Sigma_A \cup Y_A$, which are locally compact versions of the standard countable Markov shifts $\Sigma_A$. For the class of compact GCMS, which has the s-compact property, generalized versions of the full shift and the renewal shift are examples in this class; we prove that the Denker-Yuri's definition of pressure (defined originally for IFS \cite{DenYu2015}) coincides with the standard Gurevich pressure defined on $\Sigma_A$. The equality holds when the potential has uniformly bounded distortion, a class of regularity that contains the potentials such that the thermodynamic formalism is well developed for CMS, see \cite{Daon2013, Fiebig_Fiebig_Yuri, Sa4, Sarig2009}. We extended the notion of conformal measure, proving the equivalence among several definitions for general spaces, including quasi-invariant measures on generalized Renault-Deaconu groupoids. Moreover, we proved the existence of new conformal measures which are not detected by the standard theory. A new type of phase transition is introduced, namely, the length-type phase transition. This transition is the change of the set where the conformal measure lives, from $\Sigma_A$ to $Y_A$, as the temperature decreases.
Furthermore, we presented a complete description of the eigenmeasures associated with the Ruelle transformation for a couple of examples. In some cases, we also connected the eigenmeasures living on $Y_A$ with the standard ones living in $\Sigma_A$ by taking weak$^*$ limits on $\beta$. 

It is important to note that the points on $Y_A = X_A \backslash  \Sigma_A$ do not depend on any extra structure, unlike other compactifications. Recently, O. Shwartz constructed Martin boundaries for locally compact countable Markov shifts $\Sigma_A$ in \cite{Shwartz2019}, the compactification depends on the potential. After this, G. Iommi and M. Todd suggested in \cite{Iommi2020} a strategy where the compactification depends on the properties of a metric introduced in the state space $\mathbb{N}$. Different metrics in $\mathbb{N}$ generate different compactifications in their approach. Besides, they have to assume the finiteness of Gurevich entropy and the arguments only work for bounded potentials. In both papers, it is not possible to consider the full shift. We have shown that some potentials on the full shift can be encoded by potentials on the renewal shift and present the same interval of existence/absence of conformal measures living on $Y_A$, including other examples where $X_A$ has infinite  Gurevich entropy.

We presented a case where we have uncountably many extremal conformal measures living on $Y_A$ when the temperature is low enough.

Our results open many possibilities to further research directions. We highlight the following:
\begin{itemize}
    \item[$1.$] Try to use $X_A$ as the configuration space on realistic physical models. The results obtained in this paper give us a good opportunity to interact with our colleagues physicists. A possible path is to verify if $X_A$ may replace the usual configuration space in some model in statistical mechanics. More concretely, as we saw in the results, it is natural to expect to find more conformal measures using $X_A$ instead of $\Sigma_A$, mainly at low temperatures. So it is a setting more suitable to see phase transitions and find new conformal and DLR measures as in \cite{Shwartz2019}, for DLR measures on GCMS, see \cite{BisExelFrauRas2018}.
     \item[$2.$] Recently, it was discovered a volume-type phase transition on CMS, see \cite{BelBisEndo2020}. We only consider probability measures and KMS states. Under suitable hypotheses, results for infinite conformal measures and KMS weights are expected. J. Christensen and K. Thomsen already studied KMS weights on graph C$^{*}$-algebras, see \cite{Christensen2018, Christensen2021, Thomsen2017}.
    \item[$3.$] The notion of ultragraph, firstly defined in \cite{Tomforde2003}, is a generalization from both $\mathcal{O}_A$ and graph C$^{*}$-algebras. So it is natural to investigate the consequences of our results in this more general setting, see also \cite{GoncalvesTasca2020, KatsuraSimsTomforde2009, Tomforde2003_simp}.
    \item[$4.$] For compact and $s$-compact $X_A$ with finite Gurevich entropy $h_G$ we proved the existence of $e^{\beta}$-conformal probability measures living on $Y_A$ for $\beta > h_G$, see Corollary \ref{cor:phase_transition_entropy_gauge}. The results from K. Thomsen in \cite{Thomsen2017} for graph algebras and our examples suggest that for $\beta = h_G$ the $e^{h_G}$-conformal probability measures should live on $\Sigma_A$, but we don't have a proof.
    \item[$5.$] Under suitable hypotheses as in Proposition \ref{prop:phase_transition_entropy}, it should do exist a critical point $\beta_c$ where we have a phase transition in the sense of absence of new conformal measures living on $Y_A$ when $\beta < \beta_c$, and existence of such measures for $\beta > \beta_c$. The unknown region in Figure \ref{fig:conformal_general_entropy} should be artificial.
    \item[$6.$] Does exist a potential $F$ for some $X_A$ that admits both $e^{\beta F}$-conformal measures living on $\Sigma_A$ and also in $Y_A$? For same $\beta > 0$? Is it possible eigenmeasures living on $\Sigma_A$ converge to an eigenmeasure living on $Y_A$?
\end{itemize}

\section*{Acknowledgements}
RB is supported by CNPq grants 312294/2018-2 and 408851/2018-0, by FAPESP grant 16/25053-8, and by the University Center of Excellence \textquotedblleft Dynamics, Mathematical Analysis and Artificial Intelligence", at the Nicolaus Copernicus University; RE is supported by FAPESP Grant 17/26645-9 and CNPq, RF is supported by CNPq, and TR is supported by CAPES, and CNPq, and NCN (National Science Center, Poland), Grant 2019/35/D/ST1/01375. We thank Eric O. Endo for the first pictures, Elmer B\'eltran for his assistance on countable Markov shifts, Manfred Denker for very useful discussions about the reference \cite{DenYu2015}, to Italo Cipriano for helpful suggestions concerning the results about the Gurevich entropy, to Godofredo Iommi for pointing out the question about infinite Gurevich entropy, and Philippe Thieullen for pointing out a problem concerning the generality of the proof of Theorem \ref{thm:extremal_conformal_measures_general} in a previous version of this paper. Special thanks to Henrique Corsini for all the discussions and his assistance in the final formulation of the statements and proofs of the results about entropy and pressures.

\section{Appendix}
\label{ape:proof_giant_theorem_generalized_cylinders}

\subsection{Auxiliary results}
\label{sec:auxiliary_results_cylinder_intersections}

\begin{lemma}\label{lemma:CF_intersections_general} For any two positive admissible words $\alpha$ and $\gamma$ it is true that
\begin{itemize}
\item[$(i)$] $$\displaystyle
        C_\alpha \cap F_\gamma =  F_{\gamma}^{\alpha};$$
     \item[$(ii)$] 
   \begin{equation*}
        F_\alpha \cap F_\gamma = \begin{cases}
                                    F_\alpha, \quad \alpha \in \llbracket \gamma \rrbracket,\\
                                    F_\gamma, \quad \gamma \in \llbracket \alpha \rrbracket,\\
                                    F_{\alpha'}^*, \quad \text{otherwise},
                                 \end{cases}
    \end{equation*}
    where $\alpha'$ is the longest word in $\llbracket \alpha \rrbracket \cap \llbracket \gamma \rrbracket$;
    \item[$(iii)$] for fixed $n \in \mathbb{N}$ with $n \in \{0,1,\dots,|\gamma|-1\}$, let be the natural number $j \neq \gamma_n$. We have
    \begin{equation*}
        C_\alpha \cap C_{\delta^{n}(\gamma)j} = \begin{cases}
                                            C_{\delta^{n}(\gamma)j}, \quad            \alpha \in \llbracket             \gamma \rrbracket                 \text{ and } n \geq               |\alpha|,\\
                                            C_\alpha, \quad                   \alpha \notin                     \llbracket \gamma                 \rrbracket,                        \gamma \notin                     \llbracket \alpha                 \rrbracket, n =                  |\alpha'| \text{ and }           j=\alpha_{|\alpha'|},           \\
                                            \emptyset, \quad                  \text{otherwise};
                                         \end{cases}
    \end{equation*}
    where $\alpha'$ is the longest stem in $\llbracket \alpha \rrbracket \cap \llbracket \gamma \rrbracket$.
    \item[$(iv)$] for $j \in \mathbb{N}$,
    \begin{equation*} 
        C_\alpha \cap G(\gamma,j) = \begin{cases}
                                        G(\gamma,j), \quad \alpha \in \llbracket \gamma \rrbracket,\\
                                        \emptyset, \quad \text{otherwise};
                                    \end{cases}
    \end{equation*}
    \item[$(v)$] for $j \in \mathbb{N}$,
    \begin{equation*} 
        C_\alpha \cap K(\gamma,j) = \begin{cases}
                                        K(\gamma,j), \quad \alpha \in \llbracket \gamma \rrbracket,\\
                                        \emptyset, \quad \text{otherwise};
                                    \end{cases}
    \end{equation*}
    \item[$(vi)$] for $H,I \subseteq \mathbb{N}$,
    \begin{equation*} 
        G(\alpha,H) \cap G(\gamma,I) = \begin{cases}
                                        G(\alpha,H\cup I), \quad \alpha = \gamma,\\
                                        \emptyset, \quad \text{otherwise};
                                    \end{cases}
    \end{equation*}
    \item[$(vii)$] for $H,I \subseteq \mathbb{N}$,
    \begin{equation*} 
        K(\alpha,H) \cap K(\gamma,I) = \begin{cases}
                                        K(\alpha,H\cup I), \quad \alpha = \gamma,\\
                                        \emptyset, \quad \text{otherwise};
                                    \end{cases}
    \end{equation*}
    \item[$(viii)$] 
    \begin{equation*}
        F_\alpha \cap G(\gamma,j) = \begin{cases}
                                        G(\gamma,j), \quad \gamma \in \llbracket \alpha \rrbracket \setminus \{\alpha\},\\
                                        \emptyset, \quad \text{otherwise};
                                    \end{cases}
    \end{equation*}
    \item[$(ix)$] 
    \begin{equation*}
        F_\alpha \cap K(\gamma,j) = \begin{cases}
                                        K(\gamma,j), \quad \gamma \in \llbracket \alpha \rrbracket \setminus \{\alpha\},\\
                                        \emptyset, \quad \text{otherwise};
                                    \end{cases}
    \end{equation*}
\end{itemize}
\end{lemma}

\begin{proof} For $(i)$, the proof is straightforward. Now, for the proof of $(ii)$, we notice from the definition of $F_\alpha$ that 
\begin{equation*}
    F_\alpha\cap F_\gamma = \left\{\xi \in Y_A: \kappa(\xi) \in (\llbracket\alpha\rrbracket \cap \llbracket\gamma\rrbracket)\setminus\{\alpha,\gamma\}\right\}.
\end{equation*}
If $\alpha \in \llbracket\gamma\rrbracket$ then $\llbracket\alpha\rrbracket \subseteq \llbracket \gamma \rrbracket$ and therefore $(\llbracket\alpha\rrbracket \cap \llbracket\gamma\rrbracket)\setminus\{\alpha,\gamma\}= \llbracket\alpha\rrbracket\setminus\{\alpha\}$. In this case, we obtain $F_\alpha\cap F_\gamma = F_\alpha$. A similar proof holds for the case that $\gamma \in \llbracket \alpha \rrbracket$. The remaining possibility occurs when $\alpha \notin \llbracket\gamma\rrbracket$ is simultaneous to $\gamma \notin \llbracket \alpha \rrbracket$, and here we consider $\alpha'$ as in the statement. It follows immediately that $(\llbracket\alpha\rrbracket \cap \llbracket\gamma\rrbracket)\setminus\{\alpha,\gamma\}= \llbracket \alpha'\rrbracket$, since we only must have $|\alpha'|<\min\{|\alpha|,|\gamma|\}$. We conclude that $\xi  \in F_\alpha\cap F_\gamma$ if and only if $\kappa(\xi) \in \llbracket\alpha'\rrbracket$, and therefore $F_\alpha\cap F_\gamma = F_{\alpha'}^*$.

To prove $(iii)$ we use \eqref{eq:C_cap_C} and obtain
\begin{equation*}
    C_\alpha \cap C_{\delta^{n}(\gamma)j} = \begin{cases}
                                C_\alpha, \quad \text{if } \delta^{n}(\gamma)j \in \llbracket \alpha \rrbracket,\\
                                C_{\delta^{n}(\gamma)j}, \quad \text{if } \alpha \in \llbracket \delta^{n}(\gamma)j \rrbracket,\\
                                \emptyset, \quad \text{otherwise}.
                            \end{cases} 
\end{equation*}
First, we suppose that $\alpha \in \llbracket \gamma \rrbracket$. In this case, it never occurs that $\delta^{n}(\gamma) j \in \llbracket \alpha \rrbracket$, for any $n$ and $j \neq \gamma_n$. Indeed, if $\delta^{n}(\gamma) j \in \llbracket \alpha \rrbracket$ then $\delta^{n}(\gamma) j = \alpha_0 \cdots \alpha_n$ and $n \leq |\alpha|-1$, and consequently $j = \alpha_n = \gamma_n$, a contradiction. Then, if $n < |\alpha|$, we have $\alpha \notin \llbracket \delta^{n}(\gamma) j \rrbracket$, and therefore $C_\alpha \cap C_{\delta^{n}(\gamma)j} = \emptyset$. On the other hand, if $n \geq |\alpha|$, then $\alpha \in \llbracket \delta^{n}(\gamma) j \rrbracket$, and hence $C_\alpha \supseteq C_{\delta^{n}(\gamma)j}$, that is, $C_\alpha \cap C_{\delta^{n}(\gamma)j} = C_{\delta^{n}(\gamma)j}$. Now, let us consider $\alpha \notin \llbracket \gamma \rrbracket$, which implies that $|\alpha|>0$ because, otherwise, we would get $\alpha = e$, which is subword of every positive admissible word. If $\gamma \in \llbracket \alpha \rrbracket$, observe that $\delta^{n}(\gamma)j \notin \llbracket \alpha \rrbracket$ because $j \neq \gamma_n = \alpha_n$, and hence $C_{\delta^{n}(\gamma)j} \subseteq C_\alpha^c$, that is, $C_{\delta^{n}(\gamma)j} \cap C_\alpha = \emptyset$. For the case that $\gamma \notin \llbracket \alpha \rrbracket$, let $\alpha'$ be the longest word of $\llbracket \alpha \rrbracket \cap \llbracket \gamma \rrbracket$. If $n = |\alpha'|$ and $j=\alpha_{|\alpha'|}$, then it follows directly that $C_\alpha \subseteq C_{\delta^{n}(\gamma)j}$, and hence $C_\alpha \cap C_{\delta^{n}(\gamma)j}= C_\alpha$. If $n = |\alpha'|$ and $j\neq\alpha_{|\alpha'|}$ then it is clear that $C_\alpha \cap C_{\delta^{n}(\gamma)j}= \emptyset$. Now, if $n < |\alpha'|$, then also occurs $C_\alpha \cap C_{\delta^{n}(\gamma)j}= \emptyset$ because $j \neq \gamma_n = \alpha'_n = \alpha_n$. The same happens for $n > |\alpha'|$ since $\alpha \notin \llbracket \gamma' \rrbracket$ for any $\gamma' \in \llbracket\gamma\rrbracket$, and in particular for $|\gamma'|>|\alpha|$.

The proofs of $(iv)-(ix)$ are straightforward. 
\end{proof}

\begin{lemma}\label{lemma:C_delta_cap_C_gamma_p} Let $\alpha$. $\gamma$ be finite admissible words such that $\gamma \in \llbracket \alpha \rrbracket \setminus \{\alpha\}$ and consider $k,p \in \mathbb{N}$, $m \in \{0,1,..., |\alpha|-1\}$, such that $k \neq \alpha_m$. Then,
\begin{equation}\label{eq:C_delta_cap_C_gamma_p_gamma_subword}
    C_{\delta^m(\alpha)k} \cap C_{\gamma p} = \begin{cases}
               C_{\delta^m(\alpha)k}, \quad \text{if } p = \alpha_{|\gamma|}, \text{ and } m \geq |\gamma|+1;\\
               C_{\gamma p}, \quad  \text{if } m= |\gamma|, k=p;\\
               \emptyset, \quad \text{otherwise}.
    \end{cases}
\end{equation}
\end{lemma}

\begin{proof} Lemma \ref{lemma:CF_intersections_general} $(iii)$ gives 
\begin{equation}\label{eq:C_delta_cap_C_gamma_p}
    C_{\delta^m(\alpha)k} \cap C_{\gamma p} = \begin{cases}
               C_{\delta^m(\alpha)k}, \quad \gamma p \in \llbracket \alpha \rrbracket \text{ and } m \geq |\gamma|+1,\\
               C_{\gamma p}, \quad \gamma p \notin \llbracket \alpha \rrbracket, \alpha \notin \llbracket \gamma p \rrbracket, m= |\gamma|, k=(\gamma p)_{|\gamma|},\\
               \emptyset, \quad \text{otherwise};
    \end{cases}
\end{equation}
Since $\gamma \in \llbracket \alpha \rrbracket \setminus \{\alpha\}$, we have that $\gamma p \in \llbracket \alpha \rrbracket$ if and only if $p = \alpha_{|\gamma|}$. Note that $\gamma p \notin \llbracket \alpha \rrbracket$ implies $\alpha \notin \llbracket \gamma p \rrbracket$. Also, $(\gamma p)_{|\gamma|} = p$, and when $k=p$ we necessarily have $p \neq \alpha_{|\gamma|}$. The equality \eqref{eq:C_delta_cap_C_gamma_p} becomes the identity \eqref{eq:C_delta_cap_C_gamma_p_gamma_subword}.
\end{proof}

\begin{lemma} \label{lemma:C_delta_cap_C_delta} Let $\alpha$ and $\gamma$ be finite admissible words such that $\gamma \notin \llbracket \alpha \rrbracket$ and $\alpha \notin \llbracket \gamma \rrbracket$, and $\alpha'$ the longest word in $\llbracket \alpha \rrbracket \cap \llbracket \gamma \rrbracket$. Also, consider $k,p \in \mathbb{N}$, $m \in \{0,1,\dots, |\alpha|-1\}$ and $n \in \{0,1,\dots, |\gamma|-1\}$, such that $k \neq \alpha_m$ and $p \neq \gamma_n$. Then,
\begin{equation}
    C_{\delta^m(\alpha)k} \cap C_{\delta^n(\gamma) p} = \begin{cases}
               C_{\delta^m(\alpha)k} = C_{\delta^n(\gamma)p}, \quad \text{if } n=m \leq |\alpha'| \text{ and } p=k,\\
               C_{\delta^m(\alpha)k}, \quad \text{if } n = |\alpha'| < m \text{ and } p=\alpha_n;\\
               C_{\delta^n(\gamma)p}, \quad \text{if } m = |\alpha'| < n \text{ and } k=\gamma_m;\\
               \emptyset, \quad \text{otherwise}.
    \end{cases}
\end{equation}
\end{lemma}

\begin{proof}
The identity \eqref{eq:C_cap_C} gives
\begin{equation*}
    C_{\delta^m(\alpha)k} \cap C_{\delta^n(\gamma) p} = \begin{cases}
        C_{\delta^m(\alpha)k}, \quad \text{if } \delta^n(\gamma) p \in \llbracket \delta^m(\alpha)k \rrbracket,\\
         C_{\delta^n(\gamma) p}, \quad \text{if } \delta^m(\alpha)k \in \llbracket \delta^n(\gamma) p \rrbracket,\\
        \emptyset, \quad \text{otherwise}.
    \end{cases}  
\end{equation*}
The case $\delta^n(\gamma) p \in \llbracket \delta^m(\alpha)k \rrbracket$, which implies that $n \leq m$, occurs if and only if $\gamma_0 \cdots \gamma_{n-1}p = \alpha_0 \cdots \alpha_{n-1}l$, where
\begin{equation*}
    l = \begin{cases}
        k, \quad n=m;\\
        \alpha_n, \quad n<m.
    \end{cases}
\end{equation*}
Observe that, by the definition of $\alpha'$, we have $\gamma_i = \alpha_i$ for all $i = 0,..., n-1$ if and only if $n \leq |\alpha'|$. If $n = m$, then it is straightforward that $\delta^n(\gamma) p \in \llbracket \delta^m(\alpha)k \rrbracket$ if and only if $p=k$. Now, suppose that $n<m\leq |\alpha'|$, then $\delta^n(\gamma) p \in \llbracket \delta^m(\alpha)k \rrbracket$ if and only if $p = \alpha_n$, which never happens since $p \neq \gamma_n = \alpha_n$ for this case. If $n < |\alpha'|<m$, we have a similar problem as in the previous situation and hence this never happens. Finally, if $n=|\alpha'|<m$ we have that $\delta^n(\gamma) p \in \llbracket \delta^m(\alpha)k \rrbracket$ if and only if $p = \alpha_n \neq \gamma_n$. We conclude that $\delta^n(\gamma) p \in \llbracket \delta^m(\alpha)k \rrbracket$ if and only if $n = m$ and $p=k$, or $n=|\alpha'|<m$ and $p = \alpha_n$. It is analogous to prove that $\delta^m(\alpha) k \in \llbracket \delta^n(\gamma) p \rrbracket$ if and only if $n = m$ and $p=k$, or $m=|\alpha'|<n$ and $k = \gamma_m$.
\end{proof}

\begin{corollary}\label{cor:4_tuple_union_delta} Let $\alpha$ and $\gamma$ be positive admissible words. If $\alpha \in \llbracket \gamma \rrbracket$, then
\begin{align*}
    \bigsqcup_{m=0}^{|\alpha|-1}\bigsqcup_{k\neq \alpha_m} \bigsqcup_{n=0}^{|\gamma|-1}\bigsqcup_{p\neq \gamma_n}  C_{\delta^m(\alpha)k} \cap C_{\delta^n(\gamma )p} &= \bigsqcup_{m = 0}^{|\alpha| - 1} \bigsqcup_{k\neq \alpha_m} C_{\delta^m(\alpha)k}.
\end{align*}
If $\alpha \notin \llbracket \gamma \rrbracket$ and $\gamma \notin \llbracket \alpha \rrbracket$, then
\begin{align*}
    \bigsqcup_{m=0}^{|\alpha|-1}\bigsqcup_{k\neq \alpha_m} \bigsqcup_{n=0}^{|\gamma|-1}\bigsqcup_{p\neq \gamma_n}  C_{\delta^m(\alpha)k} \cap C_{\delta^n(\gamma )p} &= \left(\bigsqcup_{m = 0}^{|\alpha|-1} \bigsqcup_{k\neq \alpha_m} C_{\delta^m(\alpha)k}\right) \sqcup \left(\bigsqcup_{|\alpha'| < n \leq |\gamma|-1} \bigsqcup_{p\neq \gamma_n} C_{\delta^n(\gamma )p}\right),
\end{align*}
where $\alpha'$ is the longest word in $\llbracket \alpha \rrbracket \cap \llbracket \gamma \rrbracket$.
\end{corollary}

\begin{proof} If $\alpha \in \llbracket \gamma \rrbracket$, then we may write
\begin{align*}
    \bigsqcup_{m=0}^{|\alpha|-1}\bigsqcup_{k\neq \alpha_m} \bigsqcup_{n=0}^{|\gamma|-1}\bigsqcup_{p\neq \gamma_n}  C_{\delta^m(\alpha)k} \cap C_{\delta^n(\gamma )p} &= \left(\bigsqcup_{m = 0}^{|\alpha| - 1} \bigsqcup_{0 \leq n \leq |\alpha|-1} \bigsqcup_{k\neq \alpha_m} \bigsqcup_{p\neq \gamma_n}  C_{\delta^m(\alpha)k} \cap C_{\delta^n(\gamma )p}\right)  \\
    &\sqcup  \left(\bigsqcup_{m = 0}^{|\alpha| - 1} \bigsqcup_{|\alpha| - 1 < n \leq |\gamma|-1} \bigsqcup_{k\neq \alpha_m} \bigsqcup_{p\neq \gamma_n}  C_{\delta^m(\alpha)k} \cap C_{\delta^n(\gamma )p}\right),
\end{align*}
and it is straightforward that 
\begin{equation*}
    \bigsqcup_{m = 0}^{|\alpha| - 1} \bigsqcup_{0 \leq n \leq |\alpha|-1} \bigsqcup_{k\neq \alpha_m} \bigsqcup_{p\neq \gamma_n}  C_{\delta^m(\alpha)k} \cap C_{\delta^n(\gamma )p} = \bigsqcup_{m = 0}^{|\alpha| - 1} \bigsqcup_{k\neq \alpha_m} C_{\delta^m(\alpha)k},
\end{equation*}
because $\delta^n(\gamma) = \delta^n(\alpha)$ and $\alpha_n = \gamma_n$ for $0 \leq n \leq |\alpha|-1$. Also,
\begin{equation*}
    \bigsqcup_{m = 0}^{|\alpha| - 1} \bigsqcup_{|\alpha| - 1 < n \leq |\gamma|-1} \bigsqcup_{k\neq \alpha_m} \bigsqcup_{p\neq \gamma_n}  C_{\delta^m(\alpha)k} \cap C_{\delta^n(\gamma )p} = \emptyset,
\end{equation*}
because if $n > |\alpha|-1$ then $\delta^n(\gamma) = \alpha \gamma'$ for some admissible positive word $\gamma'$ such that $\alpha \gamma'$ is also admissible, while we have that $\delta^m(\alpha)k$ with $k \neq \alpha_m$ in the union above.

Suppose now that $\alpha \notin \llbracket \gamma \rrbracket$ and $\gamma \notin \llbracket \alpha \rrbracket$. Then, we may write
\begin{align*}
    \bigsqcup_{m=0}^{|\alpha|-1}\bigsqcup_{k\neq \alpha_m} \bigsqcup_{n=0}^{|\gamma|-1}\bigsqcup_{p\neq \gamma_n}  C_{\delta^m(\alpha)k} \cap C_{\delta^n(\gamma )p} &= \left(\bigsqcup_{0 \leq m \leq |\alpha'|} \bigsqcup_{0 \leq n \leq |\alpha'|} \bigsqcup_{k\neq \alpha_m} \bigsqcup_{p\neq \gamma_n}  C_{\delta^m(\alpha)k} \cap C_{\delta^n(\gamma )p}\right)  \\
    &\sqcup   \left(\bigsqcup_{0 \leq m \leq |\alpha'|} \bigsqcup_{|\alpha'| < n \leq |\gamma|-1} \bigsqcup_{k\neq \alpha_m} \bigsqcup_{p\neq \gamma_n}  C_{\delta^m(\alpha)k} \cap C_{\delta^n(\gamma )p}\right)  \\
    &\sqcup  \left(\bigsqcup_{|\alpha'| < m \leq |\alpha|-1} \bigsqcup_{0 \leq n \leq |\alpha'|} \bigsqcup_{k\neq \alpha_m} \bigsqcup_{p\neq \gamma_n}  C_{\delta^m(\alpha)k} \cap C_{\delta^n(\gamma )p}\right) \\
    &\sqcup \left(\bigsqcup_{|\alpha'| < m  \leq |\alpha|-1} \bigsqcup_{|\alpha'| < n \leq |\gamma|-1} \bigsqcup_{k\neq \alpha_m} \bigsqcup_{p\neq \gamma_n}  C_{\delta^m(\alpha)k} \cap C_{\delta^n(\gamma )p}\right). 
\end{align*}
The proof is concluded by applying Lemma \ref{lemma:C_delta_cap_C_delta} on each of the four parcels above. 
\end{proof}

\subsection{Proof of the Theorem \ref{thm:huge_generators_intersections}}

The proofs of the identities of the theorem \ref{thm:huge_generators_intersections} have a common strategy, which we describe now:
\begin{itemize}
    \item [$(1)$] for the elements in the form $C_{\alpha j^{-1}}$, $C_\alpha^c$, and $C_{\alpha j^{-1}}^c$, we use Propositions \ref{prop:general_C_alpha_inverse_j}, \ref{prop:general_C_alpha_complement}, and \ref{prop:general_C_alpha_j_inverse_complement}, respectively, to explicit those sets as a union of positive cylinders jointly with  subsets of $Y_A$;
    \item[$(2)$] after $(1)$, we distribute the intersections among the sets;
    \item[$(3)$] we analyze each case for each identity in theorem \ref{thm:huge_generators_intersections}, and when necessary, we use the auxiliary results in subsection \ref{sec:auxiliary_results_cylinder_intersections} to obtain the statement results.
\end{itemize}

Based on the aforementioned strategy, we present the proof of some identities of theorem \ref{thm:huge_generators_intersections}. 

For a complete description of each equality, see appendix A of \cite{Raszeja2020}.

\textbf{Proof of \eqref{eq:C_cap_C_comp}:} let $\alpha \in \llbracket \gamma \rrbracket$, then by Proposition \ref{prop:general_C_alpha_complement} we have that
\begin{equation*}
    C_\alpha \cap C_\gamma^c = (C_\alpha \cap F_\gamma) \sqcup \left(\bigsqcup_{n=0}^{|\gamma|-1}\bigsqcup_{p\neq \gamma_n}C_\alpha \cap C_{\delta^{n}(\gamma)p}\right).
\end{equation*}
By Lemma \ref{lemma:CF_intersections_general} $(i)$ and $(iii)$, it is clear that \eqref{eq:C_cap_C_comp} holds for this case. If $\alpha \notin \llbracket \gamma \rrbracket \text{ and } \gamma \notin \llbracket \alpha \rrbracket$, then $C_\alpha \cap C_\gamma = \emptyset$ and hence $C_\alpha \cap C_\gamma^c = C_\alpha$. Now, if $\gamma \in \llbracket \alpha \rrbracket$, then for every $\xi \in C_\alpha$ we have that $\xi_\gamma = 1$ and therefore $\xi \notin C_\gamma^c$, implying that $C_\alpha \cap C_\gamma^c= \emptyset$.

\textbf{Proof of \eqref{eq:C_cap_C_inverse_comp}:} Proposition \ref{prop:general_C_alpha_j_inverse_complement} gives
\begin{align*}
    C_\alpha \cap  C_{\gamma j^{-1}}^c &= (C_\alpha \cap K(\gamma,j)) \sqcup (C_\alpha \cap F_\gamma) \sqcup \left( \bigsqcup_{n=0}^{|\gamma|-1}\bigsqcup_{p\neq \gamma_n}  C_\alpha \cap C_{\delta^n(\gamma)p}\right)\sqcup \bigsqcup_{\substack{p: A(j,p)=0}} C_\alpha \cap C_{\gamma p}.
\end{align*}
By Lemma \ref{lemma:CF_intersections_general} $(iii)$ and $(v)$, if $\alpha \in \llbracket \gamma \rrbracket$ then
\begin{align*}
    C_\alpha \cap  C_{\gamma j^{-1}}^c &= K(\gamma,j) \sqcup F_\gamma^\alpha \sqcup \left( \bigsqcup_{n=|\alpha|}^{|\gamma|-1}\bigsqcup_{p\neq \gamma_m} C_{\delta^n(\gamma)p}\right)\sqcup \bigsqcup_{\substack{p: A(j,p)=0}} C_{\gamma p}.
\end{align*}
On other hand, if $\gamma \in \llbracket \alpha \rrbracket \setminus \{\alpha\}$, the same lemma gives us that 
\begin{align*}
    C_\alpha \cap  C_{\gamma j^{-1}}^c &=  \bigsqcup_{\substack{p: A(j,p)=0}}C_\alpha \cap C_{\gamma p} = \begin{cases}
                        C_\alpha, \quad A(j,\alpha_{|\gamma|})=0,\\
                        \emptyset, \quad \text{otherwise}.
                   \end{cases}
\end{align*}
Now, if $\alpha \notin \llbracket \gamma \rrbracket$ and $\gamma \notin \llbracket \alpha \rrbracket$, it is straightforward that $C_\alpha \cap  C_{\gamma j^{-1}}^c= C_\alpha$.

\textbf{Proof of \eqref{eq:C_comp_cap_C_inverse}:} when $\alpha \in \llbracket \gamma \rrbracket$, we have that $C_\alpha \supseteq C_\gamma \supseteq C_{\gamma j^{-1}}$, and therefore $C_\alpha^c \cap C_{\gamma j^{-1}} = \emptyset$. Now, if $\alpha \notin \llbracket \gamma \rrbracket$ and $\gamma \notin \llbracket \alpha \rrbracket$ we necessarily have that $C_\alpha \cap C_\gamma = \emptyset$, then $C_{\gamma j^{-1}} \subseteq C_\gamma \subseteq C_\alpha^c$, and therefore $C_{\gamma j^{-1}} \cap C_\alpha^c = C_{\gamma j^{-1}}$. For the remaining case, namely $\gamma \in \llbracket \alpha \rrbracket \setminus \{\alpha\}$, we use the identity
\begin{align*}
    C_\alpha^c \cap C_{\gamma j^{-1}} &= (F_\alpha \cap G(\gamma,j))\sqcup\left(\bigsqcup_{p:A(j,p)=1}F_\alpha \cap C_{\gamma p} \right) \sqcup \left( \bigsqcup_{m=0}^{|\alpha|-1} \bigsqcup_{q \neq \alpha_m} C_{\delta^m(\alpha)q} \cap G(\gamma,j) \right)\\ &\sqcup \left( \bigsqcup_{m=0}^{|\alpha|-1} \bigsqcup_{q \neq \alpha_m} \bigsqcup_{p:A(j,p)=1} C_{\delta^m(\alpha)q} \cap C_{\gamma p} \right),
\end{align*}
which is a direct consequence from Propositions \ref{prop:general_C_alpha_complement} and \ref{prop:general_C_alpha_inverse_j}. Lemma \ref{lemma:CF_intersections_general} $(i)$ gives $F_\alpha \cap C_{\gamma p} = F_\alpha^{\gamma p}$, which is empty for $p \neq \alpha_{|\gamma|}$. It is straightforward from Lemma \ref{lemma:CF_intersections_general} $(iv)$ and $(viii)$ that $F_\alpha \cap G(\gamma,j) = G(\gamma,j)$, and that $C_{\delta^m(\alpha)q}  \cap G(\gamma,j) = \emptyset$ for every $m$ and $q$ considered. The equality
\begin{equation}\label{eq:C_german_equals_triple_union}
     \mathfrak{C}[\alpha,\gamma,j] = \bigsqcup_{m=0}^{|\alpha|-1} \bigsqcup_{q \neq \alpha_m} \bigsqcup_{p:A(j,p)=1} C_{\delta^m(\alpha)q} \cap C_{\gamma p}.
\end{equation}
is a straightforward consequence of Lemma \ref{lemma:C_delta_cap_C_gamma_p}.

\textbf{Proof of \eqref{eq:C_comp_cap_C_inverse_comp}:} if $\alpha \in \llbracket \gamma \rrbracket$, then $C_\alpha \supseteq C_\gamma \supseteq C_{\gamma j^{-1}}$, and hence $C_\alpha^c \subseteq C_\gamma^c \subseteq C_{\gamma j^{-1}}^c$, and then $C_\alpha^c \cap C_{\gamma j^{-1}}^c = C_\alpha^c$. Now, suppose that $\gamma \in \llbracket \alpha \rrbracket \setminus \{\alpha\}$. For any $\xi \in C_{\gamma j^{-1}}^c$, we have only two possibilities, namely $\xi_\gamma = 0$ or $\xi_\gamma = 1$ with $\xi_{\gamma j^{-1}} = 0$. For the first possibility, we have by Proposition \ref{prop:general_C_alpha_j_inverse_complement} that $\xi \in F_\gamma \sqcup \bigsqcup_{n=0}^{|\gamma|-1}\bigsqcup_{p\neq \gamma_n}  C_{\delta^n(\gamma )p} \subseteq C_\alpha^c$ because $\gamma \in \llbracket \alpha \rrbracket \setminus \{\alpha\}$. For the second one, we have necessarily that $\xi \in K(\gamma,j)\sqcup \bigsqcup_{p:A(j,p) = 0}C_{\gamma p}$. However, for $\gamma \in \llbracket \alpha \rrbracket \setminus \{\alpha\}$, it is straightforward that $K(\gamma,j) \subseteq F_\alpha$ and hence $C_\alpha^c \cap K(\gamma,j) = K(\gamma,j)$. As it is shown in the proof of \eqref{eq:C_comp_cap_C_inverse}, $\gamma \in \llbracket \alpha \rrbracket \setminus \{\alpha\}$ implies that $C_{\gamma p} \cap F_\alpha = F_\alpha^{\gamma p}$, which is empty if $p \neq \alpha_{|\gamma|}$. As a consequence of Lemma \ref{lemma:C_delta_cap_C_gamma_p} we have that
\begin{equation*}
    \mathfrak{D}[\alpha,\gamma,j] = C_\alpha^c \cap \bigsqcup_{p:A(j,p)=0}  C_{\gamma p} = \bigsqcup_{m=0}^{|\alpha|-1} \bigsqcup_{q \neq \alpha_m} \bigsqcup_{p:A(j,p)=0} C_{\delta^m(\alpha)q} \cap C_{\gamma p}.
\end{equation*}
For the remaining case, $\alpha \notin \llbracket \gamma \rrbracket$ and $\gamma \notin \llbracket \alpha \rrbracket$, we use the identity
\begin{align*}
    C_\alpha^c \cap C_{\gamma j^{-1}}^c &= (F_\alpha \cap K(\gamma,j)) \sqcup (F_\alpha \cap F_\gamma) \sqcup \left(\bigsqcup_{n=0}^{|\gamma|-1}\bigsqcup_{p\neq \gamma_n}F_\alpha \cap  C_{\delta^n(\gamma )p}\right) \sqcup \left( \bigsqcup_{\substack{p: A(j,p)=0}}  F_\alpha \cap C_{\gamma p} \right)\\
    &\sqcup \left( \bigsqcup_{m=0}^{|\alpha|-1}\bigsqcup_{k\neq \alpha_m} \bigsqcup_{n=0}^{|\gamma|-1}\bigsqcup_{p\neq \gamma_n}  C_{\delta^m(\alpha)k} \cap C_{\delta^n(\gamma )p}\right) \sqcup \left( \bigsqcup_{m=0}^{|\alpha|-1}\bigsqcup_{k\neq \alpha_m} C_{\delta^m(\alpha)k} \cap K(\gamma,j)\right)  \\
    &\sqcup \left( \bigsqcup_{m=0}^{|\alpha|-1}\bigsqcup_{k\neq \alpha_m} C_{\delta^m(\alpha)k} \cap F_\gamma\right) \sqcup \left(\bigsqcup_{m=0}^{|\alpha|-1}\bigsqcup_{k\neq \alpha_m} \bigsqcup_{\substack{p: A(j,p)=0}} C_{\delta^m(\alpha)k} \cap C_{\gamma p} \right).
\end{align*}
We have $F_\alpha \cap K(\gamma,j) = \emptyset$, since $\alpha \notin \llbracket \gamma \rrbracket$ and $\gamma \notin \llbracket \alpha \rrbracket$ implies $\gamma \notin \llbracket \alpha \rrbracket \setminus \{\alpha\}$. Also, by Lemma \ref{lemma:CF_intersections_general} (ii), $F_\alpha \cap F_\gamma = F_{\alpha'}^*$, where $\alpha'$ be the longest word in $\llbracket \alpha \rrbracket \cap \llbracket \gamma \rrbracket$. By the first item of the same Lemma, we get
\begin{equation*}
    \bigsqcup_{n=0}^{|\gamma|-1}\bigsqcup_{p\neq \gamma_n}F_\alpha \cap  C_{\delta^n(\gamma )p} = \bigsqcup_{n=0}^{|\gamma|-1}\bigsqcup_{p\neq \gamma_n} F_\alpha^{\delta^n(\gamma )p} = F_{\alpha}^{\alpha'\alpha_{|\alpha'|}},
\end{equation*}
where the last equality holds because $\delta^n(\gamma) p \in \llbracket \alpha \rrbracket$ if and only if $n = |\alpha'|$ and $p = \alpha_{|\alpha'|} \neq \gamma_{|\alpha'|}$. Similarly, we obtain
\begin{equation*}
    \bigsqcup_{\substack{p: A(j,p)=0}}  F_\alpha \cap C_{\gamma p} = \bigsqcup_{\substack{p: A(j,p)=0}}  F_\alpha^{\gamma p} = \emptyset,
\end{equation*}
where the last equality holds because each $F_\alpha^{\gamma p}$ is empty. Indeed, suppose that there exists a configuration $\xi \in F_\alpha^{\gamma p}$, that is, $\kappa(\xi) \in \llbracket \alpha \rrbracket \setminus \{\alpha\}$ and $\gamma p \in \llbracket \kappa(\xi) \rrbracket$. Then, $\gamma \in \llbracket \kappa(\xi) \rrbracket \subseteq \llbracket \alpha \rrbracket$, a contradition since $\gamma \notin \llbracket \alpha \rrbracket$. Corollary \ref{cor:4_tuple_union_delta}, gives
\begin{align*}
    \bigsqcup_{m=0}^{|\alpha|-1}\bigsqcup_{k\neq \alpha_m} \bigsqcup_{n=0}^{|\gamma|-1}\bigsqcup_{p\neq \gamma_n}  C_{\delta^m(\alpha)k} \cap C_{\delta^n(\gamma )p} &= \left(\bigsqcup_{m = 0}^{|\alpha|-1} \bigsqcup_{k\neq \alpha_m} C_{\delta^m(\alpha)k}\right) \sqcup \left(\bigsqcup_{|\alpha'| < n \leq |\gamma|-1} \bigsqcup_{p\neq \gamma_n} C_{\delta^n(\gamma )p}\right).
\end{align*}
Lemma \ref{lemma:CF_intersections_general} $(v)$ gives
\begin{equation*}
    \bigsqcup_{m=0}^{|\alpha|-1}\bigsqcup_{k\neq \alpha_m} C_{\delta^m(\alpha)k} \cap K(\gamma,j) = K(\gamma,j),
\end{equation*}
because $K(\gamma,j) \subseteq C_{\delta^{|\alpha'|}(\alpha)\gamma_{|\alpha'|}}$. Also, we have
\begin{equation*}
    \bigsqcup_{m=0}^{|\alpha|-1}\bigsqcup_{k\neq \alpha_m} C_{\delta^m(\alpha)k} \cap F_\gamma = F_\gamma^{\alpha'\gamma_{|\alpha'|}},
\end{equation*}
because
\begin{equation*}
    C_{\delta^m(\alpha)k} \cap F_\gamma = \begin{cases}
        F_\gamma^{\alpha'\gamma_{|\alpha'|}}, \quad m = |\alpha'| \text{ and } k = \gamma_{|\alpha'|},\\
        \emptyset, \quad \text{otherwise}.
    \end{cases}
\end{equation*}
In addition,
\begin{equation*}
    \bigsqcup_{m=0}^{|\alpha|-1}\bigsqcup_{k\neq \alpha_m} \bigsqcup_{\substack{p: A(j,p)=0}} C_{\delta^m(\alpha)k} \cap C_{\gamma p} = \bigsqcup_{\substack{p: A(j,p)=0}} C_{\gamma p},
\end{equation*}
since $C_{\gamma p} \subseteq C_{\delta^{|\alpha'|}(\alpha)\gamma_{|\alpha'|}}$ for every $p$.

\end{document}